\documentclass[final,oneside,notitlepage,10pt]{article}

\usepackage{amsthm}
\usepackage{thmtools}
\usepackage{graphicx}
\usepackage[margin=2cm,font=normalsize,labelfont=bf]{caption}
\usepackage{verbatim}
\usepackage[shortlabels]{enumitem}
\usepackage{fancyhdr}
\usepackage[]{geometry}
\usepackage{xstring}
\usepackage{needspace}
\usepackage{color}
\usepackage{amstext}
\usepackage{nameref}
\usepackage[hypertexnames=false]{hyperref}
\usepackage{mathdots}
\usepackage{soul}
\usepackage{chngcntr}
\usepackage[section]{placeins}

\makeatletter
\newcounter{denseversion}
\newcounter{comments}
\newcounter{authorcounter}
\newcounter{adresscounter}

\def\title#1{\gdef\@title{#1}}
\def\@title{}
\def\subtitle#1{\gdef\@subtitle{#1}}
\def\@subtitle{}

\def\authortagsused{0}
\def\adresstag#1{\if!#1!\else$^{\;#1\;}$\fi}
\def\@authorsep#1{
  \ifnum\value{authorcounter}=#1 and \else\unskip, \fi
}

\renewcommand{\author}[2][]{
  \stepcounter{authorcounter}
  \if!#1!\else\gdef\authortagsused{1}\fi
  \ifnum\value{authorcounter}=1
    \def\@authorstringa{#2\adresstag{#1}}
    \def\@authorstringb{#2}
    \def\@authorstringc{#2\adresstag{#1}}
  \else
    \ifnum\value{authorcounter}=2
      \g@addto@macro\@authorstringa{\@authorsep{2}#2\adresstag{#1}}
      \g@addto@macro\@authorstringb{\@authorsep{2}#2}
      \g@addto@macro\@authorstringc{\\#2\adresstag{#1}}
    \else
      \g@addto@macro\@authorstringa{\@authorsep{3}#2\adresstag{#1}}
      \g@addto@macro\@authorstringb{\@authorsep{3}#2}
      \g@addto@macro\@authorstringc{\\#2\adresstag{#1}}
    \fi
  \fi}
\def\@author{\ifnum\value{denseversion}=0\@authorstringa\else\@authorstringb\fi}

\def\@adressstringa{}
\def\@adressstringb{}
\newcommand{\adress}[2][]{
  \stepcounter{adresscounter}
  \ifnum\value{adresscounter}=1
    \g@addto@macro\@adressstringa{\ifnum\authortagsused=0\def\br{\\}\else\def\br{, }\fi\adresstag{#1}#2}
    \g@addto@macro\@adressstringb{\def\br{\\}\adresstag{#1}\parbox[t]{14cm}{#2}}
  \else
    \g@addto@macro\@adressstringa{\\[\bigskipamount]\adresstag{#1}#2}
    \g@addto@macro\@adressstringb{\\[\medskipamount]\adresstag{#1}\parbox[t]{14cm}{#2}}
  \fi}

\def\preprint#1{\gdef\@preprint{#1}}
\def\@preprint{}
\def\keywords#1{\gdef\@keywords{#1}}
\def\@keywords{}
\def\msc#1{\gdef\@msc{#1}}
\def\@msc{}
\def\email#1{
   \gdef\@email{#1}
   \g@addto@macro\@authorstringc{ {\it (#1)}}}
\def\@email{}
\def\dedication#1{\gdef\@dedication{#1}}
\def\@dedication{}

\def\mybaselinestretch#1{
  \gdef\@mybaselinestretch{#1}
  \renewcommand{\baselinestretch}{\@mybaselinestretch}}

\def\myparskip#1{
  \gdef\@myparskip{#1}
  \setlength{\parskip}{\@myparskip}}

\mybaselinestretch{1.2}
\myparskip{1.5ex}

\binoppenalty=\maxdimen
\relpenalty=\maxdimen

\normalfont
\normalsize
\newlength{\@listleftmargin}
\def\setenumstandard{
  \setlist{leftmargin=\@listleftmargin,itemsep=0pt,topsep=0pt,partopsep=0pt,parsep=\@myparskip}
  \setlist[enumerate]{align=left,labelsep=*,leftmargin=\@listleftmargin,itemsep=0pt,topsep=0pt,partopsep=0pt,parsep=\@myparskip}
}

\def\denseversion{
  \setcounter{denseversion}{1}
  \newgeometry{left=3cm,right=3cm,top=3cm}
  \mybaselinestretch{1.1}
  \myparskip{0.8ex}
  \normalfont
  \def\possiblelinebreak{}
  \fancyfoot[C]{\itshape{--$\,\,$\thepage$\,\,$--}}}

\def\possiblelinebreak{\\}

\renewcommand{\emph}[1]{\def\reserved@a{it}\ifx\f@shape\reserved@a\ul{#1}\else\textit{#1}\fi}

\def\setcrefnames{}
\newcommand{\mytableofcontents}{
   \ifnum\value{denseversion}=0
     \tableofcontents
     \setcrefnames 
   \else
     \renewcommand{\baselinestretch}{1.1}
     \setlength{\parskip}{0ex}
     \normalfont
     \begingroup
     \def\addvspace##1{\vskip0.4em}
     \tableofcontents
     \setcrefnames 
     \endgroup
     \renewcommand{\baselinestretch}{\@mybaselinestretch}
     \setlength{\parskip}{\@myparskip}
     \normalfont
   \fi}

\newlength{\zeilenlaenge}
\def\putindent#1{
  \settowidth{\zeilenlaenge}{#1}
  \ifnum\zeilenlaenge>\textwidth
    #1
  \else
    \noindent #1
  \fi
}

\hypersetup{
    linktocpage = true,
    colorlinks,
    linkcolor=[RGB]{0,0,120},
    citecolor=[RGB]{0,0,120},
    urlcolor=[RGB]{0,0,120}
}

\def\pdfdaten{
  \hypersetup{
    pdftitle = {\@title},
    pdfauthor = {\@author},
    pdfkeywords = {\@keywords},    
    bookmarksopen = true,
    bookmarksopenlevel = 1
  }}  

\def\showkeywords{\begin{flushleft}\footnotesize\textbf{Keywords}: \@keywords\end{flushleft}}
\def\showmsc{\begin{flushleft}\footnotesize\textbf{MSC 2010}: \@msc\end{flushleft}}

\def\tocsection#1{\section*{#1}\addcontentsline{toc}{section}{#1}}

\def\mytitle{}
\def\zmptitle{
  \begin{tabular}{cc}
    \begin{minipage}[c]{0.4\textwidth}
      \begin{flushleft}
        \includegraphics[width=110pt]{../../tex/zmp}
      \end{flushleft}  
    \end{minipage}&
    \begin{minipage}[c]{0.55\textwidth}
      \begin{flushright}
      {\small\sf\@preprint}
      \end{flushright}
    \end{minipage}
  \end{tabular}
  \vskip 2cm}

\def\maketitle{
  \pdfdaten
  \noindent
  \mytitle
  \begin{center}
    \LARGE\@title\\
    \if!\@subtitle!\else\smallskip\LARGE\@subtitle\\\fi
    \bigskip
    \if!\@author!\else\bigskip\large\@author\\\fi
    \ifnum\value{denseversion}=0
      \if!\@adressstringa!\else\bigskip\normalsize\@adressstringa\\\fi
      \if!\@email!\else\ifnum\value{authorcounter}=1\bigskip\normalsize\textit{\@email}\\\else\fi\fi
    \else
    \fi
    \if!\@dedication!\else\bigskip\normalsize{\@dedication}\\\fi
  \end{center}
  \ifnum\value{denseversion}=0\vskip 1.5cm\else\vskip0.5cm\fi}

\def\kobib#1{
  \begin{raggedright}
  \ifnum\value{denseversion}=0\else\small\fi
  \Oldbibliography{#1/kobib}
  \bibliographystyle{#1/kobib}
  \end{raggedright}
  \ifnum\value{denseversion}=0\else
      \noindent
      \if!\@authorstringc!\else
        \ifnum\authortagsused=0\ifnum\value{authorcounter}>1\normalsize\@authorstringc\\[\medskipamount]\else\fi\else\normalsize\@authorstringc\\[\medskipamount]\fi
      \fi
      \if!\@adressstringb!\else\normalsize\@adressstringb\\{}\fi
      \ifnum\authortagsused=0
        \ifnum\value{authorcounter}=1
          \if!\@email!\else\linebreak\normalsize\textit{\@email}\\{}\fi
        \else
        \fi
      \else
      \fi
  \fi}

\let\Oldbibliography\bibliography
\def\bibliography#1{
  \begin{raggedright}
  \ifnum\value{denseversion}=0\else\small\fi
  \Oldbibliography{#1}
  \end{raggedright}
  \ifnum\value{denseversion}=0\else
      \medskip
      \noindent
      \if!\@authorstringc!\else
        \ifnum\authortagsused=0\ifnum\value{authorcounter}>1\normalsize\@authorstringc\\[\medskipamount]\else\fi\else\normalsize\@authorstringc\\[\medskipamount]\fi
      \fi
      \if!\@adressstringb!\else\normalsize\@adressstringb\\{}\fi
      \ifnum\authortagsused=0
        \ifnum\value{authorcounter}=1
          \if!\@email!\else\linebreak\normalsize\textit{\@email}\\{}\fi
        \else
        \fi
      \else
      \fi
  \fi
}

\newenvironment{commentfigure}{}
\newenvironment{sidewayscommentfigure}{\begin{minipage}}{\end{minipage}}
\newenvironment{displaycomment}{\begin{list}{}{\rightmargin=1cm\leftmargin=1cm}\item\sf\begin{small}}{\end{small}\end{list}}

\def\tocmark#1{}

\def\draftstamp#1{
  \def\tocmark##1{
    \ifnum\c@secnumdepth=0\section{##1}\fi
    \ifnum\c@secnumdepth=1\subsection{##1}\fi
    \ifnum\c@secnumdepth=2\subsubsection{##1}\fi
    \ifnum\c@secnumdepth=3\subsubsection{##1}\fi
  }
  \ifnum\value{comments}=0
    \gdef\@draft{DRAFT - Edited on \today\ by #1 - Comments are not displayed}
  \else
    \gdef\@draft{DRAFT - Edited on \today\ by #1 - Comments are displayed}
  \fi
  \fancyhead[C]{\footnotesize\tt\textcolor{red}{\@draft}}}

\def\skript{
  \renewenvironment{displaycomment}{}{}
  \ifnum\value{comments}=0
    \renewenvironment{example*}{\comment}{\endcomment}
    \renewenvironment{remark*}{\comment}{\endcomment}
  \else\fi
  \parindent=0mm	
}

\newcounter{sectioning}
\def\tsection#1{\ifnum\value{sectioning}=0\section{#1}\fi}
\def\lsection#1{
  \ifnum\value{sectioning}=1
    \clearpage
    \def\thesection{\lektionname~\arabic{section}:}
    \section{#1}
    \def\thesection{\arabic{section}}
  \fi}
\def\tsubsection#1{\ifnum\value{sectioning}=0\subsection{#1}\fi}
\def\lsubsection#1{\ifnum\value{sectioning}=1\subsection{#1}\fi}
\def\tsubsubsection#1{\ifnum\value{sectioning}=0\subsubsection{#1}\fi}
\def\lsubsubsection#1{\ifnum\value{sectioning}=1\subsubsection{#1}\fi}

\mybaselinestretch{}
\setenumstandard

\makeatother

\input{kobib}

\usepackage{amssymb}
\usepackage{amsmath}
\usepackage{amstext}
\usepackage{mathrsfs}
\usepackage{euscript}
\usepackage{color}
\usepackage[all,cmtip,color,arrow]{xy}
\usepackage{xstring}
\usepackage{slashed}
\usepackage{tensor}
\usepackage{tikz}

\entrymodifiers={+!!<0pt,\fontdimen22\textfont2>}

\def\Z {\mathbb{Z}}

\def\R {\mathbb{R}}
\def\C {\mathbb{C}}

\def\id{\mathrm{id}}

\def\hc#1{\mathrm{h}_{#1}}
\def\h {\mathrm{H}}

\def\subset{\subseteq}

\def\sep{\;|\;}
\def\rk{\mathrm{rk}}
\def\maps{\colon}
\def\df{:=}

\renewcommand{\varepsilon}{\epsilon}

\makeatletter
\renewenvironment{proof}[1][\nameProof]
  {\par\pushQED{\qed}%
   \normalfont \topsep6\p@\@plus6\p@\relax
   \trivlist
   \item[\hskip\labelsep
         \itshape
         #1\@addpunct{.}]
  \leavevmode}
  {\popQED\endtrivlist\@endpefalse}
\makeatother

\def\notebox#1#2{\begin{minipage}[b]{#1}\sloppy\renewcommand{\baselinestretch}{0.8}\footnotesize \begin{center}#2\end{center}\end{minipage}}

\def\mquad{\hspace{-2em}}
\def\mqquad{\hspace{-4em}}
\def\mqqquad{\hspace{-6em}}
\def\mqqqquad{\hspace{-8em}}

\newcommand{\arr}[1][r]{\ar@<0.7ex>[#1]\ar@<-0.7ex>[#1]}
\newcommand{\arrr}[1][r]{\ar@<1.4ex>[#1]\ar[#1]\ar@<-1.4ex>[#1]}
\newcommand{\arrrr}[1][r]{\ar@<2.1ex>[#1]\ar@<-2.1ex>[#1]\ar@<0.7ex>[#1]\ar@<-0.7ex>[#1]}

\def\stackref#1#2{\stackrel{\text{\ref{#1}}}{#2}}
\def\eqref#1{\stackref{#1}{=}}

\newlength{\myeqt} 
\newlength{\myeqs} 
\newlength{\myeqm} 
\newlength{\myeqn} 
\newcommand\eqtext[2][\myeqn]{\symtext[#1]{#2}{=}}
\newcommand\symtext[3][\myeqn]{
  \settowidth{\myeqt}{#2}
  \settowidth{\myeqs}{$#3$}
  \addtolength{\myeqs}{\the\myeqm}
  \ifdim\myeqt>\myeqs
    \stackrel{\hspace{-#1}\notebox{#1}{\medskip #2 \\ $\downarrow$\smallskip}\hspace{-#1}}{#3}
  \else
    \stackrel{\text{#2}}{#3}
  \fi}
\newcommand\eqcref[2][\myeqn]{\symcref[#1]{#2}{=}}

\newcommand\symcref[3][\myeqn]{\symtext[#1]{\cref{#2}}{#3}}

\def\brackets#1{\IfStrEq{#1}{-}{}{(#1)}}
\def\subindex#1{\IfStrEq{#1}{-}{}{_{#1}}}

\newcommand{\alxydim}[2]{\begin{aligned}\xymatrix#1{#2}\end{aligned}}

\makeatletter
\def\bigset#1#2{\left\lbrace\;\begin{minipage}[c]{#1}\begin{center}#2\end{center}\end{minipage}\;\right\rbrace}

\makeatother

\newlength{\myl}

\def\ddt#1#2#3{\left.\frac{\mathrm{d}^{\IfStrEq{#1}{1}{}{#1}}}{\mathrm{d}#2}\IfStrEq{#2}{#3}{\right.}{\right|_{#3}}}

\newcommand{\ueins}{{\mathrm{U}}(1)}

\def\fun{\mathcal{F}un}

\def\hom{\mathcal{H}\!om}

\def\trivtech{\mathcal{T}\hspace{-0.34em}r\hspace{-0.06em}i\hspace{-0.07em}v}
\def\triv#1{\trivtech\brackets{#1}}

\def\vect#1{\mathcal{V}\hspace{-0.1em}ect\subindex{#1}}

\def\ev{\mathrm{ev}}

\def\pr{{\mathrm{pr}}}

\newlength{\widthtmp}
\def\length#1{\settowidth{\widthtmp}{#1}\the\widthtmp}

\def\ttimes#1#2{\hspace{-0.15em}\tensor[_{#1}]{\times}{_{#2}}}

\def\buntech#1#2{\mathcal{B}\hspace{-0.01em}un_{\hspace{0.05em}#1}^{#2}}

\def\buncon#1#2{\buntech{#1}{\nabla}\hspace{-0.05em}\brackets{#2}}

\def\ubuncon#1{\buncon\relax{#1}}

\def\grbtech#1{\mathcal{G}\hspace{-0.06em}r\hspace{-0.06em}b_{\hspace{-0.07em}{#1}}}

\def\grbcon#1#2{\grbtech{#1}^{\nabla\!}\brackets{#2}}

\def\ugrbcon#1{\grbcon\relax{#1}}

\usepackage[latin1]{inputenc}
\usepackage[english]{babel}
\usepackage[english]{cleveref}

\def\refname{References}

\def\quot#1{``#1''}

\def\quand{\quad\text{ and }\quad}
\def\quomma{\quad\text{, }\quad}

\def\quith{\quad\text{ with }\quad}

\def\nameTheorem{Theorem}
\def\nameTheorems{Theorems}
\def\nameDefinition{Definition}
\def\nameDefinitions{Definitions}
\def\nameCorollary{Corollary}
\def\nameCorollaries{Corollaries}
\def\nameProposition{Proposition}
\def\namePropositions{Propositions}
\def\nameConstruction{Construction}
\def\nameConstructions{Constructions}
\def\nameLemma{Lemma}
\def\nameLemmas{Lemmas}
\def\nameRemark{Remark}
\def\nameRemarks{Remarks}
\def\nameProblem{Problem}
\def\nameProblems{Problems}
\def\nameExample{Example}
\def\nameExamples{Examples}
\def\nameExercise{Exercise}
\def\nameExercises{Exercises}

\def\nameProof{Proof}
\def\nameEquation{Eq.}
\def\nameEquations{Eqs.}
\def\nameSection{Section}
\def\nameSections{Sections}
\def\nameAppendix{Appendix}
\def\nameAppendixs{Appendices}
\def\nameFigure{Figure}
\def\nameFigures{Figures}

\input{theorems}

\title{Transgression of D-branes}

\author[a]{Severin Bunk}
\email{severin.bunk@uni-hamburg.de}

\adress[a]{
  Fachbereich Mathematik\\
  Bereich für Algebra und Zahlentheorie\\
  Bundesstraße 55\\
  D-20146 Hamburg}

\author[b]{Konrad Waldorf}
\email{konrad.waldorf@uni-greifswald.de}

\adress[b]{
  Universität Greifswald\\
  Institut für Mathematik und Informatik\\
  Walther-Rathenau-Str. 47\\
  D-17487 Greifswald}

\msc{}
\dedication{}
\pdfdaten

\denseversion
\mybaselinestretch{1}
\myparskip{0.5ex}

\allowdisplaybreaks

\def\HVect{\mathcal{H}\vect}

\def\HVbuncon#1{\mathcal{HV}\ubuncon#1}

\def\2hom{2\text{-}\hom}
\def\Des{\Delta}
\def\inf#1{\EuScript{#1}}
\def\pcomp{\star}
\def\lsg{\smash{\mathrm{LBG}}}
\def\tsg{\mathrm{TBG}}
\def\cp{\mathsf{c}}
\def\kfrob{\smash{\mathrm{K}\text{-}\mathrm{Frob}^{(I)}}}
\def\rpkfrob{\smash{\mathrm{RPK}\text{-}\mathrm{Frob}^{(I)}}}
\def\tqft{\mathcal{F}}
\def\tqftlift{\widetilde{\tqft}}
\def\rev{\mathsf{rev}}
\def\rot{\mathsf{rot}}

\def\nameEquation{\unskip}
\def\nameEquations{\unskip}

\usepackage{amstext}
\usepackage{amssymb}
\usepackage{tikz}
\usepackage{sidecap}

\begin{document}


{
\raggedleft
\vspace*{-5em}
\sf\small
Hamburger Beiträge zur Mathematik 746
\\
ZMP-HH/18-16
\\
}
\vspace{3em}

\maketitle 

\begin{abstract}
Closed strings can be seen either as one-dimensional objects in a target space or as points in the free loop space. Correspondingly, a B-field can be seen either as a connection on a gerbe over the target space, or as  a connection on a line bundle over the loop space. Transgression establishes an equivalence between these two perspectives. Open strings  require D-branes: submanifolds equipped with  vector bundles twisted by the gerbe. In this paper we develop a loop space perspective on D-branes. It involves  bundles of simple Frobenius algebras over the branes, together with  bundles of bimodules over spaces of paths connecting  two branes. We prove that the classical and our new perspectives on D-branes are equivalent.  Further, we compare our loop space perspective to Moore-Segal/Lauda-Pfeiffer data for  open-closed 2-dimensional topological quantum field theories, and exhibit it as a smooth family of reflection-positive, colored knowledgable Frobenius algebras.
\end{abstract}

\setcounter{tocdepth}{2}
\mytableofcontents

\setsecnumdepth{1}

\section{Introduction}

\label{sec:intro}

The motivation for this article comes from the Lagrangian approach to  2-dimensional field theories,  where we consider smooth maps $\phi:\Sigma \to M$ from a Riemann surface $\Sigma$ to a smooth manifold $M$.
If the surface $\Sigma$ has a boundary, then we specify a family of submanifolds $Q_i\subset M$ (the \quot{D-branes}) and require that $\phi$ map each boundary component to one of these \cite{polchinski}. 

The usual sigma model action functional on the space of maps $\phi$ requires a metric $g$  and a hermitian line bundle gerbe $\mathcal{G}$ with connection over  $M$, the \quot{B-field} \cite{alvarez1,gawedzki3,murray}. In the case that $\Sigma$ has boundaries, the D-branes $Q_i$  have to be equipped with $\mathcal{G}|_{Q_i}$-twisted vector bundles $\mathcal{E}_i$ with connections, known as \quot{twisted Chan-Paton gauge fields} in string theory, \cite{kapustin1,gawedzki1,carey2}. The twisted vector bundles make up the famous relation between D-branes and twisted K-theory \cite{Freed1999,Freed}.

In Wess-Zumino-Witten (WZW) models, for instance, $M$ is a compact Lie group $G$ and the branes  $Q_i$ are conjugacy classes of $G$ \cite{Alekseev1999}. The metric is induced by the Killing form, and the bundle gerbe $\mathcal{G}$ represents the \quot{level} via its Dixmier-Douady class $[\mathcal{G}] \in \h^3(G,\Z)\cong \Z$. The bundle gerbe $\mathcal{G}$ and the twisted vector bundles can be constructed explicitly in a Lie-theoretic way, for all compact simple Lie groups \cite{gawedzki1,meinrenken1,gawedzki2,gawedzki4}.

Especially in the discussion of WZW models it is common to consider, instead of the bundle gerbe $\mathcal{G}$, a hermitian line bundle  with connection over the loop group $LG$. One advantage of this  perspective is, for example, that one can (at least at an informal level) study the geometric quantization of the model. In fact, every smooth manifold $M$ has  a corresponding \emph{transgression} map to its free loop space $LM=C^{\infty}(S^1,M)$ \cite{gawedzki3,brylinski1,waldorf5}
\begin{equation*}
\bigset{13em}{Hermitian line bundle gerbes with connection over $M$} 
\to
\bigset{12em}{Hermitian line bundles with connection over $LM$}\text{.}
\end{equation*}
In good cases, for instance in WZW models with simply-connected target groups, this map induces an equivalence of the corresponding categories, so that the loop space perspective is  equivalent to the original setting.

In cases of more general topology, transgression is neither injective nor surjective, but the equivalence can be re-established by changing the range of transgression \cite{waldorf10}. Two changes are necessary: the first change is to add \emph{fusion products}, a structure defined for a line bundle $\inf L$ over the loop space $LM$ \cite{stolz3,waldorf10}. It consists of fiber-wise isomorphisms
\begin{equation*}
\inf L|_{\gamma_1\cup\gamma_2} \otimes \inf L|_{\gamma_2\cup\gamma_3} \to \inf L|_{\gamma_1 \cup\gamma_3}\text{,}
\end{equation*}
where $(\gamma_1,\gamma_2,\gamma_3)$ is a triple of paths in $M$ with a common initial point and a common end-point, and $\gamma_i \cup \gamma_j$ denotes the loop that first follows $\gamma_i$ and then returns along  the reverse of $\gamma_j$.
The second change is to impose various conditions on the connection on $\inf L$, most importantly the condition that it is \emph{superficial}. This is a property that can be considered for connections on bundles over mapping spaces $C^{\infty}(K,M)$, where  $K$ is a compact manifold (here, $K=S^1$ or $K=[0,1]$ later). If $\Gamma:[0,1] \to C^{\infty}(K,M)$  is a smooth path in such  a mapping space, then there is a corresponding \quot{adjoint} map $\Gamma^{\vee}: [0,1] \times K \to M$, and superficiality requires certain conditions for the parallel transport along $\Gamma$ depending on the rank of the differential of $\Gamma^{\vee}$. We refer to \cite[Sec. 2]{waldorf10} or \cref{sec:superficial} for more details. Taking fusion and superficial connections into account, transgression provides a complete loop space perspective to connections on bundle gerbes (\quot{B-fields}) for arbitrary smooth manifolds $M$.

In this article we establish an analogous loop space perspective to D-branes and their twisted Chan-Paton gauge fields. To that end, we provide a transgression map for  $\mathcal{G}|_{Q_i}$-twisted vector bundles $\mathcal{E}_i$ supported on submanifolds $Q_i$ and prove that it establishes an equivalence between the appropriate categories. Due to the twist represented by the bundle gerbe $\mathcal{G}$, transgression of D-branes will be defined relative to the above-mentioned  transgression of bundle gerbes.

First results in this direction have been obtained by Brylinski \cite{brylinski1}, Gaw\c edzki-Reis \cite{gawedzki1} and Gaw\c edzki \cite{gawedzki4}.  
In the following, we say that a \emph{$\mathcal{G}$-brane} is a pair $(Q,\mathcal{E})$ of a submanifold $Q\subset M$ and a $\mathcal{G}|_{Q}$-twisted vector bundle $\mathcal{E}$ with connection. In \cite{gawedzki1,gawedzki4} it is described how to associate to a pair of two $\mathcal{G}$-branes $(Q,\mathcal{E})$ and $(Q',\mathcal{E}')$ a hermitian vector bundle $\inf R$ with connection over the space $P$ of paths connecting $Q$ with $Q'$. In \cref{sec:transbranes} we give an alternative, but equivalent description of this vector bundle in terms of a modern, bicategorical treatment of bundle gerbes. We also clarify some  regularity aspects of the vector bundle $\inf R$ that have not been treated so far.
More precisely, we work  in the convenient setting of diffeology, which provides a simple and yet powerful framework to study   geometry over spaces of loops and paths.
Furthermore, we exhibit in \cref{sec:trans:superficialconnection} a new property of the connection on $\inf R$, which is analogous to the superficiality for line bundles over loop spaces mentioned above.

While this construction of the vector bundle $\inf  R$ over $P$ is the basis of our transgression map for D-branes, the picture is incomplete in essentially two aspects. First, it seems that the vector bundle $\inf R$ has lost its relation to the twist $\mathcal{G}$, or rather its loop space analogue, the transgressed hermitian line bundle $\inf L$. In \cref{sec:fusionrep} we discover an isomorphism
\begin{equation}
\label{eq:intro:fusionrep}
\inf L|_{\gamma_1\cup\gamma_2} \otimes \inf R|_{\gamma_2} \to \inf R|_{\gamma_1}\text{,}
\end{equation} 
where $\gamma_1,\gamma_2$ are paths with a common initial point and a common end point. It satisfies the axiom of an action with respect to the fusion product on $\inf L$; therefore, we will call it the \emph{fusion representation} of $\inf L$ on $\inf R$. Second, we discover in \cref{sec:liftpathcommp} a homomorphism
\begin{equation}
\label{eq:intro:liftedpathcomp}
\inf R_{23}|_{\gamma_{23}} \otimes \inf R_{12}|_{\gamma_{12}} \to \inf R_{13}|_{\gamma_{23} \pcomp \gamma_{12}} 
\end{equation} 
relating the vector bundles $\inf R_{ij}$ formed from three $\mathcal{G}$-branes $(Q_1,\mathcal{E}_1)$, $(Q_2,\mathcal{E}_2)$ and $(Q_3,\mathcal{E}_3)$, which lifts the concatenation of a path $\gamma_{12}$ between $Q_1$ and $Q_2$ and a path $\gamma_{23}$ from $Q_2$ to $Q_3$ to  the path $\gamma_{23} \pcomp \gamma_{12}$ between $Q_1$ and $Q_3$. The homomorphism of \cref{eq:intro:liftedpathcomp} is accompanied by additional structure rendering it compatible with identity paths and path reversal, and satisfies several compatibility conditions. We call it the \emph{lifted path concatenation}.

We assemble the data obtained by transgression abstractly   into a category $\lsg(M,Q)$ of \emph{loop space brane geometry} (LBG), depending on a target space $M$ and a family $Q=\{Q_i\}_{i\in I}$ of D-branes. Roughly, the objects are tuples consisting of a line bundle $\inf L$ over $LM$ with superficial connection, a fusion product on $\inf L$, vector bundles $\inf R_{ij}$ over the spaces $P_{ij}$ of paths connecting $Q_i$ with $Q_j$, fusion representations for each of these vector bundles, and a lifted path concatenation relating them to each other. A precise definition is given in \cref{sec:loopspacegeometry}. On the other side, the basis of transgression was a bicategory $\tsg(M,Q)$ of \emph{target space brane geometry} (TBG), whose objects consist of a bundle gerbe $\mathcal{G}$ over $M$ with connection and  $\mathcal{G}|_{Q_i}$-twisted vector bundles $\mathcal{E}_i$ with connection. A precise definition of that bicategory and some background about bundle gerbes is given in \cref{sec:targetspacegeometry}.
In \cref{sec:nattrans} we prove that transgression furnishes a functor
\begin{equation*}
\mathscr{T}: \hc1\tsg(M,Q) \to \lsg(M,Q)\text{,}
\end{equation*} 
where $\hc 1$ denotes the 1-truncation of a bicategory (identify 2-isomorphic 1-morphisms).
In \cref{sec:regression} we construct in a  natural way a functor $\mathscr{R}$ in the opposite direction called \emph{regression}.
Our main result is the following theorem.
\begin{maintheorem}
\label{th:main}
Transgression and regression form an equivalence between target space brane geometry and loop space brane geometry,
\begin{equation*}
\hc 1\tsg(M,Q) \cong \lsg(M,Q)\text{.}
\end{equation*}
\end{maintheorem}

The proof of \cref{th:main} is carried out in \cref{sec:equivalence} by constructing explicitly natural isomorphisms $\mathscr{R} \circ \mathscr{T}\cong \id$ and $\mathscr{T}\circ \mathscr{R} \cong \id$. 
\cref{th:main} means that the Lagrangian approach to 2-dimensional topological field theories with D-branes can be pursued in an equivalent way either via TBG or via LBG.

An interesting aspect of the equivalence of \cref{th:main} is the presence of algebra bundles over the D-branes, arising in both perspectives. In TBG, algebra bundles $A_i$ arise as endomorphism bundles of the $\mathcal{G}|_{Q_i}$-twisted vector bundles \cite{Schweigert2014}. They are  well-known for their relation to twisted K-theory
\cite{Karoubia}. In LBG, algebra bundles $\inf A_i$ are obtained by restricting the vector bundles $\inf R_{ii}$ to constant paths in $Q_i$, and turning  these restrictions into algebra bundles using the lifted path concatenation of \cref{eq:intro:liftedpathcomp}, see \cref{sec:algebrabundle:lsg}. We prove the following result. 

\begin{maintheorem}
\label{th:main:2}
\begin{enumerate}[(a)]

\item 
\label{th:main:2:1}
The algebra bundles $\inf A_i$ arising from loop space brane geometry are bundles of  simple Frobenius algebras.

\item
\label{th:main:2:2}
Transgression and regression induce isomorphisms $A_i \cong \inf A_i$ between the algebra bundles that arise independently from target space and loop space brane geometry. 
\end{enumerate}
\end{maintheorem}  

At first sight, it seems that \cref{th:main:2:2*} implies \cref{th:main:2:1*} because  endomorphism algebras are automatically simple Frobenius algebras. However, we prove \cref{th:main:2:1*} loop-space-intrinsically without any reference to endomorphism bundles, see \cref{prop:frob}. In fact, we use \cref{th:main:2:1*} in the construction of the regression functor $\mathscr{R}$ so that it really is an independent result. Part \cref{th:main:2:2*} is proved in \cref{sec:coincidence}.

In \cref{sec:bimodulebundles} we also discover  the following algebraic feature of the vector bundles $\inf R_{ij}$ over the spaces $P_{ij}$ of paths connecting $Q_i$ with $Q_j$. 

\begin{maintheorem}
\label{th:main:3}
\begin{enumerate}[(a)]

\item 
\label{th:main:3:1}
The vector bundles $\inf R_{ij}$ are bundles of bimodules over the algebra bundles $\ev_0^{*}\inf A_i$ and $\ev_1^{*}\inf A_j$, where $\ev_0: P_{ij} \to Q_i$ and $\ev_1: P_{ij} \to Q_j$ are the end point evaluations. 

\item
\label{th:main:3:2}
Lifted path concatenation induces connection-preserving  isomorphisms
$\inf R_{jk} \otimes_{\inf A_j} \inf R_{ij} \cong \inf R_{ik}$. 

\end{enumerate}
\end{maintheorem}

Our motivation for exhibiting these algebraic features is the \emph{functorial perspective} to topological field theories \cite{segal1,stolz1,lurie1}, where  field theories are regarded as  functors from certain categories of bordisms to certain algebraic categories. The bordisms are equipped with smooth maps to $M$ in order to incorporate a target space. Bunke-Turner-Willerton have shown that 2-dimensional topological  functorial field theories are equivalent to bundle gerbes with connection over $M$ \cite{Bunke2004}.
 In an upcoming article \cite{Bunk2018} we will extend this result to an equivalence between TBG and functorial field  theories with D-branes, whose values will be given by the algebraic data provided in  \cref{th:main:3}.

In the setting of functorial field theories, \emph{quantum} theories are defined on bordisms without a map to a target space. Equivalently, a quantum field theory is one with $M=\{\ast\}$. Quantization is supposed to correspond to the pushforward  to a point in a suitable generalized cohomology theory \cite{stolz1}. For quantum field theories, it is common to describe functorial field theories by algebraic structure obtained from a presentation of the bordism category in terms of generators and relations. For the closed sector and two dimensions, this leads to commutative Frobenius algebras \cite{Abrams1996}. 
For quantum theories with D-branes, the algebraic structure has been determined  by 
Lazaroiu, Moore-Segal, and Lauda-Pfeiffer \cite{Lazaroiu2001,Moore2006,Lauda2008}, of which the last reference termed it a \emph{colored knowledgable Frobenius algebra} (the \quot{colors} are the brane indices $i\in I$). 
One can now perform an interesting consistency check between classical and quantum field theories: the restriction of a classical field theory to a point; here, this means putting $M=Q_i=\{\ast\}$ for all brane indices $i$. Of course, this will not give \quot{the} quantization of the original theory, but it does give \quot{a} quantum theory and thus should fit into that framework. We show that our LBG passes this consistency check in the sense that there is a naturally defined faithful functor
\begin{equation*}
\mathcal{F}:\lsg(\ast,\{\ast\}_{i\in I}) \to \kfrob
\end{equation*}
to the category of $I$-colored knowledgable Frobenius algebras. Our final result determines the image of this functor. We show that all colored knowledgable Frobenius algebras in its image are equipped with an additional structure that we call a \emph{positive reflection structure}.
In our upcoming article \cite{Bunk2018} we will show that it is equivalent -- under the correspondence to topological quantum field theories -- to a positive reflection structure in the sense of functorial field theories \cite{Freed2016}. We prove the following.

\begin{maintheorem}
\label{th:main:4}
The functor $\mathcal{F}$ induces an equivalence between loop space brane geometries of a point and reflection-positive, $I$-colored knowledgable Frobenius algebras whose underlying Frobenius algebra is $\C$. 
\end{maintheorem}

Thus, LBG is a target space family of reflection-positive, colored knowledgable structures on the Frobenius algebra $\C$.   All details and the proof of  \cref{th:main:4} are given out in \cref{sec:mooresegal}.

This paper is organized in the following way. In  \cref{sec:branegeometry} we give precise definitions of the categories $\tsg(M,Q)$ and $\lsg(M,Q)$, and we recall some relevant aspects about bundle gerbes. In \cref{sec:alglsg} we derive the algebraic structures induced by LBG as stated in \cref{th:main:2:1}  and \cref{th:main:3} and discuss the reduction to the work of Moore-Segal \cite{Moore2006} and Lauda-Pfeiffer \cite{Lauda2008}. In \cref{sec:transgression,sec:regression} we construct the transgression and regression functors $\mathscr{T}$ and $\mathscr{R}$, respectively. In \cref{sec:equivalence} we prove \cref{th:main} and \cref{th:main:2:2}. We conclude with an appendix providing technical background material, in particular about diffeological vector bundles, superficial connections on path spaces, and  bundles of algebras and  bimodules. For the benefit of the reader we include a small table of notation on page \pageref{sec:notation}.

\paragraph{Acknowledgements.} KW was supported by the German Research Foundation under project code WA 3300/1-1. SB was partly supported by the RTG 1670 \quot{Mathematics inspired by string theory and quantum field theory}.
We would like to thank Ulrich Bunke, Lukas Müller, Ingo Runkel, Christoph Schweigert, Richard Szabo, Peter Teichner and Stephan Stolz
for valuable conversations. 

\setsecnumdepth{2}

\section{Brane geometries}

\label{sec:branegeometry}

 In the following, a
 \emph{target space} will be a pair $(M,Q)$ of a connected smooth manifold $M$  and  a family $Q=\{Q_i\}_{i\in I}$ of submanifolds $Q_i \subset M$. We  work over a fixed target space; but all definitions and constructions will be natural under maps between target spaces, i.e. smooth maps $f:M \to M'$ such that $f(Q_i)\subset Q_i'$.

\subsection{Target space brane geometry}

\label{sec:targetspacegeometry}

 As motivated in \cref{sec:intro}, an object in the bicategory $\tsg(M,Q)$ is a pair $(\mathcal{G},\mathcal{E})$, consisting of a hermitian line bundle gerbe $\mathcal{G}$ with connection over $M$ and a family $\mathcal{E}=\{\mathcal{E}_i\}_{i\in I}$ of $\mathcal{G}|_{Q_i}$-twisted vector bundles $\mathcal{E}_i$ with connections.

We recall some minimal facts to explain these notions. Hermitian line bundle gerbes with connection over $M$ form a bicategory $\ugrbcon M$ that can be defined in the following very elegant way \cite{nikolaus2}. Let $\HVbuncon M$ denote the symmetric monoidal category of hermitian vector bundles equipped with (unitary) connections, with the morphisms connection-preserving, unitary bundle morphisms. 
The assignment $M \mapsto \HVbuncon M$ forms a sheaf of symmetric monoidal categories over the site of smooth manifolds (with surjective submersions as coverings). 
We consider the following bicategory $\triv-\ugrbcon M$:
\begin{enumerate}[(a)]

\item 
Its objects are 2-forms $B \in \Omega^2(M)$, which we will denote by $\mathcal{I}_B$ in this context. 

\item
The Hom-category $\hom(\mathcal{I}_{B_1},\mathcal{I}_{B_2})$ is by definition the full subcategory $\HVbuncon M^{B_2-B_1}$  of $\HVbuncon M$  on the objects $(E,\nabla)$ with
\begin{equation}
\label{eq:curvcond}
\frac{1}{\mathrm{rk}(E)}\,\mathrm{tr}(\mathrm{curv}(\nabla))=B_2-B_1\text{.}
\end{equation}

\item
Composition is the tensor product in $\HVbuncon M$. 

\end{enumerate}
The assignment $M \mapsto \triv-\ugrbcon M$ is a presheaf of bicategories. The sheaf $\ugrbcon -$ is by definition its sheafification, i.e.
\begin{equation*}
\ugrbcon - := (\triv-\ugrbcon -)^{+}\text{.}
\end{equation*}
The objects $\mathcal{I}_B \in \triv-\ugrbcon M \subset \ugrbcon M$ are now called \emph{trivial} bundle gerbes. 
Spelling this out results exactly in the usual notion of bundle gerbes \cite{murray,murray2,stevenson1} with the bicategorical structure described in \cite{waldorf1}. 
In particular,  1-morphisms in $\ugrbcon M$ have a well-defined rank, and  a 1-morphism is a 1-\emph{iso}morphism (\quot{stable isomorphism}) if and only if it has rank one \cite[Prop. 3]{waldorf1}.

We will use a slightly generalized version of this bicategory, obtained by performing the sheafification in the site of diffeological spaces (with respect to the Grothendieck topology  subductions), and then restricting  again to smooth manifolds. The occurring vector bundles have to be treated as vector bundles over diffeological spaces (see \cref{sec:diffvectorbundle}). This generalization results in an equivalent bicategory, see \cite[Section 3.1]{waldorf10} for more information. It
is necessary because the regression functor defined in \cref{sec:regression} takes values in the generalized version of bundle gerbes.

If $\mathcal{G}$ is a bundle gerbe with connection over $M$, then a \emph{$\mathcal{G}$-twisted vector bundle with connection} is a 2-form $\omega\in \Omega^2(M)$ together with a 1-morphism $\mathcal{E}: \mathcal{G} \to \mathcal{I}_{\omega}$ in  $\ugrbcon M$. This gives precisely the usual notion of a \quot{bundle gerbe module with connection} introduced in \cite{carey2}, see \cite{waldorf1}. A twisted vector bundle of rank one is called a \emph{trivialization} of $\mathcal{G}$.

We return to the definition of the bicategory $\tsg(M,Q)$  
of target space brane geometry, whose objects are pairs $(\mathcal{G},\mathcal{E})$ of a hermitian line bundle gerbe $\mathcal{G}$ with connection over $M$ and a family $\mathcal{E}=\{\mathcal{E}_i\}_{i\in I}$ of $\mathcal{G}|_{Q_i}$-twisted vector bundles $\mathcal{E}_i$ with connections. The 1-morphisms between $(\mathcal{G},\mathcal{E})$ and $(\mathcal{G}',\mathcal{E}')$ are pairs $(\mathcal{A},\psi)$, consisting of a 1-isomorphism $\mathcal{A}: \mathcal{G} \to \mathcal{G}'$ in $\ugrbcon M$ and a family $\psi=\{\psi_{i}\}_{i\in I}$ of 2-morphisms $\psi_i: \mathcal{E}_i \Rightarrow \mathcal{E}_i' \circ \mathcal{A}|_{Q_i}$ in $\ugrbcon {Q_i}$. The 2-morphisms between $(\mathcal{A},\psi)$ and $(\mathcal{A}',\psi')$ are 2-morphisms $\varphi: \mathcal{A} \Rightarrow \mathcal{A}'$ in $\ugrbcon M$ such that the diagram
\begin{equation*}
\small
\alxydim{@R=3em@C=1em}{\quad\mathcal{E}_i \ar@{=>}[rr]^-{\psi_i} \ar@{=>}[dr]_-{\psi_i'} && \mathcal{E}_i' \circ \mathcal{A}|_{Q_i}\mquad \ar@{=>}[dl]^{\id_{\mathcal{E}_i'} \circ \varphi|_{Q_i}} \\ & \mathcal{E}_i' \circ \mathcal{A}'|_{Q_i} &}
\end{equation*}
of 2-morphisms in $\ugrbcon {Q_i}$ is commutative for all $i\in I$.

In the remainder of this subsection we discuss an operation on 1-morphisms between bundle gerbes that will be used frequently  throughout this article. 
Namely, for bundle gerbes $\mathcal{G}$, $\mathcal{H} \in \ugrbcon M$ there is a functor
\begin{equation}
\label{eq:des}
\Des: \hom(\mathcal{G},\mathcal{I}_{\omega_2}) \times \hom(\mathcal{G},\mathcal{I}_{\omega_1})^{op} \to \HVbuncon M^{\omega_2-\omega_1}\text{,}
\end{equation}
that can be seen as   an enriched version of an internal hom, see \cite{Bunk2016} and \cite[Thm. 3.6.3]{Bunk2017}, and \cite{Bunk2017a}  for a less technical overview.
It is defined in the following way. We first consider the functor $()^{*}: \hom(\mathcal{G},\mathcal{H})^{op} \to \hom(\mathcal{H},\mathcal{G})$, which  on the level of the presheaf $\triv-\ugrbcon M$ is just the dualization $E \mapsto E^{*}$. Since this is a morphism of presheaves, it survives the sheafification and induces the claimed functor $()^{*}$.  Now,  $\Des$ is defined by
\begin{multline*}
\alxydim{@C=4em}{\hom(\mathcal{G},\mathcal{I}_{\omega_2}) \times \hom(\mathcal{G},\mathcal{I}_{\omega_1})^{op}\ar[r]^{\id \times ()^{*}} & \hom(\mathcal{G},\mathcal{I}_{\omega_2}) \times \hom(\mathcal{I}_{\omega_1},\mathcal{G})}
\\
\alxydim{@C=4em}{\ar[r]^-{\circ} & \hom(\mathcal{I}_{\omega_1},\mathcal{I}_{\omega_2})\cong \HVbuncon M^{\omega_2-\omega_1}\text{.}}
\end{multline*}

Since the functor $\Des$ is quite important in this article, we spell out its definition in terms of the explicit definition of bundle gerbes. Suppose  $\mathcal{G}$ consists of a surjective submersion $Y \to M$, a hermitian line bundle $L$ with connection over $Y^{[2]}:= Y \times_M Y$, and a bundle gerbe product $\mu$ over $Y^{[3]}$. Twisted vector bundles $\mathcal{E}_i: \mathcal{G} \to \mathcal{I}_{\omega_i}$ ($i=1,2$) consist of hermitian vector bundles $E_i$ with connection over $Y$, and of connection-preserving, unitary bundle isomorphisms
\begin{equation*}
\zeta_i: L \otimes \pr_2^{*}E_i \to \pr_1^{*}E_i
\end{equation*}  
over $Y^{[2]}$, satisfying a compatibility condition with the bundle gerbe product $\mu$ over $Y^{[3]}$, see \cite{waldorf1}. 
In order to define the  vector bundle $\Des(\mathcal{E}_2,\mathcal{E}_1)$, we consider the vector bundle $E_1^{*} \otimes E_2$ over $Y$, and over $Y^{[2]}$ the bundle isomorphism $\tilde\zeta$ defined by 
\begin{equation}
\label{eq:tildeepsilon}
\alxydim{@C=3.8em}{\pr_2^{*}(E_1^{*} \otimes E_2) \cong \pr_2^{*}E_1^{*} \otimes L^{*} \otimes L \otimes \pr_2^{*}E_2 \ar[r]^-{\zeta_1^{tr-1} \otimes \zeta_2} & \pr_1^{*}E_1^{*} \otimes \pr_1^{*}E_2  \cong \pr_1^{*}(E_1^{*} \otimes E_2)\text{.}}
\end{equation} 
Here we have used the coevaluation isomorphisms between a complex line bundle and its dual, and $(..)^{tr}$ denotes the transpose of a linear map. 
The conditions for $\zeta_1$ and $\zeta_2$ imply a cocycle condition for $\tilde\zeta$ over $Y^{[3]}$. Since $\tilde \zeta$ is connection-preserving and unitary, $E_1^{*} \otimes E_2$ descends to a hermitian vector bundle $\Des(\mathcal{E}_2,\mathcal{E}_1)$ with  connection over $M$; this gives the definition of $\Des$.  

\begin{remark}
\label{rem:propDes}
We observe the following features of the functor $\Des$:
\begin{enumerate}[(a)]

\item 
\label{rem:propDes:1}
For three twisted vector bundles $\mathcal{E}_1,\mathcal{E}_2,\mathcal{E}_3$ we have a  morphism
\begin{equation*}
\Des(\mathcal{E}_3,\mathcal{E}_2) \otimes \Des(\mathcal{E}_2,\mathcal{E}_1) \to \Des(\mathcal{E}_3,\mathcal{E}_1)
\end{equation*}
in $\HVbuncon M$, 
which is an isomorphism if and only if $\mathcal{E}_2$ is invertible. It is induced by a  2-morphism $\mathcal{E}_2^{*} \circ \mathcal{E}_2 \Rightarrow \id_{\mathcal{I}_{\omega_2}}$, which is in turn induced by the evaluation $E_2 \otimes E_2^{*} \to \C$ of dual vector bundles.

\item
For  two twisted vector bundles $\mathcal{E}_i:\mathcal{G} \to \mathcal{I}_{\omega_i}$ we have an  isomorphism $\Des(\mathcal{E}_1,\mathcal{E}_2) \to \Des(\mathcal{E}_2,\mathcal{E}_1)^{*}$ in $\HVbuncon M$, induced by the canonical isomorphism $(\mathcal{E}_1\circ \mathcal{E}_2^{*})^{*}=\mathcal{E}_2 \circ \mathcal{E}_1^{*}$.

\item
\label{rem:propDes:2}
For  two twisted vector bundles $\mathcal{E}_i:\mathcal{G} \to \mathcal{I}_{\omega_i}$ and a 1-morphism $\mathcal{A}: \mathcal{G}' \to \mathcal{G}$ in $\ugrbcon M$ we have a morphism
\begin{equation*}
\Des( \mathcal{E}_1\circ \mathcal{A}, \mathcal{E}_2\circ \mathcal{A})\to \Des(\mathcal{E}_1,\mathcal{E}_2)
\end{equation*}
in $\HVbuncon M$, which is an isomorphism if and only if $\mathcal{A}$ is invertible. It is induced from the 2-morphism $\mathcal{A}\circ \mathcal{A}^{*}\Rightarrow \id_{\mathcal{G}}$ and  a 2-isomorphism $(\mathcal{E}_2\circ \mathcal{A})^{*} \Rightarrow \mathcal{A}^{*} \circ \mathcal{E}_2^{*}$, which is in turn induced by the canonical isomorphism $(A \circ E_2)^{*}\cong E_2^{*} \otimes A^{*}$. 

\end{enumerate}
\end{remark}

\begin{remark}
\label{re:algebrabundle}
If $\mathcal{E}$ is a $\mathcal{G}$-twisted vector bundle, we obtain a hermitian vector bundle $\mathrm{End}(\mathcal{E}) := \Des(\mathcal{E},\mathcal{E})$ with  connection.  \cref{rem:propDes:1} provides a connection-preserving bundle morphism
\begin{equation}
\label{eq:endEmult}
\mathrm{End}(\mathcal{E}) \otimes \mathrm{End}(\mathcal{E}) \to \mathrm{End}(\mathcal{E})\text{,}
\end{equation}
which endows the fibers of $\mathrm{End}(\mathcal{E})$  with (associative, unital, complex) algebra structures. The fact that $\mathrm{End}(\mathcal{E})$ carries a connection for which \cref{eq:endEmult} is connection-preserving, together with the fact that all fibers are simple algebras,  assures that $\mathrm{End}(\mathcal{E})$ is an algebra bundle, see \cref{lem:algloctriv,re:loctriv:c}.  For an explicit construction, one notices that the descent isomorphism $\tilde\zeta$ of \cref{eq:tildeepsilon} is an isomorphism between algebra bundles  \cite{Schweigert2014}. In relation to twisted K-theory the same algebra bundle has been constructed in \cite{Karoubia}.
\end{remark}

\begin{remark}
Applying the previous remark to a TBG object $(\mathcal{G},\mathcal{E})$, we obtain bundles of central simple algebras $\mathrm{End}(\mathcal{E}_i)$ over the branes $Q_i$.
The assignment of these bundles is functorial with respect to TBG morphisms. Indeed, if $(\mathcal{A},\psi)$ is  a 1-morphism, then we obtain a bundle morphism
\begin{equation*}
\alxydim{@C=4em}{\Des(\mathcal{E}_i,\mathcal{E}_i) \ar[r]^-{\Des(\psi_i,\psi_i^{-1})} & \Des(\mathcal{E}_i' \circ \mathcal{A},\mathcal{E}_i' \circ \mathcal{A}) \ar[r]  & \Des(\mathcal{E}_i',\mathcal{E}_i')}
\end{equation*}
by \cref{rem:propDes:2}. Following  \cref{rem:propDes} it is easy to check that this bundle morphism preserves the algebra structures.
Furthermore, one can check that it only depends on the 2-isomorphism class of $(\mathcal{A},\psi)$.
Finally, the composition of TBG morphisms induces the composition of algebra bundle homomorphisms. Summarizing, for every brane index $i\in I$ we have constructed a functor $(\mathcal{G},\mathcal{E}) \mapsto \mathrm{End}(\mathcal{E}_i)$ from $\hc 1 \tsg(M,Q)$ to the category of central simple algebra bundles over $Q_i$. 
\end{remark}

\begin{remark}
\label{rem:sectionsandtrivializations}
We recall the following facts about trivializations of bundle gerbes, described in \cite[Lemma 3.2.3]{waldorf10}. Suppose $\mathcal{G}\in \ugrbcon M$  and $s: M \to Y$ is a smooth section into the surjective submersion of $\mathcal{G}$. Then, we obtain a  1-morphism $\mathcal{T}_s: \mathcal{G} \to \mathcal{I}_{s^{*}B}$ in $\ugrbcon M$, where $B$ is the curving of $\mathcal{G}$. Explicitly, it is  defined by the hermitian line bundle $T_s :=  (\id_{Y}, s \circ \pi)^{*}L$ over $Y$ and  the isomorphism $\sigma_s := (\pr_1,\pr_2,s\circ \pi)^{*}\mu$, where $L$  is the line bundle of $\mathcal{G}$, and $\mu$ is its bundle gerbe product. 
We call $\mathcal{T}_s$ the \emph{trivialization associated to the section $s$}.
If $s':M \to Y$ is another section, and   $\tilde s: M \to L$ is a parallel unit-length section along $(s,s'):M \to Y^{[2]}$, then we obtain a 2-isomorphism $\psi_{\tilde s}: \mathcal{T}_s \Rightarrow \mathcal{T}_{s'}$ in $\ugrbcon M$, given fiber-wise by
\begin{equation*}
\alxydim{@C=4em}{T_s|_y = L_{y,s(\pi(y))} \ar[r]^-{\id \otimes \tilde s} & L_{y,s(\pi(y))} \otimes L_{s(\pi(y)),s'(\pi(y))}  \ar[r]^-{\mu} &  L_{y,s'(\pi(y))} = T_{s'}|_y}\text{.}
\end{equation*}
Now, consider a 1-morphism $\mathcal{E}: \mathcal{G} \to \mathcal{I}_{\rho}$ in $\ugrbcon M$, consisting of a vector bundle $E$ over $Y$ and a bundle isomorphism $\zeta$ over $Y^{[2]}$, and a section $s: M \to Y$.
Then, there is a connection-preserving, unitary bundle isomorphism 
$\varphi_s: \Des(\mathcal{E},\mathcal{T}_s) \to s^{*}E$,
which is induced by 
\begin{equation*}
\small
\alxydim{@C=2cm}{T_s|^{*}_y \otimes E|_y   = L|_{y,s(\pi(y))}^{*}\otimes E|_y \ar[r]^-{\tilde\mu \otimes \id}  &  L|_{s(\pi(y)),y} \otimes E|_y \ar[r]^-{\zeta|_{s(\pi(y)),y}} & E|_{s(\pi(y))}\text{.}}
\end{equation*}  
Here, $\tilde\mu: L|_{y_1,y_2}^{*} \to L|_{y_2,y_1}$ is induced from the bundle gerbe product. 
In the presence of a second section $s'$ and a parallel unit-length section $\tilde s$  of $L$ along $(s,s')$ we obtain from the definitions a commutative diagram
\begin{equation*}
\small
\alxydim{@C=4em@R=3em}{\Des(\mathcal{E},\mathcal{T}_{s'}) \ar[d]_{\Des(\id,\psi_{\tilde s})} \ar[r]^-{\varphi_{s'}} & s'^{*}E \ar[d]^{\zeta(\tilde s \otimes -)} \\ \Des(\mathcal{E},\mathcal{T}_s) \ar[r]_-{\varphi_{s}} & s^{*}E\text{.}}
\end{equation*}
\end{remark}

\subsection{Loop space brane geometry}

\label{sec:loopspacegeometry}

We let $LM := C^{\infty}(S^1,M)$ be the  loop space of $M$, considered as a Fréchet manifold or, equivalently, as a diffeological space. We let $PM \subset C^{\infty}([0,1],M)$ be the subset of paths with sitting instants, considered as a diffeological space. Further, we let $P_{ij}\subset PM$ be the subspace of paths $\gamma$ with $\gamma(0)\in Q_i$ and $\gamma(1)\in Q_j$.
Vector bundles and connections on diffeological spaces are discussed in \cref{sec:diffvectorbundle}; essentially, they can be treated just like vector bundles over smooth manifolds.

We now describe the objects of the category $\lsg(M,Q)$, making the announcements of \cref{sec:intro} precise. An object is a septuple $(\inf L,\lambda, \inf R,\phi,\chi,\epsilon,\alpha)$ consisting of  the following structure:
\begin{enumerate}[(1)]

\item
\label{LBGstr:1}
A hermitian line bundle $\inf L$ over  $LM$ with a superficial connection. 

The definition of  \quot{superficial} is \cite[Definition 2.2.1]{waldorf10}. We recall that a \emph{thin homotopy} between loops $\tau,\tau':S^1 \to M$ is a path $h:[0,1] \to LM$ such that the differential of its adjoint map $h^{\vee}:[0,1] \times S^1 \to M$  has rank less than two everywhere. If $\tau,\tau'$ are thin homotopic, then the parallel transport of a superficial connection along a thin homotopy $h$ between $\tau$ and $\tau'$ is independent of the choice of $h$, and hence determines a \emph{canonical} unitary isomorphism $d_{\tau,\tau'}:\inf L|_{\tau} \to \inf L|_{\tau'}$, independent of the thin homotopy.

\item
\label{LBGstr:2}
A  fusion product $\lambda$ on $\inf L$, i.e. unitary isomorphisms
\begin{equation*}
\lambda_{\gamma_1,\gamma_2,\gamma_3}: \inf L|_{\gamma_1 \cup \gamma_2} \otimes \inf L|_{\gamma_2 \cup\gamma_3} \to \inf L|_{\gamma_1\cup\gamma_3}
\end{equation*}
for all triples $(\gamma_1,\gamma_2,\gamma_3)$ of paths in $M$ with a common initial point and a common end point, forming a connection-preserving unitary bundle isomorphism over the 3-fold fibre product $PM^{[3]}$ of $\ev:PM \to M \times M$ with itself. The fusion product is required to be associative over $PM^{[4]}$ and to be symmetrizing
with respect to the connection (\cite[Def. 2.1.5]{waldorf10}); this latter property only plays a minor role in the present article, though.

The fusion product induces a parallel, unit-length section $PM \to \inf L:\gamma \mapsto \nu_\gamma$ along the inclusion of \quot{flat} loops $PM \to LM: \gamma \mapsto \gamma\cup\gamma$, which is neutral with respect to fusion \cite[Lemma 2.1.4]{waldorf10}. Further, we obtain a connection-preserving, unitary isomorphism $\tilde\lambda_{\gamma_1,\gamma_2}:\inf L|_{\gamma_1\cup\gamma_2} \to \inf L^{*}|_{\gamma_2\cup\gamma_1}$  that we will often combine with the identification $\inf L^{*} \cong \overline{\inf L}$ between the dual and the complex conjugate line bundle, induced by the hermitian metric on $\inf L$.

\item
\label{LBGstr:3}
A family $\inf R=\{\inf R_{ij}\}_{i,j\in I}$ of hermitian vector bundles $\inf R_{ij}$ over $P_{ij}$  with superficial connections $pt_{ij}$. 

We refer to \cref{def:superficial} in \cref{sec:superficial} for the definition of \quot{superficial} for connections on vector bundles over path spaces. Similarly to \cref{LBGstr:1}, a superficial connection determines -- via parallel transport -- a canonical unitary isomorphism $d_{\gamma,\gamma'}: \inf R_{ij}|_{\gamma} \to \inf R_{ij}|_{\gamma'}$ between the fibers of $\inf R_{ij}$ over all pairs of fixed-ends thin homotopic paths $\gamma,\gamma' \in P_{ij}$.  

\item 
\label{LBGstr:4}
A family $\phi=\{\phi_{ij}\}_{i,j\in I}$ of  unitary isomorphisms
\begin{equation*}
\phi_{ij}|_{\gamma_1,\gamma_2}: \inf L|_{\gamma_1 \cup \gamma_2} \otimes \inf R_{ij}|_{\gamma_2} \to \inf R_{ij}|_{\gamma_1}
\end{equation*}
for all $\gamma_1,\gamma_2\in P_{ij}$ with common initial point and common end point. These have to form connection-preserving bundle isomorphisms $\phi_{ij}$ over  $P_{ij}^{[2]}:= P_{ij} \times_{Q_i \times Q_j} P_{ij}$.

\item
\label{LBGstr:5}
A family $\chi=\{\chi_{ijk}\}_{i,j,k\in I}$ of linear maps 
\begin{equation*}
\chi_{ijk}|_{\gamma_{12},\gamma_{23}}: \inf R_{jk}|_{\gamma_{23}} \otimes \inf R_{ij}|_{\gamma_{12}} \to \inf R_{ik}|_{\gamma_{23} \pcomp \gamma_{12}}
\end{equation*} 
for all composable paths $\gamma_{12}$ and $\gamma_{23}$. These have to form a connection-preserving bundle morphism $\chi_{ijk}$ over $P_{jk} \times_{Q_j} P_{ij}$.

\item 
\label{LBGstr:6}
A family $\epsilon=\{\epsilon_i\}_{i\in I}$ of  parallel  sections $\epsilon_i: Q_i \to  \inf R_{ii}$ along the inclusion $x \mapsto \cp_x$ of constant paths. 
\item
\label{LBGstr:7}
A family $\alpha=\{\alpha_{ij}\}_{i,j\in I}$ of  unitary isomorphisms
\begin{equation*}
\alpha_{ij}|_{\gamma}: \inf R_{ij}|_{\gamma} \to \overline{\inf R_{ji}}|_{\overline{\gamma}}
\end{equation*}
for all $\gamma\in P_{ij}$, where $\overline{\gamma}$ denotes the reversed path. These have to form a connection-preserving  bundle isomorphism over $P_{ij}$. 

\end{enumerate}
This structure has to satisfy the following axioms:
\begin{enumerate}[({LBG}1),leftmargin=*,widest=10]

\item
\label{eq:lsg:fusionass}
The isomorphisms $\phi_{ij}$ are compatible with the fusion product $\lambda$; that is, the diagram
\begin{equation*}
\small
\alxydim{@C=2cm@R=3em}{\inf L|_{\gamma_1 \cup \gamma_2} \otimes \inf L|_{\gamma_2\cup\gamma_3} \otimes\inf  R_{ij}|_{\gamma_3} \ar[d]_{\lambda_{\gamma_1,\gamma_2,\gamma_3} \otimes \id} \ar[r]^-{\id \otimes \phi_{ij}|_{\gamma_2,\gamma_3}} & \inf L|_{\gamma_1 \cup \gamma_2} \otimes \inf R_{ij}|_{\gamma_2} \ar[d]^{\phi_{ij}|_{\gamma_1,\gamma_2}} \\ \inf L|_{\gamma_1\cup \gamma_3} \otimes \inf R_{ij}|_{\gamma_3} \ar[r]_-{\phi_{ij}|_{\gamma_1,\gamma_3}} & \inf R_{ij}|_{\gamma_1}}
\end{equation*}
is commutative for all $(\gamma_1,\gamma_2,\gamma_3)\in P_{ij}^{[3]}$. We  say that $\phi_{ij}$ is a \emph{fusion representation} of $(\inf L,\lambda)$ on $\inf R_{ij}$. 

\item
\label{eq:lsg:pentagon}
The maps $\chi_{ijk}$ are associative up to parallel transport along a reparameterization; that is, the diagram
\begin{equation*}
\small
\alxydim{@C=-6.5em@R=3em}{&& \inf R_{kl}|_{\gamma_{34}} \otimes \inf R_{jk}|_{\gamma_{23}} \otimes  \inf R_{ij}|_{\gamma_{12}} \ar[lld]_{\chi_{jkl}|_{\gamma_{23},\gamma_{34}} \otimes \id} \ar[rrd]^{\id \otimes \chi_{ijk}|_{\gamma_{12},\gamma_{23}}} && \\  \inf R_{jl}|_{\gamma_{34} \pcomp \gamma_{23}}  \otimes  \inf R_{ij}|_{\gamma_{12}} \ar[dr]_{\chi_{ijl}|_{\gamma_{12},\gamma_{34} \pcomp \gamma_{23}}} &&&&  \inf R_{kl}|_{\gamma_{34}} \otimes  \inf R_{ik}|_{\gamma_{23} \pcomp \gamma_{12}} \ar[dl]^{\chi_{ikl}|_{\gamma_{23} \pcomp \gamma_{12},\gamma_{34}}} \\ &  \inf R_{il}|_{(\gamma_{34} \pcomp \gamma_{23}) \pcomp \gamma_{12}} \ar[rr]_{d} &\hspace{6cm}&  \inf R_{il}|_{\gamma_{34} \pcomp (\gamma_{23} \pcomp \gamma_{12})}\text{.} }
\end{equation*}
is commutative for all triples $\gamma_{12}\in P_{ij}, \,\gamma_{23}\in P_{jk}$ and $\gamma_{13}\in P_{ik}$ of composable paths, and $d$ is the canonical isomorphism of \cref{LBGstr:3}. We  say that $\chi$ is a \emph{lifted path concatenation} on $\inf R$. 

\item
\label{eq:lsg:comp}
Fusion representation and lifted path concatenation are compatible in the following sense. Suppose $\gamma_{12},\gamma_{12}'\in P_{ij}$ and $\gamma_{23},\gamma_{23}'\in P_{jk}$  connect three points in the following way:
 \begin{equation*}
\small
\alxydim{@C=2cm}{x \ar@/^1.5pc/[r]^{\gamma_{12}}\ar@/_1.5pc/[r]_{\gamma_{12}'} & y \ar@/^1.5pc/[r]^{\gamma_{23}}\ar@/_1.5pc/[r]_{\gamma_{23}'} & z\text{.}}
\end{equation*}
Then, the diagram
\begin{equation*}
\small
\alxydim{@C=1.2cm@R=3em}{\inf L|_{\gamma_{23} \cup \gamma_{23}'} \otimes \inf L|_{\gamma_{12} \cup \gamma_{12}'} \otimes \inf R_{jk}|_{\gamma_{23}'} \otimes \inf R_{ij}|_{\gamma_{12}'} \ar[r]^-{\lambda' \otimes \chi_{ijk}}  \ar[d]_{\id \otimes braid \otimes \id} & \inf L|_{(\gamma_{23} \pcomp \gamma_{12}) \cup (\gamma_{23}' \pcomp \gamma_{12}')} \otimes \inf R_{ik}|_{\gamma_{23}' \pcomp \gamma_{12}'}  \ar[dd]^{\phi_{ik}} \\ \inf L|_{\gamma_{23} \cup \gamma_{23}'} \otimes \inf R_{jk}|_{\gamma_{23}'} \otimes \inf L|_{\gamma_{12} \cup \gamma_{12}'} \otimes \inf R_{ij}|_{\gamma_{12}'} \ar[d]_{\phi_{jk} \otimes \phi_{ij}}  & \\ \inf R_{jk}|_{\gamma_{23}} \otimes \inf R_{ij}|_{\gamma_{12}} \ar[r]_{\chi_{ijk}} & \inf R_{ik}|_{\gamma_{23} \pcomp \gamma_{12}} }
\end{equation*}
has to be commutative. The isomorphism $\lambda'$ on the top of this diagram is given by
\begin{equation*}
\small
\alxydim{@R=3em}{\inf L|_{\gamma_{23} \cup \gamma_{23}'} \otimes \inf L|_{\gamma_{12} \cup \gamma_{12}'} \ar[r]^-{d \otimes d} & \inf L|_{(\overline{\gamma_{23}'} \pcomp \gamma_{23}) \cup \cp_y} \otimes \inf L|_{\cp_y \cup (\gamma'_{12} \pcomp \overline{\gamma_{12}})} \ar[d]^-{\lambda} \\ & \inf L|_{(\overline{\gamma_{23}'} \pcomp \gamma_{23}) \cup  (\gamma'_{12} \pcomp \overline{\gamma_{12}})} \ar[r]_-{d} & \inf L|_{(\gamma_{23} \pcomp \gamma_{12}) \cup (\gamma_{23}' \pcomp \gamma_{12}')}\text{.}}
\end{equation*}

\item
\label{eq:lsg:unit}
The sections $\epsilon_i$ provide units (up to reparameterization) for the lifted path concatenation, i.e. 
\begin{equation*}
\chi_{iij}|_{\cp_x,\gamma}(v,\epsilon_i(x)) = d_{\gamma,\gamma\pcomp \cp_x}(v)
\quand
\chi_{ijj}|_{\gamma,\cp_y}(\epsilon_j(y),v) = d_{\gamma,\cp_y \pcomp \gamma}(v)
\end{equation*}
for all paths $\gamma\in P_{ij}$ with $x:= \gamma(0)$ and $y:=\gamma(1)$,  and all $v\in  \inf R_{ij}|_{\gamma}$. We say that $\epsilon_i$  is a \emph{lifted constant path}. 

\item
\label{eq:lsg:invariance}
The isomorphisms $\alpha_{ij}$ satisfy the following compatibility condition with the hermitian metric $h_{ij}$ on $\inf R_{ij}$ and the lifted path concatenation: for $\gamma \in P_{ij}$ and elements $v,w\in \inf R_{ij}|_{\gamma}$ we have
\begin{multline*}
h_{ii}(\chi_{iji}|_{\gamma,\overline{\gamma}}(\alpha_{ij}(w)\otimes v) , d_{\cp_x,\overline{\gamma} \pcomp\gamma}(\epsilon_i(x)))
\\=h_{ij}(v,w) = h_{jj}(d_{\cp_x,\gamma\pcomp \overline{\gamma}}(\epsilon_j(x)),\chi_{jij}|_{\overline{\gamma},\gamma}(w \otimes \alpha_{ij}(v))) \text{.}
\end{multline*}
We will say that $\alpha_{ij}$ is a \emph{lifted path reversal}.

\item
\label{eq:lsg:unitinversion}
Lifted constant paths are invariant under the lifted path reversal:
\begin{equation*}
\alpha_{ii}|_{\cp_x}(\epsilon_i(x)) = \epsilon_i(x)\text{.}
\end{equation*} 

\item
\label{eq:lsg:involutive}
Lifted path reversal is involutive:
$\alpha_{ji} \circ \alpha_{ij} = \id_{\inf R_{ij}}$.

\item
\label{eq:lsg:symmetrizing}
Lifted path reversal is an anti-homomorphism with respect to lifted path concatenation: the diagram
\begin{equation*}
\small
\alxydim{@C=1.8cm@R=3em}{\inf R_{jk}|_{\gamma_2} \otimes \inf R_{ij}|_{\gamma_1} \ar[d]_{\alpha_{jk} \otimes \alpha_{ij}} \ar[r]^-{\chi_{ijk}|_{\gamma_1,\gamma_2}} & \inf R_{ik}|_{\gamma_2\pcomp \gamma_1} \ar[dd]^{\alpha_{ik}}  \\  \overline{\inf R_{kj}}|_{\overline{\gamma_2}} \otimes \overline{\inf R_{ji}}|_{\overline{\gamma_1}} \ar[d]_{braid} &\\ \overline{\inf R_{ji}}|_{\overline{\gamma_1}} \otimes \overline{\inf R_{kj}}|_{\overline{\gamma_2}} \ar[r]_-{\chi_{kji}|_{\overline{\gamma_2},\overline{\gamma_1}}} & \overline{\inf R_{ki}}|_{\overline{\gamma_1}\pcomp \overline{\gamma_2}}}
\end{equation*}
is commutative for all $\gamma_1\in P_{ij}$ and $\gamma_2\in P_{jk}$.

\item
\label{eq:lsg:symmetrizingfusion}
Lifted path reversal intertwines the fusion representation (up to reparameterization), in the sense that the diagram
\begin{equation*}
\small
\alxydim{@R=3em}{\inf L|_{\gamma_1\cup\gamma_2} \otimes \inf R_{ij}|_{\gamma_2} \ar[d]_{\tilde\lambda \otimes \alpha_{ij}}\ar[r]^-{\phi_{ij}} & \inf R_{ij}|_{\gamma_1} \ar[dd]^{\alpha_{ij}}\\ \overline{\inf L}|_{\gamma_2\cup\gamma_1} \otimes \overline{\inf R}_{ji}|_{\overline{\gamma_2}}  \ar[d]_{d \otimes \id}  & \\ \overline{\inf L}|_{\overline{\gamma_1} \cup \overline{\gamma_2}} \otimes \overline{\inf R}_{ji}|_{\overline{\gamma_2}} \ar[r]_-{\phi_{ji}} & \overline{\inf R_{ji}}|_{\overline{\gamma_1}} }
\end{equation*}
is commutative, where $\tilde\lambda$ was explained in \cref{LBGstr:2}.

\item
\label{eq:lsg:cardy}
The following normalization condition holds
 between the lifted path concatenation and the hermitian metric $h_{ij}$ on $\inf R_{ij}$: for all $\gamma \in P_{ij}$ with $x:=\gamma(0)$ and $y:=\gamma(1)$, every orthonormal basis  $(v_1,...,v_n)$  of $\inf R_{ij}|_{\gamma}$ and every $v\in \inf R_{ii}|_{\cp_x}$:
\begin{equation*}
\sum_{k=1}^{n} \chi_{jij}|_{\overline{\gamma},\gamma \pcomp \cp_x}(\chi_{iij}|_{\cp_x,\gamma}(v_k \otimes v) \otimes \alpha_{ij}(v_k))=h_{ii}(\epsilon_i(x),v)\cdot  d_{\cp_y,\gamma \pcomp \cp_x\pcomp  \overline{\gamma}}(\epsilon_j(y))\text{.}
\end{equation*}
\end{enumerate}

Loop space brane geometry as defined above forms a groupoid $\lsg(M,Q)$ in a natural way: a morphism between $(\inf L,\lambda, \inf R,\phi,\chi,\epsilon,\alpha)$ and $(\inf L',\lambda',\inf R',\phi',\chi',\epsilon',\alpha')$ is a pair $(\varphi,\xi)$ consisting of a connection-preserving unitary isomorphism $\varphi:\inf L\to \inf L'$ that preserves  the fusion products $\lambda$ and $\lambda'$, and of a family $\xi=\{\xi_{ij}\}_{i,j\in I}$ of connection-preserving, unitary bundle isomorphisms $\xi_{ij}:\inf R_{ij} \to \inf R_{ij}'$ that are compatible with the remaining structure in the obvious way. 
In the remainder of this subsection we deduce some direct  consequences from the axioms of LBG in the following remarks.

\begin{remark}
\label{re:fusrepcan}
From \cref{eq:lsg:fusionass} one can easily deduce that the section  $\nu$ in $\inf L$ described in \cref{LBGstr:2} is neutral for the fusion representation, i.e. 
$\phi_{ij}|_{\gamma,\gamma}(\nu_{\gamma},v)=v$ for all $v\in \inf R_{ij}|_\gamma$. 
\end{remark}

\begin{remark}
In the context of field theories, \cref{eq:lsg:cardy} will  be responsible for the Cardy condition, see \cref{sec:mooresegal}. Here we remark the following consequences of \cref{eq:lsg:cardy} for the rank of the vector bundles $\in \inf R_{ij}$:
\begin{enumerate}[(a)]

\item
\label{re:lsg:rank}
We have
\begin{equation*}
\| \epsilon_i(x) \| = \sqrt[4]{\mathrm{rk}(\inf R_{ii})}\text{,}
\end{equation*}
where the norm is formed fibrewise using the hermitian metric on $\inf R_{ii}$.
Indeed, we set $i=j$, $\gamma=\cp_x$  and $v=\epsilon_i(x)$ in \cref{eq:lsg:cardy}, and obtain
using \cref{eq:lsg:unit}: 
\begin{equation*}
\sum_{k=1}^{n} \chi_{iii}|_{\cp_x,\cp_x}(v_k  \otimes \alpha_{ii}(v_k))= h_{ii}(\epsilon_i(x),\epsilon_i(x))\cdot \epsilon_i(x)\text{.}
\end{equation*}
Inserting this into $h_{ii}(\epsilon_i(x),-)$ 
and using  \cref{eq:lsg:invariance}
gives
\begin{equation*}
\sum_{k=1}^{n} h_{ii}(v_k,v_k)= h_{ii}(\epsilon_i(x),\epsilon_i(x))^2\text{.}
\end{equation*}
Since $(v_1,...,v_n)$ is an orthonormal basis, we have the claim.

\item
\label{re:lsg:rank2}
Putting $v=\epsilon_i(x)$  in \cref{eq:lsg:cardy} we obtain using \cref{re:lsg:rank*} 
\begin{equation}
\label{re:lsg:unit}
\sum_{k=1}^{n} \chi_{jij}|_{\overline{\gamma},\gamma}(v_k  \otimes \alpha_{ij}(v_k))=\sqrt{\mathrm{rk}(\inf R_{ii})} \cdot  d_{\cp_y,\gamma \pcomp  \overline{\gamma}}(\epsilon_j(y))\text{.}
\end{equation}
We insert 
\cref{re:lsg:unit}
 into $h_{jj}(-, d_{\cp_y,\gamma \pcomp  \overline{\gamma}}(\epsilon_j(y)))$ and obtain
\begin{equation*}
\sum_{k=1}^{n} h_{jj}(\chi_{jij}|_{\overline{\gamma},\gamma}(v_k  \otimes \alpha_{ij}(v_k)), d_{\cp_y,\gamma \pcomp  \overline{\gamma}}(\epsilon_j(y)))=\sqrt{\mathrm{rk}(\inf R_{ii})} \cdot h_{jj}(\epsilon_j(y), \epsilon_j(y))\text{.}
\end{equation*}
With \cref{eq:lsg:invariance} the left hand side sums up to $n=\mathrm{rk}(\inf R_{ij})$, so that we obtain:
\begin{equation*}
\mathrm{rk}(\inf R_{ij})^2 = \mathrm{rk}(\inf R_{ii})\cdot \mathrm{rk}(\inf R_{jj})\text{.}
\end{equation*}
\end{enumerate}
\end{remark}

\begin{remark}
\label{re:closingandopeningstrings}
Let $L_i := \{\gamma \in P_{ii}\sep \gamma(0)=\gamma(1)\}$, coming with an injective smooth map $L_i \to LM$ which we usually drop from notation (the image consists of loops based in $Q_i$ with sitting instants at $1\in S^1$).  Consider the map  $L_i \to P_{ii} \times_{Q_i \times Q_i} P_{ii}: \gamma\mapsto (\cp_{x},\gamma)$, where $x:=\gamma(0)$. The pullback of the fusion representation along this map is a connection-preserving  bundle isomorphism
\begin{equation*}
\phi_{ii}|_{\cp_x,\gamma}:\inf L_{\cp_x \cup\gamma} \otimes  \inf R_{ii}|_{\cp_x} \to \inf R_{ii}|_{\gamma}\text{.}
\end{equation*} 
Using  the superficial connection on $\inf L$ to identify $\inf L_{\cp_x \cup \gamma}\cong \inf L_{\gamma}$, and the lifted constant path $\epsilon_i(x)$  we obtain a connection-preserving bundle morphism
\begin{equation*}
\iota_i: \inf L \to \inf R_{ii}|_{L_i}
\end{equation*}
over $L_i$. We call it the \emph{string opening morphism}; in the context of topological field theories it will describe the phenomenon that a closed string opens up in the presence of a D-brane \cite{Bunk2018}. In order to describe the converse phenomenon, we define a smooth, fiber-wise linear map
$\theta_i: \cp^{*}\inf R_{ii} \to \C$ by $\theta_i(v):= h_{ii}(\epsilon_i(x),v)$, where $v\in \inf R_{ii}|_{\cp_x}$. We will reveal it in \cref{sec:algebrabundle:lsg} as the trace of a certain Frobenius algebra. Now, the inverse of $\phi_{ii}$ in combination with  $\theta_i$ gives a connection-preserving bundle morphism
\begin{equation*}
\iota^{*}_i: \inf R_{ii}|_{L_i} \to \inf L
\end{equation*}
over $L_i$, which is called the \emph{string closing morphism}.
\end{remark}

\setsecnumdepth{2}

\section{Algebraic structures in loop space brane geometry}

\label{sec:alglsg}

In this section we study LBG intrinsically, without relation to TBG.
We show that LBG  induces   bundles of simple Frobenius algebras over the branes, together with  bundles of bimodules over spaces of paths connecting the branes. By an algebra we  always mean a complex, associative, unital, finite-dimensional algebra, and algebra homomorphisms are assumed to be unital.

\subsection{Frobenius algebra bundles over the branes}

\label{sec:algebrabundle:lsg}

Let $(\inf L,\lambda,\inf R,\phi,\chi,\epsilon,\alpha) \in \lsg(M,Q)$.
A  vector bundle  over $Q_i$ is obtained by putting $\inf A_i := \cp^{*} \inf R_{ii}$, where $\cp :Q_i \to P_{ii}$ associates to each point $x\in Q_i$ the constant path $\cp_x$ at that point. 
The pullback of  $\chi_{iii}$ along $Q_i \to P_{ii} \times_{Q_i} P_{ii}:x \mapsto (\cp_x,\cp_x)$ gives a bundle morphism
$\mu_i:\inf A_i \otimes \inf A_i \to \inf A_i$. 
We consider condition \cref{eq:lsg:pentagon} restricted to a triple of constant paths, $(\cp_x,\cp_x,\cp_x)$. The reparameterization at the bottom of the diagram is trivial in this case, and we obtain that $\mu_i$ is associative. Further, the lifted constant paths $\epsilon_i: Q_i \to \inf R_{ii}$ induce a smooth  section of $\inf A_i$ that by \cref{eq:lsg:unit} provides a unit for each fiber.

\begin{lemma}
\label{lem:aibundle}
$\inf A_i$ is a bundle of simple algebras over $Q_i$. 
\end{lemma}

\begin{proof}
So far we have constructed an algebra structure on the vector bundle  $\inf A_i$, see \cref{sec:algebrabundles}. The connection on $\inf R_{ii}$ induces a connection on $\inf A_i$, for which  the multiplication $\mu_i$ is connection-preserving. By \cref{lem:algloctriv}  the algebra structure is local.  We will show in \cref{co:aisimple} that all algebras $\inf A_i|_x$ are simple, so that  $\inf  A_i$ is  a genuine algebra bundle (\cref{re:loctriv:c}). \end{proof}

We continue writing $a\cdot b := \mu_i(a \otimes b)$ and  $1 := \epsilon_i(x)$ for short. Further, we consider the following additional structures: first, the lifted path reversal induces an isomorphism $\cp^{*}\alpha_{ii}: \inf A_i \to \overline{\inf A_i}$ that is unital \cref{eq:lsg:unitinversion},  involutive \cref{eq:lsg:involutive} and anti-multiplicative \cref{eq:lsg:symmetrizing}; we write $a^{*} := \cp^{*}\alpha_{ii}(a)$. With these operations, $\inf A_i$ becomes a bundle of involutive algebras. Second, the hermitian metric $h_{ii}$ on $\inf R_{ii}$ induces a hermitian metric $\left \langle  -,- \right \rangle: \inf A_i \times \inf A_i \to \C$ on $\inf A_i$.  The induced norm satisfies $\|1\|=\sqrt[4]{\rk(\inf A_i)}$ by \cref{re:lsg:rank}. In particular, $\inf A_i$ is not a bundle of C$^{*}$-algebras unless $\mathrm{rk}(\inf A_i)=\mathrm{rk}(\inf R_{ii})=1$.

\begin{remark}
\label{re:epsilonunique}
Since units in algebras are uniquely determined, we obtain that the lifted constant path is fiber-wise uniquely determined. More precisely, if $(\inf L,\lambda, \inf R,\phi,\chi,\epsilon,\alpha)$ and $(\inf L,\lambda, \inf R,\phi,\chi,\epsilon',\alpha')$ are LBG objects, then $\epsilon=\epsilon'$. \end{remark}

Using the involution we can turn the hermitian metric into a  bilinear product,
\begin{equation*}
\sigma_i: \inf A_i \times \inf A_i \to \C: (v,w) \mapsto \left \langle v^{*},w \right \rangle\text{.}
\end{equation*}

\begin{lemma}
\label{eq:trace:invariance}
\label{eq:trace:skewsymmetry}
The bilinear product $\sigma_i$ is fiber-wise non-degenerate, symmetric, and invariant, i.e. $\sigma_i(v\cdot w,x)=\sigma_i(v,w\cdot x)$ for all $v,w,x\in \inf A_i$.
\end{lemma}

\begin{proof}
It is fiber-wise non-degenerate because $h_{ii}$ is so and  because $\alpha_{ii}$ is an isomorphism. The unitarity of $\alpha_{ii}$ implies that $\left \langle a^{*},b^{*}  \right \rangle = \left \langle b,a  \right \rangle$; this shows that $\sigma_i$ is symmetric.
For the invariance, we use \cref{eq:lsg:invariance} to check:
\begin{equation*}
\left\langle(v \cdot w)^{*},x\right\rangle= \left\langle x^{*}\cdot(v \cdot w)^{*} , 1\right\rangle
=\left\langle(w\cdot x)^{*}\cdot v^{*},1\right\rangle=\left\langle v^{*},w \cdot x\right\rangle\text{,}
\end{equation*}
which proves the claim.
\end{proof}

This shows that $\inf A_i|_x$ is a symmetric Frobenius algebra for each $x\in Q_i$, and together with \cref{lem:aibundle} we have the following result:

\begin{proposition}
\label{prop:frob}
$\inf A_i$ is a  bundle of simple, symmetric Frobenius algebras over $Q_i$. 
\end{proposition}

\begin{remark}
\label{eq:trace:star}
The trace $\inf A_i \to \C$ corresponding to the bilinear form $\sigma_i$ is $a \mapsto \sigma_i(a,1)$. In terms of the hermitian metric, it is $a \mapsto \left \langle a^{*},1  \right \rangle =\left \langle 1,a  \right \rangle$, so that it is precisely the map $\theta_i$ defined in \cref{re:closingandopeningstrings}. From this it is easy to see  that it respects the involution: $\theta_i(a^{*})=\overline{\theta_i(a)}$.
It is non-degenerate in the sense that its kernel contains no non-trivial left ideals; this follows from the non-degeneracy of $\sigma_i$, see \cite[Lemma 2.2.4]{Kock2003}. 
\end{remark}

\subsection{Bimodule bundles and Morita equivalences}

\label{sec:bimodulebundles}

In this section we show that the vector bundles $\inf R_{ij}$ are  bundles of bimodules over the  algebra bundles defined in \cref{sec:algebrabundle:lsg}.
For this purpose, we define  bundle morphisms
\begin{equation*}
\lambda_{ij}: \ev_1^{*}\inf A_{j} \otimes \inf R_{ij} \to \inf R_{ij}
\quand
\rho_{ij}:\inf R_{ij} \otimes \ev_0^{*}\inf A_i \to \inf R_{ij}
\end{equation*}
over $P_{ij}$ in the following way. For a path $\gamma \in P_{ij}$ with $x:=\gamma(0)$ and $y:=\gamma(1)$, we let $\lambda_{ij}|_{\gamma}$ be defined by
\begin{equation*}
\alxydim{@C=4em}{\inf A_j|_{y} \otimes \inf R_{ij}|_{\gamma} \ar[r]^-{\chi_{ijj}|_{\gamma,\cp_y}} & \inf R_{ij}|_{\cp_y \pcomp \gamma} \ar[r]^-{d_{\cp_y\pcomp\gamma,\gamma}} & \inf R_{ij}|_{\gamma}\text{,}}
\end{equation*}
and we let $\rho_{ij}|_{\gamma}$  be defined by
\begin{equation*}
\alxydim{@C=4em}{\inf R_{ij}|_{\gamma} \otimes \inf A_i|_x \ar[r]^-{\chi_{iij}|_{\cp_x,\gamma}} & \inf R_{ij}|_{\gamma \pcomp \cp_x} \ar[r]^-{d_{\gamma \pcomp \cp_x,\gamma}} & \inf R_{ij}|_{\gamma}\text{.}}
\end{equation*}

\begin{lemma}
\label{lem:Rijbimodule}
The bundle morphisms $\lambda_{ij}$ and $\rho_{ij}$ define commuting left and right actions of $\inf A_j$ and $\inf A_i$ on $\inf R_{ij}$.
\end{lemma}

\begin{proof}
We consider the following diagram; the commutativity of its  outer square corresponds to the statement that left and right actions commute:
\begin{equation*}
\small
\alxydim{@C=0.8cm@R=3em}{\inf A_j|_{y} \otimes \inf R_{ij}|_{\gamma} \otimes \inf A_i|_x \ar[dd]_{\chi_{ijj}|_{\gamma,\cp_y} \otimes \id} \ar[rrr]^{\id \otimes \chi_{iij}|_{\cp_x,\gamma}} &&& \inf A_j|_y \otimes \inf R_{ij}|_{\gamma \pcomp \cp_x} \ar[rr]^{\id \otimes d_{}}  \ar[d]_{\chi_{ijj}|_{\gamma\pcomp \cp_x,\cp_y}} && \inf A_j|_y \otimes \inf R_{ij}|_{\gamma}\ar[d]^{\chi_{ijj}|_{\gamma,\cp_y}} \\ &&& \mquad\inf R_{ij}|_{\cp_y \pcomp (\gamma \pcomp \cp_x)} \ar[rr]_-{d_{}} && \inf R_{ij}|_{\cp_y \pcomp \gamma}  \ar[dd]^{d}  \\ \inf R_{ij}|_{\cp_y \pcomp \gamma} \otimes \inf A_i|_x \ar[d]_{d_{} \otimes \id} \ar[rr]^-{\chi_{iij}|_{\cp_x,\cp_y \pcomp \gamma}} &&  \inf R_{ij}|_{(\cp_y \pcomp \gamma)\circ \cp_x}\mquad \ar[d]^{d_{}} \ar[ur]^-{d}  \\ \inf R_{ij}|_{\gamma} \otimes \inf A_i|_x \ar[rr]_{\chi_{iij}|_{\cp_x,\gamma}} && \inf R_{ij}|_{\gamma \circ \cp_x} \ar[rrr]_{d_{}} &&& \inf R_{ij}|_{\gamma}}
\end{equation*}
The diagram in the upper left corner is the pentagon diagram of \cref{eq:lsg:pentagon}. The two rectangular diagrams commute because $\chi_{iij}$ is connection-preserving, and the diagram in the lower right corner is commutative because the isomorphism $d$ is independent of the ways the thin homotopy is performed. 

Associativity of the left action follows from  \cref{eq:lsg:pentagon} in a very similar way, 
and the associativity of the right action is seen analogously. The fact that the actions are unital can easily be deduced from \cref{eq:lsg:unit}.
\end{proof}

\begin{lemma}
\label{lem:actionsconnectionpreserving}
The bundle morphisms $\lambda_{ij}$ and $\rho_{ij}$ are connection-preserving. \end{lemma}

\begin{proof}
For $\lambda_{ij}$, we have to show that for each path $\Gamma:\gamma \to \gamma'$ in $P_{ij}$ there is a commutative diagram
\begin{equation*}
\small
\alxydim{@C=2cm@R=3em}{\inf A_j|_{y} \otimes \inf R_{ij}|_{\gamma} \ar[d]_{pt_{\eta} \otimes pt_{ij}|_{\Gamma}} \ar[r]^-{\chi_{ijj}|_{\gamma,\cp_y}} & \inf R_{ij}|_{\cp_y \pcomp \gamma} \ar[d]^{pt_{ij}|_{\Gamma_{\eta}}} \ar[r]^-{d_{\cp_y\circ\gamma,\gamma}} & \inf R_{ij}|_{\gamma} \ar[d]^{pt_{ij}|_{\Gamma}} \\ \inf A_j|_{y'} \otimes \inf R_{ij}|_{\gamma'} \ar[r]_-{\chi_{ijj}|_{\gamma',\cp_y}} & \inf R_{ij}|_{\cp_y' \pcomp \gamma'} \ar[r]_-{d_{\cp_{y'}\circ\gamma',\gamma'}} & \inf R_{ij}|_{\gamma'}}
\end{equation*}
Here, $\eta\in PQ_j$ is the path formed by the end points of $\Gamma$, i.e. $\eta(t):=\Gamma(t)(1)$, and $\Gamma_{\eta} \in PPM$  is defined by $\Gamma_{\eta}(t) := \cp_{\eta(t)} \pcomp \Gamma(t)$. Further, we have set $y:= \gamma(1)$ and $y':=\gamma'(1)$. The diagram on the left is commutative because $\chi_{ijj}$ is connection-preserving. The diagram on the right is commutative because the connection on $\inf R_{ij}$ is superficial.
Indeed, we note that the obvious homotopy $h\in PPPM$ between $\Gamma_{\eta}$ and $\Gamma$ fixes the  paths of end points (namely, $t \mapsto \Gamma(t)(0)$ and $\eta$). Further, the adjoint map $h^{\vee}:[0,1]^3 \to M$  factors through $\Gamma^{\vee}:[0,1]^2 \to M$, which implies that its rank is less or equal than  two. Then, the diagram on the right is precisely an instance of the property \cref{def:superficial:ii} of a superficial connection.  The discussion for $\rho_{ij}$ is analogous. 
\end{proof}

\begin{proposition}
\label{prop:bimodule}
The vector bundle $\inf R_{ij}$ is bundle of   $\ev_1^{*}\inf A_j$-$\ev_0^{*}\inf A_i$-bimodules.   
\end{proposition}

\begin{proof}
By now we have equipped $\inf R_{ij}$ with a bimodule structure $(\lambda_{ij},\rho_{ij})$. As a consequence of  \cref{lem:Rijbimodule,lem:actionsconnectionpreserving,lem:bimodloctriv} this bimodule structure is local. We show below (\cref{rem:faithfullybalanced}) that the bimodule structure is faithfully balanced, so that \cref{re:bimod:a} implies that $\inf R_{ij}$ is a   bundle of bimodules. 
\end{proof}

\begin{remark}
\begin{enumerate}[(a)]

\item
\label{re:bimodule:a}
$\inf R_{ii}|_{\cp_x}$ is the identity $\inf A_i|_x$-$\inf A_i|_x$-bimodule. 

\item 
\label{re:bimodule:b}
\cref{lem:actionsconnectionpreserving} implies that parallel transport along a fixed-ends path $\Gamma$ in $P_{ij}$ is an intertwiner of $\ev_1^{*}\inf A_j$-$\ev_0^{*}\inf A_i$-bimodules. In particular, the canonical isometries $d_{\gamma,\gamma'}: \inf R_{ij}|_{\gamma} \to \inf R_{ij}|_{\gamma'}$ for a  pair $(\gamma,\gamma')$ of thin fixed-ends homotopic paths are intertwiners. 
If a path $\Gamma$ does not fix the endpoints, then parallel transport is an intertwiner with respect to  the parallel transport  in $\inf A_i$  and $\inf A_j$ along the paths of end points, respectively. 

\item
\label{re:bimodule:c}
The lifted path reversal $\alpha_{ij}:\inf R_{ij} \to \overline{\inf R_{ji}}$ exchanges left and right actions under the involutions of the algebra bundles.
More precisely, for a path $\gamma\in P_{ij}$ with $x:=\gamma(0)$ and $y:=\gamma(1)$  the diagram
\begin{equation*}
\small
\alxydim{@C=2cm@R=3em}{\inf A_j|_y \otimes \inf R_{ij}|_{\gamma} \ar[d]_{\lambda_{ij}} \ar[r]^-{\cp^{*}\alpha_{ii} \otimes \alpha_{ij}} & \overline{\inf A_j|_y} \otimes \overline{\inf R_{ij}|_{\gamma}} \ar[r]^{braid} & \overline{\inf R_{ji}}|_{\gamma} \otimes \overline{\inf A_j|_y} \ar[d]^{\rho_{ij}} \\ \inf R_{ij}|_{\gamma} \ar[rr]_-{\alpha_{ij}} && \overline{\inf R_{ji}}|_{\gamma}}
\end{equation*}
is commutative. This follows  from \cref{eq:lsg:symmetrizing} and the fact that $\alpha_{ij}$ is connection-preserving. 

\end{enumerate}
\end{remark}

Next, we consider the space $P_{jk} \times_{Q_j} P_{ij}$ of composable paths between three D-branes, equipped with the  projections $p_{ij}$ to $P_{ij}$ and $p_{jk}$ to $P_{jk}$, the composition $c$ to $P_{ik}$, and the projections $p_i,p_j,p_k$ to the end points of the paths.  

\begin{proposition}
\label{lem:compositionofbimodules}
Lifted path concatenation $\chi_{ijk}$ induces an isomorphism
\begin{equation*}
p_{jk}^{*}\inf R_{jk} \otimes_{p_j^{*}\inf  A_j} p_{ij}^{*}\inf R_{ij} \cong c^{*}\inf R_{ik}
\end{equation*}
of bundles of  $p_k^{*}\inf A_k$-$p_i^{*}\inf A_i$-bimodules over  $P_{jk} \times_{Q_j} P_{ij}$. 
\end{proposition}

\begin{proof}
In order to see that $\chi_{ijk}$ is well-defined on $p_{jk}^{*}\inf R_{jk} \otimes_{p_j^{*}\inf  A_j} p_{ij}^{*}\inf R_{ij} $   it suffices to show that it vanishes on elements of the form $\rho_{jk}|_{\gamma'}(w \otimes a) \otimes v - w \otimes \lambda_{ij}|_{\gamma}(a\otimes v)$, where $a\in \inf A_j|_y$, $v\in \inf R_{ij}|_{\gamma}$ and $w\in \inf R_{jk}|_{\gamma'}$, and $\gamma(1)=y=\gamma'(0)$. This follows  from the definitions of $\lambda_{ij}$ and $\rho_{jk}$ and \cref{eq:lsg:pentagon} via a direct calculation.

Similarly one checks that $\chi_{ijk}$ is an intertwiner for both actions. 
In order to show that it is an isomorphism we construct an inverse map in the fiber over a point $(\gamma',\gamma) \in P_{jk} \times_{Q_j} P_{ij}$.
Let $(v_1,...,v_n)$ be an orthonormal basis of $\inf R_{ij}|_{\gamma}$. For $x\in \inf R_{ik}|_{\gamma'\pcomp \gamma}$ we consider the element
\begin{equation*}
\psi(x) := \frac{1}{\sqrt{\mathrm{rk}(\inf R_{jj})}} \sum_{l=1}^{n} d_{(\gamma'\pcomp\gamma) \pcomp \overline{\gamma},\gamma'}\big (\chi_{jik}|_{\overline{\gamma},\gamma'\pcomp \gamma}(x \otimes \alpha_{ij}(v_l)) \big) \otimes v_l \in \inf R_{jk}|_{\gamma'} \otimes_{} \inf R_{ij}|_{\gamma}\text{.}
\end{equation*}
Using the fact that $\chi_{ijk}$ is connection-preserving and \cref{eq:lsg:pentagon}, we compute
\begin{equation*}
\chi_{ijk}|_{\gamma,\gamma'}(\psi(x)) 
=\frac{1}{\sqrt{\mathrm{rk}(\inf R_{jj})}} \sum_{l=1}^{n} d_{(\gamma' \pcomp \gamma )\pcomp (\overline{\gamma}\pcomp \gamma),\gamma'\pcomp \gamma}\big (\chi_{iik}|_{\overline{\gamma} \pcomp \gamma,\gamma'\pcomp \gamma}(x \otimes \chi_{iji}|_{\gamma,\overline{\gamma}}(\alpha_{ij}(v_l) \otimes v_l))\big)\text{.}
\end{equation*}
With \cref{re:lsg:unit} the latter becomes
\begin{equation*}
\chi_{ijk}|_{\gamma,\gamma'}(\psi(x)) 
=d_{(\gamma' \pcomp \gamma )\pcomp (\overline{\gamma}\pcomp \gamma),\gamma'\pcomp \gamma}(\chi_{iik}|_{\overline{\gamma} \pcomp \gamma,\gamma'\pcomp \gamma}(x \otimes d_{\cp_x,\overline{\gamma} \pcomp  \gamma}(\epsilon_i(x))))\text{.}
\end{equation*}
Via \cref{eq:lsg:unit} this is equal to $x$. 
Conversely, we compute $\psi(\chi_{ijk}|_{\gamma,\gamma'}(w \otimes v))$ for $v\in \inf R_{ij}|_{\gamma}$ and $w\in \inf R_{jk}|_{\gamma'}$. Using the definitions, the fact that $\chi_{ijk}$ is connection-preserving, and \cref{eq:lsg:pentagon},  we obtain 
\begin{equation*}
\psi(\chi_{ijk}|_{\gamma,\gamma'}(w \otimes v))=\frac{1}{\sqrt{\mathrm{rk}(\inf R_{jj})}} \sum_{l=1}^{n} \rho_{jk}\big (w \otimes d_{\gamma\pcomp \overline{\gamma},\cp_y}(\chi_{jij}|_{\overline{\gamma},\gamma}(v \otimes \alpha_{ij}(v_l)))\big ) \otimes v_l\text{.}
\end{equation*}
Under the quotient by the $\inf A_j$-action, the right hand side is identified with 
\begin{equation*}
\psi(\chi_{ijk}|_{\gamma,\gamma'}(w \otimes v))=\frac{1}{\sqrt{\mathrm{rk}(\inf R_{jj})}} \sum_{l=1}^{n} w \otimes \lambda_{ij} (d_{\gamma\pcomp \overline{\gamma},\cp_y}(\chi_{jij}|_{\overline{\gamma},\gamma}(v \otimes \alpha_{ij}(v_l))) \otimes v_l)\text{.}
\end{equation*} 
Again, from the definitions, the fact that $\chi_{ijk}$ is connection-preserving, and \cref{eq:lsg:pentagon} we obtain
\begin{equation*}
\psi(\chi_{ijk}|_{\gamma,\gamma'}(w \otimes v))=\frac{1}{\sqrt{\mathrm{rk}(\inf R_{jj})}} \sum_{l=1}^{n} w \otimes d_{ \gamma \pcomp (\overline{\gamma}\pcomp \gamma),\gamma}(\chi_{iij} |_{\overline{\gamma}\pcomp\gamma,\gamma}(v \otimes \chi_{iji}|_{\gamma,\overline{\gamma}}(\alpha_{ij}(v_l) \otimes v_l))\text{.}
\end{equation*}
Now we use \cref{re:lsg:unit} to obtain
\begin{equation*}
\psi(\chi_{ijk}|_{\gamma,\gamma'}(w \otimes v))=w \otimes  d_{ \gamma \pcomp (\overline{\gamma}\pcomp \gamma),\gamma}(\chi_{iij} |_{\overline{\gamma}\pcomp\gamma,\gamma}(v \otimes d_{\cp_x,\overline{\gamma} \pcomp  \gamma}(\epsilon_i(x))))\text{.}
\end{equation*}
Finally, by \cref{eq:lsg:unit}, the latter is equal to $w \otimes v$.
\end{proof}

\begin{corollary}
\label{lem:morita}
The  bimodule bundles $\inf R_{ij}$ and $\inf R_{ji}$ are invers to each other, in the sense that
there exist bimodules isomorphisms   
\begin{equation*}
\inf R_{ji}|_{\overline{\gamma}} \otimes_{\inf A_j|_y} \inf R_{ij}|_{\gamma} \cong \inf A_i|_x
\quand
\inf R_{ij}|_{\gamma} \otimes_{\inf A_i|_x} \inf R_{ji}|_{\overline{\gamma}} \cong \inf A_j|_y\text{,}
\end{equation*}
for every $\gamma\in P_{ij}$ with $x:=\gamma(0)$ and $y:=\gamma(1)$, forming bundle isomorphisms over $P_{ij}$. In particular, the bimodule $\inf R_{ij}|_{\gamma}$ establishes a Morita equivalence.  
\end{corollary}

\begin{proof}
\cref{lem:compositionofbimodules} provides bimodule isomorphisms
\begin{equation*}
\inf R_{ji}|_{\overline{\gamma}} \otimes_{\inf A_i|_x} \inf R_{ij}|_{\gamma} \cong \inf R_{ii}|_{\overline{\gamma}\pcomp\gamma}
\quand
\inf R_{ij}|_{\gamma} \otimes_{\inf A_j|_y} \inf R_{ji}|_{\overline{\gamma}} \cong \inf R_{jj}|_{\gamma\pcomp\overline{\gamma}}\text{,}
\end{equation*}
and by \cref{re:bimodule:b} we have bimodule isomorphisms $\inf R_{ii}|_{\overline{\gamma}\pcomp\gamma} \cong \inf R_{ii}|_{\cp_x}=\inf A_i|_x$  and $\inf R_{jj}|_{\gamma \pcomp \overline{\gamma}} \cong \inf R_{jj}|_{\cp_y}=\inf A_j|_y$. \end{proof}

\begin{remark}
\begin{enumerate}[(a)]

\item
Since all algebras $\inf A_i|_x$ are simple, it is clear that they are all Morita equivalent (to $\C$); the point is that  LBG provides a consistent choice of these Morita equivalences, parameterized by paths.

\item
\label{rem:alphaunique}
Inverses of invertible bimodules are unique up to unique intertwiners, and a canonical choice is to take the  complex conjugate vector space with swapped left and right actions. \Cref{re:bimodule:c} shows that $\alpha_{ij}:\inf R_{ij} \to \overline{\inf R_{ji}}$ is that unique intertwiner. In particular, if $(\inf L,\lambda, \inf R,\phi,\chi,\epsilon,\alpha)$ and $(\inf L,\lambda, \inf R,\phi,\chi,\epsilon,\alpha')$ are LBG objects, then $\alpha=\alpha'$.

\item 
\label{rem:faithfullybalanced}
By \cref{lem:Rijbimodule} we have morphisms between algebra bundles
\begin{equation*}
\ev_1^{*}\inf A_j \to \mathrm{End}_{\ev_0^{*}\inf A_i}(\inf R_{ij})
\quand
\ev_0^{*}\inf A_i \to (\mathrm{End}_{\ev_1^{*}\inf A_j}(\inf R_{ij}))^{op}
\end{equation*}
over $P_{ij}$, where $\mathrm{End}_{\ev_0^{*}\inf A_i}(\inf R_{ij})$ and $(\mathrm{End}_{\ev_1^{*}\inf A_j}(\inf R_{ij}))^{op}$ are bundles of endomorphisms of right $\ev_0^{*}\inf A_i$-module bundles and of left $\ev_1^{*}\inf A_j$-module bundles, respectively. Since the bimodule $\inf R_{ij}$ establishes a Morita equivalence, it is  faithfully balanced; i.e., the above algebra bundle homomorphisms are isomorphisms. 

\end{enumerate}

\end{remark}

The following statement combines the bimodules structure with the fusion representation, and will be useful in \cref{sec:regression}.

\begin{lemma}
\label{lem:fusionandbimodule}
The fusion representation commutes with the bimodule actions of \cref{lem:Rijbimodule}; more precisely, the diagrams
\begin{equation*}
\small
\alxydim{@C=2cm@R=3em}{\inf A_j|_y \otimes (\inf L|_{\gamma_1 \cup \gamma_2} \otimes \inf R_{ij}|_{\gamma_2}) \ar[r]^-{\id \otimes \phi_{ij}|_{\gamma_1,\gamma_2}} \ar[d]_{braid \otimes \id} &  \inf A_j|_y \otimes \inf R_{ij}|_{\gamma_1} \ar[dd]^{\lambda_{ij}|_{\gamma_1}} \\ \inf L|_{\gamma_1 \cup \gamma_2} \otimes( \inf A_j|_y  \otimes \inf R_{ij}|_{\gamma_2}) \ar[d]_{\id \otimes \lambda_{ij}|_{\gamma_2}} & \\  \inf L|_{\gamma_1\cup\gamma_2} \otimes \inf R_{ij}|_{\gamma_2} \ar[r]_-{\phi_{ij}|_{\gamma_1,\gamma_2}} &  \inf R_{ij}|_{\gamma_1} }
\end{equation*}
and
\begin{equation*}
\small
\alxydim{@C=2cm@R=3em}{\inf  L|_{\gamma_1 \cup \gamma_2} \otimes \inf R_{ij}|_{\gamma_2} \otimes \inf A_i|_x \ar[r]^-{\phi_{ij}|_{\gamma_1,\gamma_2} \otimes \id} \ar[d]_{\id \otimes \rho_{ij}|_{\gamma_2}} &  \inf R_{ij}|_{\gamma_1} \otimes \inf A_i|_x  \ar[d]^{\rho_{ij}|_{\gamma_1}} \\   \inf L|_{\gamma_1\cup\gamma_2} \otimes \inf R_{ij}|_{\gamma_2} \ar[r]_-{\phi_{ij}|_{\gamma_1,\gamma_2}} & \inf R_{ij}|_{\gamma_1} }
\end{equation*}
are commutative for all $(\gamma_1,\gamma_2)\in P_{ij}^{[2]}$ with $x := \gamma_1(0)=\gamma_2(0)$ and $y:=\gamma_1(1)=\gamma_2(1)$.
\end{lemma}

\begin{proof}
For the first diagram, we insert  $k=j$, $\gamma_{12}:=\gamma_1$, $\gamma_{12}':= \gamma_2$, $\gamma_{23}:=\gamma_{23}':=\cp_y$ into \cref{eq:lsg:comp} and use the section $\nu_{\cp_y}$ into $\inf L_{\gamma_{23} \cup \gamma_{23}'}$. As $\nu_{\cp_y}$ is neutral with respect to the fusion product and  the fusion representation (\cref{re:fusrepcan})  we obtain from  \cref{eq:lsg:comp} the commutativity of the diagram 
\begin{equation*}
\small
\alxydim{@C=2cm@R=3em}{ \inf L|_{\gamma_{1} \cup \gamma_{2}} \otimes \inf R_{jj}|_{\cp_y} \otimes \inf R_{ij}|_{\gamma_{2}} \ar[r]^-{d \otimes \chi_{ijj}|_{\gamma_{12}',\cp_y}}  \ar[d]_{ braid \otimes \id} & \inf L|_{(\cp_y \pcomp \gamma_{1}) \cup (\cp_y \pcomp \gamma_{2})} \otimes \inf R_{ij}|_{\cp_y \pcomp \gamma_{2}}  \ar[dd]^{\phi_{ij}|_{\cp_y \pcomp \gamma_{1} ,\cp_y \pcomp \gamma_{2}}} \\ \inf R_{jj}|_{\cp_y} \otimes \inf L|_{\gamma_{1} \cup \gamma_{2}} \otimes \inf R_{ij}|_{\gamma_{2}} \ar[d]_{\id\otimes \phi_{ij}|_{\gamma_{1},\gamma_{2}}}  & \\ \inf R_{jj}|_{\cp_y} \otimes \inf R_{ij}|_{\gamma_{1}} \ar[r]_{\chi_{ijj}|_{\gamma_{1},\cp_y}} & \inf R_{ij}|_{\cp_y \pcomp \gamma_{1}\text{.}} }
\end{equation*}
Using the obvious reparameterizations and the superficial connection, as well as the fact that $\phi_{ij}$ is connection-preserving, this shows the commutativity of the first diagram. The second diagram follows analogously.
\end{proof}

\subsection{Reduction to the point}

\label{sec:mooresegal}

One important insight of the Stolz-Teichner programme \cite{stolz1} is that  \emph{classical} field theories (with a target space) and \emph{quantum} field theories   fit into the same framework of functorial field theories, in such a way that the target space of a quantum theory is a point. Under a conjectural  identification between certain types of field theories  and generalized cohomology theories, quantization is the pushforward to the point in that cohomology theory. 

As we will show in \cite{Bunk2018}, our LBG is precisely the data for a  2-dimensional open-closed topological field theory with target space $M$. On the other hand, Lazaroiu, Lauda-Pfeiffer and Moore-Segal determined the data for a 2-dimensional open-closed topological
 \emph{quantum} field theory \cite{Lazaroiu2001,Moore2006,Lauda2008}, resulting in a structure called an $I$-colored knowledgable Frobenius algebra in \cite{Lauda2008}. In this subsection we prove that both data are consistent in the sense that in the case $M=\{\ast\}$   our LBG reduces to an $I$-colored knowledgable Frobenius algebra (with  additional structure and properties). Loosely speaking, LBG is a family of  $I$-colored knowledgable Frobenius algebras. We start with the following simple reduction of LBG to a point:

\begin{lemma}
\label{lem:LBGstrpt}
Consider the target space $(M,Q)$ with $M=\{\ast\}$ and $Q=\{\ast\}_{i\in I}$. Then, the category $\lsg(M,Q)$ is canonically equivalent to a category $\lsg^{(I)}$ defined as follows. An object is a septuple $(\inf L,\lambda, \inf R,\phi,\chi,\epsilon,\alpha)$ consisting of  the following structure:
\begin{enumerate}[\normalfont (1$^\ast$),leftmargin=*]

\item
\label{LBGstrpt:1}
\label{LBGstrpt:2}
A complex inner product space $\inf L$ together with 
a unitary isomorphism
$\lambda: \inf L \otimes \inf L \to \inf L$
such that $(\inf L,\lambda)$ is a commutative algebra.

\item
\label{LBGstrpt:3}
A family $\inf R=\{\inf R_{ij}\}_{i,j\in I}$ of complex inner product spaces.

\item 
\label{LBGstrpt:4}
A family $\phi=\{\phi_{ij}\}_{i,j\in I}$ of  unitary isomorphisms
$\phi_{ij}: \inf L\otimes \inf R_{ij} \to \inf R_{ij}$,
forming a representation of $(\inf L,\lambda)$ on $\inf R_{ij}$. 

\item
\label{LBGstrpt:5}
\label{LBGstrpt:6}
A family $\chi=\{\chi_{ijk}\}_{i,j,k\in I}$ of associative, linear maps 
$\chi_{ijk}: \inf R_{jk} \otimes \inf R_{ij} \to \inf R_{ik}$, 
and a family $\epsilon=\{\epsilon_i\}_{i\in I}$ of   elements $\epsilon_i\in\inf R_{ii}$ that are neutral with respect to $\chi_{iij}$ and $\chi_{ijj}$.

\item
\label{LBGstrpt:7}
A family $\alpha=\{\alpha_{ij}\}_{i,j\in I}$ of  unitary isomorphisms
$\alpha_{ij}: \inf R_{ij} \to \overline{\inf R_{ji}}$
that are unital ($\alpha_{ii}(\epsilon_i)=\epsilon_i$), involutive $(\alpha_{ji} \circ \alpha_{ij}=\id_{\inf R_{ij}})$ and anti-multiplicative, i.e.  $\chi_{kji}(\alpha_{ij}(v) \otimes \alpha_{jk}(w))=\alpha_{ik}(\chi_{ijk}(w \otimes v))$ for all $v\in \inf R_{ij}$, $w\in \inf R_{jk}$.

\end{enumerate}
This structure has to satisfy the following axioms:
\begin{enumerate}[\normalfont (LBG1$^{\ast}$),leftmargin=*]

\item
\label{eq:lsgpt:comp}
The following diagram is commutative:
\begin{equation*}
\small
\alxydim{@C=2cm@R=3em}{\inf L \otimes \inf L \otimes \inf R_{jk} \otimes \inf R_{ij} \ar[r]^-{\lambda \otimes \chi_{ijk}}  \ar[d]_{\id \otimes braid \otimes \id} & \inf L \otimes \inf R_{ik}  \ar[dd]^{\phi_{ik}} \\ \inf L\otimes \inf R_{jk} \otimes \inf L \otimes \inf R_{ij} \ar[d]_{\phi_{jk} \otimes \phi_{ij}}  & \\ \inf R_{jk} \otimes \inf R_{ij} \ar[r]_{\chi_{ijk}} & \inf R_{ik} }
\end{equation*}

\item
\label{eq:lsgpt:invariance}
For any elements $v,w\in \inf R_{ij}$ we have
\begin{equation*}
h_{ii}(\chi_{iji}(\alpha_{ij}(w)\otimes v) , \epsilon_i)=h_{ij}(v,w) = h_{jj}(\epsilon_j,\chi_{jij}(w \otimes \alpha_{ij}(v))) \text{.}
\end{equation*}

\item
\label{eq:lsgpt:symmetrizingfusion}
The following diagram, where  $\tilde\lambda$ is defined as in \cref{LBGstr:2*}, is commutative:
\begin{equation*}
\small
\alxydim{@R=3em}{\inf L \otimes \inf R_{ij} \ar[d]_{\tilde\lambda \otimes \alpha_{ij}}\ar[r]^-{\phi_{ij}} & \inf R_{ij} \ar[d]^{\alpha_{ij}} \\ \overline{\inf L} \otimes \overline{\inf R}_{ji} \ar[r]_-{\phi_{ji}} & \overline{\inf R_{ji}} }
\end{equation*}

\item
\label{eq:lsgpt:cardy}
For  every orthonormal basis  $(v_1,...,v_n)$  of $\inf R_{ij}$ and every $v\in \inf R_{ii}$ we have
\begin{equation*}
\sum_{k=1}^{n} \chi_{jij}(\chi_{iij}(v_k \otimes v) \otimes \alpha_{ij}(v_k))=h_{ii}(\epsilon_i,v)\cdot  \epsilon_j\text{.}
\end{equation*}
\end{enumerate}
Finally, a morphism in $\lsg^{(I)}$ is a pair $(\varphi,\xi)$, consisting of a  unitary algebra isomorphism $\varphi:\inf L\to \inf L'$  and of a family $\xi=\{\xi_{ij}\}_{i,j\in I}$ of unitary  isomorphisms $\xi_{ij}:\inf R_{ij} \to \inf R_{ij}'$ that are compatible with the remaining structure in the obvious way. 
\qed
\end{lemma}

\begin{remark}
\label{re:lambdatildept}
The unit-length section $\nu$ of \cref{LBGstr:2*} yields a unit-length element $1\in \inf L$ that is neutral with respect to $\lambda$. 
By definition, the map $\tilde \lambda: \inf L \to \inf L^{*}$ mentioned in \cref{LBGstr:2*} is given by
\begin{equation*}
\tilde\lambda(\ell)(\ell')\cdot 1 =\lambda(\ell\otimes \ell')\text{.}
\end{equation*} 
Using that $\{1\}$ is an orthonormal basis of $\inf L$ and $\lambda$ is unitary, one can show that $\tilde\lambda$ (under the isomorphism $\inf L^{*} \cong \overline{\inf L}$) is given by
\begin{equation*}
\tilde\lambda:\inf L \to \overline{\inf L}: \ell \mapsto h(\ell,1)\cdot 1\text{.}
\end{equation*}
We find that $\tilde \lambda$ is a unital, involutive, algebra homomorphism.   
The map $\C \to \inf L: z \mapsto z\cdot 1$ is a unitary algebra isomorphism, and under this algebra homomorphism, $\tilde\lambda$ becomes  complex conjugation. 
\end{remark}

Our goal is to compare the category $\lsg^{(I)}$  with the following category, which in turn is equivalent to the category of topological open-closed quantum field theories with  boundary labels in $I$, valued in the symmetric monoidal category of finite-dimensional complex vector spaces \cite[Theorem 5.6]{Lauda2008} and \cite{Moore2006}. The relationship to quantum field theories will be further elaborated in \cite{Bunk2018}. 

\begin{definition}
\label{def:kFrob}
Let $I$ be a set.
An \emph{$I$-colored knowledgable Frobenius algebra} is a septuple $(\inf L,\inf R,\chi,\epsilon,\theta,\iota,\iota^{*})$ consisting of the following structure:
\begin{enumerate}[({CFa}1),leftmargin=*]

\item 
\label{CFa:1}
$\inf L$ is a  commutative  Frobenius algebra whose trace will be denoted by   $\vartheta$.

\item
\label{CFa:2}
$\inf R=\{\inf R_{ij}\}_{i,j\in I}$ is a family of finite-dimensional complex vector spaces.

\item
\label{CFa:3}
$\chi=\{\chi_{ijk}\}_{i,j,k\in I}$ is a family of linear maps $\chi_{ijk}: \inf R_{jk} \otimes \inf R_{ij} \to \inf R_{ik}$ satisfying an associativity condition for four indices, and $\epsilon=\{\epsilon_i\}_{i\in I}$ is a family of elements $\epsilon_i\in \inf R_{ii}$ that are neutral with respect to $\chi$. In particular, $\inf R_{ii}$ is an algebra.

\item
\label{CFa:4}
$\theta=\{\theta_i\}_{i\in I}$ is a family of linear maps $\theta_i: \inf R_{ii} \to \C$. 

\item
\label{CFa:5}
$\iota=\{\iota_i\}_{i\in I}$ is a family of  algebra homomorphisms   $\iota_i: \inf L \to \inf R_{ii}$ which are central in the sense that $\chi_{iij}(v \otimes \iota_i(\ell))=\chi_{ijj}(\iota_j(\ell) \otimes v)$ for all $\ell \in \inf L$ and $v\in \inf R_{ij}$.

\item
\label{CFa:6}
$\iota^{*}=\{\iota_i^{*}\}_{i\in I}$ is a family of linear maps $\iota_i^{*}: \inf R_{ii} \to \inf L$ which are adjoint to $\iota_i$ in the sense that $\vartheta(\ell\cdot \iota_{i}^{*}(v))=\theta_i(\chi_{iii}(\iota_i(\ell) \otimes v))$ for all $v\in \inf R_{ii}$ and $\ell \in \inf L$.

\end{enumerate}
This structure defines a pairing $\sigma_{ij}$ by
\begin{equation}
\label{eq:pairing}
\alxydim{}{\inf R_{ji} \otimes \inf R_{ij} \ar[r]^-{\chi_{iji}} & \inf R_{ii} \ar[r]^{\theta_i} & \C}
\end{equation}
which is supposed to satisfy the following three axioms:
\begin{enumerate}[({CFa}1),leftmargin=*]

\setcounter{enumi}{6}

\item 
\label{CFa:7}
$\sigma_{ij}$ is non-degenerate, i.e. the induced map $\Phi_{ij}:\inf R_{ij} \to \inf R_{ji}^{*}$ is bijective.

\item
\label{CFa:8}
$\sigma_{ij}$ and $\sigma_{ji}$ are symmetric: $\sigma_{ij}(a \otimes b)=\sigma_{ji}(b \otimes a)$ for all $b\in \inf R_{ij}$ and $a\in \inf R_{ji}$.

\item
\label{CFa:9}
If $(v_1,...,v_n)$ is a basis of $\inf R_{ij}$, $(v^1,...,v^n)$ is the dual basis of $\inf R_{ji}$ with respect to $\sigma_{ij}$, and $v\in \inf R_{ii}$, then
\begin{equation*}
\displaystyle (\iota_j \circ \iota_{i}^{*})(v)=\sum_{k=1}^{n}\chi_{jij}(\chi_{iij}(v_{k} \otimes v) \otimes v^{k})\text{.}
\end{equation*}  

\end{enumerate}
\end{definition}

Note that \cref{CFa:7*} and \cref{CFa:8*} imply that $\inf R_{ii}$ is a symmetric Frobenius algebra. 
A \emph{homomorphism} between $I$-colored knowledgable Frobenius algebras is a pair $(\varphi,\xi)$, consisting of a Frobenius algebra homomorphism $\varphi:\inf L \to \inf L'$ and of a family $\xi=\{\xi_{ij}\}_{i,j\in I}$ of linear maps $\xi_{ij}:\inf R_{ij} \to \inf R'_{ij}$ that respect products, units, and traces in the obvious sense, and satisfy $\iota'_{i}\circ \varphi = \xi_{ii}\circ \iota_i$ and $\iota'^{*}_{i}\circ \xi_{ii} = \varphi \circ \iota_i^{*}$.
The category of $I$-colored knowledgable Frobenius algebras is denoted by $\kfrob$. 

\begin{remark}
\cref{def:kFrob} was described first in \cite[Section 2.2]{Moore2006}, where \cref{CFa:9*} was related to the Cardy condition.   
The terminology \quot{$I$-colored knowledgable Frobenius algebra} was coined in \cite[Def. 5.1]{Lauda2008} for a slightly different but equivalent structure -- the equivalence is established by an \quot{$I$-colored version} of  the well-known equivalent ways to define  a Frobenius algebra (one by a non-degenerate inner product and the other one by a co-product, see \cite[Prop. 2.3.22 \& 2.3.24]{Kock2003}).
\end{remark}

Moore and Segal prove the following theorem \cite[Theorem 2]{Moore2006}; their proof applies without changes.

\begin{theorem}
\label{th:simple}
Let $(\inf L,\inf R,\chi,\epsilon,\theta,\iota,\iota^{*})$ be an $I$-colored knowledgable Frobenius algebra with $\inf L$ one-dimensional. Then,  the Frobenius algebra $\inf R_{ii}$ is simple, for every $i\in I$. 
\qed
\end{theorem}

In the following we construct a functor
\begin{equation}
\label{eq:functorkfrob}
\tqft:\lsg^{(I)} \to  \kfrob\text{.}
\end{equation}
Suppose  $(\inf L,\lambda, \inf R,\phi,\chi,\epsilon,\alpha)$ is an object in $\lsg^{(I)}$. Thus, $\inf L$ is a commutative algebra.  The hermitian metric $h$ on $\inf L$ induces a non-degenerate trace $\vartheta(\ell) := h(1,\ell)$, making $\inf L$  into a commutative  Frobenius algebra; this is \cref{CFa:1*}. The vector spaces $\inf R_{ij}$, products $\chi_{ijk}$, and units $\epsilon_i$  give \cref{CFa:2*,CFa:3*}. The trace $\theta_i$ is the one given in \cref{eq:trace:star}, i.e. $\theta_i(v):=h_{ii}(\epsilon_i,v)$; this gives \cref{CFa:4*}. The linear map    $\iota_i: \inf L \to \inf R_{ii}$ is the string opening morphism of \cref{re:closingandopeningstrings}, i.e. $\iota_i(\ell) := \phi_{ii}  (\ell \otimes \epsilon_i)$. In order to check \cref{CFa:5*} we have to prove that $\iota_i$ has the following properties:
\begin{enumerate}[(a)]

\item 
It is an algebra homomorphism. 
This follows by reducing \cref{eq:lsgpt:comp*} to $i=j=k$, inserting $\epsilon_i \in \inf{R}_{ii}$, and using that $\chi_{iii}(\epsilon_i \otimes \epsilon_i)=\epsilon_i$. 

\item
It is unital. This follows from \cref{re:fusrepcan}. 

\item
It is central. This follows by reducing  \cref{eq:lsgpt:comp*} to the cases $i=j$ and $j=k$, and by using the neutrality of $\epsilon_i$.

\end{enumerate} 
The linear map $\iota_i^{*}:\inf R_{ii} \to \inf L$ is the string closing morphism, i.e.  $\iota_i^{*}(v) :=1\cdot\theta_i(v)$. In order to check  \cref{CFa:6*} we show that it is adjoint to $\iota_i$: 
\begin{align*}
\vartheta(\iota_{i}^{*}(v)\cdot \ell) 
= h_{ii}(\epsilon_i,h(1,\ell)\cdot v)
&= h_{ii}(\epsilon_i,\alpha_{ii}(\phi_{ii}(h(\ell,1)\cdot 1 \otimes \alpha_{ii}(v))))
\\&\eqcref{re:lambdatildept} h_{ii}(\epsilon_i,\alpha_{ii}(\phi_{ii}(\tilde\lambda(\ell) \otimes \alpha_{ii}(v))))
\\&\eqcref{eq:lsgpt:symmetrizingfusion*}h_{ii}(\epsilon_i,\phi_{ii}(\ell \otimes v))
\\&\eqcref{eq:lsgpt:comp*}h_{ii}(\epsilon_i,\chi_{iii}(v \otimes \phi_{ii}(\ell \otimes \epsilon_i)))
=\theta_i(\chi_{iii}(v\otimes \iota_i(\ell)))  \text{.}
\end{align*}
Now it remains to verify the axioms that involve the pairings $\sigma_{ij}$. For \cref{CFa:7*}, the non-degeneracy of $\sigma_{ij}$, we show that the corresponding map $\Phi_{ij}$ coincides with the isomorphism $\alpha_{ij}$ under the isomorphism $()^{\flat}:\overline{\inf R_{ji}} \to \inf R_{ji}^{*}$, i.e.
we show that  $\Phi_{ij}(v)=\alpha_{ij}(v)^{\flat}$ for all $v\in \inf R_{ij}$. To that end, let $w\in \inf R_{ji}$. We compute: 
\begin{align*}
\Phi_{ij}(v)(w)&= h_{ii}(\epsilon_i,\chi_{iji}(w \otimes v))
\\[-2em]&= h_{ii}(\chi_{iji}(\alpha_{ij}(v) \otimes \alpha_{ji}(w))),\epsilon_i)
\eqcref{eq:lsgpt:invariance*} h_{ij}(\alpha_{ji}(w),v)
= \overline{h_{ji}(w,\alpha_{ij}(v))}
=\alpha_{ij}(v)^{\flat}(w)\text{.}
\end{align*}
The next calculation verifies \cref{CFa:8*}, the symmetry between $\sigma_{ij}$ and $\sigma_{ji}$:
\begin{align*}
\sigma_{ij}(w \otimes v)
=h_{ii}(\chi_{iji}(\alpha_{ij}(v) \otimes \alpha_{ji}(w)) , \epsilon_i)
\eqcref{eq:lsgpt:invariance*} h_{jj}(\epsilon_j,\chi_{jij}(v \otimes w))=\sigma_{ji}(v \otimes w)\text{.}
\end{align*}
Finally, we verify the Cardy condition \cref{CFa:9*}.
Suppose $(v_1,...,v_n)$ is an orthonormal basis  of $\inf R_{ij}$ with respect to $h_{ij}$, and suppose $(v^1,...,v^n)$ is the dual basis of $\inf R_{ji}$ with respect to $\sigma_{ij}$. We have already seen that $\Phi_{ij}=()^{\flat} \circ \alpha_{ij}$; this means that $v^{k}=\alpha_{ij}(v_k)$ for all $k=1,...,n$.
We get for $v\in \inf R_{ii}$:
\begin{equation*}
\iota_j(\iota_{i}^{*}(v))=\iota_j(1\cdot\theta_i(v)) = \theta_i(v)\cdot \epsilon_j =h_{ii}(\epsilon_i,v)\cdot  \epsilon_j\text{.}
\end{equation*}
Now, \cref{eq:lsgpt:cardy*} implies \cref{CFa:9*}. This finishes the proof that $(\inf L,\inf R,\chi,\epsilon,\theta,\iota,\iota^{*})$ is an $I$-colored knowledgable Frobenius algebra.

In order to define the functor $\tqft$ on the level of morphisms, we consider  a morphism $(\varphi,\xi)$ in $\lsg^{(I)}$. We claim that the same maps $\varphi:\inf L \to \inf L'$ and $\xi_{ij}:\inf R_{ij} \to \inf R_{ij}'$  form a morphism of the corresponding objects in $\kfrob$. First of all, the unitarity of $\varphi$ implies that  $\vartheta' \circ \varphi=\vartheta$, i.e. $\varphi$ is a Frobenius algebra homomorphism. Second, by assumption, $\xi_{ij}$ respects products and units, and since $\xi_{ij}$ is unitary, it also respects the traces $\theta_{i}$. Proving the identities
$\iota'_{i}\circ \varphi = \xi_{ii}\circ \iota_i$ and $\iota'^{*}_{i}\circ \xi_{ii} = \varphi \circ \iota_i^{*}$ is straightforward.
This completes the definition of the functor $\tqft$.

Now we are in position to deliver the remaining part of the proof of \cref{prop:frob}, namely  that the algebras $\inf A_i|_x$, for $x\in Q_i$ and $i\in I$, are simple.
Restricting a given LBG object to $x$ yields an object in $\lsg^{(I)}$ with $\inf L:=\inf L|_{\cp_x}$ and $\inf R_{ii} := \inf R_{ii}|_{c_x}=\inf A_i|_x$. Under the        functor $\mathcal{F}$, it becomes an $I$-colored knowledgable Frobenius algebra with $\inf L$ one-dimensional. Applying \cref{th:simple}, we obtain:

\begin{corollary}
\label{co:aisimple}
The algebras $\inf A_i|_x$ obtained from LBG in \cref{sec:algebrabundle:lsg} are all simple.  
\end{corollary}

We continue studying the functor $\tqft$. 
It is obviously faithful, but neither full (i.e., non-isomorphic LBG may give isomorphic $I$-colored knowledgable Frobenius algebras), nor essentially surjective. In order to properly understand these phenomena, we lift the functor $\tqft$ to a new category $\rpkfrob$, where the $I$-colored knowledgable Frobenius algebras are equipped with so-called \emph{positive reflection structures}, which we introduce next.
We will discuss in \cite{Bunk2018} the relation to positive reflection structures in functorial field theories, in the sense of Freed-Hopkins \cite{Freed2016}.

\begin{definition}
A \emph{reflection structure} on an $I$-colored knowledgable Frobenius algebra $(\inf L,\inf R,\chi,\epsilon,\theta,\iota,\iota^{*})$ is a pair $(\tilde\lambda,\alpha)$ consisting of
an involutive  algebra isomorphism $\tilde\lambda: \inf L \to \overline{\inf L}$ and of a family $\alpha=\{\alpha_{ij}\}_{i,j\in I}$ of  involutive (i.e., $\alpha_{ji} \circ \alpha_{ij}=\id$), anti-multiplicative isomorphisms $\alpha_{ij}: \inf R_{ij} \to \overline{\inf R_{ji}}$ 
such that the conditions 
\begin{equation}
\label{eq:refscomp}
\vartheta(\tilde\lambda(\ell))=\overline{\vartheta(\ell)}
\quomma
\alpha_{ii}(\epsilon_i)=\epsilon_i
\quomma
\theta_i(\alpha_{ii}(v))=\overline{\theta_i(v)}
\quand
\alpha_{ii} \circ \iota_i = \iota_i \circ \tilde\lambda
\end{equation}
are satisfied for all $i\in I$.
A reflection structure is called \emph{positive} if the sesquilinear pairings
\begin{equation*}
(v,w) \mapsto \sigma_{ij}(\alpha_{ji}^{-1}(v) \otimes w)
\quand
(\ell,\ell') \mapsto \vartheta(\tilde\lambda^{-1}(\ell)\cdot \ell')
\end{equation*}
on $\inf R_{ij}$ and $\inf L$, respectively, are positive-definite for all $i,j\in I$. 
\end{definition}

We remark that the last equation in \labelcref{eq:refscomp} implies the analogous condition  $\tilde\lambda\circ \iota_i^{*} = \iota_i^{*} \circ \alpha_{ii}$ for $\iota^{*}$, due to the adjointness in  \cref{CFa:6*} and the non-degeneracy of $\theta$. 
A homomorphism $(\varphi,\xi)$ between $I$-colored knowledgable Frobenius algebras is called \emph{reflection-preserving}, if $\tilde\lambda' \circ \varphi = \varphi \circ \tilde\lambda$ and $\alpha_{ij}'\circ \xi_{ij} = \xi_{ji} \circ \alpha_{ij}$.  
The category of $I$-colored knowledgable Frobenius algebras with positive reflection structures is denoted by $\rpkfrob$. There is an obvious forgetful functor $\rpkfrob \to \kfrob$.
Next, we construct a lift of the functor $\tqft$ to $\rpkfrob$,
\begin{equation*}
\small
\alxydim{@R=3em}{&\rpkfrob \ar[d] \\\lsg^{(I)} \ar@/^1pc/[ur]^-{\tqftlift} \ar[r]_-{\tqft} & \kfrob\text{.}}
\end{equation*}

In order to define the lift $\tqftlift$, we have to define positive reflection structures on the $I$-colored knowledgable Frobenius algebras in the image of $\tqft$. The isomorphisms $\tilde\lambda$ and $\alpha_{ij}$ are the given ones. \Cref{re:lambdatildept} shows that $\tilde\lambda$ has the required properties; the properties of $\alpha_{ij}$ follow directly from \cref{LBGstrpt:7*}. The compatibility condition \cref{eq:refscomp} follows from the definition of $\iota$ and \cref{eq:lsgpt:symmetrizingfusion*}. The pairing reproduces exactly the given hermitian metric on $\inf R_{ij}$, \begin{equation*}
(v,w) \mapsto \theta_i(\chi_{iji}(\alpha_{ij}(v) \otimes w))=h_{ii}(\epsilon_i,\chi_{iji}(\alpha_{ij}(v) \otimes w))\eqcref{eq:lsgpt:invariance*} h_{ji}(\alpha_{ij}(w),\alpha_{ij}(v))=h_{ij}(v,w)\text{,}
\end{equation*} 
and is hence positive definite.
Similarly, the pairing on $\inf L$ reproduces  the metric $h$:
\begin{equation*}
(\ell,\ell') \mapsto \vartheta(\tilde\lambda(\ell)\cdot \ell') = h(1,\tilde\lambda(\ell) \cdot \ell')=\overline{h(\tilde\lambda(\ell),1)} \cdot h(1,\ell')=h(\ell,1)\cdot h(1,\ell')=h(\ell,\ell')\text{.}
\end{equation*}
Finally, we observe that the morphisms of $\lsg^{(I)}$ respect  the reflection structure by definition. This completes the definition of the lift $\tqftlift$.

\begin{proposition}
\label{prop:rpkfrob}
The lifted functor  $\tqftlift$  is full and faithful, and  surjective onto those  objects with $\inf L \cong \C$ as  Frobenius algebras (with $\id_{\C}$ as the trace on $\C$). In particular, it induces an  equivalence
\begin{equation*}
\lsg^{(I)} \cong \rpkfrob_{\C}\text{,}
\end{equation*}
where $\rpkfrob_{\C}$ is the full subcategory on all objects with $\inf L=\C$ as  Frobenius algebras.
\end{proposition}

\begin{proof}
It is obvious that $\tqftlift$  is faithful. In order to show that it is full, we consider two objects $(\inf L,\lambda, \inf R,\phi,\chi,\epsilon,\alpha)$ and $(\inf L',\lambda', \inf R',\phi',\chi',\epsilon',\alpha')$ in $\lsg^{(I)}$, and consider a morphism $(\varphi,\xi)$ between the corresponding $I$-colored knowledgable Frobenius algebras with reflection structures. The compatibility of $\varphi$ and $\xi_{ij}$ with the involutions $\tilde\lambda$ and $\alpha_{ij}$ follows because $(\varphi,\xi)$ is reflection-preserving.  
It remains to show that $\varphi$ and $\xi_{ij}$ are unitary, and that the fusion representations are preserved. We have seen above that the metrics $h$ and $h_{ij}$ are determined by the traces, the multiplications, and the involutions, which are all preserved by $\varphi$ and $\xi_{ij}$; this implies  unitarity. 
Concerning the fusion representation,  \cref{eq:lsgpt:comp*} implies that $\phi_{ij} = \chi_{ijj} \circ (\iota_j \otimes \id_{\inf R_{ij}})$.
Thus, the fusion representation $\phi_{ij}$ is determined by the product $\chi_{ijj}$ and the algebra homomorphism $\iota_j$. Since these are preserved by $\xi_{ij}$ and $\varphi$, the fusion representation is preserved. 
This completes the proof that $\smash{\tqftlift}$ is full.

Now we assume that $(\inf L,\inf R,\chi,\epsilon,\theta,\iota,\iota^{*})$ is an $I$-colored knowledgable Frobenius algebra with a positive reflection structure $(\tilde\lambda,\alpha_{ij})$, such that $\inf L \cong \C$ as Frobenius algebras. First of all, we remark that if $\inf L\cong \C$, preservation of traces implies that there is only one such isomorphism, namely the trace $\vartheta$ itself. In particular, $\vartheta:\inf L\to \C$ is a unital algebra isomorphism. We equip $\inf L$ with the sesquilinear form $h(\ell,\ell') := \vartheta(\tilde\lambda(\ell)\cdot \ell')$.  Using that $\vartheta$ is  involutive and an algebra homomorphism, we get $h(\ell,\ell')=\overline{\vartheta(\ell)} \cdot \vartheta(\ell')$; this shows  that $h$ is a hermitian metric and that the product $\lambda$ of $\inf L$ is unitary. 
We equip $\inf R_{ij}$ with the sesquilinear form $h_{ij}(v,w) := \theta_i(\chi_{iji}(\alpha_{ji}^{-1}(v) \otimes w))$. Using that $\alpha_{ij}$ is involutive, anti-multiplicative and compatible with $\theta_i$, we see that $h_{ij}$  is hermitian. It is non-degenerate, since the pairing $\sigma_{ij}$ is non-degenerate and $\alpha_{ij}$ is an isomorphism. Finally, it is positive-definite since the reflection structure is positive.
The involutions $\alpha_{ij}$ are unitary with respect to the metrics $h_{ij}$ and $\overline{h_{ji}}$: we have
\begin{align*}
h_{ji}(\alpha_{ij}(v),\alpha_{ij}(w)) = \sigma_{ji}(v \otimes \alpha_{ij}(w))\eqcref{CFa:8*}\sigma_{ij}(\alpha_{ij}(w) \otimes v)=h_{ij}(w,v)=\overline{h_{ij}(v,w)}\text{.}
\end{align*}
We define the fusion representation by $\phi_{ij} := \chi_{ijj} \circ (\iota_j \otimes \id_{\inf R_{ij}})$. This is a representation 
because $\chi_{ijk}$ are associative and $\iota_i$ is an algebra homomorphism. Further, $\phi_{ij}$
is unitary:
\begin{align*}
h_{ij}(\phi_{ij}(\ell\otimes v),\phi_{ij}(\ell'\otimes v'))
&\eqcref{eq:refscomp}\theta_i(\chi_{iji}(\chi_{jii}(\iota_i(\tilde\lambda(\ell))\otimes \alpha_{ji}^{-1}(v))\otimes \chi_{ijj}(\iota_j(\ell')\otimes v')))
\\&\eqcref{CFa:5*} \theta_i(\chi_{iii}(\chi_{iii}(\iota_i(\tilde\lambda(\ell))\otimes\iota_i(\ell') )\otimes \chi_{iji}(\alpha_{ji}^{-1}(v)\otimes v')))
\\&= \theta_i(\chi_{iii}(\iota_i(\tilde\lambda(\ell)\cdot \ell')\otimes \chi_{iji}(\alpha_{ji}^{-1}(v)\otimes v')))
\\&\eqcref{CFa:6*} \vartheta(\tilde\lambda(\ell)\cdot \ell'\cdot \iota_i^{*} (\chi_{iji}(\alpha_{ji}^{-1}(v)\otimes v')))
\\&\eqtext[5cm]{$\vartheta$ is an algebra homomorphism}\vartheta(\tilde\lambda(\ell)\cdot \ell')\cdot \vartheta(1\cdot \iota_i^{*} (\chi_{iji}(\alpha_{ji}^{-1}(v)\otimes v')))
\\&\eqcref{CFa:6*}h(\ell, \ell')\cdot \theta_i(\chi_{iii}(\epsilon_i \otimes \chi_{iji}(\alpha_{ji}^{-1}(v)\otimes v')))
\\&= h(\ell,\ell')\cdot h_{ij}(v,v')\text{.}
\end{align*}     
Now we have provided the structure \cref{LBGstrpt:1*,LBGstrpt:3*,LBGstrpt:4*,LBGstrpt:5*,LBGstrpt:7*} of an object of $\lsg^{(I)}$. It remains to check the axioms. Axioms \cref{eq:lsgpt:comp*,eq:lsgpt:symmetrizingfusion*} follow from the definition of $\phi_{ij}$, the centrality of $\iota_i$ in \cref{CFa:5*} and its compatibility with $\alpha$ in \cref{eq:refscomp}. The first part of \cref{eq:lsgpt:invariance*} follows from the associativity of $\chi_{ijk}$ and its compatibility with $\alpha_{ij}$, and the second part additionally from the symmetry of the pairing $\sigma_{ij}$ in \cref{CFa:8*}.
Finally, for Axiom \cref{eq:lsgpt:cardy*} we have seen above that $v^{k}=\alpha_{ij}(v_k)$ for a basis $(v_1,...,v_n)$  of $\inf R_{ij}$ and its dual basis $(v^1,...,v^n)$ with respect to $\sigma_{ij}$. With \cref{CFa:9*}, it thus remains to prove that
\begin{equation*}
h_{ii}(\epsilon_i,v)\cdot  \epsilon_j=(\iota_j \circ \iota_{i}^{*})(v)\text{.}
\end{equation*}
To see this, we first note that $\vartheta(\iota_i^{*}(v))=\theta_i(v)$ due to the adjointness in \cref{CFa:6*}.
Since $\vartheta$ is a unital algebra isomorphism, this implies $\iota_i^{*}(v)=\theta_i(v)\cdot 1$. Then, we obtain $(\iota_j \circ \iota_{i}^{*})(v)=\theta_i(v) \cdot \epsilon_j$, which coincides with $h_{ii}(\epsilon_i,v)\cdot  \epsilon_j$. Summarizing, we have constructed an object of $\lsg^{(I)}$, which is (by construction) sent by the functor $\tqftlift$ to the $I$-colored knowledgable Frobenius algebra with reflection structure we started with.    
\end{proof}

\Cref{prop:rpkfrob} shows that by restriction of LBG to a point one obtains all  reflection-positive open-closed topological quantum field theories whose bulk algebra is $\C$. This will be further investigated in our upcoming paper \cite{Bunk2018}.

\setsecnumdepth{2}

\section{Transgression}

\label{sec:transgression}

In this section we construct our transgression functor $\mathscr{T}: \tsg(M,Q) \to \lsg(M,Q)$. In  \cref{sec:trans:gerbes,sec:transbranes,sec:trans:superficialconnection,sec:fusionrep,sec:liftpathcommp,sec:liftedconstandrev,sec:trans:liftedpathreversal} we describe its action on the level of objects, i.e. the transgression of a TBG object $(\mathcal{G},\mathcal{E})$. In \cref{sec:nattrans} we treat the morphisms, and in \cref{sec:trivialgerbe} we consider the situation where the bundle gerbe $\mathcal{G}$ is trivial.

As announced in the introduction, we will treat analytical regularity over loop spaces and path spaces in the framework of diffeology. The reason is that we use extensively the concatenation of arbitrary paths, whenever they have a common end point. This requires sitting instants (i.e., the map $\gamma:[0,1] \to M$ is constant around $\{0\}$ and $\{1\}$). Spaces of paths with sitting instants are not manifolds in any way; this forces us to use diffeological spaces. An introduction to diffeology can be found in \cite{iglesias1,baez6}. A systematic application to spaces of paths and loops in the context of parallel transport has been pursued in \cite{baez3,schreiber2}, and in the context of transgression in \cite{waldorf9,waldorf10}. For the sake of self-containedness, we have included an \cref{sec:diffvectorbundle} about vector bundles over diffeological spaces.

\subsection{The line bundle over the loop space}

\label{sec:trans:gerbes}

The transgression of  gerbes with connection to the loop space was described first by Gaw\c edzki \cite{gawedzki3} (in terms of Deligne cohomology) and Brylinski \cite{brylinski1} (in terms of Dixmier-Douady sheaves of groupoids). A transgression for bundle gerbes was described by Gaw\c edzki and Reis in \cite{gawedzki1}. An adaption to bundle gerbes of Brylinski's transgression was described in \cite{waldorf5}, and then extended in \cite{waldorf10} to include fusion products. In the following we recall this approach briefly.

Let $\mathcal{G}$ be a bundle gerbe with connection over $M$. We define a principal $\ueins$-bundle $L\mathcal{G}$ over $LM$ in the following way. The fiber $L\mathcal{G}|_{\tau}$ over a loop $\tau$ is the set of 2-isomorphism classes of trivializations $\mathcal{T}: \tau^{*}\mathcal{G} \to \mathcal{I}_0$, i.e. 1-isomorphisms in $\ugrbcon{S^1}$. This set is a torsor over the group of isomorphism classes of principal $\ueins$-bundles with connection over $S^1$, which can be identified canonically with $\ueins$ by taking the holonomy around the base $S^1$; this establishes the $\ueins$-action on $L\mathcal{G}|_{\tau}$.  The total space of $L\mathcal{G}$ is the disjoint union of the fibers $L\mathcal{G}|_{\tau}$. A diffeology on $L\mathcal{G}$ is defined as follows. A map $\tilde c:U \to L\mathcal{G}$ is a plot if the projection $c:=\pi \circ \tilde c: U \to LM$ is a plot of $LM$, and every point $u \in U$ has an open neighborhood $u \in W \subset U$ and a trivialization $\mathcal{T}$ of $(c^{\vee})^{*}\mathcal{G}$ over $W \times S^1$ such that $\tilde c(w) = [i_w^{*}\mathcal{T}] \in L\mathcal{G}|_{\tau}$ for all $w\in W$. Here $c^{\vee}: U \times S^1  \to M$ is the map $c^{\vee}(u,z):= c(u)(z)$, and   $i_w: S^1 \to W \times S^1 $ is defined by $i_w(t) := (w,t)$. It is proved in  \cite[Sec. 4.1]{waldorf10} that this makes $L\mathcal{G}$ into a diffeological principal $\ueins$-bundle over $LM$. By \cref{lem:assbun}, the associated vector bundle $\inf L := L\mathcal{G} \times_{\ueins} \C$ with respect to the standard representation of $\ueins$ on $\C$ is a hermitian line bundle over $LM$.

The fusion product $\lambda$ on $\inf  L$ is defined first on the principal $\ueins$-bundle $L\mathcal{G}$ as described in  \cite[Sec. 4.2]{waldorf10}. We consider $(\gamma_1,\gamma_2,\gamma_3) \in PM^{[3]}$ and the corresponding loops $\tau_{ab} := \gamma_a \cup \gamma_b$, for $a,b=1,2.3$. Let $\mathcal{T}_{ab}: \tau_{ab}^{*}\mathcal{G} \to \mathcal{I}_0$ be trivializations for $ab\in \{12,23,13\}$. We work with $S^1=\R/\Z$, and consider the two maps   $\iota_1,\iota_2\maps [0,1] \to S^1$  defined by $\iota_1(t)\df  \frac{1}{2}t$ and $\iota_2(t) \df  1-\frac{1}{2}t$.
Suppose there exist 2-isomorphisms
$\phi_1\maps \iota_1^{*}\mathcal{T}_{12}  \Rightarrow \iota_1^{*}\mathcal{T}_{13}$, $\phi_2\maps \iota_2^{*}\mathcal{T}_{12} \Rightarrow \iota_1^{*}\mathcal{T}_{23}$ and $\phi_3\maps \iota_2^{*}\mathcal{T}_{23} \Rightarrow \iota_2^{*}\mathcal{T}_{13}$
over the interval $[0,1]$ such that
$\phi_1|_0 = \phi_{3}|_0 \bullet \phi_{2}|_0$ and $\phi_1|_1 = \phi_{3}|_1 \bullet \phi_{2}|_1$, 
where $\bullet$ denotes the vertical composition of 2-morphisms in $\ugrbcon-$. 
 Then, we set 
\begin{equation*}
\lambda|_{\gamma_1,\gamma_2,\gamma_3}(\mathcal{T}_{12} \otimes  \mathcal{T}_{23}) := \mathcal{T}_{13}\text{.}
\end{equation*}
It is proved in  \cite[Sec. 4.2]{waldorf10}
 that this sufficiently characterizes $\lambda$ as a bundle morphism. On the associated bundle $\inf L$, the fusion product is then defined by
\begin{equation*}
\lambda|_{\gamma_1,\gamma_2,\gamma_3}([\mathcal{T}_{12},z_{12}] \otimes [\inf \mathcal{T}_{23},z_{23}]) := [\lambda(\mathcal{T}_{12} \otimes \inf \mathcal{T}_{23}),z_{12}\cdot z_{23}]\text{.} 
\end{equation*}

\begin{remark}
\label{eq:transtildelambda}
Pullback along  $\rev : S^1 \to S^1:t \mapsto 1-t$  defines a bundle morphism $L\mathcal{G} \to L\mathcal{G}^{*}$ that covers the loop reflection $\widetilde {\rev}: LM \to LM: \tau \mapsto \tau \circ \rev$. We note that $\widetilde \rev(\gamma_1\cup\gamma_2)=\gamma_2\cup\gamma_1$. On the associated line bundle, $\widetilde \rev$ induces the bundle morphism 
\begin{equation*}
\tilde\lambda_{\gamma_1,\gamma_2} : \inf L_{\gamma_1\cup\gamma_2} \to \overline{\inf L}|_{\gamma_2\cup\gamma_1}: [\mathcal{T},z] \mapsto [\rev^{*}\mathcal{T},\overline{z}]\text{,}
\end{equation*} 
which is part of the LBG structure \cref{LBGstr:2}. \end{remark}

A connection $\omega$ on $L\mathcal{G}$ is defined in  \cite[Sec. 4.3]{waldorf10}. Here, we describe its parallel transport $\tau^{\omega}$, combining the definition of the 1-form $\omega$  \cite[Sec. 4.3]{waldorf10} with the derivation of the parallel transport described in \cite[Def. 3.2.9]{waldorf9}. If $\Gamma\in PLM$ is a path,  $\mathcal{T}: (\Gamma^{\vee})^{*}\mathcal{G} \to \mathcal{I}_{\rho}$ is a trivialization, and $\mathcal{T}_0$ and $\mathcal{T}_1$ are its restrictions to $\{0\} \times S^1$ and $\{1\} \times S^1$, respectively, then
\begin{equation*}
\tau_{\Gamma}^{\omega}(\mathcal{T}_0) :=\mathcal{T}_1 \cdot \exp \left (\int_{[0,1] \times S^1} \rho \right ) \text{.}
\end{equation*}
The connection $\omega$ on $L\mathcal{G}$ is superficial \cite[Cor. 4.3.3]{waldorf10} and symmetrizes the fusion product \cite[Prop. 4.3.5]{waldorf10}.
We obtain an associated connection $pt$ on the associated line bundle $\inf L$, see \cref{lem:assbuncon}. This completes the construction of LBG structures \cref{LBGstr:1,LBGstr:2} from a bundle gerbe $\mathcal{G}$ with connection over $M$.

\subsection{The vector bundle over the path space}

\label{sec:transbranes}

In this section we construct the vector bundles $\inf R_{ij}$ over the path spaces $P_{ij}$, from a given TBG object $(\mathcal{G},\mathcal{E})$.
It is different from a construction of Gaw\c edzki and Reis \cite{gawedzki1}, but results into an isomorphic bundle. We use the following notation: if $p\in M$ is a point, we denote by $\mathcal{G}|_p$ the pullback of $\mathcal{G}$ along the map $\{\ast\} \to M: \ast \mapsto p$. 

The fiber of $\inf R_{ij}$ over a path $\gamma\in P_{ij}$ with $x:=\gamma(0)$ and $y:= \gamma(1)$ is defined as follows. Let $\mathcal{T}:\gamma^{*}\mathcal{G} \to \mathcal{I}_0$ be a trivialization. Over the point we have the two 1-morphisms $\mathcal{T}|_0: \mathcal{G}|_x \to \mathcal{I}_0$ and $\mathcal{E}_i|_x: \mathcal{G}|_x \to \mathcal{I}_0$, so that we obtain a hermitian vector bundle $\Des(\mathcal{E}_i|_x,\mathcal{T}|_0)$ with connection over the point, i.e. a complex inner product space; see \cref{sec:targetspacegeometry}. Similarly we obtain a complex inner product space $\Des(\mathcal{E}_j|_y,\mathcal{T}|_1)$. We define
\begin{equation*}
\inf R_{ij}|_{\gamma}(\mathcal{T}) :=  \mathrm{Hom} \big(\Des(\mathcal{E}_i|_x,\mathcal{T}|_0),\Des(\mathcal{E}_j|_y,\mathcal{T}|_1)\big)\text{,}
\end{equation*}
which is a vector space of  dimension $\mathrm{rk}(\mathcal{E}_i) \cdot \mathrm{rk}(\mathcal{E}_j)$, and comes equipped with the complex inner product 
$h_{ij}(\varphi,\psi) := \mathrm{tr}(\varphi^{*}\circ \psi)$,
where $\varphi^{*}$ denotes the adjoint map.

\begin{remark}
Under the identification $\mathrm{Hom}(V,W)=V^{*}\otimes W$, this inner product is the  one induced from the inner products on $V$ and $W$. The induced norm is the Frobenius norm.
\end{remark}

We show that the vector spaces $\inf R_{ij}|_{\gamma}(\mathcal{T})$ and $\inf R_{ij}|_{\gamma}(\mathcal{T}')$ associated to two trivializations are canonically isomorphic. 
If $\psi: \mathcal{T} \Rightarrow \mathcal{T}'$ is a  2-isomorphism in $\ugrbcon{[0,1]}$, then we have unitary isomorphisms
\begin{align*}
\psi_0:=\Des(\id,\psi|_0) &:\Des(\mathcal{E}_i|_x,\mathcal{T}'|_0) \to \Des(\mathcal{E}_i|_x,\mathcal{T}|_0)
\\
\psi_1:=\Des(\id,\psi|_1) &:\Des(\mathcal{E}_j|_y,\mathcal{T}'|_1) \to \Des(\mathcal{E}_j|_y,\mathcal{T}|_1)\text{,}
\end{align*}
combining into a unitary isomorphism
\begin{equation*}
r_{\psi} : \inf R_{ij}|_{\gamma}(\mathcal{T}') \to\ \inf R_{ij}|_{\gamma}(\mathcal{T}) : \varphi \mapsto \psi_1\circ \varphi \circ \psi_0^{-1}\text{.}
\end{equation*} 
This isomorphism is in fact independent of $\psi$: since $[0,1]$ is connected, any other 2-isomorphism $\psi': \mathcal{T} \Rightarrow \mathcal{T}'$  satisfies $\psi'=\psi \cdot z$ for a constant $z\in \ueins$; and hence we have $r_\psi=r_{\psi'}$. Moreover, since $[0,1]$ is contractible and 1-dimensional, any two trivializations over $[0,1]$ are 2-isomorphic.

Alternatively, we give another description of $r_{\psi}$. The two trivializations $\mathcal{T}$ and $\mathcal{T}'$ determine a hermitian line bundle $L := \Des(\mathcal{T},\mathcal{T}')$ with connection over $[0,1]$. By \cref{rem:propDes:1} we have  canonical isomorphisms
\begin{align*}
\Des(\mathcal{E}_i|_x,\mathcal{T}|_0) \otimes L|_0 &= \Des(\mathcal{E}_i|_x,\mathcal{T}|_0)\otimes \Des(\mathcal{T}|_0,\mathcal{T}'|_0) \cong \Des(\mathcal{E}_i|_x,\mathcal{T}'|_0)
\\
\Des(\mathcal{E}_j|_y,\mathcal{T}|_1) \otimes L|_1 &= \Des(\mathcal{E}_j|_y,\mathcal{T}|_1)\otimes \Des(\mathcal{T}|_1,\mathcal{T}'|_1) \cong \Des(\mathcal{E}_j|_y,\mathcal{T}'|_1)
\end{align*}
that together yield an isomorphism
\begin{equation*}
L|_0^{*} \otimes \inf R_{ij}|_{\gamma}(\mathcal{T}) \otimes L|_1 \cong \inf R_{ij}|_{\gamma}(\mathcal{T}')\text{.}
\end{equation*}
Any parallel unit-length section $\sigma$ in $L$ then determines an isomorphism
\begin{equation*}
\tilde\sigma: \inf R_{ij}|_{\gamma}(\mathcal{T}) \to \inf R_{ij}|_{\gamma}(\mathcal{T}')\text{.}
\end{equation*}
This isomorphism is again independent of $\sigma$: if $\sigma'$ is another parallel unit-length section, we have $\sigma'=\sigma\cdot z$ for a constant $z\in \ueins$. This shows that $\tilde\sigma=\tilde\sigma'$.

A parallel unit-length section $\sigma$ in $\Des(\mathcal{T},\mathcal{T}')$ is the same as a 2-isomorphism $\psi:\mathcal{T} \Rightarrow \mathcal{T}'$, and under this correspondence we have $r_{\psi}=\tilde\sigma$. Thus, we have two descriptions of the  same, canonical isomorphism, which we will denote by
$r_{\mathcal{T},\mathcal{T}'}$. On the class of pairs $(\mathcal{T},\varphi)$, where $\varphi \in \inf R_{ij}|_{\gamma}(\mathcal{T})$, we define the relation $(\mathcal{T}',\varphi') \sim (\mathcal{T},r_{\mathcal{T},\mathcal{T}'}(\varphi'))$. This is an equivalence relation, and $\inf R_{ij}|_{\gamma}$ is, by definition, the set of equivalence classes. 
It is a complex inner product space of  dimension $\mathrm{rk}(\mathcal{E}_i) \cdot \mathrm{rk}(\mathcal{E}_j)$. 
We note that the choice of any trivialization $\mathcal{T}$ determines a unitary isomorphism 
\begin{equation*}
\inf R_{ij}|_{\gamma}(\mathcal{T}) \to \inf R_{ij}|_{\gamma}: \varphi \mapsto [(\mathcal{T},\varphi)]\text{.}
\end{equation*}

We will shortly need to consider the vector spaces $\inf R_{ij}$ in smooth families, and we thus prepare the following notation and \cref{lem:changeoftrivsmooth} below. Suppose $U$ is a  smooth manifold, possibly with boundary, and $f:U \to P_{ij}$ is a smooth map. By definition of the diffeology on $P_{ij}$, this means that $f^{\vee}:U \times [0,1] \to M$ is smooth. For $u\in U$ and $t\in [0,1]$ we will use the maps two maps $i_u: [0,1] \to U \times [0,1]$ and $j_t: U \to U \times [0,1]$ defined by $i_u(t):=j_t(u) := 
(u,t)$.
If  $\mathcal{T}:(f^{\vee})^{*}\mathcal{G} \to \mathcal{I}_{\rho}$ is a trivialization, then we define the hermitian vector bundles \begin{equation*}
\inf V_{\mathcal{T}} := \Des((\ev_0 \circ f)^{*}\mathcal{E}_i,j_0^{*}\mathcal{T})
\quand
\inf W_{\mathcal{T}} := \Des((\ev_1\circ f)^{*}\mathcal{E}_j,j_1^{*}\mathcal{T})
\end{equation*}
with connection over $U$. In terms of this notation, we have $\inf R_{ij}|_{f(u)}(i_u^{*}\mathcal{T})=\mathrm{Hom}(\inf V_{\mathcal{T}},\inf W_{\mathcal{T}})|_{u}$ for all $u\in U$. We need the following result about a change of trivialization in smooth families. 

\begin{lemma}
\label{lem:changeoftrivsmooth}
Let $\mathcal{T}:(f^{\vee})^{*}\mathcal{G} \to \mathcal{I}_{\rho}$ and $\mathcal{T}': (f^{\vee})^{*}\mathcal{G} \to \mathcal{I}_{\rho'}$ be two trivializations. If $U$ is contractible, then there exist the following:
\begin{itemize}
\item 
a 1-form $\eta\in \Omega^1(U \times [0,1])$ such that $\rho-\rho'=\mathrm{d}\eta$ and $i_u^{*}\eta=0$ for all $u\in U$

\item
a connection-preserving bundle isomorphism \begin{equation*}
\tilde\sigma: \mathrm{Hom}(\inf V_{\mathcal{T}},\inf W_{\mathcal{T}}) \otimes \C_{j_1^{*}\eta-j_0^{*}\eta} \to \mathrm{Hom}(\inf V_{\mathcal{T}'},\inf W_{\mathcal{T}'})
\end{equation*}
over $U$, such that $\tilde\sigma|_u= r_{i_u^{*}\mathcal{T},i_u^{*}\mathcal{T}'}$ for all $u\in U$. Here, $\C_{\omega}$ denotes the trivial line bundle equipped with the connection induced by a 1-form $\omega$.
\end{itemize}
\end{lemma}

\begin{proof}
We use a $U$-family version of the canonical isomorphism $r_{\mathcal{T},\mathcal{T}'}$.
We consider the hermitian line bundle $L := \Des(\mathcal{T},\mathcal{T}')$  over $U \times [0,1]$ with connection of curvature $\rho-\rho'$. By \cref{rem:propDes:1} we have an  isomorphism
\begin{equation*}
\inf V_{\mathcal{T}} \otimes j_0^{*}L = \Des((\ev_0 \circ c)^{*}\mathcal{E}_i,j_0^{*}\mathcal{T})\otimes \Des(j_0^{*}\mathcal{T},j_0^{*}\mathcal{T}') \cong \Des((\ev_0 \circ c)^{*}\mathcal{E}_i,j_0^{*}\mathcal{T}')=\inf V_{\mathcal{T}'}
\end{equation*}
A similar construction works for $\inf W_{\mathcal{T}}$ and $\inf W_{\mathcal{T}'}$, and together we obtain a   bundle isomorphism
\begin{equation*}
j_0^{*}L^{*} \otimes\mathrm{Hom}(\inf V_{\mathcal{T}},\inf W_{\mathcal{T}}) \otimes  j_1^{*}L \cong \mathrm{Hom}(\inf V_{\mathcal{T}'},\inf W_{\mathcal{T}'})\text{.}
\end{equation*}
Now, $L$ may not have a parallel unit-length section. However, it always admits a unit-length section $\sigma$ that is parallel along $i_u$. Indeed, since $U$ is contractible, there exists a smooth unit-length section $\sigma_0:U \to i_0^{*}L$. We define $\sigma(u,t) := pt_{\gamma_{u,t}}(\sigma_0(u))$, where $\gamma_{u,t}$ is the path $\tau \mapsto i_u(\tau)$ restricted to $[0,t]$. Since the parallel transport depends smoothly on the path, this gives a smooth section of $L$, as desired. 
We denote by $\eta\in \Omega^1(U \times [0,1])$ its covariant derivative, so that $\mathrm{d}\eta=\rho-\rho'$. Therefore, $\sigma$ determines the claimed connection-preserving bundle isomorphism $\tilde\sigma$, as claimed.
Since $\sigma$ is parallel along $i_u$, we have $i_u^{*}\eta=0$ and $\tilde\sigma|_u = r_{i_u^{*}\mathcal{T},i_u^{*}\mathcal{T}'}$. \end{proof}

We are now in position to assemble the fibers $\inf R_{ij}|_{\gamma}$ into a diffeological vector bundle. Its total space $\inf R_{ij}$ is the disjoint union of the fibers $\inf R_{ij}|_{\gamma}$ for $\gamma\in P_{ij}$, equipped with the obvious projection $\pi:\inf R_{ij} \to P_{ij}$.  Let $U \subset \R^{k}$ be open. We define a map $\tilde c:U \to \inf R_{ij}$ to be a plot of $\inf R_{ij}$ if the following conditions hold: 
\begin{enumerate}[(a)]

\item 
The composition $c:=\pi \circ \tilde c: U \to P_{ij}$  is a plot of $P_{ij}$.

\item
Every point $u\in U$ has an open neighborhood $W \subset U$, a trivialization $\mathcal{T}:(c|_{W}^{\vee})^{*}\mathcal{G} \to \mathcal{I}_{\rho}$, and a  smooth section $\tau$ into the bundle $\mathrm{Hom}(\inf V_{\mathcal{T}},\inf W_{\mathcal{T}})$ over $W$, such that  $\tilde c(w) =[i_w^{*}\mathcal{T},\tau(w)]$ for all $w\in W$.

\end{enumerate}
It is straightforward to show that this indeed defines  a diffeology. 
 
\begin{proposition}
\label{prop:Rijvectorbundle}
$\inf R_{ij}$ is a hermitian vector bundle over $P_{ij}$.
\end{proposition}

\begin{proof}
Let $c:U \to P_{ij}$ be a plot and let $u\in U$. We choose a contractible open neighborhood $ W \subset U$ of $u$. Since $W \times [0,1]$ is contractible, there exists a trivialization $\mathcal{T}: (c|_W^{\vee})^{*}\mathcal{G} \to \mathcal{I}_{\rho}$, and since $W$ is contractible, there exists a trivialization 
\begin{equation*}
\tau:W \times \C^{k} \to \mathrm{Hom}(\inf V_{\mathcal{T}},\inf W_{\mathcal{T}})\text{,}
\end{equation*}
where $k=\mathrm{rk}(\mathcal{E}_i) \cdot \mathrm{rk}(\mathcal{E}_j)$. 
We define a local trivialization $\phi$ of $\inf R_{ij}$ by the formula 
\begin{equation*}
\phi: W \times \C^{k} \to W \times_{P_{ij}} \inf R_{ij}:(w,v) \mapsto (w,[i_w^{*}\mathcal{T},\tau(w,v)])\text{.}
\end{equation*}
This map is smooth (by definition of the plots) and fiber-wise linear. 
We have to show that it is a diffeomorphism.

The inverse map of $\phi$ can be described in the following way. Given $(w,[\mathcal{T}_w,\varphi])\in W \times_{P_{ij}} \inf R_{ij}$ with a trivialization $\mathcal{T}_w:c(w)^{*}\mathcal{G} \to \mathcal{I}_0$ and $\varphi\in \inf R_{ij}|_{c(w)}(\mathcal{T}_w)$, we have
\begin{equation*}
\phi^{-1}(w,[\mathcal{T}_w,\varphi]) = \tau^{-1}(r_{\iota_w^{*}\mathcal{T},\mathcal{T}_w}(\varphi))\text{.}
\end{equation*}
In order to show that $\phi^{-1}$ is smooth, we consider a plot of $W \times_{P_{ij}} \inf R_{ij}$, which is a pair $(f,\tilde c)$ of a smooth map $f:U' \to W$ and a plot $\tilde c: U'\to \inf R_{ij}$ such that $c \circ f = \pi \circ \tilde c =:c'$. Since $\tilde c$ is a plot, there exists (possibly after replacing $U'$ by a smaller subset) a trivialization $\mathcal{T}' : (c'^{\vee})^{*}\mathcal{G} \to \mathcal{I}_{\rho'}$ and a smooth section $\tau'$ of $\mathrm{Hom}(\inf V_{\mathcal{T}'},\inf W_{\mathcal{T}'})$ such that $\tilde c(u)=[i_u^{*}\mathcal{T},\tau'(u)]$ for all $u\in U'$. Now, we have to show that the map $U' \to W \times \C^{k}$ given by
\begin{equation*}
u \mapsto \phi^{-1}(f(u),[i_u^{*}\mathcal{T}',\tau'(u)])=\tau^{-1}(r_{i_{u}^{*}\mathcal{T},i_u^{*}\mathcal{T}'}(\tau'(u)))
\end{equation*} 
is smooth. Indeed, by \cref{lem:changeoftrivsmooth} it is the composition of smooth bundle morphisms. 
\end{proof}

\subsection{The superficial connection}

\label{sec:trans:superficialconnection}

We define a connection on the vector bundle $\inf R_{ij}$, in the sense of \cref{def:connection}. A  similar definition of that connection has been sketched in \cite{gawedzki1}.  
Let $\Gamma\in PP_{ij}$ be a path in $P_{ij}$, and let $\gamma_s := \Gamma(s)\in P_{ij}$. 
Thus, $\Gamma$ is a path from $\gamma_0$ to $\gamma_1$. 
We consider the adjoint map $\Gamma^{\vee}: [0,1]^2 \to M$ (i.e. $\Gamma^{\vee}(s,t):= \gamma_s(t)$) and choose a trivialization $\mathcal{T}: (\Gamma^{\vee})^{*}\mathcal{G} \to \mathcal{I}_{\rho}$. 
Let $\mathcal{T}_s := \mathcal{T}|_{\{s\} \times [0,1]}$; these are trivializations of $(\gamma_s^{\vee})^{*}\mathcal{G}$. 
Further, let $\mathcal{T}^0:=\mathcal{T}|_{[0,1]\times \{0\}}$ and $\mathcal{T}^1:=\mathcal{T}|_{[0,1]\times \{1\}}$, which are trivializations along the paths of end points.  
In the notation of the previous subsection, we consider the hermitian vector bundles $\inf V_{\mathcal{T}^0}$ and $\inf W_{\mathcal{T}^1}$ with connection over $[0,1]$, 
and the corresponding Hom-bundle $\mathrm{Hom}(\inf V_{\mathcal{T}^{0}},\inf W_{\mathcal{T}^{1}})$ over $[0,1]$ with its induced connection.
Now we have 
\begin{align*}
\inf R_{ij}|_{\gamma_0}(\mathcal{T}_0) =\mathrm{Hom}(\inf V_{\mathcal{T}^{0}},\inf W_{\mathcal{T}^{1}})|_0
\quand
\inf R_{ij}|_{\gamma_1}(\mathcal{T}_1) =\mathrm{Hom}(\inf V_{\mathcal{T}^{0}},\inf W_{\mathcal{T}^{1}})|_1\text{.}
\end{align*}
Then, we define 
\begin{equation}
\label{def:connectionRij}
pt_{ij}|_{\Gamma}: \inf R_{ij}|_{\gamma_0}(\mathcal{T}_0) \to \inf R_{ij}|_{\gamma_1}(\mathcal{T}_1): \varphi \mapsto  \exp \left (\int_{[0,1]^2}\rho \right ) \cdot pt(\varphi)\text{,}
\end{equation}
where $pt$ is the parallel transport in $\mathrm{Hom}(\inf V_{\mathcal{T}^{0}},\inf W_{\mathcal{T}^{1}})$ along the linear path in $[0,1]$ from $0$ to $1$.

\begin{proposition}
\label{lem:connection}
\cref{def:connectionRij} defines a superficial, unitary connection $pt_{ij}$ on $\inf R_{ij}$.
\end{proposition}

\begin{proof}
The proof is split into five parts. 

Part I: the definition of $pt_{ij}|_{\Gamma}$ is independent of the choice of the trivialization $\mathcal{T}$. In order to prove this,
let $\mathcal{T}': (\Gamma^{\vee})^{*}\mathcal{G} \to \mathcal{I}_{\rho'}$ be another trivialization. We consider a 1-form $\eta\in \Omega^1([0,1]^2)$ and a connection-preserving bundle isomorphism  \begin{equation*}
\tilde\sigma: \mathrm{Hom}(\inf V_{\mathcal{T}},\inf W_{\mathcal{T}}) \otimes \C_{j_1^{*}\eta-j_0^{*}\eta} \to \mathrm{Hom}(\inf V_{\mathcal{T}'},\inf W_{\mathcal{T}'})
\end{equation*}
over $U$, as in \cref{lem:changeoftrivsmooth}. Thus, the parallel transport $pt$ in $\mathrm{Hom}(\inf V_{\mathcal{T}},\inf W_{\mathcal{T}})$ and $pt'$ in $\mathrm{Hom}(\inf V_{\mathcal{T}'},\inf W_{\mathcal{T}'})$ differ by the parallel transport  of $\C_{j_1^{*}\eta-j_0^{*}\eta}$, i.e.,
\begin{equation*}
pt'(\tilde \sigma|_0(\varphi)) =\tilde\sigma|_1( pt(\varphi)) \cdot \exp\left (  \int_{[0,1]} j_1^{*}\eta-j_0^{*}\eta \right )
\end{equation*}
for all $\varphi \in \mathrm{Hom}(\inf V_{\mathcal{T}},\inf W_{\mathcal{T}})|_0$. Due to Stokes' Theorem, and the properties $\rho'-\rho=\mathrm{d}\eta$ and $i_0^{*}\eta=i_1^{*}\eta=0$ of $\eta$, we can write this as
\begin{equation*}
pt'(\tilde \sigma|_0(\varphi)) \cdot \exp\left (  \int_{[0,1]^2} \rho' \right )=\tilde\sigma|_1( pt(\varphi)) \cdot \exp\left (  \int_{[0,1]^2} \rho \right )\text{.}
\end{equation*}
Since $\tilde\sigma|_s$ establishes the isomorphism $r_{\mathcal{T}_s,\mathcal{T}_s'}$ between $\inf R_{ij}|_{\gamma_s}(\mathcal{T}_s)$ and $\inf R_{ij}|_{\gamma_s}(\mathcal{T}'_s)$; this shows the claimed independence. 

Part II is the verification of \ref{def:connection:a}:  $pt_{ij}|_{\Gamma}$ depends only on the thin homotopy class of $\Gamma$.  In order to prove this, we suppose $\Gamma_1,\Gamma_2\in PP_{ij}$ are thin homotopic  (\cref{def:thinhomotopic} for $X=P_{ij}$). Thus, there exists $h\in PPP_{ij}$ such that
\begin{enumerate}[(a)]

\item 
$h(0)=\Gamma_1$ and $ h(1)=\Gamma_2$

\item
$h(r)(0)=\Gamma_1(0)=\Gamma_2(0)$ and $h(r)(1)=\Gamma_1(1)=\Gamma_2(1)$ for all $r\in [0,1]$

\item
$h^{\vee}:[0,1]^2\to P_{ij}:(r,s) \mapsto h(r)(s)$ has rank   one. 

\end{enumerate}
Here, a smooth map $f$ between smooth manifolds is said to have rank $k$ if $\mathrm{rk}(\mathrm{d}f_x)\leq k$ for all points $x$ in its domain.
We let $H:[0,1]^3 \to M$ denote the map $H(r,s,t)=h(r)(s)(t)=h^{\vee}(r,s)(t)$; which by (c) has rank two. We choose a trivialization $\mathcal{T}: H^{*}\mathcal{G} \to \mathcal{I}_{\rho}$; since $H$   has rank two, $H^{*}\mathcal{G}$ is flat and hence $\mathrm{d}\rho = 0$.

The restriction of $H$ to the $(t=0)$ face of the cube $[0,1]^3$ has rank  one due to (c). 
By \cref{th:rankonepullback} there exists a parallel trivialization $\mathcal{T}'_{t=0}$, a hermitian vector bundle $E_0$ with flat connection, and a 2-isomorphism $\mathcal{T}'_{t=0} \otimes E_0 \cong H|_{t=0}^{*}\mathcal{E}_i$. Thus,
\begin{equation}
\label{eq:thinbrane}
\inf V_{\mathcal{T}}=\Des(H|_{t=0}^{*}\mathcal{E}_i,\mathcal{T}_{t=0})\cong \Des(\mathcal{T}'_{t=0} \otimes E_0,\mathcal{T}_{t=0})\cong \Des(\mathcal{T}'_{t=0} ,\mathcal{T}_{t=0})\otimes E_0\text{.}
\end{equation}
We note that the line bundle $\Des(\mathcal{T}'_{t=0} ,\mathcal{T}_{t=0})$ carries a connection of curvature $-\rho_{t=0}$. Analogously, $\inf W_{\mathcal{T}}\cong \Des(\mathcal{T}'_{t=1},\mathcal{T}_{t=1}) \otimes E_1$, for another flat hermitian vector bundle $E_1$ and a parallel trivialization $\mathcal{T}'_{t=1}$. Computing the holonomy of the vector bundle $\mathrm{Hom}(\inf V_{\mathcal{T}},\inf W_{\mathcal{T}})$ over $[0,1]^2$ around the boundary, we obtain
\begin{equation}
\label{eq:proofthinhomotopyinvariance:2}
pt_{r=1}  \circ pt_{s=0} =  pt_{s=1} \circ pt_{r=0} \cdot \exp\left ( \int_{[0,1]^2} -\rho_{t=0} - \rho_{t=1} \right )\text{.}
\end{equation}
Here and in the following, the orientation on the faces of the cube $[0,1]^3$ are always the ones induced on the boundary by the standard orientation on $[0,1]^3$.

Next, we consider the $(s=0)$ face of the cube $[0,1]^3$, where $H$ is constant in $r$ due to (b).  
Let $\mathcal{T}_{s=0} := \mathcal{T}|_{[0,1] \times \{0\} \times [0,1]}$ be the restriction of $\mathcal{T}$ to that face, and let $\rho_{s=0}$ be its 2-form. For $p:[0,1]^2 \to [0,1]^2: (r,t) \mapsto (0,t)$ we have another trivialization $\mathcal{T}_{s=0}' := p^{*}\mathcal{T}|_{\{0\} \times \{0\} \times[0,1]}$ with vanishing 2-form, since $p$ factors through $[0,1]$.  Applying \cref{lem:changeoftrivsmooth} to the pair  $(\mathcal{T}_{s=0},\mathcal{T}_{s=0}')$,  we obtain a 1-form $\eta_0\in \Omega^1([0,1]^2)$ such that $\rho_{s=0}=\mathrm{d}\eta_0$ and $i_r^{*}\eta_0=0$ for all $r\in [0,1]$, and a connection-preserving bundle isomorphism  
\begin{equation*}
\tilde\sigma_0: \mathrm{Hom}(\inf V_{\mathcal{T}_{s=0}},\inf W_{\mathcal{T}_{s=0}}) \otimes \C_{j_1^{*}\eta_0-j_0^{*}\eta_0} \to \mathrm{Hom}(\inf V_{\mathcal{T}'_{s=0}},\inf W_{\mathcal{T}'_{s=0}})
\end{equation*}
over $[0,1]$. The parallel transport from $0$ to $1$ gives the identity
\begin{equation*}
\tilde\sigma_0|_1 \circ pt_{s=0} =pt'_{s=0} \circ \tilde\sigma_0|_0 \;\cdot\;  \exp\left ( \int_{0}^1 j_0^{*}\eta_0-j_1^{*}\eta_0  \right ) \text{.}
\end{equation*}
In this identity, we used the fact that  over each point $r\in [0,1]$  $\tilde\sigma_0$ restricts to the canonical isomorphism 
\begin{equation*}
r_{\mathcal{T}_{s,r=0},\mathcal{T}'_{s,r=0}}: \inf R_{ij}|_{\gamma_0}(\mathcal{T}_{s,r=0}) \to \inf R_{ij}|_{\gamma_0}(\mathcal{T}'_{s,r=0})\text{,}
\end{equation*}
see \cref{lem:changeoftrivsmooth}.
Likewise, the parallel transport $pt'_{s=0}$ in $\mathrm{Hom}(\inf V_{\mathcal{T}'_{s=0}},\inf W_{\mathcal{T}'_{s=0}})$ coincides with the canonical isomorphism $r_{\mathcal{T}'_{s,r=0},\mathcal{T}'_{s=0,r=1}}$, since it describes a  trivialization of $\Des(\mathcal{T}_{s,r=0}',\mathcal{T}'_{s=0,r=1})$.
Finally, we use the properties of $\eta_0$ and Stokes' Theorem, and  obtain
\begin{equation}
\label{eq:proofthinhomotopyinvariance:1}
pt_{s=0} = r_{\mathcal{T}_{s,r=0},\mathcal{T}_{s=0,r=1}} \cdot  \exp\left (  \int_{[0,1]^2} \rho_{s=0} \right )\text{.}
\end{equation}
The $(s=1)$ face is treated analogously, just that we get $-\rho_{s=1}$ under the integral.

Now we are in position to show that the two parallel transports $pt_{ij}|_{\Gamma_1}$ and $pt_{ij}|_{\Gamma_2}$
coincide, under the canonical isomorphisms $r$. Indeed,  for $\varphi\in \inf R_{ij}|_{\gamma_0}(\mathcal{T}_{s,r=0})$ we have:
\begin{align*}
&\mquad pt_{ij}|_{\Gamma_2}(r_{\mathcal{T}_{s,r=0},\mathcal{T}_{s=0,r=1}}(\varphi)) \\&=  pt_{r=1}(r_{\mathcal{T}_{s,r=0},\mathcal{T}_{s=0,r=1}}(\varphi))\cdot \exp \left (\int_{[0,1]^2}-\rho_{r=1} \right )
 \\&\eqcref{eq:proofthinhomotopyinvariance:1}pt_{r=1}(pt_{s=0}(\varphi))\cdot  \exp\left (  \int_{[0,1]^2} -\rho_{s=0}-\rho_{r=1} \right ) 
 \\&\eqcref{eq:proofthinhomotopyinvariance:2} pt_{s=1} ( pt_{r=0}(\varphi)) \cdot \exp\left ( \int_{[0,1]^2} -\rho_{t=0} - \rho_{t=1}-\rho_{s=0}-\rho_{r=1} \right )
 \\&\eqcref{eq:proofthinhomotopyinvariance:1} r_{\mathcal{T}_{s=1,r=0},\mathcal{T}_{s,r=1}} ( pt_{r=0}(\varphi)) \cdot \exp\left ( \int_{[0,1]^2} -\rho_{t=0} - \rho_{t=1}-\rho_{s=0}-\rho_{s=1}-\rho_{r=1} \right )
 \\&= r_{\mathcal{T}_{s=1,r=0},\mathcal{T}_{s,r=1}} ( pt_{r=0}(\varphi)) \cdot \exp\left ( \int_{[0,1]^2} \rho_{r=0} \right )
 \\&= r_{\mathcal{T}_{s=1,r=0},\mathcal{T}_{s,r=1}} (pt_{ij}|_{\Gamma_1}(\varphi) ) 
\end{align*}
In the last-but-one step we have used Stokes' Theorem for the closed 2-form $\rho$.

Part III is the verification of \ref{def:connection:b}: parallel transport is compatible with path concatenation; this  follows directly from the definition.

Part IV is the verification of \ref{def:connection:c}:  $pt_{ij}$ is compatible with local trivializations. Let $c:U \to P_{ij}$ be a plot and let  $\phi: W \times \C^{k} \to W \times_{P_{ij}} \inf R_{ij}$ be a local trivialization with $W \subset U$. We can assume that $\phi(w,v) = (w,[\iota_w^{*}\mathcal{T},\tau(w,v)])$, where  $\mathcal{T}: (c^{\vee})^{*}\mathcal{G} \to \mathcal{I}_{\rho}$ is a trivialization over $[0,1] \times W$, and $\tau$ is a bundle isomorphism $\tau:W \times \C^{k} \to \mathrm{Hom}(\inf V_{\mathcal{T}},\inf W_{\mathcal{T}})$. We let $\omega_{\tau}\in\Omega^1(W,\mathfrak{gl}(\C^{k}))$ be the corresponding connection 1-form, i.e. it induces the unique connection on $W \times \C^{k}$ such that $\tau$ is connection-preserving. This means that for a path $\gamma\in PW$ we have
\begin{equation}
\label{eq:localoneformtau}
\tau(\gamma(1),\exp(\omega_{\tau})(\gamma) \cdot v) = pt_{\gamma} (\tau(\gamma(0),v))\text{,}
\end{equation}  
where $\exp(\omega_{\tau})(\gamma)\in \mathrm{GL}(\C^{k})$ is the path-ordered exponential of $\omega_{\tau}$ along $\gamma$.
Then we define
\begin{equation}
\label{eq:localoneform}
\omega_{\phi} := \omega_{\tau} + \int_{[0,1]}\rho \in \Omega^1(W,\mathfrak{gl}(\C^{k}))\text{,}
\end{equation}
with the addition performed under the diagonal embedding $\R \subset \mathfrak{gl}(\C^{k})$. We note that 
\begin{equation}
\label{eq:localoneformexp}
\exp(\omega_{\phi})(\gamma)=\exp(\omega_{\tau})(\gamma) \cdot \exp\left ( \int_{[0,1]^2} (\id_{[0,1]} \times \gamma)^{*} \rho \right )\text{.}
\end{equation}
Now we consider the relevant diagram, whose commutativity is to check:
\begin{equation*}
\small
\alxydim{@C=2cm@R=3em}{\C^{k} \ar[r]^{\exp(\omega_{\phi})(\gamma)}  \ar[d]_{\phi|_{\gamma(0)}} & \C^{k} \ar[d]^{\phi|_{\gamma(1)}} \\\inf R_{ij}|_{c(\gamma(0))} \ar[r]_{pt_{ij}|_{c\circ\gamma}} & \inf R_{ij}|_{c(\gamma(1))}\text{.}}
\end{equation*}
Clockwise,  using \cref{eq:localoneformtau,eq:localoneformexp}  we obtain the map
\begin{equation*}
v \mapsto \left [\iota_{\gamma(1)}^{*}\mathcal{T},\exp \left (\int_{[0,1]^2}(\id \times \gamma)^{*}\rho \right )\cdot pt_{\gamma} (\tau(\gamma(0),v)))\right]\text{.}
\end{equation*}
Counter-clockwise, we first have
$v \mapsto [\iota_{\gamma(0)}^{*}\mathcal{T},\tau(\gamma(0),v)]$ and then obtain, by definition of $pt_{ij}$, precisely the same result. 

Part V: the connection $pt_{ij}$ is superficial (\cref{def:superficial}).
For condition \ref{def:superficial:i}, we may equivalently show that a thin, fixed-ends loop $\Gamma: S^1 \to P_{ij}$ has trivial holonomy, $pt_{ij}|_{\Gamma}=\id$. Here, by a thin loop we mean that its adjoint $\Gamma^{\vee}: S^1 \times [0,1] \to M$ has rank one.
Since $\Gamma^{\vee}|_{S^1 \times \{0\}}$ and $\Gamma^{\vee}|_{S^1 \times \{1\}}$ are constant and all paths have sitting instants, we can extend $\Gamma^{\vee}$ constantly to discs glued along their boundary to $S^1 \times \{0\}$ and $S^1 \times \{1\}$.  By \cref{th:rankonepullback:a} there exists a trivialization $\mathcal{T}:(\Gamma^{\vee})^{*}\mathcal{G} \to \mathcal{I}_0$. Thus, $pt_{ij}|_{\Gamma}$ is determined completely by the holonomy of $\mathrm{Hom}(\inf V_{\mathcal{T}},\inf W_{\mathcal{T}})$ around $S^1=\partial D^2$.
The vector bundles $\inf V_{\mathcal{T}}$ and $\inf W_{\mathcal{T}}$ are defined over $D^2$ and are flat by \cref{th:rankonepullback:b}; hence, their holonomy vanishes. 

For condition \ref{def:superficial:ii}, we consider a rank-two-homotopy $h\in PPP_{ij}$ between paths $\Gamma_1=h(0)$ and $\Gamma_2=h(1)$. Let $\mathcal{T}:(h^{\vee})^{*}\mathcal{G} \to \mathcal{I}_{\rho}$ be a trivialization. Due to condition \ref{def:ranktwohomotopic:c}, $h^\vee$ has rank two, so that $(h^{\vee})^{*}\mathcal{G}$ is flat and $\mathrm{d}\rho=0$. Over the $(t=0)$ face, $h^{\vee}$ is constant in $r$ due to  \ref{def:ranktwohomotopic:b}, and hence is of rank one. Using \cref{th:rankonepullback}, there exists a trivialization $\mathcal{T}_0:(h^{\vee})|_{t=0}^{*}\mathcal{G} \to \mathcal{I}_0$ and a flat hermitian vector bundle $E_0$ over $[0,1]^2$ such that $\inf V_{\mathcal{T}}\cong \Des(\mathcal{T}_{0} ,\mathcal{T}_{t=0})\otimes E_0$, see \cref{eq:thinbrane}. Since $\Des(\mathcal{T}_{0} ,\mathcal{T}_{t=0})$ is a hermitian line bundle with connection of curvature $-\rho|_{t=0}$, we have
\begin{equation}
\label{eq:proofofsuperficialityii:a}
\mathrm{Hol}_{\inf V_{\mathcal{T}}}(\partial [0,1]^2)=\exp\left ( \int_{[0,1]^2} -\rho|_{t=0}  \right )\text{.}
\end{equation}
We treat the $(t=1)$ face analogously, producing the same formula for $\inf W_{\mathcal{T}}$ and $\rho|_{t=1}$. Now we are in position to prove condition \ref{def:superficial:ii}; we need to check that
\begin{equation*}
pt_{ij}|_{h_1} \circ pt_{ij}|_{\gamma_1}=pt_{ij}|_{\gamma_2} \circ pt_{ij}|_{h_0}\text{.}
\end{equation*}
Substituting our above findings, we obtain the integral of $\rho$ over four faces of the cube, as well as the holonomy of $\mathrm{Hom}(\inf V_{\mathcal{T}},\inf W_{\mathcal{T}})$ around $\partial [0,1]^2$. Using \cref{eq:proofofsuperficialityii:a} the latter provides integrals of $\rho$ over the remaining faces. All together, these integrals vanish due to $\mathrm{d}\rho=0$ by Stokes' Theorem.
\end{proof}

We conclude the discussion of the connection $pt_{ij}$ with the following result. It coincides with a corresponding claim in \cite{gawedzki1} (for rank one D-branes).

\begin{lemma}
The curvature of the connection $pt_{ij}$ on $\inf R_{ij}$ satisfies
\begin{equation*}
\frac{1}{\mathrm{rk}(\inf R_{ij})}\,\mathrm{tr}(\mathrm{curv}(pt_{ij}))=\int_{[0,1]} \ev^{*}\mathrm{curv}(\mathcal{G}) + \ev_1^{*}\omega_j - \ev_0^{*}\omega_i \text{.}
\end{equation*}
\end{lemma}

\begin{proof}
According to \cref{re:curvature}, we have to compute $\mathrm{tr}(\mathrm{d}\omega_{\phi})$, where  $\omega_{\phi}$ is the 1-form encountered in the preceding proof in \cref{eq:localoneform}, which in turn was obtained from a trivialization $\mathcal{T}:(c^{\vee})^{*}\mathcal{G} \to \mathcal{I}_{\rho}$. For the first summand in \cref{eq:localoneform} we have to compute $\mathrm{tr}(\mathrm{d}\omega_{\tau})$, i.e. the trace of the curvature of the vector bundle $\mathrm{Hom}(\inf V_{\mathcal{T}},\inf W_{\mathcal{T}})$. From the definition of the functor $\Des$ and the vector bundles $\inf V_{\mathcal{T}}$ and $\inf W_{\mathcal{T}}$ we see that 
\begin{equation*}
\frac{1}{\mathrm{rk}(\mathcal{E}_i)}\mathrm{tr}(\mathrm{curv}(\inf V_{\mathcal{T}})) = (\ev_0\circ c)^{*}\omega_i-j_0^{*}\rho
\quand
\frac{1}{\mathrm{rk}(\mathcal{E}_j)}\mathrm{tr}(\mathrm{curv}(\inf W_{\mathcal{T}}))=(\ev_0\circ c)^{*}\omega_j - j_1^{*}\rho\text{.}
\end{equation*}
Thus, we obtain
\begin{equation*}
\frac{1}{\mathrm{rk}(\mathcal{E}_i)\mathrm{rk}(\mathcal{E}_j)}\mathrm{tr}(\mathrm{curv}(\mathrm{Hom}(\inf V_{\mathcal{T}},\inf W_{\mathcal{T}}))) = c^{*}(\ev_1^{*}\omega_j-\ev_0^{*}\omega_i)-j_1^{*}\rho+j_0^{*}\rho \end{equation*}
in $\in \Omega^2(W)$.
For the second summand in \cref{eq:localoneform} we recall that integration over a fiber with boundary satisfies a version of Stokes' Theorem,
\begin{equation*}
\mathrm{d} \int_{[0,1]}\rho = \int_{[0,1]}\mathrm{d}\rho+j_1^{*}\rho-j_0^{*}\rho\text{.}
\end{equation*}
We have $\mathrm{d}\rho=(c^{\vee})^{*}\mathrm{curv}(\mathcal{G})$ and obtain --
under the sum of  \cref{eq:localoneform} -- the claimed formula. \end{proof}

\subsection{Fusion representation}

\label{sec:fusionrep}

We equip the vector bundle $\inf R_{ij}$ over $P_{ij}$ with a fusion representation of the line bundle $\inf L$ over $LM$. We start by constructing  isomorphisms
\begin{equation*}
\phi_{ij}|_{\gamma_1,\gamma_2}: \inf L|_{\tau} \otimes \inf R_{ij}|_{\gamma_2} \to \inf R_{ij}|_{\gamma_1}
\end{equation*}
for $(\gamma_1,\gamma_2) \in P_{ij}^{[2]}:= P_{ij} \times_{Q_i \times Q_j} P_{ij}$, with $\tau := \gamma_1 \cup \gamma_2$.
Let $\mathcal{T}$ be a trivialization of $\tau^{*}\mathcal{G}$. We recall that $S^1=\R/\Z$ and consider the two maps   $\iota_1,\iota_2\maps [0,1] \to S^1$  defined by $\iota_1(t)\df  \frac{1}{2}t$ and $\iota_2(t) \df  1-\frac{1}{2}t$. Then, $\mathcal{T}_1:=\iota_1^{*}\mathcal{T}$ is a trivialization of $\gamma_1^{*}\mathcal{G}$, and $\mathcal{T}_2:=\iota_2^{*}\mathcal{T}$ is a trivialization of $\gamma_2^{*}\mathcal{G}$. We note that
\begin{align*}
\inf R_{ij}|_{\gamma_2}(\mathcal{T}_2) &= \mathrm{Hom}(\Des(\mathcal{E}_i|_x,\mathcal{T}|_{0}),\Des(\mathcal{E}_j|_y,\mathcal{T}|_{\frac{1}{2}}))=\inf R_{ij}|_{\gamma_1}(\mathcal{T}_1)\text{,}
\end{align*}
and define:
\begin{equation*}
\phi_{ij}|_{\gamma_1,\gamma_2} : \inf L|_{\tau} \otimes \inf R_{ij}|_{\gamma_2}(\mathcal{T}_2) \to \inf R_{ij}|_{\gamma_1}(\mathcal{T}_1): [\mathcal{T},z] \otimes \varphi \mapsto z\varphi\text{.} \end{equation*}

\begin{lemma}
This defines a connection-preserving, unitary bundle isomorphism $\phi_{ij}$.
\end{lemma}

\begin{proof}
(1) The map $\phi_{ij}|_{\gamma_1,\gamma_2}$ is  independent of the choice of the trivialization $\mathcal{T}$. Indeed, if $\mathcal{T}'$ is another trivialization of $\tau^{*}\mathcal{G}$, let $P$ be a hermitian line bundle with connection over $S^1$ such that $\mathcal{T}'\cong\mathcal{T} \otimes P$. Let $\sigma_1$ and $\sigma_2$ be parallel unit-length sections into $\iota_1^{*}P$ and $\iota_2^{*}P$, inducing 2-isomorphisms $\psi_1: \mathcal{T}_1 \Rightarrow \mathcal{T}_1'$ and $\psi_2: \mathcal{T}_2 \Rightarrow \mathcal{T}_2'$. These differ over the endpoints by numbers $a_0,a_1\in \ueins$, say $\psi_1|_0=\psi_2|_0\cdot a_0$ and $\psi_1|_1=\psi_2|_1\cdot a_1$. We obtain $\mathrm{Hol}_P(\id_{S^1})=a_0a_1^{-1}$, and $[\mathcal{T},z]=[\mathcal{T}',a_1a_0^{-1}\cdot z]$ in $\inf L|_{\tau}$. Further, we have (in the notation of \cref{sec:transbranes})
\begin{equation*}
r_{\psi_2}(\varphi)=(\psi_2)_1\circ \varphi \circ (\psi_2)_0^{-1}=a_1^{-1} a_0\cdot   (\psi_1)_1\circ \varphi \circ (\psi_1)_0^{-1}=a_1^{-1}a_0\cdot r_{\psi_1}(\varphi)\text{.}
\end{equation*}
We obtain
\begin{equation*}
(a_1a_0^{-1}\cdot z)\cdot r_{\psi_2}(\varphi) =r_{\psi_1}(z\varphi)\text{;}
\end{equation*}
this shows the independence.

(2) The map $\phi_{ij}$ is smooth. Consider a plot $c:U \to P_{ij}^{[2]}$. We consider an open subset $W\subset U$ with a trivialization $\mathcal{T}:(\tilde c^{\vee})^{*}\mathcal{G}\to \mathcal{I}_{\rho}$, where $\tilde c: W \to LM$ is the induced plot of $LM$, i.e.,  $\tilde c := \cup \circ c$. We obtain induced trivializations $\mathcal{T}_1 := (\iota_1 \times \id_W)^{*}\mathcal{T}$ and $\mathcal{T}_2:=(\iota_2\times \id_W)^{*}\mathcal{T}$. After choosing a trivialization  $\tau$ of the vector bundle $\mathrm{Hom}(\inf V_{\mathcal{T}},\inf W_{\mathcal{T}})$ over $W$, we obtain local trivializations $\phi_1$ of $\pr_1^{*}\inf R_{ij}$ and $\phi_2$ of $\pr_2^{*}\inf R_{ij}$ over $\smash{P_{ij}^{[2]}}$, defined by $\phi_1(w,v) := (w,[\iota_w^{*}\mathcal{T}_1,\tau(w,v)])$ and $\phi_2(w,v) := (w,[\iota_w^{*}\mathcal{T}_2,\tau(w,v)])$, according to the proof of \cref{prop:Rijvectorbundle}. Further, we have a local trivialization $\phi$ of $\cup^{*}\inf L$ defined (see \cref{lem:assbun}) by $\phi(w,v) := (w,[\iota_w^{*}\mathcal{T},v])$. By definition of the isomorphism $\phi_{ij}$, we have a commutative diagram
\begin{equation*}
\small
\alxydim{@C=4em@R=3em}{(
(W \times \C) \otimes (W \times \C^{k})\ar[d]_{\phi \otimes \phi_2} \ar[r] & W \times \C^{k} \ar[d]^{\phi_1}
\\(W \times_{P_{ij}^{[2]}} \cup^{*}\inf L) \otimes (W \times_{P_{ij}^{[2]}} \pr_2^{*}\inf R_{ij})  \ar[r]_-{\phi_{ij}} & W \times_{P_{ij}^{[2]}} \pr_1^{*}\inf R_{ij} }
\end{equation*}
whose top arrow is pointwise scalar multiplication. This shows that $\phi_{ij}$ is smooth.

(3) The bundle morphism $\phi_{ij}$  is connection-preserving. 
Let $\Gamma=(\Gamma_1,\Gamma_2)$ be a path in $\smash{P_{ij}^{[2]}}$, and let $\tau := \Gamma_1 \cup \Gamma_2$ be the induced path in $LM$. Let $\mathcal{T}: (\tau^{\vee})^{*}\mathcal{G} \to \mathcal{I}_{\rho}$ be a trivialization over $ [0,1] \times S^1$. Let $\mathcal{T}_1: (\Gamma_1^{\vee})^{*}\mathcal{G} \to \mathcal{I}_{\rho_1}$ and $\mathcal{T}_2:(\Gamma_2^{\vee})^{*}\mathcal{G} \to \mathcal{I}_{\rho_2}$ be the trivializations obtained by pullback of $\mathcal{T}$   along the maps $\id \times \iota_1,\id \times \iota_2:[0,1]^2 \to [0,1] \times S^1$, respectively. Let $\mathcal{T}_{k}(s)$ be the restriction of $\mathcal{T}_k$ to $\{s\} \times [0,1]$, for $s\in \{0,1\}$. 
By construction, we have $\inf V_{\mathcal{T}_1}=\inf V_{\mathcal{T}_2}$ and $\inf W_{\mathcal{T}_1}=\inf W_{\mathcal{T}_2}$, as well as
\begin{equation*}
\int_{[0,1] \times S^1} \rho  + \int_{[0,1]^2}\rho_2 = \int_{[0,1]^2} \rho_1\text{.}
\end{equation*}
From the definitions of the parallel transport in $\inf L$ and $\inf R_{ij}$, we conclude that the diagram
\begin{equation*}
\small
\alxydim{@C=3cm@R=3em}{\inf L|_{\tau(0)} \otimes \inf R_{ij}|_{\Gamma_2(0)}(\mathcal{T}_{2}(0)) \ar[d]_{pt_{\tau} \otimes pt_{ij}|_{\Gamma_2}} \ar[r]^-{\phi_{ij}|_{\Gamma_1(0),\Gamma_2(0)}} & \inf R_{ij}|_{\Gamma_1(0)}(\mathcal{T}_{1}(0))\ar[d]^{pt_{ij}|_{\Gamma_1}} \\
\inf L|_{\tau(1)} \otimes \inf R_{ij}|_{\Gamma_2(1)}(\mathcal{T}_{2}(1)) \ar[r]_-{\phi_{ij}|_{\Gamma_1(1),\Gamma_2(1)}} & \inf R_{ij}|_{\Gamma_1(1)}(\mathcal{T}_{1}(1))}
\end{equation*}
is commutative; this shows that $\phi_{ij}$ is connection-preserving.
\end{proof}

Next we show that the first  LBG axiom is satisfied.
\begin{lemma}
Axiom \cref{eq:lsg:fusionass}  is satisfied: $\phi_{ij}$ is a fusion representation.
\end{lemma}

\begin{proof}
We consider a triple $(\gamma_1,\gamma_2,\gamma_3) \in PM^{[3]}$ and the corresponding loops $\tau_{ab} := \gamma_a \cup \gamma_b$. Using trivializations $\mathcal{T}_{ab}: \tau_{ab}^{*}\mathcal{G} \to \mathcal{I}_0$, and 2-isomorphisms 
$\phi_1\maps \iota_1^{*}\mathcal{T}_{12} \Rightarrow \iota_1^{*}\mathcal{T}_{13}$, $\phi_2\maps \iota_2^{*}\mathcal{T}_{12} \Rightarrow \iota_1^{*}\mathcal{T}_{23}$ and $\phi_3\maps \iota_2^{*}\mathcal{T}_{23} \Rightarrow \iota_2^{*}\mathcal{T}_{13}$,
\cref{eq:lsg:fusionass} implies the commutativity of the following diagram:
\begin{equation*}
\small
\alxydim{@C=2cm@R=3em}{\inf L|_{\gamma_1 \cup \gamma_2}  \otimes \inf L|_{\gamma_2\cup\gamma_3}  \otimes\inf  R_{ij}|_{\gamma_3}(\iota_2^{*}\mathcal{T}_{23}) \ar[dd]_{\lambda_{\gamma_1,\gamma_2,\gamma_3} \otimes \id} \ar[r]^-{\id \otimes \phi_{ij}|_{\gamma_2,\gamma_3}} &  \inf L_{\gamma_1 \cup \gamma_2} \otimes \inf R_{ij}|_{\gamma_2}(\iota_1^{*}\mathcal{T}_{23}) \ar[d]^{\id \otimes r_{\phi_2}} \\ & \inf L_{\gamma_1 \cup \gamma_2} \otimes \inf R_{ij}|_{\gamma_2}(\iota_2^{*}\mathcal{T}_{12}) \ar[d]^{\phi_{ij}|_{\gamma_1,\gamma_2}}\\ \inf L|_{\gamma_1 \cup \gamma_3}  \otimes\inf  R_{ij}|_{\gamma_3}(\iota_2^{*}\mathcal{T}_{23}) \ar[d]_{\id \otimes r_{\phi_3}^{-1}} & \inf R_{ij}|_{\gamma_1}(\iota_1^{*}\mathcal{T}_{12}) \ar[d]^{(r_{\phi_1})^{-1}} \\ \inf L_{\gamma_1\cup \gamma_3} \otimes \inf R_{ij}|_{\gamma_3}(\iota_2^{*}\mathcal{T}_{13}) \ar[r]_-{\phi_{ij}|_{\gamma_1,\gamma_3}} & \inf R_{ij}|_{\gamma_1}(\iota_1^{*}\mathcal{T}_{13})\text{.}}
\end{equation*}
In this diagram, the bundle morphisms $\phi_{ij}$, $\phi_{jk}$ and $\phi_{ik}$ are just scalar multiplication. 
We may additionally assume that $\lambda([\mathcal{T}_{12},z_{12}] \otimes [\mathcal{T}_{23},z_{23}])=[\mathcal{T}_{13},z_{12}z_{23}]$, reducing the commutativity of the diagram to the equation $r_{\phi_1}^{-1} \circ r_{\phi_2} = r_{\phi_3}^{-1}$. By definition of the fusion product and of the morphisms $r_{\phi_a}$, this equation follows from the identities between the 2-isomorphisms $\phi_1$, $\phi_2$ and $\phi_3$.
\end{proof}

\subsection{Lifted path concatenation}

\label{sec:liftpathcommp}

We equip the vector bundle $\inf R_{ij}$ over $P_{ij}$ with a lifted path concatenation, and start by constructing a linear map 
\begin{equation*}
\chi_{ijk}|_{\gamma_{12},\gamma_{23}}: \inf R_{jk}|_{\gamma_{23}} \otimes \inf R_{ij}|_{\gamma_{12}} \to \inf R_{ik}|_{\gamma_{23} \pcomp \gamma_{12}}
\end{equation*} 
for all $(\gamma_{23},\gamma_{12}) \in P_{jk} \times_{Q_j} P_{ij}$ and $i,j,k\in I$. Let $x:= \gamma_{12}(0)$, $y:= \gamma_{12}(1)=\gamma_{23}(0)$ and $z := \gamma_{23}(1)$. Consider a trivialization $\mathcal{T}$ of $(\gamma_{23}\pcomp \gamma_{12})^{*}\mathcal{G}$, and let trivializations $\mathcal{T}_{12}$ of $\gamma_{12}^{*}\mathcal{G}$ and $\mathcal{T}_{23}$ of $\gamma_{23}^{*}\mathcal{G}$ be defined via restriction, i.e. $\mathcal{T}_{12}:= \iota_1^{*}\mathcal{T}$ and $\mathcal{T}_{23} :=  \iota_2^{*}\mathcal{T}$, where  $\iota_1,\iota_2: [0,1] \to [0,1]$ are defined by $\iota_1(t):= \frac{1}{2}t$ and $\iota_2(t) :=\frac{1}{2}+\frac{1}{2}t$. We find
\begin{align*}
\inf R_{ij}|_{\gamma_{12}}(\mathcal{T}_{12}) &= \mathrm{Hom}(\Des(\mathcal{E}_i|_x,\mathcal{T}|_{0}),\Des(\mathcal{E}_j|_y,\mathcal{T}|_{\frac{1}{2}}))
\\  \inf R_{jk}|_{\gamma_{23}}(\mathcal{T}_{23}) &= \mathrm{Hom}(\Des(\mathcal{E}_j|_y,\mathcal{T}|_{\frac{1}{2}}), \Des(\mathcal{E}_k|_z,\mathcal{T}|_{1}))
\\\inf R_{ik}|_{\gamma_{23} \circ \gamma_{12}} (\mathcal{T}) &= \mathrm{Hom}(\Des(\mathcal{E}_i|_x,\mathcal{T}|_{0}), \Des(\mathcal{E}_k|_z,\mathcal{T}|_{1}))\text{.}
\end{align*}
Thus, we define 
\begin{equation*}
\inf R_{jk}|_{\gamma_{23}}(\mathcal{T}_{23})\otimes \inf R_{ij}|_{\gamma_{12}}(\mathcal{T}_{12}) \to \inf R_{ik}|_{\gamma_{23} \pcomp \gamma_{12}} (\mathcal{T})
\end{equation*}
simply as the composition of linear maps, i.e., $\varphi_{23} \otimes \varphi_{12} \mapsto \varphi_{23} \circ \varphi_{12}$.

\begin{lemma}
This defines a connection-preserving bundle morphism 
\begin{equation*}
\chi_{ijk}:\pr_1^{*}\inf R_{jk} \otimes \pr_2^{*}\inf R_{ij} \to \star^{*}\inf R_{ik}\text{.}
\end{equation*}
\end{lemma}

\begin{proof}
(1) Independence of the choice of the trivialization is  straightforward to see and left out for brevity.

(2) In order to see that $\chi_{ijk}$ is a smooth bundle morphism we represent it in local trivializations. Consider a plot $c:U \to P_{jk} \times_{Q_j} P_{ij}$, and let $c_{12}$, $c_{23}$ and $c_{13}$ be the pointwise projections to $P_{ij}$ and $P_{jk}$, and the pointwise concatenation, respectively. We restrict to a contractible open subset $W\subset U$ such that there is a trivialization $\mathcal{T}: (c_{13}^{\vee})^{*}\mathcal{G} \to \mathcal{I}_{\rho}$. We induce trivializations $\mathcal{T}_{12}$ and $\mathcal{T}_{23}$ of $(c_{12}^{\vee})^{*}\mathcal{G}$ and $(c_{23}^{\vee})^{*}\mathcal{G}$ by pullback along $\id_W \times \iota_1$ and $\id_W \times \iota_2$, respectively. We then consider the vector bundles 
\begin{equation*}
\inf V_1 := \Des((\ev_0 \circ c_{13})^{*}\mathcal{E}_i,j_0^{*}\mathcal{T})
\text{, }
 \inf V_2 := \Des((\ev_{\frac{1}{2}} \circ  c_{13})^{*}\mathcal{E}_j,j_{\frac{1}{2}}^{*}\mathcal{T})
\text{ and }
\inf V_3 := \Des((\ev_1 \circ  c_{13})^{*}\mathcal{E}_i,j_1^{*}\mathcal{T})
\end{equation*}
 over $W$. We claim that there exist trivializations $\tau_{ab}$ of the Hom-bundles $\mathrm{Hom}(\inf V_a,\inf V_b) $ such that the diagram
\begin{equation*}
\small
\alxydim{@C=5em@R=3em}{(W \times \C^{n_3\times n_2}) \otimes (W \times \C^{n_2 \times n_1}) \ar[d]_{\tau_{23} \otimes \tau_{12}} \ar[r]^-{\cdot} & W \times \C^{n_3 \times n_1} \ar[d]^{\tau_{13}}  
\\\mathrm{Hom}(\inf V_2,\inf V_2) \otimes \mathrm{Hom}(\inf V_1,\inf V_2) \ar[r]_-{\circ} & \mathrm{Hom}(\inf V_1,\inf V_3)}
\end{equation*} 
is commutative, with pointwise matrix multiplication in the top row.  These can be found by choosing trivializations of the bundles $\inf V_a$ separately, and then inducing trivializations of the Hom-bundles via $\mathrm{Hom}(\inf V_a,\inf V_b)=\inf V_a^{*}\otimes \inf V_b$. 
 According to the proof of \cref{prop:Rijvectorbundle}, the gerbe trivializations $\mathcal{T}_{ab}$ and the bundle isomorphisms $\tau_{ab}$ induce local trivializations $\phi_{jk}$ of $c_{23}^{*}\inf R_{jk}$, $\phi_{ij}$ of $c_{12}^{*}\inf R_{ij}$, and $\phi_{ik}$ of $c_{13}^{*}\inf R_{ik}$. By construction, the diagram
\begin{equation*}
\small
\alxydim{@C=6em@R=3em}{W \times (\C^{n_{k}\times n_j} \otimes \C^{n_{j}\times n_i}) \ar[d]_{\phi_{jk}\otimes \phi_{ij}} \ar[r]^-{\cdot} & W \times \C^{n_{k}\times n_i} \ar[d]^{\phi_{ik}} \\ W \times_{P_{jk} \times_{Q_j} P_{ij}} (\pr_{jk}^{*}\inf R_{jk} \otimes \pr_{ij}^{*}\inf R_{ij}) \ar[r]_-{\id \times \chi_{ijk}} & W \times_{P_{jk} \times_{Q_j} P_{ij}} (\pcomp^{*}\inf R_{ik}) }
\end{equation*}
is commutative. This shows that $\chi_{ijk}$ is smooth under local trivializations, and hence smooth.

(3) The bundle morphism $\chi_{ijk}$ is connection-preserving.
Indeed, consider a path $(\Gamma_{23},\Gamma_{12})$ in $P_{jk} \times_{Q_j} P_{ij}$, and let $\Gamma$ be its pointwise concatenation. Let $\mathcal{T}: (\Gamma^{\vee})^{*}\mathcal{G} \to \mathcal{I}_{\rho}$ be a trivialization, and let $\mathcal{T}_{12}: (\Gamma_{12}^{\vee})^{*}\mathcal{G} \to \mathcal{I}_{\rho_{12}}$ and $\mathcal{T}_{23}: (\Gamma_{23}^{\vee})^{*}\mathcal{G} \to \mathcal{I}_{\rho_{23}}$ be the pullbacks of $\mathcal{T}$ under $(\id \times \iota_1)$ and $(\id \times \iota_2)$, respectively. We consider the  hermitian vector bundles with connections
\begin{align*}
\inf U := \Des((\ev_0 \circ \Gamma)^{*}\mathcal{E}_i,j_0^{*}\mathcal{T}) \text{, }\quad
\inf V := \Des((\ev_{\frac{1}{2}} \circ \Gamma)^{*}\mathcal{E}_j,j_{\frac{1}{2}}^{*}\mathcal{T}) \;\text{ and }\;
\inf W := \Des((\ev_1 \circ \Gamma)^{*}\mathcal{E}_k,j_1^{*}\mathcal{T}) \end{align*}
over $[0,1]$, so that 
\begin{align*}
\inf R_{ij}|_{\Gamma_{12}(s)}(\mathcal{T}_{12}|_{\{s\} \times [0,1]}) = \mathrm{Hom}(\inf U,\inf V)|_s
\quomma
\inf R_{jk}|_{\Gamma_{23}(s)}(\mathcal{T}_{23}|_{\{s\} \times [0,1]}) &= \mathrm{Hom}(\inf V,\inf W)|_s
\\\quand
\inf R_{ik}|_{\Gamma(s)}(\mathcal{T}|_{\{s\} \times [0,1]}) &= \mathrm{Hom}(\inf U,\inf W)|_s\text{,}
\end{align*}
and  lifted path concatenation is  the composition \begin{equation}
\label{eq:comphombundles}
\mathrm{Hom}(\inf V,\inf W) \otimes\mathrm{Hom}(\inf U,\inf V) \to \mathrm{Hom}(\inf U,\inf W)\text{.} \end{equation}
It is elementary to see that \cref{eq:comphombundles} is a connection-preserving bundle morphism over $[0,1]$. 
Further, we find 
\begin{equation*}
\int_{[0,1]^2} (\Gamma^{\vee})^{*}\rho =\int_{[0,1]^2} (\Gamma_{12}^{\vee})^{*}\rho_{12}+\int_{[0,1]^2} (\Gamma_{23}^{\vee})^{*}\rho_{23}
\end{equation*}
for the integrals of the 2-form of the trivializations.
These results prove the claim.
\end{proof}

Now that we have established the lifted path concatenation in $\inf R_{ij}$, we are in position to show that the next two LBG axioms are satisfied.

\begin{lemma}
\label{lem:LBG2}
Axiom \cref{eq:lsg:pentagon} is satisfied: the maps  $\chi_{ijk}$ are associative up to reparameterization. 
\end{lemma}

\begin{proof}
In order to prove the pentagon diagram of \cref{eq:lsg:pentagon},
we have to choose a path $\Gamma$ in $P_{il}$ connecting  $(\gamma_{34} \pcomp \gamma_{23}) \pcomp \gamma_{12}$ with $\gamma_{34} \pcomp (\gamma_{23} \pcomp \gamma_{12})$. In order to do so, we let $\varphi:[0,1] \to [0,1]$ be a smooth reparameterization such that $((\gamma_{34} \pcomp \gamma_{23}) \pcomp \gamma_{12})(t) =(\gamma_{34} \pcomp (\gamma_{23} \pcomp \gamma_{12}))(\varphi(t))$ for $t\in[0,1]$. 
We let $\Gamma$ be induced from a fixed-ends homotopy $h$ between $\id_{[0,1]}$ and $\varphi$. We let $\mathcal{T}_1$ be a trivialization of  $(\gamma_{34} \pcomp (\gamma_{23} \pcomp \gamma_{12}))^{*}\mathcal{G}$, and then define a trivialization of $\mathcal{G}$ along $\Gamma$  by $\mathcal{T} :=h^{*}\mathcal{T}_1$.  We set $\mathcal{T}_0 := \mathcal{T}|_{\{0\} \times [0,1]} = \varphi^{*}\mathcal{T}_1$. We note that $\Gamma$ is a thin path with fixed endpoints, and that $\mathcal{T}$ is a trivialization with vanishing 2-form. Thus, by definition of the parallel transport $pt_{il}|_{\Gamma}$, 
\begin{equation*}
d_{(\gamma_{34} \pcomp \gamma_{23}) \pcomp \gamma_{12},\gamma_{34} \pcomp (\gamma_{23} \pcomp \gamma_{12})}: \inf R_{il}|_{(\gamma_{34} \pcomp \gamma_{23}) \pcomp \gamma_{12}}(\mathcal{T}_0) \to \inf R_{il}|_{\gamma_{34} \pcomp (\gamma_{23} \pcomp \gamma_{12})}(\mathcal{T}_1)
\end{equation*}
is the identity map. The pentagon diagram we have to prove thus reads as
\begin{equation*}
\small
\alxydim{@C=0.4cm@R=3em}{ 
\inf R_{kl}|_{\gamma_{34}} (\iota_2^{*}\iota_2^{*}\mathcal{T}_0)\otimes \inf R_{jk}|_{\gamma_{23}} (\iota_1^{*}\iota_2^{*}\mathcal{T}_0)\otimes  \inf R_{ij}|_{\gamma_{12}}(\iota_1^{*}\mathcal{T}_0) 
\ar[r] \ar[d]_{\chi_{jkl}|_{\gamma_{23},\gamma_{34}} \otimes \id} 
& \inf R_{kl}|_{\gamma_{34}} (\iota_2^{*}\mathcal{T}_1)\otimes \inf R_{jk}|_{\gamma_{23}}(\iota_2^{*}\iota_1^{*}\mathcal{T}_1) \otimes  \inf R_{ij}|_{\gamma_{12}}(\iota_1^{*}\iota_1^{*}\mathcal{T}_1)  
\ar[d]^{\id \otimes \chi_{ijk}|_{\gamma_{12},\gamma_{23}}}  
\\  \inf R_{jl}|_{\gamma_{34} \pcomp \gamma_{23}}  (\iota_2^{*}\mathcal{T}_0)\otimes  \inf R_{ij}|_{\gamma_{12}}(\iota_1^{*}\mathcal{T}_0) \ar[d]_{\chi_{ijl}|_{\gamma_{12},\gamma_{34} \pcomp \gamma_{23}}} 
&  \inf R_{kl}|_{\gamma_{34}} (\iota_2^{*}\mathcal{T}_1)\otimes  \inf R_{ik}|_{\gamma_{23} \pcomp \gamma_{12}}(\iota_1^{*}\mathcal{T}_1) \ar[d]^{\chi_{ikl}|_{\gamma_{23} \pcomp \gamma_{12},\gamma_{34}}} 
\\   \inf R_{il}|_{(\gamma_{34} \pcomp \gamma_{23}) \pcomp \gamma_{12}}(\mathcal{T}_0) \ar@{=}[r] 
&  \inf R_{il}|_{\gamma_{34} \pcomp (\gamma_{23} \pcomp \gamma_{12})}(\mathcal{T}_1)\text{.} }
\end{equation*}
The arrow on top consists of the canonical identifications $r$ relating different choices of trivializations, and they are all identities according to the subsequent \cref{rem:reparameterizations}. The remaining diagram then commutes  due to the associativity of the composition of maps. 
\end{proof}

\begin{remark}
\label{rem:reparameterizations}
Suppose $\gamma \in P_{ij}$ and $\varphi:[0,1] \to [0,1]$ is a smooth map with $\varphi(0)=0$ and $\varphi(1)=1$, such that $\gamma \circ \varphi=\gamma$. Let $\mathcal{T}:\gamma^{*}\mathcal{G} \to \mathcal{I}_0$ be a trivialization, and let $\mathcal{T}':= \varphi^{*}\mathcal{T}$. We have $\inf R_{ij}|_{\gamma}(\mathcal{T}) = \inf R_{ij}|_{\gamma}(\mathcal{T}')$, and we claim that $r_{\mathcal{T},\mathcal{T}'}=\id$. In order to see this, we may choose a smooth homotopy $h:[0,1]^2 \to [0,1]$ between the identity and $\varphi$, with fixed ends ($h(s,0)=0$, $h(s,1)=1$ for all $s\in[0,1]$)  that fixes the path, $\gamma(h(s,t))=\gamma(t)$ for all $s,t\in[0,1]$.  Then, we consider the trivializations $h^{*}\mathcal{T}$ and $\pr_2^{*}\mathcal{T}$, where $\pr_2: [0,1]^2 \to [0,1]$ is the projection. The hermitian line bundle $L := \Des(h^{*}\mathcal{T},\pr_2^{*}\mathcal{T})$ has a flat connection and thus admits a parallel unit-length section over $(\{0\} \times [0,1]) \cup ([0,1] \times \{0\}) \cup ([0,1] \times \{1\})$, where the two trivializations agree. With the connection, the section can be extended to all of $[0,1]^2$. In particular, we obtain a parallel section $\sigma$ of $\Des(\mathcal{T}',\mathcal{T})$ that is the identity over the end points. This shows  $r_{\mathcal{T},\mathcal{T}'}=\id$.

\end{remark}

\begin{lemma}
Axiom \cref{eq:lsg:comp} is satisfied: $\phi_{ij}$ and $\chi_{ijk}$ are compatible with each other. 
\end{lemma}

\begin{proof}
We consider paths  $\gamma_{12},\gamma_{12}'\in P_{ij}$ and $\gamma_{23},\gamma_{23}'\in P_{jk}$ as in  \cref{eq:lsg:comp}, and form the loops 
\begin{equation*}
\tau := (\gamma_{23} \pcomp \gamma_{12}) \cup (\gamma_{23}'\pcomp \gamma_{12}')
\quomma
\tau_{12} := \gamma_{12} \cup \gamma_{12}'
\quomma
\tau_{23} := \gamma_{23}\cup\gamma_{23}'\text{.}
\end{equation*} 
We consider corresponding trivializations of $\mathcal{G}$, namely  $\mathcal{T}$ along $\tau$, $\mathcal{T}_{12}$ along $\tau_{12}$, and $\mathcal{T}_{23}$ along $\tau_{23}$. 
The first objective is to compare these trivializations  wherever two of them are defined. Over the common point $y$, we fix a 2-isomorphism $\sigma: \smash{\mathcal{T}_{12}|_{\frac{1}{2}}}\Rightarrow \mathcal{T}_{23}|_0$. 
Since all paths have sitting instants there exist maps $\varphi,\varphi_{12},\varphi_{23}: S^1 \to S^1$ such that 
\begin{align*}
\tau_{23}\circ \varphi_{23}\circ \iota_1 &=\tau_{23}= \tau \circ \varphi\circ \iota_1\text{,}
\\
\tau_{23} \circ \varphi_{23}\circ \iota_1&=\cp_y=\tau_{12}\circ \varphi_{12}\circ \iota_1\text{,}
\\
\tau_{12} \circ \varphi_{12} \circ \iota_2 &=\tau_{12}\circ \mathrm{rot}_{\pi} = \tau \circ \varphi \circ \iota_2\text{,}
\end{align*}
where $\mathrm{rot}_{\pi}:S^1 \to S^1$ is the rotation by an angle of $\pi$. Due to the first and the third of these identities, we can find 2-isomorphisms $\phi_1: \iota_1^{*}\varphi_{23}^{*}\mathcal{T}_{23} \Rightarrow \iota_1^{*}\varphi^{*}\mathcal{T}$ and $\phi_3: \iota_2^{*}\varphi_{12}^{*}\mathcal{T}_{12}\Rightarrow \iota_2^{*}\varphi^{*}\mathcal{T}$. Because of the second identity, we have a 2-isomorphism $\phi_2 := \cp_y^{*}\sigma: \iota_2^{*}\varphi_{23}^{*}\mathcal{T}_{23} \Rightarrow \iota_1^{*}\varphi_{12}^{*}\mathcal{T}_{12}$. We claim that we can choose these 2-isomorphism such that 
\begin{equation}
\label{eq:LBG3:fusion}
\phi_1|_0 = \phi_3|_0 \bullet \phi_2|_0
\quand
\phi_1|_1 = \phi_3|_1 \bullet \phi_2|_1\text{.}
\end{equation}
Indeed, the first equation can be used  to re-define $\sigma$ such that this first equation is satisfied. Since $\phi_2=\cp_y^{*}\sigma$, we cannot repeat this for the second equation, so that we first obtain an error, a number $z\in \ueins$. We consider the hermitian line bundle $L_z$ over $S^1$ with connection of holonomy $z$. Then we replace $\mathcal{T}_{12}$ by $\mathcal{T}_{12} \otimes L_z$ and repeat the whole construction by fixing a new 2-isomorphism $\sigma$; then both equations in \cref{eq:LBG3:fusion} are satisfied. Comparing with the definition of the fusion product given in \cref{sec:trans:gerbes}, we have that
\begin{equation*}
[\mathcal{T}_{23},z_{23}] \otimes [\mathcal{T}_{12},z_{12}] \mapsto [\varphi_{23}^{*}\mathcal{T}_{23},z_{23}] \otimes [\varphi_{12}^{*}\mathcal{T}_{12},z_{12}] \mapsto [\varphi^{*}\mathcal{T},z_{12}z_{23}]\mapsto [\mathcal{T},z_{12}z_{23}]
\end{equation*} 
realizes the isomorphism $\lambda'$ on top of the diagram of \cref{eq:lsg:comp}.
The remaining arrows of the diagram are labelled with the lifted path concatenation and the fusion representation, which, by our choices of trivializations, are just composition and scalar multiplication of linear maps, and the commutativity of the diagram reduces to the trivial fact that $z_{12}z_{23}\cdot (\varphi_{23}\circ \varphi_{12})=(z_{23}\varphi_{23})\circ (z_{12}\varphi_{12})$.  
\end{proof}

\subsection{Lifted constant paths}

\label{sec:liftedconstandrev}

We equip the vector bundle $\inf R_{ii}$ with lifted constant paths.
 To that end, we consider $x\in Q_i$ and choose a trivialization $\mathcal{T}: \mathcal{G}|_x \to \mathcal{I}_0$. We have
$\inf R_{ii}|_{\cp_x}(\cp_x^{*}\mathcal{T}) = \mathrm{End}(\Des(\mathcal{E}_i|_x,\mathcal{T}))$, and readily  define $\epsilon_i(x) := \id_{\Des(\mathcal{E}_i|_x,\mathcal{T})}$. 

\begin{lemma}
The assignment $x \mapsto \epsilon_i(x)$ defines a smooth, parallel section along $\cp: Q_i \to P_{ii}$.
\end{lemma}

\begin{proof}
If $\mathcal{T}'$ is another trivialization of $\mathcal{G}|_x$, then there exists a 2-isomorphism $\psi: \mathcal{T} \Rightarrow \mathcal{T}'$ inducing the linear map $r_{\cp_x^{*}\psi}: \inf R_{ii}|_{\cp_x}(\cp_x^{*}\mathcal{T}') \to \inf R_{ii}|_{\cp_x}(\cp_x^{*}\mathcal{T})$. We have $r_{\cp_x^{*}\psi}(\id_{\Des(\mathcal{E}_i|_x,\mathcal{T}')})=\id_{\Des(\mathcal{E}_i|_x,\mathcal{T})}$; this shows well-definedness.
Next, we show directly that $\epsilon_i: Q_i \to \inf R_{ii}$ is smooth. Let $f:U \to Q_i$ be a plot of $Q_i$, i.e., a smooth map defined on an open subset $U\subset \R^{n}$. We have to show that $\tilde c:=\epsilon_i \circ f: U \to \inf R_{ii}$ is a plot. Note that $c:=\pi\circ \tilde c=\cp\circ f$, where $\cp:Q_i\to P_{ii}$ is  the assignment of constant paths, and $c^{\vee}:U \times [0,1]\to P_{ii}$ is the map $(u,t)\mapsto \cp_{f(u)}(t)=f(u)$. We choose an open subset  $W\subset U$ such that there exists a trivialization $\mathcal{T}: f^{*}\mathcal{G}|_W \to \mathcal{I}_{\rho}$. Then we set $\mathcal{T}' := \pr_W^{*}\mathcal{T}$ for $\pr_W: [0,1]\times W \to W$, so that $\mathcal{T}'$ is a trivialization of $(c^{\vee})^{*}\mathcal{G}$. The relevant vector bundle over $W$ is 
\begin{align*}
\mathrm{Hom}(\Des((\ev_0 \circ c)^{*}\mathcal{E}_i,j_0^{*}\mathcal{T}'),\Des((\ev_1\circ c)^{*}\mathcal{E}_i,j_1^{*}\mathcal{T}'))
&=\mathrm{End}(\Des(f^{*}\mathcal{E}_i,\mathcal{T}))\text{.}
\end{align*}
We note that $w\mapsto \tau(w) := \id_{\Des(f^{*}\mathcal{E}_i,\mathcal{T})|_w}$ is a smooth section of this bundle over $W$. By definition of the diffeology on $\inf R_{ii}$, we have that $w\mapsto [i_w^{*}\mathcal{T}',\tau(w)]=[\cp_{w}^{*}\mathcal{T},\id_{\Des(\mathcal{E}_i|_{f(w)},\mathcal{T}|_w)}]=\epsilon_i(f(w))=\tilde c(w)$ is a plot, which was to show.

If $\gamma$ is a path in $Q_i$, let $\mathcal{T}$ be a trivialization along $\gamma$, and consider the trivialization $\mathcal{S}:=\pr_2^{*}\mathcal{T}$ over $[0,1]^2$. In particular, $\mathcal{S}$ is a trivialization along the path $t\mapsto \cp_{\gamma(t)}$ connecting $\cp_x$ with $\cp_y$ in $P_{ii}$, and it has vanishing 2-form.   We let $\inf V := \Des(\gamma^{*}\mathcal{E}_i,\mathcal{T})$. Then, the parallel transport in $\inf R_{ii}$ along $\cp_{\gamma(-)}$ is just the parallel transport in $\mathrm{End}(\inf V)$:
\begin{equation*}
pt_{ii}|_{\cp_{\gamma(-)}}: \mathrm{End}(\inf V)|_0 \to \mathrm{End}(\inf V)|_1\text{.}
\end{equation*}
Since the section $\id$ into $\mathrm{End}(\inf V)=\inf V^{*} \otimes \inf V$ is parallel, we have the claim.
\end{proof}

\begin{lemma}
\label{lem:LBG4}
Axiom \cref{eq:lsg:unit} is satisfied: $\epsilon_i$ provides units up to reparameterization  for $\chi_{ijk}$.
\end{lemma}

\begin{proof}
Let $\mathcal{T}_0$ be a trivialization of $\mathcal{G}$ along a path $\gamma$ from $x$ to $y$, and let $\varphi:[0,1] \to [0,1]$ be a smooth map such that $\gamma\circ \varphi=\gamma\pcomp \cp_x$ (this uses the sitting instants of $\gamma$). Let $h$ be a fixed-ends homotopy between $\id_{[0,1]}$ and $\varphi$, and let $\mathcal{T}:=h^{*}\mathcal{T}$ as well as $\mathcal{T}_1 := \mathcal{T}|_{\{1\} \times [0,1]}=\varphi^{*}\mathcal{T}_0$. These are  trivializations with vanishing 2-forms and they agree on $[0,1] \times \{0,1\}$, so that
\begin{equation*}
d_{\gamma,\gamma\pcomp \cp_x}: \inf R_{ij}|_{\gamma}(\mathcal{T}_0) \to \inf R_{ij}|_{\gamma \pcomp \cp_x}(\mathcal{T}_1)
\end{equation*}
is the identity by definition of the connection $pt_{ij}$.  Moreover,
\begin{equation*}
\chi_{iij}|_{\cp_x,\gamma}: \inf R_{ij}|_{\gamma}(\iota_2^{*}\mathcal{T}_1) \otimes \inf R_{ii}|_{\cp_x}(\iota_1^{*}\mathcal{T}_1) \to \inf R_{ij}|_{\gamma\pcomp \cp_x}(\mathcal{T}_1)
\end{equation*}
is the composition. Due to \cref{rem:reparameterizations} we have identities $\inf R_{ii}|_{\cp_x}(\iota_1^{*}\mathcal{T}_1) =\inf R_{ii}|_{\cp_x}(\cp_0^{*}(\mathcal{T}_0|_0))$ and $\inf R_{ij}|_{\gamma}(\iota_2^{*}\mathcal{T}_1) = \inf R_{ij}|_{\gamma}(\mathcal{T}_0)$. The first shows that $\epsilon_i(x)=\id\in\inf R_{ii}|_{\cp_x}(\iota_1^{*}\mathcal{T}_1)$. The second shows that $\inf R_{ij}|_{\gamma}(\iota_2^{*}\varphi^{*}\mathcal{T})= \inf R_{ij}|_{\gamma}(\mathcal{T}_0)$, and together with the first we have
\begin{equation*}
\chi_{iij}|_{\cp_x,\gamma}(v,\epsilon_i(x)) = d_{\gamma,\gamma\pcomp \cp_x}(v)\text{,}
\end{equation*} 
which is the first half of \cref{eq:lsg:unit}. The second half is proved analogously. 
\end{proof}

\subsection{Lifted path reversal}

\label{sec:trans:liftedpathreversal}

We equip the vector bundle $\inf R_{ij}$ with a lifted path reversal.
Let $\gamma\in P_{ij}$ with $x:= \gamma(0)$ and $y:=\gamma(1)$. We choose a trivialization $\mathcal{T}$ of $\gamma^{*}\mathcal{G}$, and let $\overline{\mathcal{T}} := \rev^{*}\mathcal{T}$, where $\rev:[0,1] \to [0,1]$ is defined by $\rev(t):=1-t$. We have
\begin{equation*}
\inf R_{ji}|_{\overline{\gamma}}(\overline{\mathcal{T}}) = \mathrm{Hom}(\Des(\mathcal{E}_j|_y,\overline{\mathcal{T}}|_0),\Des(\mathcal{E}_i|_x,\overline{\mathcal{T}}|_1))= \mathrm{Hom}(\Des(\mathcal{E}_j|_y,\mathcal{T}|_1),\Des(\mathcal{E}_i|_x,\mathcal{T}|_0))\text{.}
\end{equation*} 
We  define the lifted path reversal by forming the adjoint linear map with respect to the given hermitian metrics, 
\begin{equation*}
\alpha_{ij}|_{\gamma}: \inf R_{ij}|_{\gamma}(\mathcal{T}) \to \overline{\inf R_{ji}}|_{\overline{\gamma}}(\overline{\mathcal{T}}):\varphi \mapsto \varphi^{*}\text{.}
\end{equation*}

\begin{lemma}
This defines  a connection-preserving  bundle morphism.
\end{lemma}

\begin{proof}
If $\mathcal{T}'$ is another trivialization we consider a 2-isomorphism $\psi: \mathcal{T} \Rightarrow \mathcal{T}'$
and the corresponding canonical identification $r_{\psi}(\varphi) = \psi_1\circ \varphi \circ \psi_0^{-1}$, for  $\varphi \in \inf R_{ij}|_{\gamma}(\mathcal{T}')$.
Since $\psi_0$ and $\psi_1$ are isometric isomorphisms, we have
\begin{equation*}
(r_{\psi}(\varphi))^{*} = (\psi_1\circ \varphi \circ \psi_0^{-1})^{*} =\psi_0\circ \varphi^{*} \circ \psi_1^{-1} =(\rev^{*}\psi)_1^{*}\circ \varphi^{*} \circ (\rev^{*}\psi)_0^{-1}=r_{\rev^{*}\psi}(\varphi^{*})\text{,} \end{equation*}
where $\rev^{*}\psi: \overline{\mathcal{T}} \to \overline{\mathcal{T}'}$. This shows the independence of the choice of the trivialization.    

In the following we consider a smooth map $f:W \to P_{ij}$, where $W$ is a smooth manifold, that admits  a trivialization $\mathcal{T}:(f^{\vee})^{*}\mathcal{G} \to \mathcal{I}_{\rho}$ defining the vector bundles $\inf V_{\mathcal{T}}$ and $\inf W_{\mathcal{T}}$ over $W$. We consider the pointwise path reversal, $\rev_W: W \times [0,1] \to W \times [0,1]: (w,t) \mapsto (w,1-t)$. Then, $\overline{\mathcal{T}} := \rev_W^{*}\mathcal{T}$ is a trivialization of $\mathcal{G}$ along $f^{\vee} \circ \rev_W$. We have 
\begin{equation*}
\inf V_{\mathcal{T}} = \Des((\ev_0 \circ f^{\vee})^{*}\mathcal{E}_i,j_0^{*}\mathcal{T}) =\Des((\ev_1 \circ f^{\vee} \circ R)^{*}\mathcal{E}_i,j_1^{*}\overline{\mathcal{T}})= \inf W_{\overline{\mathcal{T}}}\text{,}
\end{equation*}
and similarly $\inf W_{\mathcal{T}}=\inf V_{\overline{\mathcal{T}}}$. Now suppose $f:=c|_W$ is the restriction of a plot $c:U \to P_{ij}$ to an open subset $W \subset U$, and we have a trivialization $\tau:W \times \C^{k} \to \mathrm{Hom}(\inf V_{\mathcal{T}},\inf W_{\mathcal{T}})$, inducing a local trivialization of $\inf R_{ij}$. Taking  the pointwise adjoint  defines a local trivialization $\overline{\tau}: W \times \C^{k} \to \overline{\mathrm{Hom}(\inf W_{\mathcal{T}},\inf V_{\mathcal{T}})}=\overline{\mathrm{Hom}(\inf V_{\overline{\mathcal{T}}},\inf W_{\overline{\mathcal{T}}})}$, inducing a local trivialization of $\overline{\inf R_{ji}}$. Under these local trivializations, $\alpha_{ij}$  is the identity, and hence smooth.

In order to see that $\alpha_{ij}$ is connection-preserving, we consider a path $\Gamma \in PP_{ij}$ and compare the parallel transport $pt_{ij}|_{\Gamma}$ with $pt_{ji}|_{\widetilde\rev\circ \Gamma}$, where $\tilde\rev:P_{ij} \to P_{ji}:\gamma\mapsto \bar\gamma$. We put $f := \Gamma$ in the situation described above. It is elementary to see that  taking adjoints in the bundle of homomorphisms between hermitian vector bundles preserves induced connections. In the present case, we have a commutative diagram
\begin{equation*}
\small
\alxydim{@R=3em}{\mathrm{Hom}(\inf V_{\mathcal{T}},\inf W_{\mathcal{T}})|_0 \ar[d]_{pt} \ar[r]^-{()^{*}} & \overline{\mathrm{Hom}(\inf W_{\mathcal{T}},\inf V_{\mathcal{T}})}|_0  \ar[d]^{pt} \ar@{=}[r] & \overline{\mathrm{Hom}(\inf V_{\overline{\mathcal{T}}},\inf W_{\overline{\mathcal{T}}})}|_0 \ar[d]^{pt}
\\
\mathrm{Hom}(\inf V_{\mathcal{T}},\inf W_{\mathcal{T}})|_1 \ar[r]_-{()^{*}} & \overline{\mathrm{Hom}(\inf W_{\mathcal{T}},\inf V_{\mathcal{T}})}|_1 \ar@{=}[r]   & \overline{\mathrm{Hom}(\inf V_{\overline{\mathcal{T}}},\inf W_{\overline{\mathcal{T}}})}|_1}
\end{equation*}
It remains to compare the integral of the 2-form $\rho$ over $\Gamma^{\vee}$ with the integral of $\rev_W^{*}\rho$ over $(\widetilde\rev\circ\Gamma)^{\vee}= \rev_{W}\circ \Gamma^{\vee}$. Since $\rev_W$ is orientation-reversing, we get the opposite sign. In the definition of $pt_{ji}|_{\widetilde\rev\circ\Gamma}$, the exponential of this integral is considered as a scalar, which is multiplied with the parallel transport in the Hom-bundle. Since we are concerned with the complex conjugate vector bundle $\overline{\inf R_{ji}}$, the  product has again the correct sign. This shows that the diagram  
\begin{equation*}
\small
\alxydim{@R=3em}{\inf R_{ij}|_{\Gamma(0)}(\mathcal{T}_0) \ar[d]_{pt_{ij}|_{\Gamma}} \ar[r]^{\alpha_{ij}} & \overline{\inf R_{ji}}|_{\overline{\Gamma(0)}}(\overline{\mathcal{T}}_0) \ar[d]^{pt_{ji}|_{\widetilde\rev\circ\Gamma}} \\ \inf R_{ij}|_{\Gamma(1)}(\mathcal{T}_1) \ar[r]_{\alpha_{ij}} & \overline{\inf R_{ji}}|_{\overline{\Gamma(1)}}(\overline{\mathcal{T}}_1)}
\end{equation*}
is commutative; hence, $\alpha_{ij}$ is connection-preserving. 
\end{proof}

\begin{lemma}
Axiom \cref{eq:lsg:invariance} is satisfied: lifted path reversal is compatible with the metrics.
\end{lemma}

\begin{proof}
As in the proofs of \cref{lem:LBG2,lem:LBG4}, the endomorphisms   $d_{\cp_x,\overline{\gamma} \pcomp\gamma}(\epsilon_i(x))$ and $d_{\cp_x,\gamma\pcomp \overline{\gamma}}(\epsilon_j(x))$ are identities, and the remaining equality is the obvious identity
\begin{equation*}
\mathrm{tr}((\psi^{*}\circ \kappa)^{*} \circ  \id)=\mathrm{tr}(\kappa^{*}\circ\psi) = \mathrm{tr}(\id \circ \psi \circ \kappa^{*})
\end{equation*}
for traces and adjoints of endomorphisms $\psi,\kappa$ of a finite-dimensional complex inner product space.
\end{proof}

Axioms \cref{eq:lsg:unitinversion}, \cref{eq:lsg:involutive} and \cref{eq:lsg:symmetrizing}
are  obvious identities for adjoints of linear operators. The last two axioms are again more involved.

\begin{lemma}
Axiom \cref{eq:lsg:symmetrizingfusion} is satisfied: lifted path reversal is compatible with the fusion representation. 
\end{lemma}

\begin{proof}
We consider $(\gamma_1,\gamma_2)\in P_{ij} \times_{Q_i \times Q_j}P_{ij}$, the corresponding loop $\tau := \gamma_1\cup\gamma_2$, and a trivialization $\mathcal{T}: \tau^{*}\mathcal{G}  \to \mathcal{I}_0$. In the diagram of \cref{eq:lsg:symmetrizingfusion}, the clockwise composition 
\begin{equation*}
\alxydim{}{\inf L|_{\tau} \otimes \inf R_{ij}|_{\gamma_2}(\iota_2^{*}\mathcal{T}) \ar[r]^-{\phi_{ij}} & \inf R_{ij}|_{\gamma_1}(\iota_1^{*}\mathcal{T}) \ar[r]^-{\alpha_{ij}} & \overline{\inf R_{ji}}|_{\overline{\gamma_1}}(\overline{\iota_1^{*}\mathcal{T}})}
\end{equation*} 
sends $[\mathcal{T},z] \otimes \varphi$ to $\overline{z}\cdot \varphi^{*}$. For the counter-clockwise direction, we let $\rot_{\pi}:S^1 \to S^1$ denote the rotation by an angle of $\pi$. Then, the map $d: \inf L|_{\gamma_2\cup\gamma_1} \to \inf L|_{\overline{\gamma_1}\cup\overline{\gamma_2}}$ is given by $[\mathcal{T},z] \mapsto [\rot_{\pi}^{*}\mathcal{T},z]$, see \cite[Lemma 4.3.6]{waldorf10}. We have to compose this with the isomorphism $\tilde\lambda$ of \cref{LBGstr:2}, which was computed in \cref{sec:trans:gerbes}, see \cref{eq:transtildelambda}.
Thus, the map
\begin{equation*}
\alxydim{@C=4em}{\inf L|_{\tau} \otimes \inf R_{ij}|_{\gamma_2}(\iota_2^{*}\mathcal{T}) \ar[r]^-{\tilde\lambda \otimes \alpha_{ij}} & \overline{\inf L}|_{\gamma_2\cup\gamma_1} \otimes \overline{\inf R}_{ji}|_{\overline{\gamma_2}}(\overline{\iota_2^{*}\mathcal{T}}) \ar[r]^{d \otimes \id} & \overline{\inf L}|_{\overline{\gamma_1} \cup \overline{\gamma_2}} \otimes \overline{\inf R}_{ji}|_{\overline{\gamma_2}}(\overline{\iota_2^{*}\mathcal{T}})}
\end{equation*}
is given by $[\mathcal{T},z] \otimes \varphi \mapsto [\rev^{*}\mathcal{T},\overline{z}]\otimes \varphi^{*} \mapsto [\mathcal{T}',\overline{z}]\otimes \varphi^{*}$, where $\mathcal{T}':=\rot_{\pi}^{*}\rev^{*}\mathcal{T}$. In order to obtain the counter-clockwise composition, it remains to compose the with fusion representation. We observe that $\iota_2^{*}\mathcal{T}' = \overline{\iota_2^{*}\mathcal{T}}$ and $\iota_1^{*}\mathcal{T}' = \overline{\iota_1^{*}\mathcal{T}}$; hence, the fusion representation results in  $\phi_{ji}([\mathcal{T}',\overline{z}]\otimes \varphi^{*}) = \overline{z}\cdot \varphi^{*}\in \overline{\inf R_{ji}}|_{\overline{\gamma_1}}(\overline{\iota_1^{*}\mathcal{T}})$. 
This shows that the diagram of \cref{eq:lsg:symmetrizingfusion} is commutative. \end{proof}

\begin{lemma}
The Cardy condition \cref{eq:lsg:cardy} is satisfied.  
\end{lemma}

\begin{proof}
As in the proofs of \cref{lem:LBG2,lem:LBG4}, the endomorphism $d_{\cp_y,\gamma \pcomp \cp_x\pcomp  \overline{\gamma}}(\epsilon_j(y))$ is the identity. The remaining equality is the identity
\begin{equation*}
\sum_{k=1}^{n} \varphi_k \circ \varphi \circ \varphi^{*}_k= \mathrm{tr}(\varphi)\cdot  \id\text{,}
\end{equation*}
where $(\varphi_1,...,\varphi_n)$  is an orthonormal basis of $\inf R_{ij}|_{\gamma}(\mathcal{T})$ with respect to the metric defined in \cref{sec:transbranes}. It is straightforward to check this identity, for instance using elementary matrices. 
\end{proof}

\subsection{Functoriality of transgression}

\label{sec:nattrans}

In  \cref{sec:trans:gerbes,sec:transbranes,sec:trans:superficialconnection,sec:fusionrep,sec:liftpathcommp,sec:liftedconstandrev,sec:trans:liftedpathreversal}  we have defined our transgression functor on the level of objects. Now we provide its definition on the level of morphisms: we associate to a TBG 1-morphism  $(\mathcal{A},\psi):(\mathcal{G},\mathcal{E})\to (\mathcal{G}',\mathcal{E}')$  a morphism $(\varphi,\xi)$ between the transgressed LBG objects.

To start with, the transgression of $\ueins$-bundle gerbes is functorial  \cite{waldorf5}: from the isomorphism $\mathcal{A}:\mathcal{G} \to \mathcal{G}'$ we obtain an isomorphism $\varphi:\inf L \to \inf L'$, which  over a loop $\tau\in LM$ is given by $[\mathcal{T},z] \mapsto [\mathcal{T}\circ \tau^{*}\mathcal{A}^{-1},z]$, for $\mathcal{T}$ a trivialization of $\tau^{*}\mathcal{G}$. The isomorphism $\varphi$  is connection-preserving and fusion-preserving.

Next, we define a vector bundle isomorphism $\xi_{ij}:\inf R_{ij} \to \inf R_{ij}'$ for all $i,j\in I$, using the 2-isomorphisms $\psi_i: \mathcal{E}_i \Rightarrow \mathcal{E}_i' \circ \mathcal{A}|_{Q_i}$.  Let $\gamma\in P_{ij}$ be a path with $x := \gamma(0)$ and $y:= \gamma(1)$,  let  $\mathcal{T}:\gamma^{*}\mathcal{G} \to \mathcal{I}_0$ be a trivialization, and let $\mathcal{T}':= \mathcal{T} \circ \gamma^{*}\mathcal{A}^{-1}$. We consider  vector bundle isomorphisms
\begin{equation*}
\psi_{i,0}:\Des(\mathcal{E}_i|_x,\mathcal{T}|_0) \to \Des(\mathcal{E}_i'|_x,\mathcal{T}'|_0) \quand
\psi_{i,1}:\Des(\mathcal{E}_j|_y,\mathcal{T}|_1) \to \Des(\mathcal{E}_j'|_y,\mathcal{T}'|_1) \end{equation*}
defined as follows. The bicategory $\ugrbcon {M}$  provides a  2-isomorphism $\delta: \mathcal{A}^{-1}\circ \mathcal{A} \Rightarrow \id$, which  induces a 2-isomorphism $\id \circ \delta^{-1}: \mathcal{T} \Rightarrow \mathcal{T}' \circ \gamma^{*}\mathcal{A}$. Now, $\psi_{i,0}$ is the composite
\begin{equation*}
\alxydim{@C=0.5cm}{\Des(\mathcal{E}_i|_x,\mathcal{T}|_0) \ar[rr]^-{\Des(\psi_i,\id)} && \Des(\mathcal{E}_i'|_x \circ \mathcal{A}|_x,\mathcal{T}|_0)\ar[rrr]^-{\Des(\id,\id \circ \delta^{-1})} &&& \Des(\mathcal{E}_i'|_x \circ \mathcal{A}|_x,\mathcal{T}|_0 '\circ \mathcal{A}|_x) \ar[r]^-{\Des_{\mathcal{A}|_x}} & \Des(\mathcal{E}_i'|_x ,\mathcal{T}'|_0)\text{,}} 
\end{equation*}
where  $\Delta_{\mathcal{A}|_x}$ is defined in \cref{rem:propDes:2}. 
$\psi_{i,1}$ is defined analogously. We have
\begin{align*}
\inf R_{ij}|_{\gamma}(\mathcal{T}) &= \mathrm{Hom}(\Des(\mathcal{E}_i|_x,\mathcal{T}|_0),\Des(\mathcal{E}_j|_y,\mathcal{T}|_1)) \\
\inf R_{ij}'|_{\gamma}(\mathcal{T}') &= \mathrm{Hom}(\Des(\mathcal{E}_i'|_x,\mathcal{T'}|_0),\Des(\mathcal{E}_j'|_y,\mathcal{T}'|_1)) \text{,}
\end{align*}
and we finally define
\begin{equation*}
\xi_{ij}: \inf R_{ij}|_{\gamma}(\mathcal{T}) \to \inf R'_{ij}|_{\gamma}(\mathcal{T}'): \varphi \mapsto \psi_{i,1} \circ \varphi \circ \psi_{i,0}^{-1}\text{.}
\end{equation*}

\begin{lemma}
This defines a connection-preserving bundle isomorphism, which  depends only on the 2-isomorphism class of $(\mathcal{A},\psi)$.
\end{lemma}

\begin{proof}
Independence of the choice of $\mathcal{T}$ is routine. For the smoothness, we consider a plot $c: U \to P_{ij}$, an open subset $W \subset U$, a trivialization $\mathcal{T}: (c|_W^{\vee})^{*}\mathcal{G} \to \mathcal{I}_{\rho}$, and a bundle isomorphism $\tau: W \times \C^{n} \to \mathrm{Hom}(\inf V_{\mathcal{T}},\inf W_{\mathcal{T}})$ over $W$, so that we have a local trivialization $\phi$ of $\inf R_{ij}$. We set  $\mathcal{T}' := \mathcal{T} \circ (c|_W^{\vee})^{*}\mathcal{A}^{-1}$
and obtain bundle isomorphisms $\psi_{i,0}:\inf V_{\mathcal{T}} \to \inf V_{\mathcal{T}'}$ and $\psi_{i,1}:\inf W_{\mathcal{T}} \to \inf W_{\mathcal{T}'}$ over $W$ from exactly the same definition as above. Then we obtain a local trivialization of $\inf R_{ij}'$ using $\mathcal{T}'$ and $\tau'(w,v) :=  \psi_{i,1}\circ \tau(w,v) \circ \psi_{i,0}^{-1}$ for $w\in W$ and $v\in \C^{n}$; in these local trivializations the map $\xi_{ij}$ is the identity, and hence smooth.

If $\Gamma:[0,1] \to P_{ij}$ is a path,  we choose a trivialization $\mathcal{T}:(\Gamma^{\vee})^{*}\mathcal{G} \to \mathcal{I}_{\rho}$. Then, the trivialization $\mathcal{T}' := \mathcal{T} \circ (\Gamma^{\vee})^{*}\mathcal{A}^{-1}$ has the same 2-form $\rho$. Since the isomorphisms $\psi_{i,0}$ and $\psi_{i,1}$ over $[0,1]$ are connection-preserving, the induced isomorphism between Hom-bundles is connection-preserving, too. This shows that $\xi_{ij}$ preserves  connections.

We consider a 2-isomorphism $\varphi$ between $(\mathcal{A},\psi)$ and $(\mathcal{A}',\psi')$
and the corresponding vector bundle isomorphisms $\xi_{ij}$ and $\xi_{ij}'$. For a path $\gamma\in P_{ij}$, let $\mathcal{T}:(\gamma^{\vee})^{*}\mathcal{G} \to \mathcal{I}_0$ be a trivialization,  $\mathcal{T}' := \mathcal{T} \circ \gamma^{*}\mathcal{A}^{-1}$ and  $\mathcal{T}'' := \mathcal{T} \circ \gamma^{*}\mathcal{A}'^{-1}$. From the 2-isomorphism  $\varphi: \mathcal{A} \Rightarrow \mathcal{A}'$ we construct another 2-isomorphism $\eta := \id\circ \gamma^{*}\varphi^{\#}: \mathcal{T}' \Rightarrow \mathcal{T}''$, where $\varphi^{\#}:\mathcal{A}^{-1} \Rightarrow \mathcal{A}'^{-1}$ denotes the inverse with respect to horizontal composition. We consider the following diagram of linear maps:
\begin{equation*}
\small
\alxydim{@C=0.9cm@R=1.5em}{ & \Des(\mathcal{E}_i'|_x \circ \mathcal{A}|_x,\mathcal{T}|_0) \ar[dd]^{\Des(\id \circ \varphi,\id)} \ar[rr]^-{\Des(\id,\id \circ \delta_{\mathcal{A}}^{-1})} && \Des(\mathcal{E}_i'|_x \circ \mathcal{A}|_x,\mathcal{T}'|_0 \circ \mathcal{A}|_x) \ar[dd]^{\Des(\id \circ \varphi,\eta|_0\circ \varphi)} \ar[r]^-{\Des_{\mathcal{A}|_x}} & \Des(\mathcal{E}_i'|_x,\mathcal{T}'|_0 )\ar[dd]^{\Des(\id,\eta|_0)}
\\
\Des(\mathcal{E}_i|_x,\mathcal{T}|_0) \mquad  \ar[dr]_{\Des(\psi_i',\id)} \ar[ur]^-{\Des(\psi_i,\id)}
\\ & \Des(\mathcal{E}_i'|_x \circ \mathcal{A}'|_x,\mathcal{T}|_0) \ar[rr]_-{\Des(\id,\id \circ \delta_{\mathcal{A}'}^{-1})} && \Des(\mathcal{E}_i'|_x \circ \mathcal{A}|_x',\mathcal{T}''|_0 \circ \mathcal{A}'|_x) \ar[r]_-{\Des_{\mathcal{A}|_y}} & \Des(\mathcal{E}_i'|_x ,\mathcal{T}''|_0)} 
\end{equation*}
The diagram is commutative: the triangular diagram on the left commutes due to the condition for TBG 2-isomorphisms, $\psi_i' = (\id \circ \varphi) \bullet \psi_i$. The diagram in the middle commutes because the 2-isomorphism $\delta_{\mathcal{A}}: \mathcal{A}^{-1} \circ \mathcal{A} \Rightarrow \id$ is natural with respect to 2-morphisms, 
and the diagram on the right-hand side commutes by definition of  $\Delta_{\mathcal{A}|_x}$. The diagram gives the equation $\psi'_{i,0}=\Des(\id,\eta|_0) \circ \psi_{i,0}$. Analogously, we have $\psi'_{i,1}=\Delta(\id,\eta|_1) \circ \psi_{i,1}$. Thus, for $\inf R_{ij}|_{\gamma}(\mathcal{T})$ we have:
\begin{multline*}
\xi'_{ij}(\varphi) = \psi_{i,1}' \circ \varphi \circ \psi_{i,0}^{\prime-1}  =\Des(\id,\eta|_0) \circ \psi_{i,0}\circ \varphi \circ \psi_{i,1}^{-1} \circ \Delta(\id,\eta|_1)^{-1}\\=\Des(\id,\eta|_0) \circ \xi_{ij}(\varphi) \circ \Delta(\id,\eta|_1)^{-1} =r_{\eta}(\xi_{ij}(\varphi))\text{.}
\end{multline*}
This shows that $\xi_{ij}=\xi_{ij}'$.
\end{proof}

The following lemma is a straightforward check, only using the definitions; its proof is left out for brevity. 

\begin{lemma}
The bundle isomorphism $\xi_{ij}$ intertwines the fusion representations, and respects the lifted path composition,  lifted constant paths, and lifted path reversal. 
\qed
\end{lemma}

Thus, the pair $(\varphi,\xi)$, with $\xi=\{\xi_{ij}\}_{i,j\in I}$ is a morphism in $\lsg(M,Q)$, and we call it the \emph{transgression of the 1-morphism $(\mathcal{A},\psi)$}. It is again straightforward to show that the assignment $(\mathcal{A},\psi) \mapsto (\varphi,\xi)$ preserves identities and  the composition, so that we have the following:

\begin{proposition}
\label{prop:transfunct}
Transgression is a functor $\mathscr{T}:\tsg(M,Q) \to \lsg(M,Q)$.
\qed
\end{proposition}

\subsection{Transgression of a trivial bundle gerbe, and reduction to the point}

\label{sec:trivialgerbe}

We analyze the transgression of TBG with a trivial bundle gerbe, i.e. $\mathcal{G}=\mathcal{I}_{\rho}$. We note that  the bundle gerbe modules $\mathcal{E}_i$ are nothing but ordinary hermitian vector bundles $E_i$ with connections over $Q_i$, and the 2-forms $\omega_i$ are determined by
\begin{equation*}
\omega_i  = \rho +  \frac{1}{\mathrm{rk}(E_i)}\mathrm{tr}(\mathrm{curv}(E_i))\text{.}
\end{equation*}
Under transgression, we obtain the following. The line bundle is $\inf L = \C_{\rho_{LM}}$,  the trivial line bundle equipped with the connection 1-form $\rho_{LM} :=\int_{S^1}\ev^{*}\rho$, and the fusion product is just multiplication in the fibers \cite[Lemma 3.6]{waldorf13}. Concerning the vector bundles $\inf R_{ij}$, we claim that we have a canonical, connection-preserving, unitary bundle isomorphism
\begin{equation}
\label{eq:trivgrbiso}
\inf R_{ij} \cong \mathrm{Hom}(\ev_0^{*}E_i, \ev_1^{*}E_j) \otimes \C_{\rho_{PM}}
\end{equation}
over $P_{ij}$, where $\rho_{PM} := \int_{[0,1]}\ev^{*}\rho$. Thus, $\inf R_{ij}$ is globally a bundle of homomorphisms between two vector bundles.

The isomorphism \cref{eq:trivgrbiso} is defined in the following way. We consider a path $\gamma\in P_{ij}$ from $x\in Q_i$ to $y\in Q_j$ and the trivialization $\mathcal{T}:=\id_{\mathcal{I}_0}$ of $\gamma^{*}\mathcal{G}=\gamma^{*}\mathcal{I}_{\rho}=\mathcal{I}_0$. Then, $\Des(\mathcal{E}_i|_x,\mathcal{T}|_0)=E_i|_x$ and $\Des(\mathcal{E}_j|_y,\mathcal{T}|_1)= E_j|_y$, so that  $\inf R_{ij}|_{\gamma}(\id_{\mathcal{I}_0}) = \mathrm{Hom}(E_i|_x,E_j|_y)$. It is straightforward to see that this defines a smooth bundle morphism \cref{eq:trivgrbiso}.
Concerning the connection, let $\Gamma:[0,1] \to P_{ij}$ be a path. Considering the trivialization $\mathcal{T}:=\smash{\id_{\mathcal{I}_{(\Gamma^{\vee})^{*}\rho}}}$ of $(\Gamma^{\vee})^{*}\mathcal{I}_{\rho}$, we obtain $\inf V_{\mathcal{T}}=(\ev_0\circ \Gamma)^{*}E_i$ and $\inf W_{\mathcal{T}}=(\ev_1\circ \Gamma)^{*}E_j$. Now, the parallel transport in $\inf R_{ij}$ is
\begin{equation*}
pt_{ij}|_{\Gamma}: \inf R_{ij}|_{\Gamma(0)}(\id_{\mathcal{I}_0}) \to \inf R_{ij}|_{\Gamma(1)}(\id_{\mathcal{I}_0}): \varphi \mapsto  \exp \left (\int_{[0,1]^2}(\Gamma^{\vee})^{*}\rho \right )\cdot pt(\varphi)\text{.}
\end{equation*}
This is precisely the parallel transport in the vector bundle $\mathrm{Hom}(\ev_0^{*}E_i, \ev_1^{*}E_j) \otimes \C_{\rho_{PM}}$.

The fusion representation is trivial. 
The lifted  path concatenation is the composition, under the isomorphisms \cref{eq:trivgrbiso}, which is connection-preserving because $\pcomp^{*}\rho_{PM}=\pr_{jk}^{*}\rho_{PM}+\pr_{ij}^{*}\rho_{PM}$ on $P_{jk} \times_{Q_j} P_{ij}$.
The lifted constant path $\epsilon_i$ is the identity section, and the lifted path reversal is given by taking adjoint homomorphisms.

Finally, we consider the yet more trivial case of a target space  $(M,Q)$ with $M=\{\ast\}$ and $Q=\{\ast\}_{i\in I}$. The bundle gerbe $\mathcal{G}$ is trivial, as before, while the twisted vector bundles $\mathcal{E}_i$ give just a family $\{E_i\}_{i\in I}$ of finite-dimensional hermitian vector spaces. In other words,
\begin{equation*}
\tsg^{I}:= \hc 1\tsg(\ast,\{\ast\}_{i\in I}) = \coprod_{i\in I} \HVect {fin}-\text{.}
\end{equation*}
Transgression to the category $\lsg^{(I)}$ described in \cref{lem:LBGstrpt} takes a family $\{E_i\}_{i\in I}$ to the object $(\inf L,\lambda, \inf R,\phi,\chi,\epsilon,\alpha)$ with $\inf L:=\C$, $\lambda := \id$, $\inf R_{ij} := \mathrm{Hom}(E_i,E_j)$, $\phi := \id$, $\chi_{ijk} := \circ$, $\epsilon_i:=\id_{E_i}$, $\alpha_{ij} = ()^{*}$.

\setsecnumdepth{1}

\section{Regression}

\label{sec:regression}

In this section we construct our regression functor $\mathscr{R}_{i_0,x_0,\inf F_0}: \lsg(M,Q) \to \tsg(M,Q)$, depending on three parameters as explained below. We consider a LBG object $(\inf L,\lambda, \inf R,\phi,\chi,\epsilon,\alpha)$.

  We first treat the line bundle $\inf L$ with its superficial connection and the fusion product $\lambda$. The regression of $(\inf L,\lambda)$ is discussed in \cite{waldorf10}.
It requires the choice of a base  point $x_0\in M$. The regressed bundle gerbe $\mathcal{G}:=\mathscr{R}_{x_0}(\inf L,\lambda)$ has the subduction (the diffeological analogue of a surjective submersion) $\ev_1: P_{x_0}M \to M$, where $P_{x_0}M:= \{\gamma\in PM \sep \gamma(0)=x_0\}$. The 2-fold fiber product is equipped with the smooth map $\cup: P_{x_0}M^{[2]} \to LM$, and the hermitian line bundle of $\mathcal{G}$ is $\cup^{*}\inf L$. The bundle gerbe product is the restriction of the fusion product $\lambda$ to $P_{x_0}M^{[3]}$.

In the presence of branes, we need to choose a brane index $i_0\in I$ such that $x_0\in Q_{i_0}$. Over $x_0$ we have the simple algebra $\inf A_0 := \inf A_{i_0}|_{x_0}$ (\cref{co:aisimple}), which is isomorphic to a Matrix algebra $\C^{n_0 \times n_0}$, where $n_0=\sqrt{\mathrm{rk}(\inf R_{i_0i_0})}=\|\epsilon_{i_0}(x_0)\|^2$ by  \cref{re:lsg:rank}. Thus, there exists  an irreducible  $n_0$-dimensional left $\inf A_{0}$-module $\inf F_0$, which we  fix, too. In the following we consider $\inf A_0$ and $\inf F_0$ as bundles of algebras and left $\inf A_0$-modules over the point $x_0$, respectively.

For $i\in I$ we consider the subspace $P_{i} :=\{\gamma\in P_{i_0i} \sep \gamma(0)=x_0\} \subset P_{x_0}M$, over which we find the bundle $\inf R_{i_0i}$ of  $\ev_1^{*}\inf A_{i}$-$\ev_0^{*}\inf A_{0}$-bimodules (\cref{prop:bimodule}). Tensoring over $\inf A_0$ with the bundle $\ev_0^{*}\inf F_0$ of left $\ev_0^{*}\inf A_{0}$-modules, we obtain a bundle $E_i :=  \inf R_{i_0i} \otimes _{\ev_0^{*}\inf A_0} \ev_0^{*}\inf F_0 $ of left $\ev_1^{*}\inf A_i$-modules over $P_i$. We forget the left $\ev_1^{*}\inf A_i$-action and consider $E_i$ as a hermitian vector bundle. For later purpose, we note the following result.

\begin{lemma}
\label{lem:rkEi}
We have $\mathrm{rk}(E_i)=\sqrt{\mathrm{rk}(\inf R_{ii})}$ for all $i\in I$. 
\end{lemma}

\begin{proof}
By \cref{lem:morita}, $\inf R_{i_0i}|_{\gamma}$ establishes, for any $\gamma\in P_{ij}$ with $\gamma(1)=:x$, a Morita equivalence between $\inf A_0$ and $\inf A_i|_x$. The module $E_i|_{\gamma}$ is the image of $\inf F_0$ under the corresponding functor between representation categories. Irreducibility is preserved under Morita equivalence; hence, $E_i|_{\gamma}$  is an irreducible module of the algebra $\inf A_i|_x$, which is a simple algebra (\cref{co:aisimple}) and has dimension $\mathrm{rk}(\inf R_{ii})$ (\cref{re:lsg:rank}). Matrix algebras have only one irreducible module up to isomorphism, namely the standard one; thus it has dimension $\sqrt{\mathrm{rk}(\inf R_{ii})}$.
\end{proof}

As bundles over a point,  $\inf F_0$ and $\inf A_0$ are equipped with trivial connections, for which the module action $\inf F_0 \otimes \inf A_0 \to \inf F_0$ is connection-preserving. Since  the right module action $\inf R_{i_0i} \otimes \ev_0^{*}\inf A_0\to \inf R_{i_0i}$ is connection-preserving (\cref{lem:actionsconnectionpreserving}), it follows that $E_i$ comes equipped with a connection, see \cref{re:bimodulebundle:connections}.

The fusion representation $\phi$ induces a connection-preserving, unitary  bundle isomorphism
\begin{equation*}
\zeta_i:\cup^{*}\inf L \otimes \pr_2^{*}E_i \to \pr_1^{*}E_i\text{.}
\end{equation*} 
over $P_i^{[2]}$; explicitly, it is induced by  
\begin{equation*}
\alxydim{@C=1.4cm}{ \cup^{*}\inf L \otimes \pr_2^{*}\inf R_{i_0i} \otimes\ev_0^{*}\inf F_0  \ar[r]^-{ \rho_{i_0i} \otimes \id} & \pr_1^{*}\inf R_{i_0i} \otimes  \ev_0^{*}\inf F_0\text{.}}
\end{equation*}
Since the fusion representation intertwines the algebra actions  (\cref{lem:fusionandbimodule}), this is well-defined under taking the quotient by the $\ev_0^{*}\inf A_0$-action. The compatibility between fusion representation and the fusion product \cref{eq:lsg:fusionass} ensures that the bundle isomorphism $\zeta_i$ satisfies the condition for making $\mathcal{E}_i:=(E_i,\zeta_i)$ a $\mathcal{G}|_{Q_i}$-module.
This  completes the definition of our regression functor on the level of objects.

Next, we consider an LBG morphism $(\varphi,\xi)$ between LBG objects $(\inf L,\lambda, \inf R,\phi,\chi,\epsilon,\alpha)$ and $(\inf L',\lambda', \inf R',\phi',\chi',\epsilon',\alpha')$. 
The regression of  $\varphi:\inf L \to \inf L'$ to a 1-morphism $\mathcal{A}:\mathcal{G} \to \mathcal{G}'$ in $\ugrbcon M$ is discussed in \cite{waldorf10}; it is a \quot{refinement} consisting simply of the bundle morphism $\cup^{*}\varphi:\cup^{*}\inf L \to \cup^{*}\inf L'$. The bundle isomorphisms $\xi_{ij}:\inf R_{ij} \to \inf R_{ij}'$ induce in the obvious way isomorphisms of the corresponding  bundles of algebras and bimodules. In particular, we obtain an algebra isomorphism $\xi_0:\inf A_{0} \to \inf A'_{0}$. Again, since the modules $\inf F_0$ and $\inf F_0'$ are irreducible modules over isomorphic matrix algebras, they must be isomorphic, and we can choose an intertwiner $f_0: \inf F_0 \to \inf F_0'$.  Then, the  bundle isomorphism \begin{equation*}
\xi_{i_0i} \otimes \ev_0^{*}f_0:\inf R_{i_0i} \otimes \ev_0^{*}\inf F_0 \to \inf R_{i_0i}' \otimes \ev_0^{*}\inf F_0
\end{equation*}
induces a connection-preserving, unitary bundle isomorphism $\tilde\psi_i:E_i \to E_i'$, since $\xi_{i_0i}$ and $f_0$ intertwine the $\inf A_0$ and $\inf A_0'$ actions. The diagram  
\begin{equation}
\label{eq:regfunc}
\small
\alxydim{@R=3em}{\cup^{*}\inf L \otimes \pr_2^{*}E_i \ar[d]_{\cup^{*}\varphi \otimes \pr_2^{*}\tilde\psi_i} \ar[r]^-{\epsilon_i} & \pr_1^{*}E_i \ar[d]^{\pr_1^{*}\tilde\psi_i} \\ \cup^{*}\inf L' \otimes \pr_2^{*}E_i' \ar[r]_-{\epsilon_i'} & \pr_1^{*}E_i'}
\end{equation}
is commutative, because $\xi_{i_0i}$ commutes with the fusion representation. 
This means that $\tilde\psi_i$ gives rise to a 2-isomorphism $\psi_i: \mathcal{E}_i \Rightarrow \mathcal{E}_i'\circ \mathcal{A}$. 
Thus, forming the collection $\psi=\{\psi_i\}_{i\in I}$,  the pair $(\mathcal{A},\psi)$ is a TBG 1-morphism. It  depends on the (non-canonical) choice of the intertwiner $f_0$; however, we have the following.

\begin{lemma}
The 2-isomorphism class of $(\mathcal{A},\psi)$ is independent of the choice of  $f_0$.
\end{lemma}

\begin{proof}
If $f_0$ and $f_0'$ are two choices, then, by Schur's lemma and because both respect hermitian metrics, they differ by a number $z\in \ueins$, say $f_0'=f_0\cdot z$. This number determines a 2-isomorphism $\tilde\varphi: \mathcal{A}\Rightarrow \mathcal{A}$, by inducing it via the automorphism $z\cdot \id$ of $\cup^{*}\inf L'$.  Let $\tilde \psi_i$ and $\tilde \psi'_i$ be the vector bundle morphisms determined by $f_0$ and $f_0'$, respectively, so that $\tilde\psi_i' = \tilde\psi_i \cdot z$. 
Then, we have $\psi_i' = (\id_{\mathcal{E}_i'} \circ \tilde \varphi) \bullet \psi_i$; this shows that $\tilde\varphi$ is the required 2-isomorphism. 
\end{proof}

We have thus provided the data for our regression functor $\mathscr{R}_{i_0,x_0,\inf F_0}$. The fact that it preserves the composition is easy to see by choosing, for the composed morphism, the composition of the intertwiners $f_0$ of the separate morphisms. 

\begin{remark}
Up to canonical natural isomorphism, the regression functor is independent of all choices. This can be seen either manually, or it can be deduced from \cref{th:main}, which says that it is inverse to one fixed functor, the transgression functor. 
\end{remark}

\setsecnumdepth{2}

\section{Equivalence of target space and loop space perspectives}

\label{sec:equivalence}

In this section we prove our main results: \cref{th:main:2:2} in \cref{sec:coincidence}, and \cref{th:main} in \cref{sec:regaftertr,sec:traftereg}.
Throughout this section, we fix a target space $(M,Q)$, with $Q=\{Q_i\}_{i\in I}$.

\subsection{Coincidence of the algebra bundles}

\label{sec:coincidence}

We recall that TBG and LBG independently induce   algebra bundles over the branes (\cref{re:algebrabundle,sec:algebrabundle:lsg}). We show that they coincide under transgression and regression.

We start with  a TBG object $(\mathcal{G},\mathcal{E})$. Let $\inf R_{ij}$ be the transgressed vector bundle over $P_{ij}$, and let $\inf A_i = \cp^{*}\inf R_{ii}$ be the induced  algebra bundle over $Q_i$. We let $\pi:Y \to M$ be the surjective submersion of the bundle gerbe $\mathcal{G}$. For a point $y\in Y$ with $x := \pi(y)\in Q_i$, we denote by $\mathcal{T}_y$ the corresponding trivialization of $\mathcal{G}|_{x}$, see \cref{rem:sectionsandtrivializations}. Applying that remark to $s=\cp_y$ as a section of $\pi:Y \to M$ along $\{x\} \to M$,  we obtain a canonical vector bundle isomorphism $\varphi_y: \Des(\mathcal{E}_i|_x,\mathcal{T}_y) \to E_i|_y$.
It is straightforward to see that
\begin{equation*}
\psi_y:\mathrm{End}(E_i)|_y \to \inf A_i|_{x}  : \varphi \mapsto (\cp_{x}^{*}\mathcal{T}_y,\varphi_y^{-1}\circ \varphi \circ \varphi_y)
\end{equation*}
induces a smooth, connection-preserving vector bundle isomorphism $\psi: \mathrm{End}(E_i) \to \pi^{*}\inf A_i$ over $\pi^{-1}(Q_i)$ that preserves the fiber-wise algebra structures. If $y'\in Y$ is another point with $\pi(y')=x$, we choose $\ell \in L|_{y,y'}$ and obtain a 2-isomorphism $\psi: \mathcal{T}_y \Rightarrow \mathcal{T}_{y'}$ in such a way that the diagram
\begin{equation}
\label{eq:algbun:diag}
\small
\alxydim{@R=3em}{\Des(\mathcal{E}_i|_x,\mathcal{T}_{y'}) \ar[d]_{\Des(\psi,\id)} \ar[r]^-{\varphi_{y'}} & E_i|_{y'} \ar[d]^{\zeta_i(\ell \otimes -)} \\ \Des(\mathcal{E}_i|_x,\mathcal{T}_y) \ar[r]_-{\varphi_{y}} & E_i|_y\text{.}}
\end{equation}
is commutative (see again \cref{rem:sectionsandtrivializations}).
By \cref{eq:algbun:diag}, the bundle morphism 
\begin{equation}
\label{eq:algbundesc}
\pr_1^{*}\psi^{-1} \circ \pr_2^{*}\psi: \pr_2^{*}\mathrm{End}(E_i) \to \pr_1^{*}\mathrm{End}(E_i)
\end{equation}
over $Y^{[2]}$ is given  at $(y,y')$  by $\varphi \mapsto \zeta_i(\ell \otimes -)\circ \varphi \circ \zeta_i(\ell\otimes -)^{-1}$. 
One can check that this is independent of $\ell$; in fact, under the isomorphism $\mathrm{End}(E_i)=E_i^{*}\otimes E_i$,  it is precisely the descent isomorphism $\tilde \zeta_i$ that defines the bundle $\mathrm{End}(\mathcal{E}_i)$.    
This proves that $\psi$ descends to an  isomorphism $\mathrm{End}(\mathcal{E}_i) \cong \inf A_i$ of  algebra bundles over $Q_i$.

 Now, consider an  LBG object $(\inf L,\lambda, \inf R,\phi,\chi,\epsilon,\alpha)$, and the corresponding regression $(\mathcal{G},\mathcal{E})$. We recall from \cref{sec:regression} that $E_i =  \inf R_{i_0i} \otimes _{\ev_0^{*}\inf A_0} \ev_0^{*}\inf F_0$, which is a bundle of left $\ev_1^{*}\inf A_i$-modules over $P_i$. In particular, we have a  homomorphism of  algebra bundles
\begin{equation*}
\ev_1^{*}\inf A_i \to \mathrm{End}(E_i)\text{. }
\end{equation*}
We remark that $E_i$ is the composition of two Morita equivalences: $\inf R_{i_0i}$ is a Morita equivalence by \cref{lem:morita}, and, since it is an irreducible module, $\inf F_0$ is a Morita equivalence between $\inf A_0$ and $\C$. Hence, $E_i$  is a Morita equivalence; in particular it is faithfully balanced, which implies that the above map is an isomorphism. 
The  bundle  of algebras $\mathrm{End}(\mathcal{E}_i)$ of TBG is now obtained via descent of $\mathrm{End}(E_i)$ along the path evaluation $\ev_1: P_i\to Q_i$. Using that the bundle morphism $\zeta_i$ is defined from the fusion representation (see \cref{sec:regression}) and that the fusion representation commutes with the action of $\inf A_i$ (\cref{lem:fusionandbimodule}) one can show that the diagram
\begin{equation*}
\small
\alxydim{@R=3em}{&\ev_1^{*}\inf A_i \ar[dl] \ar[dr] & \\  \pr_2^{*}\mathrm{End}(E_i) \ar[rr]_{\tilde\zeta_i} &&\pr_1^{*}\mathrm{End}(E_i)}
\end{equation*} 
is commutative.  
Thus, we obtain an algebra bundle isomorphism $\inf A_i \cong \mathrm{End}(\mathcal{E}_i)$.

\subsection{Regression after Transgression}

\label{sec:regaftertr}

In this section we provide a natural isomorphism $\mathscr{R}_{i_0,x_0,\inf F_0} \circ \mathscr{T}\cong \id_{\tsg(M,Q)}$, contributing one half of the equivalence $\tsg(M,Q)\cong \lsg(M,Q)$.
Thus, we construct for each TBG object $(\mathcal{G},\mathcal{E})$ a 1-morphism $(\mathcal{A},\psi): (\mathscr{R}_{i_0,x_0,\inf F_0} \circ \mathscr{T})(\mathcal{G},\mathcal{E}) \to (\mathcal{G},\mathcal{E})$, and show that this is natural with respect to TBG morphisms.

We consider a TBG object $(\mathcal{G},\mathcal{E})$,  form the transgressed LBG object $(\inf L,\lambda, \inf R,\phi,\chi,\epsilon,\alpha)$, and consider the regressed bundle gerbe  $\mathscr{R}_{x_0}(\inf L,\lambda)$. 
The bundle gerbe $\mathcal{G}$ consists of a surjective submersion $\pi:Y \to M$, a hermitian line bundle $P$ over $Y^{[2]}$, and a bundle gerbe product $\mu$.
A 1-isomorphism 
$\mathcal{A}:\mathscr{R}_{x_0}(\inf L,\lambda)  \to \mathcal{G}$
was constructed in \cite[Lemma 6.1.1]{waldorf10}
over the fiber product $Z:= P_{x_0}M \times_M Y$. We review and slightly reformulate its construction in the following. The 1-isomorphism $\mathcal{A}$ consists   of a hermitian line bundle $K$ with connection over $Z$. Its construction involves the choice of a point $y_0\in Y$ with $\pi(y_0)=x_0$. 
Let $\mathcal{T}: \gamma^{*}\mathcal{G} \to \mathcal{I}_0$ be a trivialization of $\mathcal{G}$ along $\gamma\in P_{x_0}M$, consisting of a hermitian line bundle $T$ over $Y_{\gamma} :=[0,1] \ttimes{\gamma}{\pi} Y$ and a bundle isomorphism $\tau$ over $Y_{\gamma}^{[2]}$. We define the complex inner product space
\begin{equation}
\label{eq:defK}
K|_{\gamma,y}(\mathcal{T}) := T|_{(0,y_0)} \otimes T|_{(1,y)}^{*} \text{.}
\end{equation} 
If $\psi: \mathcal{T} \Rightarrow \mathcal{T}'$ is a 2-isomorphism in $\ugrbcon{[0,1]}$, then it is given by a bundle isomorphism $\psi: T \to T'$ over $Y_{\gamma}$, and hence induces an obvious isomorphism $K|_{\gamma,y}(\mathcal{T}) \cong K|_{\gamma,y}(\mathcal{T}')$. We let $K|_{\gamma,y}$ be the set of equivalence classes of pairs $(\mathcal{T},s\otimes \sigma)$ with a trivialization $\mathcal{T}$ and $s\otimes\sigma \in K|_{\gamma,y}(\mathcal{T})$.

The disjoint union of all these fibers define the total space of $K$. It can be equipped with a diffeology such that it becomes a hermitian  line bundle. Further, one can define a connection on $K$ in a canonical way.
It remains to provide the bundle isomorphism over $Z^{[2]}$ which is part of the 1-isomorphism $\mathcal{A}$. Here, 
\begin{equation*}
\kappa|_{(\gamma_1,y_1),(\gamma_2,y_2)} : \inf L_{\gamma_1 \cup \gamma_2} \otimes K|_{\gamma_2,y_2} \to K|_{\gamma_1,y_1} \otimes P_{y_1,y_2}  
\end{equation*}
is defined as follows.  We choose a trivialization $\mathcal{T}: (\gamma_1 \cup \gamma_2)^{*}\mathcal{G} \to \mathcal{I}_0$ and consider the induced trivializations $\mathcal{T}_1:\gamma_1^{*}\mathcal{G} \to \mathcal{I}_0$ and $\mathcal{T}_2:\gamma_2^{*}\mathcal{G} \to \mathcal{I}_0$. Let $\tau$ be the bundle isomorphism of the trivialization $\mathcal{T}$. 
We define for $s_0\in T|_{(0,y_0)}$ and $\sigma\in T|_{(1,y_2)}^{*}$:
\begin{equation*}
\kappa|_{(\gamma_1,y_1),(\gamma_2,y_2)}(( \mathcal{T},z) \otimes (\mathcal{T}_2,s_0\otimes  \sigma )) := z\cdot (\mathcal{T}_1,s_0 \otimes \tau^{tr-1}|_{(1,y_1),(1,y_2)}(\sigma))\in K|_{\gamma_1,y_1}(\mathcal{T}_1) \otimes P_{y_1,y_2}   \text{.}
\end{equation*}
It is proved in \cite[Lemmata 6.1.1 \& 6.1.3]{waldorf10} that $\kappa$ defines a connection-preserving unitary bundle isomorphism. This finishes the definition of the 1-isomorphism $\mathcal{A}:\mathscr{R}_{x_0}(\inf L,\lambda)  \to \mathcal{G}$.

Next, we construct, for each brane index $i\in I$, a 2-isomorphism 
$\varphi_i: \mathcal{E}_i \circ \mathcal{A}|_{Q_i} \Rightarrow \mathcal{E}_i'$, 
where $\mathcal{E}_i'$ is the regressed bundle gerbe module. The inverses $\varphi:=\{\varphi_i^{-1}\}_{i\in I}$ will then complete $\mathcal{A}$ into a TBG morphism $(\mathcal{A},\varphi)$.  
We set $Z_i := Z|_{Q_i}$ and $K_i := K|_{Z_i}$. The 2-isomorphism $\varphi_i$ consists of a  bundle isomorphism 
$\varphi_i: K_i \otimes \pr_1^{*}E_i  \to \pr_2^{*}E_i'$
over $Z_i$, where $E_i'$ is the vector bundle of the regressed module $\mathcal{E}_i'$. Over a point $(y,\gamma) \in Z_i$ with $\pi(y)=\gamma(1)\in Q_i$ and $\gamma(0)=x_0$, this is an isomorphism
\begin{equation}
\label{eq:varphiiygamma}
\varphi_i|_{y,\gamma}: K_i|_{y,\gamma} \otimes E_i|_{y}  \to  \inf R_{i_0i}|_{\gamma}\otimes_{\inf A_{0}}\inf F_0 \text{.}
\end{equation}  
Here, $\inf A_{0}=\inf R_{i_0i_0}|_{\cp_{x_0}}$ is the fiber of the algebra bundle $\inf A_{i_0}=\cp^{*}\inf R_{i_0i_0}$ over $Q_{i_0}$ over the point $x_0$.  
In order to define $\varphi_i|_{y,\gamma}$, we make two observations. The first is the algebra isomorphism $\varphi_{y_0}: \mathrm{End}(E_{i_0})|_{y_0} \to \inf A_{0}$ constructed in \cref{sec:coincidence}, that allows us to choose $\inf F_0 = E_{i_0}|_{y_0}$ in the definition of the regression functor.
The second observation is the following: the chosen lifts $y_0$ of $x_0=\gamma(0)$ and $y$ of $\gamma(1)$ yield a vector space isomorphism
\begin{multline}
\label{eq:transreg:1}
\inf R_{i_0i}|_{\gamma}(\mathcal{T}) =\mathrm{Hom}(\Des(\mathcal{E}_{i_0}|_{x_0},\mathcal{T}|_0),\Des(\mathcal{E}_i|_{\gamma(1)},\mathcal{T}|_1))
\\ \alxydim{}{\ar[r]^{\delta_{y_0,y}} &} \mathrm{Hom}(T|_{(0,y_0)}^{*} \otimes E_{i_0}|_{y_0} ,T|_{(1,y)}^{*}\otimes E_i|_y)\\\cong E_{i_0}|_{y_0}^{*} \otimes T|_{(0,y_0)} \otimes  T|_{(1,y)}^{*}\otimes E_i|_y\text{.}
\end{multline}
for any trivialization $\mathcal{T}:\gamma^{*}\mathcal{G} \to \mathcal{I}_0$, where $\delta_{y_0,y}$ is obtained from the isomorphisms $\Des(\mathcal{E}_{i_0}|_{x_0},\mathcal{T}|_0)\cong T|_{(0,y_0)}^{*} \otimes E_{i_0}|_{y_0}$ and $\Des(\mathcal{E}_i|_{\gamma(1)},\mathcal{T}|_1)\cong T|_{(1,y)}^{*}\otimes E_i|_y$ that result from the definition of $\Des$ via descent.  
For another point $y'\in Y$ with $\pi(y)=\gamma(1)$, the two isomorphisms $\delta_{y_0,y}$ and $\delta_{y_0,y'}$ are related by the descent isomorphism
\begin{equation}
\label{eq:pointydiff}
\alxydim{@C=5em}{T^{*}|_{(1,y')} \otimes E_i|_{y'} \cong T^{*}|_{(1,y')} \otimes P_{y,y'}^{*} \otimes P_{y,y'} \otimes E_i|_{y'} \ar[r]^-{\tau^{tr-1} \otimes \zeta_i} &  T^{*}|_{(1,y)} \otimes E_i|_{y}}
\end{equation}
that produces $\Des(\mathcal{E}_i|_{\gamma(1)},\mathcal{T}|_1)$. 
Now we are ready to give the definition of the isomorphism $\varphi_i|_{y,\gamma}$ of \cref{eq:varphiiygamma}. We consider the linear map
\begin{equation*}
T|_{(0,y_0)} \otimes T|_{(1,y)}^{*}\otimes E_i|_{y}  \to  E_{i_0}|_{y_0}^{*} \otimes T|_{(0,y_0)} \otimes T|_{(1,y)}^{*}\otimes E_i|_y \otimes E_{i_0}|_{y_0}
\end{equation*}
induced by the counit $\C \to E^{*} \otimes E$ of the vector space $E=E_{i_0}|_{y_0}$. The vector space on the left-hand side is isomorphic to $K_i|_{y,\gamma}(\mathcal{T}) \otimes E_i|_{y}$, and the vector space on the right-hand side is isomorphic to $\inf R_{i_0i}|_{\gamma}(\mathcal{T})\otimes E_{i_0}|_{y_0}$ via $\delta_{y_0,y}^{-1} \otimes \id$. Composed with the projection to the quotient $\inf R_{i_0i}|_{\gamma}(\mathcal{T})\otimes E_{i_0}|_{y_0} \to \inf R_{i_0i}|_{\gamma}(\mathcal{T})\otimes_{\inf A_0} E_{i_0}|_{y_0}$, it becomes an isomorphism, and this is the definition of $\varphi_i|_{y,\gamma}$. It is straightforward to check that this definition is compatible with a change of the trivialization, smooth, and connection-preserving.

\begin{lemma}
For all $i\in I$, the bundle isomorphism $\varphi_i$ induces a 2-isomorphism $\varphi_i: \mathcal{E}_i \circ \mathcal{A}|_{Q_i} \Rightarrow \mathcal{E}_i'$, i.e., the diagram
\begin{equation*}
\small
\alxydim{@C=1.8cm@R=3em}{\inf L|_{\gamma_1\cup \gamma_2}\otimes K_i|_{y_2,\gamma_2}\otimes E_i|_{y_2} \ar[d]_{\id \otimes \varphi_i} \ar[r]^-{\kappa \otimes \id} &  K_i|_{\gamma_1,y_1} \otimes P_{y_1,y_2}\otimes E_i|_{y_2} \ar[r]^-{\id \otimes \zeta_i} & K_i|_{y_1,\gamma_1}\otimes E_i|_{y_1}  \ar[d]^{\varphi_i} \\ \inf L|_{\gamma_1\cup \gamma_2}\otimes E_i'|_{\gamma_2} \ar[rr]_{\zeta_i'} && E_i'|_{\gamma_1}}
\end{equation*}
is commutative for any point $((y_1,\gamma_1),(y_2,\gamma_2)) \in Z_i^{[2]}$. \end{lemma}

\begin{proof}
Choosing a trivialization $\mathcal{T}$ of $\mathcal{G}$ along $\gamma_1\cup\gamma_2$, we obtain an isomorphism $\inf L_{\gamma_1\cup\gamma_2}\cong \C$. Substituting  the definitions  $\zeta_i'$ (via the fusion representation), of $K_i$, and of and $\varphi_i$ (on the left using $\delta_{y_0,y_2}$ and on the right using $\delta_{y_0,y_1}$) and using \cref{eq:pointydiff}, the diagram can be reduced to 
\begin{equation*}
\small
\alxydim{@C=-3cm@R=3em}{ &  T|_{(0,y_0)} \otimes T|_{(1,y_1)}^{*}\otimes P_{y_1,y_2}\otimes E_i|_{y_2} \ar[rd]^-{\id \otimes \zeta_i} &  \\  T|_{(0,y_0)} \otimes T|_{(1,y_2)}^{*}\otimes E_i|_{y_2} \ar[ur]^-{\kappa|_{\inf L=\C} \otimes \id} \ar[dr] & &
T|_{(0,y_0)} \otimes T|_{(1,y_1)}^{*}\otimes E_i|_{y_1} \\& T|_{(0,y_0)} \otimes T^{*}|_{(1,y_2)} \otimes P_{y_1,y_2}^{*} \otimes P_{y_1,y_2} \otimes E_i|_{y_2} \ar[ur]_-{\id \otimes \tau^{tr-1} \otimes \zeta_i} &  \text{.}}
\end{equation*}
Consulting the definitions, we observe that this diagram is commutative. \end{proof}

So far we have constructed a TBG 1-isomorphism $(\mathcal{A},\varphi): (\mathcal{G}',\mathcal{E}_i') \to  (\mathcal{G},\mathcal{E})$ for every TBG object $(\mathcal{G},\mathcal{E})$. With the following lemma we show that this assignment is natural with respect to TBG morphisms, and hence results in  a natural isomorphism.  

\begin{lemma}
Let $(\mathcal{A},\psi): (\mathcal{G}_1,\mathcal{E}_1) \to (\mathcal{G}_2,\mathcal{E}_2)$ be a TBG 1-morphism. For $a=1,2$, let $(\mathcal{G}_a',\mathcal{E}_{a}')$ be the TBG objects obtained by transgressing and then regressing $(\mathcal{G}_a,\mathcal{E}_a)$, and let $(\mathcal{A}_a,\varphi_a)$ be the corresponding 1-isomorphisms constructed above. Further, let $(\mathcal{A}',\psi')$ be obtained by transgressing and regressing $(\mathcal{A},\psi)$. Then, there exists a 2-isomorphism 
\begin{equation*}
\small
\alxydim{@=3em}{(\mathcal{G}_1',\mathcal{E}'_1) \ar[d]_{(\mathcal{A}_1,\varphi_1)} \ar[r]^-{(\mathcal{A}',\psi')} & (\mathcal{G}'_2,\mathcal{E}'_2) \ar@{=>}[dl]|*+{\phi} \ar[d]^{(\mathcal{A}_2,\varphi_2)} \\ (\mathcal{G}_1,\mathcal{E}_1)  \ar[r]_{(\mathcal{A},\psi)} & (\mathcal{G}_2,\mathcal{E}_2) }
\end{equation*}
in the bicategory $\tsg(M,Q)$. In other words, the diagram 
\begin{equation*}
\small
\alxydim{@=3em}{(\mathcal{G}_1',\mathcal{E}'_1) \ar[d]_{(\mathcal{A}_1,\varphi_1)} \ar[r]^-{(\mathcal{A}',\psi')} & (\mathcal{G}'_2,\mathcal{E}'_2)  \ar[d]^{(\mathcal{A}_2,\varphi_2)} \\ (\mathcal{G}_1,\mathcal{E}_1)  \ar[r]_{(\mathcal{A},\psi)} & (\mathcal{G}_2,\mathcal{E}_2) \text{.}}
\end{equation*}
is commutative in $\hc 1 (\tsg(M,Q))$.
\end{lemma}

\begin{proof}
The 2-isomorphism $\phi$ itself is already part of the naturality of ordinary bundle gerbe transgression and has been constructed in \cite[Lemma 6.1.4]{waldorf10}. 
It remains to show that it is compatible with the 2-isomorphisms of the TBG morphisms. This is equivalent to the commutativity of the following pentagon diagram:
\begin{equation*}
\small
\alxydim{@R=3em@C=0.5em}{&& \mathcal{E}_{2,i}\circ \mathcal{A}_2\circ \mathcal{A}' \ar@{<=}[drr]^-{\varphi_{2,i}\circ \id} \ar@{=>}[dll]_-{\id \circ \phi} && \\ \mquad\mathcal{E}_{2,i} \circ \mathcal{A} \circ \mathcal{A}_1\mquad \ar@{<=}[dr]_{\psi_i \circ \id} &&&& \mathcal{E}_{2,i}' \circ \mathcal{A}'\quad \ar@{<=}[dl]^{\psi_i'} \\ & \mathcal{E}_{1,i} \circ \mathcal{A}_1 \ar@{<=}[rr]_-{\varphi_{1,i}} && \mathcal{E}'_{1,i}}
\end{equation*}
We translate this further into a diagram of the bundle isomorphisms underlying these 2-isomorphisms. Let $Q$ be the vector bundle of the 1-isomorphism $\mathcal{A}:\mathcal{G}_1\to \mathcal{G}_2$, defined over the fibre product $Y_1 \times_M Y_2$ of the surjective submersions of the two bundle gerbes.  
Further, we recall that $\mathcal{A}': \mathcal{G}_1' \to \mathcal{G}_2'$ is the 1-isomorphism induced by the identity map on $P_{x_0}M$ and the bundle isomorphism $\varphi:\inf L_1 \to \inf L_2: [\mathcal{T},z] \mapsto [\mathcal{T} \circ \mathcal{A}^{-1},z]$. Thus, its hermitian line bundle is $\inf L_2$.
In the following we work over a point $(y_1,y_2,\gamma_1,\gamma_2)\in Y_1 \times_M Y_2 \times_M P_{x_0}M^{[2]}$, over which our diagram becomes
 the following diagram of linear maps:
\begin{equation*}
\small
\alxydim{@R=3em@C=-2em}{&& \inf L_2|_{\gamma_1\cup\gamma_2} \otimes K_2|_{y_2,\gamma_2}\otimes  E_{2,i}|_{y_2} \ar@{<-}[drr]^-{\id \otimes \varphi_{2,i}} \ar[dll]_-{\phi \otimes \id} && \\ \mqquad K_1|_{y_1,\gamma_1} \otimes  Q_{y_1,y_2} \otimes  E_{2,i}|_{y_2}\mqqquad \ar@{<-}[dr]_{\id \otimes \psi_i} &&&& \inf L_2|_{\gamma_1\cup\gamma_2} \otimes E'_{2,i}|_{\gamma_2} \ar@{<-}[dl]^{\psi'_i} \\ & \mquad K_1|_{y_1,\gamma_1} \otimes E_{1,i}|_{y_1} \ar@{<-}[rr]_-{\varphi_{1,i}} && E'_{1,i}|_{\gamma_1}}
\end{equation*}
For further elaboration, we choose a trivialization $\mathcal{T}_{1}$ of $\mathcal{G}_1$ along $\gamma_1\cup\gamma_2$, and consider its two halves $\iota_1^{*}\mathcal{T}_{1}$ and $\iota_2^{*}\mathcal{T}_{1}$, as well as the composite $\mathcal{T}_2 := \mathcal{T}_1 \circ \mathcal{A}^{-1}$ and its two halves $\iota_1^{*}\mathcal{T}_{2}$ and $\iota_2^{*}\mathcal{T}_{2}$. Then,  using \cref{eq:defK},     we are able to express the line bundles $K_1$ and $K_2$ in terms of the line bundles $T_1$ of $\mathcal{T}_1$ and the line bundle $Q$, and we can express using \cref{eq:transreg:1} and the points $y_1,y_2$  the vector bundle        $E_{1,i}'$ and $E_{2,i}$ in terms of $E_{1,i}$, $E_{2,i}$, $T_1$, and $Q$. 
The main point now is to compute the 2-isomorphism $\psi_i': \mathcal{E}'_{1,i} \Rightarrow \mathcal{E}'_{2,i} \circ \mathcal{A}'$, defined as the regression of the isomorphism $(\varphi,\xi)$, which in turn is the transgression of $(\mathcal{A},\psi)$. For the regression we only need  
\begin{equation*}
\xi_{i_0i}:\inf R_{1,i_0i}|_{\gamma_1}(\iota_1^{*}\mathcal{T}_{1}) \to  \inf R_{2,i_0i}|_{\gamma_1}(\iota_1^{*}\mathcal{T}_{2})\text{,} 
\end{equation*}
which is -- under the above isomorphisms -- determined by $\psi_i$ via
\begin{align*}
&\mquad  E_{1,i_0}|_{y_{1,0}}^{*} \otimes T_1|_{(0,y_{1,0})} \otimes T_1|_{(1,y_1)} ^{*}\otimes E_{1,i}|_{y_1}
\\*&\quad\alxydim{@=7em}{\ar[r]^-{\psi_{i_0}^{tr-1} \otimes \id \otimes \id \otimes \psi_i} & } E_{2,i_0}|_{y_{2,0}}^{*} \otimes Q^{*}_{y_{1,0},y_{2,0}} \otimes T_1|_{(0,y_{1,0})} \otimes  T_1|_{(1,y_1)}^{*}\otimes Q_{y_1,y_2} \otimes E_{2,i}|_{y_2}\text{.}
\end{align*}
For the subsequent regression of $\xi_{i_0i}$ we need to fix an intertwiner  $f_0: E_{1,i_0}|_{y_{1,0}} \to E_{2,i_{0}}|_{y_{2,0}}$, see \cref{sec:regression}. Then  -- under our isomorphisms -- the bundle isomorphism  $\psi_i':E'_{1,i}|_{\gamma_1} \to \inf L_2|_{\gamma_1\cup\gamma_2} \otimes E'_{2,i}|_{\gamma_2}$ becomes a map
\begin{multline*}
E_{1,i_0}|_{y_{1,0}}^{*} \otimes T_1|_{(0,y_{1,0})} \otimes  T_1|_{(1,y_1)}^{*}\otimes E_{1,i}|_{y_1} \otimes E_{1,i_0}|_{y_{1,0}} \\\to \inf L_2|_{\gamma_1\cup\gamma_2} \otimes E_{2,i_0}|_{y_{2,0}}^{*} \otimes Q^{*}_{y_{1,0},y_{2,0}} \otimes T_1|_{(0,y_{1,0})} \otimes  T_1|_{(1,y_1)}^{*}\otimes Q_{y_1,y_2}\otimes E_{2,i}|_{y_2}\otimes E_{2,i_0}|_{y_{2,0}}
\end{multline*}
given by
\begin{equation*}
\omega \otimes s\otimes\sigma\otimes v \otimes v_0 \mapsto [\mathcal{T}_2,1] \otimes \psi_{i_0}^{tr-1}(\omega) \otimes s \otimes \sigma \otimes  \psi_i(v) \otimes f_0(v_0 )\text{.}
\end{equation*}
In the following we will use a specific intertwiner $f_0$, obtained by choosing an element $\chi_0\in Q^{*}|_{y_{1,0},y_{2,0}}$ and setting $f_0:= (\chi_0 \otimes \id) \circ \psi_{i_0}|_{y_{1,0},y_{2,0}}$. 

Finally, we recall the definition of the 2-isomorphism $\phi$ from \cite[Lemma 6.1.4]{waldorf10}. Its bundle morphism
\begin{align*}
\phi: \inf L_2|_{\gamma_1\cup\gamma_2} \otimes Q^{*}_{y_{1,0},y_{2,0}} \otimes T_1|_{(0,y_{1,0})} \otimes T_1|_{(1,y_1)}^{*}\otimes Q_{y_1,y_2} &\to T_1|_{(0,y_{1,0})} \otimes T_1|_{(1,y_1)}\otimes  Q_{y_1,y_2}
\end{align*}
is determined so that it sends $[\mathcal{T}_2,1] \otimes  \chi_0 \otimes s \otimes \sigma \otimes q$ to $s \otimes \sigma \otimes q$,
where $\chi_0$ is the element fixed above. We can now prove the commutativity of the pentagon diagram. The counter-clockwise composition acts as
\begin{equation*}
\omega \otimes s \otimes \sigma \otimes v \otimes v_0 \mapsto \omega_0(f_0(v_0)) \cdot  s \otimes \sigma \otimes  q \otimes v'\text{,}
\end{equation*}
where $\omega_0$ is determined so that $\psi_{i_0}^{tr-1}(\omega)  = \omega_0 \otimes \chi_0$, and $q \otimes v' := \psi_i(v)$. On the other hand, the clock-wise composition yields
\begin{equation*}
\omega \otimes s \otimes \sigma \otimes v \otimes v_0 \mapsto \omega(v_0) \cdot s \otimes\sigma \otimes  q \otimes v'\text{,}
\end{equation*}
and it is straightforward to deduce the coincidence $\omega(v_0)= \omega_0(f_0(v_0))$ from the given definitions.
Thus, the pentagon diagram is commutative. 
\end{proof}

\subsection{Transgression after Regression}

\label{sec:traftereg}

In this section we provide a natural equivalence $\mathscr{T}\circ \mathscr{R}_{i_0,x_0,\inf F_0}\cong \id_{\lsg(M,Q)}$, thus establishing the second half of the equivalence $\tsg(M,Q)\cong \lsg(M,Q)$.
To that end, we construct for each LBG object $(\inf L,\lambda, \inf R,\phi,\chi,\epsilon,\alpha)$  a 1-morphism 
\begin{equation*} 
(\varphi,\xi): (\mathscr{T}\circ \mathscr{R}_{i_0,x_0,\inf F_0})(\inf L,\lambda, \inf R,\phi,\chi,\epsilon,\alpha) \to (\inf L,\lambda, \inf R,\phi,\chi,\epsilon,\alpha)\text{,}
\end{equation*}
and show that it is natural with respect to LBG morphisms.

We let $(\mathcal{G},\mathcal{E}) := \mathscr{R}_{i_0,x_0,\inf F_0}(\inf L,\lambda, \inf R,\phi,\chi,\epsilon,\alpha)$ be the regressed TBG, and  denote its transgressed LBG by $(\inf L',\lambda', \inf R',\phi',\chi',\epsilon',\alpha') := \mathscr{T}(\mathcal{G},\mathcal{E})$. In \cite[Section 6.2]{waldorf10} we have constructed an isomorphism $\varphi: \inf L' \to \inf L$ of hermitian line bundles over $LM$. We recall this construction in a slightly adapted version. We consider a loop $\tau \in LM$  of the form $\tau =(\gamma_1 \pcomp \gamma_0)\cup (\gamma_2\pcomp \gamma_0)$, where $\gamma_0\in P_{x_0}M$ and $(\gamma_1,\gamma_2)\in PM^{[2]}$. Up to thin homotopy, every loop in $M$ is of this form. Now let $\mathcal{T}:\tau^{*}\mathcal{G} \to \mathcal{I}_0$ be a trivialization, including a hermitian line bundle $T$ with connection over $Z := S^1 \ttimes{\tau}{\ev_1} P_{x_0}M$ and a connection preserving bundle isomorphism
\begin{equation*}
\delta|_{(t,\beta_1),(t,\beta_2)}: \inf L|_{\beta_1\cup\beta_2} \otimes T_{(t,\beta_2)} \to T_{(t,\beta_1)}
\end{equation*}
over $Z^{[2]}$. We consider the obvious paths $\alpha_i$ in $Z$ connecting $(0,\id\pcomp\gamma_0)$ with $(\frac{1}{2},\gamma_i\pcomp \gamma_0)$. Then, there exists a unique element $p\in \inf L|_{\tau}=\inf L|_{(\gamma_1 \pcomp \gamma_0)\cup (\gamma_2\pcomp \gamma_0)}$ such that
\begin{equation*}
\delta(p\otimes pt_{\alpha_2}(q)) = pt_{\alpha_1}(q)
\end{equation*}
for all $q\in T_{(0,\id\pcomp \gamma_0)}$. We define $\varphi([\mathcal{T},z]) := p\cdot z$. Using the superficial connection on $\inf L$ and its parallel transport along thin homotopies, this definition can be extended to all loops $\tau\in LM$. This results in a fusion-preserving bundle morphism $\varphi:\inf L' \to \inf L$ that is independent of all involved choices, see \cite[Lemma 6.2.1 \& 6.2.2]{waldorf10}.

Next, we construct, for $i,j\in I$, a vector bundle isomorphism
\begin{equation*}
\xi_{ij}:\inf R'_{ij} \to \inf R_{ij}
\end{equation*}
over $P_{ij}$. First, we compute $\inf R'_{ij}$ over a path $\gamma\in P_{ij}$ connecting $x\in Q_i$ with $y\in Q_j$. Let $x_0\in Q_{i_0}$ be the point and $\inf F_0$ be the $\inf A_{x_0}$-module chosen for regression,  and recall that we constructed the bundle $E_i :=  \inf R_{i_0i} \otimes _{\ev_0^{*}\inf A_0} \inf F_0 $ over $P_i$. We choose a path $\gamma_0\in P_{i_0i}$ with $\gamma_0(0)=x_0$ and $\gamma_0(1)= x$, and obtain a smooth map
\begin{equation}
\label{eq:usualsection}
s:[0,1] \to P_{x_0}M: t\mapsto  \gamma^t \star \gamma_0\text{,}
\end{equation}
where $\gamma^t$ denotes the restriction of a path to $[0,t]$, reparameterized to $[0,1]$ (with sitting instants). 
In fact, $s$ is a section along $\gamma$ into the surjective submersion of $\gamma^{*}\mathcal{G}$. By \cref{rem:sectionsandtrivializations} this induces  a trivialization $\mathcal{T}: \gamma^{*}\mathcal{G} \to \mathcal{I}_0$, together with an isomorphism
\begin{align}
\label{eq:ident:1}
\inf R'_{ij}|_{\gamma}(\mathcal{T}) &= \mathrm{Hom}(\Des(\mathcal{E}_i|_{x},\mathcal{T}|_0),\Des(\mathcal{E}_j|_{y}),\mathcal{T}|_1) \nonumber
\\&\cong  \mathrm{Hom} ((s^{*}E_i)|_0,(s^{*}E_j)|_1)
= \mathrm{Hom}(E_i|_{\cp_{x} \star \gamma_0} , E_j|_{\gamma \star \gamma_0})\text{.}
\end{align}
Using the definition of $E_i$ and $E_j$ via regression,  as well as the  lifted path concatenation, we obtain another isomorphism
\begin{align}
\mathrm{Hom}(E_i|_{\cp_{x} \star \gamma_0} , E_j|_{\gamma \star \gamma_0})&=\mathrm{Hom}(\inf R_{i_0i}|_{\cp_x \pcomp \gamma_0} \otimes _{\inf A_{x_0}} \inf F_0,   \inf R_{i_0j}|_{\gamma\star \gamma_0}\otimes _{\inf A_{x_0}} \inf F_0)\nonumber \\ &\cong \mathrm{Hom}(\inf R_{i_0i}|_{\gamma_0} \otimes _{\inf A_{x_0}} \inf F_0,\inf R_{ij}|_{\gamma}\otimes _{\inf A_i|_{x}}    \inf R_{i_0i}|_{\gamma_0}\otimes _{\inf A_{x_0}} \inf F_0) \nonumber
\\
&=\mathrm{Hom}(E_i|_{\gamma_0},\inf R_{ij}|_{\gamma}\otimes _{\inf A_i|_{x}}  E_i|_{\gamma_0})\text{.}
\label{eq:ident:2}
\end{align}
Next, we consider the linear map
\begin{equation}
\label{eq:ident:3}
\inf R_{ij}|_{\gamma} \to \mathrm{Hom}(E_i|_{\gamma_0},\inf R_{ij}|_{\gamma}\otimes _{\inf A_i|_{x}}  E_i|_{\gamma_0})
\end{equation}
that sends a vector $v\in \inf R_{ij}|_{\gamma}$ to the homomorphism $w \mapsto v \otimes_{\inf A_i|_x} w$. Using that $\inf A_i|_x$ is simple and that $E_i|_{\gamma_0}$ is irreducible, it is straightforward to see that  \cref{eq:ident:3} is injective. 
Moreover, the dimensions on both sides coincide, so that it is in fact an isomorphism,
\begin{multline*}
\dim(\inf R_{ij}|_{\gamma}) \eqcref{re:lsg:rank2}  \sqrt{\mathrm{rk}(\inf R_{ii})\cdot \mathrm{rk}(\inf R_{jj})}
 \eqcref{lem:rkEi} \mathrm{rk}(E_i)\cdot \mathrm{rk}(E_j) \\[-2em]=\dim(\mathrm{Hom}(E_i|_{\cp_{x} \star \gamma_0} , E_j|_{\gamma \star \gamma_0})) \eqcref{eq:ident:2}  \dim(\mathrm{Hom}(E_i|_{\gamma_0},\inf R_{ij}|_{\gamma}\otimes _{\inf A_i|_{x}}  E_i|_{\gamma_0}))\text{.}
\end{multline*}
Combining \cref{eq:ident:1,eq:ident:2} with the inverse of \cref{eq:ident:3}, we obtain the isomorphism
\begin{equation*}
\xi_{ij}|_{\gamma}:\inf R_{ij}'|_{\gamma} \to \inf R_{ij}|_{\gamma}\text{.}
\end{equation*}

\begin{lemma}
 $\xi_{ij}|_{\gamma}$ is independent of the choice of the path $\gamma_0$.
\end{lemma}

\begin{proof}
If $\gamma_0'$ is another path connecting $x_0$ with $x=\gamma(0)$, then we have the two sections $s,s'$ into the submersion of $\gamma^{*}\mathcal{G}$, inducing trivializations $\mathcal{T}$ and $\mathcal{T}'$, respectively. Any element $p\in \inf L_{\gamma_0\cup\gamma_0'}$ determines (via a thin homotopy) a section of $\inf L$  along $(s,s')$, and hence a hence a 2-isomorphism $\psi_p: \mathcal{T} \Rightarrow \mathcal{T}'$, see \Cref{rem:sectionsandtrivializations}. 
Our aim is to show that the following diagram is commutative, whose top row and bottom row are the isomorphisms $\xi_{ij}^{-1}$, defined using $\gamma_0'$ and $\gamma_0$, respectively.
\begin{equation*}
\small
\alxydim{@R=1.5em@C=4em}{\inf R_{ij}|_{\gamma} \ar@{=}[dd] \ar[r]^-{\labelcref{eq:ident:3}} & \mathrm{Hom}(E_i|_{\gamma_0'},\inf R_{ij}|_{\gamma}\otimes _{\inf A_i|_{x}}  E_i|_{\gamma_0'}) \ar[dd]|{((\id\otimes \rho_p \otimes \id_{\inf F_0} ) \circ - \circ(\rho_p \otimes \id_{\inf F_0} )^{-1})} \ar[r]^-{\labelcref{eq:ident:2}} & \mathrm{Hom}(E_i|_{\cp_{x} \star \gamma_0'} , E_j|_{\gamma \star \gamma_0'}) \ar[dd]|{\zeta'_j \circ - \circ \zeta_i^{\prime-1}} \ar[r]^-{\labelcref{eq:ident:1}} & \inf R'_{ij}|_{\gamma}(\mathcal{T}') \ar[dd]^{r_{\psi_p}}  \\  & \\ \inf R_{ij}|_{\gamma} \ar[r]_-{\labelcref{eq:ident:3}}  & \mathrm{Hom}(E_i|_{\gamma_0},\inf R_{ij}|_{\gamma}\otimes _{\inf A_i|_{x}}  E_i|_{\gamma_0}) \ar[r]_-{\labelcref{eq:ident:2}} & \mathrm{Hom}(E_i|_{\cp_{x} \star \gamma_0} , E_j|_{\gamma \star \gamma_0})\ar[r]_-{\labelcref{eq:ident:1}} & \inf R'_{ij}|_{\gamma}(\mathcal{T})\text{.}}
\end{equation*}
The vertical maps in this diagram all depend  on the choice of the point $ p\in \inf L_{\gamma_0\cup\gamma_0'}$: first, we have an isomorphism $\rho_p: \inf R_{i_0i}|_{\gamma_0'} \to \inf R_{i_0i}|_{\gamma_0}$ of right $\inf A_{x_0}$-modules, obtained by fusion  with $p$.
Second, we have  linear isomorphisms $\zeta_i': E_i|_{\cp_{x} \star \gamma_0'} \to  E_i|_{\cp_{x} \star \gamma_0}$ and  $\zeta_j': E_j|_{\gamma \star \gamma_0'} \to  E_j|_{\gamma \star \gamma_0}$ obtained using the gerbe module isomorphisms $\zeta_i$ and $\zeta_j$, with the first argument fixed by the elements in $\inf L_{(\cp_x \pcomp \gamma_0) \cup (\cp_x \pcomp \gamma_0')}$  and $\inf L_{(\gamma \star\gamma_0)\cup (\gamma\pcomp \gamma_0')}$ determined by $p$ under thin homotopies. Now, all three subdiagrams are commutative: the one on the left commutes obviously by inspection,  the one in the middle commutes due to the fact that $\zeta_i$ is defined (in the process of regression) by the fusion representation, and the one on the right-hand side commutes by definition of the isomorphism $r_{\psi_p}$ (see \cref{sec:transbranes}) and \cref{rem:sectionsandtrivializations}. 
\end{proof}

By the previous lemma, the isomorphisms $\xi_{ij}|_{\gamma}$ assemble into a well-defined map $\xi_{ij}:\inf R_{ij}'\to\inf R_{ij}$. 

\begin{lemma}
\label{lem:632}
For all $i,j\in I$, the map $\xi_{ij}$ is a connection-preserving bundle morphism. 
\end{lemma}

\begin{proof}
We first show the smoothness of $\xi_{ij}$.
Let $c: U \to P_{ij}$ be a plot and let  $W \subset U$ be a contractible open subset such that we can find a local trivialization $\phi: W \times \C^{k} \to W \times_{P_{ij}} \inf R_{ij}$. Due to the contractibility of $W$, there exists a point $w_0\in W$ and a smooth map $W \to PW:w \mapsto \gamma_w$ with $\gamma_w(0)= w_0$ and $\gamma_w(1)=w$ for all $w\in W$. We let $c_0 := \ev_0 \circ c$ be the plot of initial points, set $x := c_0(w_0)$ and choose a path $\gamma_0$ in $M$ connecting $x_0$ with $x$. For $w\in W$, we have the path $\gamma^{c}_w :=c_0 \circ \gamma_w\in P_{ii}$ connecting $x$  with $c_0(w)$. Finally (upon choosing a smoothing function which we suppress from the notation), for every $w\in W$ and $t\in[0,1]$ we have a path $\beta_{w,t}\in PM$ defined by $\beta_{w,t}(s) := c(w)(ts)$; it connects $\beta_{w,t}(0)=c_0(w)$ with $\beta_{w,t}(1)=c(w)(t)=c^{\vee}(w,t)$. The concatenation of these paths defines a smooth map
\begin{equation}
\label{eq:transreg:sigma}
\sigma: [0,1] \times W \to P_{x_0}M: (t,w) \mapsto \beta_{w,t} \pcomp (\gamma_w^{c} \pcomp \gamma_0)
\end{equation}
such that $\ev_1 \circ \sigma=c^{\vee}$. In other words, $\sigma$ is a lift of $c^{\vee}$ to the surjective submersion of the regressed bundle gerbe $\mathcal{G}$ and hence defines a trivialization $\mathcal{T}: (c^{\vee})^{*}\mathcal{G} \to \mathcal{I}_{\sigma^{*}B}$. We  want to compute  the corresponding  vector bundles $\inf V_{\mathcal{T}}$ and $\inf W_{\mathcal{T}}$ over $W$, as defined in \cref{sec:transbranes}.
For this purpose, we consider the map $\sigma': W \to P_i$ defined by $\sigma'(w):=\gamma_w^{c} \pcomp \gamma_0$, so that $(\sigma\circ j_0)(w)= \cp_{c_0(w)} \pcomp \sigma'(w)$.
We obtain an isomorphism of vector bundles over $W$:
\begin{multline}
\label{eq:transreg:VT}
\inf V_{\mathcal{T}}\cong (\sigma\circ j_0)^{*}E_i =(\sigma\circ j_0)^{*}\inf R_{i_0 i} \otimes_{\inf A_{x_0}} \inf F_0 \\\cong(\cp \circ c_0)^{*}\inf R_{ii} \otimes_{c_0^{*}\inf A_i} \sigma'^{*}\inf R_{i_0 i} \otimes_{\inf A_{x_0}} \inf F_0\cong  \sigma'^{*}\inf R_{i_0 i} \otimes_{\inf A_{x_0}} \inf F_0
\end{multline}
using, respectively,  \cref{rem:sectionsandtrivializations}, the definition of $E_i$ under  regression, \cref{lem:compositionofbimodules}, and the definition  $\inf A_i|_{c_0(w)} = \inf R_{ii}|_{\cp_{c_0(w)}}$. 
Analogously, we obtain an isomorphism
\begin{equation}
\label{eq:transreg:WT}
\inf W_{\mathcal{T}} \cong c^{*}\inf R_{ij} \otimes_{c_0^{*}\inf A_i} \sigma'^{*}\inf R_{i_0i}\otimes_{\inf A_{x_0}} \inf F_0\text{.} 
\end{equation}
Combining these two isomorphisms, we have an isomorphism 
\begin{equation*}
\psi_{ij}:\mathrm{Hom}(\inf V_{\mathcal{T}},\inf W_{\mathcal{T}}) \to c^{*}\inf R_{ij} \otimes_{c_0^{*}\inf A_i} \mathrm{End}(\inf V_{\mathcal{T}})\text{,}
\end{equation*} 
 of vector bundles over $W$,
and by construction it restricts over each point $w\in W$ to the isomorphism $\xi_{ij}|_{c(w)}$. Using the given local trivialization $\phi$ of $\inf R_{ij}$, we consider the  trivialization
$\tau': W \times \C^{k} \to \mathrm{Hom}(\inf V_{\mathcal{T}},\inf W_{\mathcal{T}})$
defined by $\tau'(w,v) := \psi_{ij}^{-1}(\phi(w,v) \otimes \id)$, and consider the associated local trivialization $\phi'$ of $\inf R_{ij}'$, which reads as
\begin{equation*}
\phi'(w,v) :=(w,[i_w^{*}\mathcal{T},\tau'(w,v)])\text{.}
\end{equation*}
Under  the two local trivializations $\phi$ and $\phi'$, the map $\xi_{ij}$ corresponds to the identity, and hence it is smooth. 

In order to show that $\xi_{ij}$  is connection-preserving, 
we compute the local connection 1-form $\omega_{\phi'}$ associated to the local trivialization $\phi'$, and prove that $\omega_{\phi'}=\omega_{\phi}$. 
According to the proof of \cref{lem:connection} (see \cref{eq:localoneform}), we have
\begin{equation*}
\omega_{\phi'} = \omega_{\tau'} + \int_{[0,1]}\sigma^{*}B' \in \Omega^1(W,\mathfrak{gl}(\C^{k}))\text{,}
\end{equation*}
where $B'$ is the curving of $\mathcal{G}'$. 
We observe from the constructions that the bundle isomorphisms \cref{eq:transreg:VT,eq:transreg:WT}, and hence the bundle isomorphism $\psi_{ij}$ are connection-preserving. The definition of $\tau'$, together with the fact that the identity section in $\mathrm{End}(\inf V_{\mathcal{T}})$ is parallel, show that $\omega_{\tau'} = \omega_{\phi}$. 
It remains to prove that the 1-form
\begin{equation}
\label{eq:fibreintegralvanishes}
\int_{[0,1]}\sigma^{*}B'  \in \Omega^1(W)
\end{equation}
vanishes. This is complicated by the fact that the regressed curving $B'$ is defined using a correspondence between 2-forms and smooth functions on a space of bigons; see \cite{schreiber5} for the general theory of this correspondence.
In our case, the 2-form $B' \in \Omega^2(P_{x_0}M)$ corresponds to  a certain  map $G_{\inf L}$ whose definition we will recall below. In \cref{sec:fibreintegration} we show that the 1-form \cref{eq:fibreintegralvanishes} corresponds to the map
\begin{equation*}
PW \to \ueins: \gamma \mapsto G_{\inf L}(\sigma_{*}(\Sigma_{\gamma}))\text{,}
\end{equation*}
where the bigon $\Sigma_{\gamma}$ is defined by $\Sigma_{\gamma} := (\id \times \gamma)_{*}(\Sigma_{1,1})$, where $\Sigma_{1,1}:[0,1]^2 \to [0,1]^2$ is the so-called standard bigon.
Using the definition of $\sigma$ from \cref{eq:transreg:sigma}, the bigon $ \sigma_{*}(\Sigma_{\gamma})$ is 
\begin{equation*}
(s,t) \mapsto \beta_{\gamma(\xi_2(s,t)),\xi_1(s,t)} \pcomp (\gamma_{\gamma(\xi_2(s,t))}^{c} \pcomp \gamma_0)\text{,}
\end{equation*}
where $\xi_1,\xi_2$ are the two components of $\Sigma_{1,1}$, i.e. $\Sigma_{1,1}=(\xi_1,\xi_2)$. In the following we prove that $G_{\inf L}(\sigma_{*}(\Sigma_{\gamma}))=1$ for all $\gamma\in PW$; this shows that \cref{eq:fibreintegralvanishes} is zero. In order to do this,  we recall the definition of $G_{\inf L}(\Sigma)$ given in  \cite[Section 5.2]{waldorf10}, which is essentially by parallel transport in $\inf L$  along a path $\gamma_{\Sigma}$  in $LM$ (see Figure 1 in Section 5.2 of \cite{waldorf10} for a picture of this path).   We write $\chi_{\gamma} := \gamma_{\sigma_{*}(\Sigma_{\gamma})}$ for simplicity; this is a path in $LM$ given by
\begin{multline}
\label{eq:gammaSigma}
\chi_\gamma(t) =(\beta_{\gamma(\xi_2(-,t)),\xi_1(-,t)} (1) \pcomp \beta_{\gamma(\xi_2(0,t)),\xi_1(0,t)} \pcomp \gamma_{\gamma(\xi_2(0,t))}^{c} \pcomp \gamma_0)\\ \cup (\cp \pcomp  \beta_{\gamma(\xi_2(1,t)),\xi_1(1,t)} \pcomp\gamma_{\gamma(\xi_2(1,t))}^{c} \pcomp \gamma_0) \in LM\text{.}
\end{multline}
The end loops $\chi_\gamma(0)$ and $\chi_\gamma(1)$ are \quot{flat}, i.e. in the image of the map $\gamma \mapsto \gamma \cup \gamma$, along which $\inf L$ has the canonical flat section $\nu$.
In particular, this defines elements $\nu_0,\nu_1 \in \inf L$ projecting to $\chi_\gamma(0)$ and $\chi_\gamma(1)$, respectively.  Then, $G_{\inf L}(\sigma_{*}(\Sigma_{\gamma}))$  is defined by
\begin{equation}
\label{eq:transreg:ptcan}
pt_{\chi_{\gamma}}(\nu_0) = \nu_1\cdot G_{\inf L}(\sigma_{*}(\Sigma_{\gamma}))\text{.}
\end{equation} 
We consider for $r\in [0,1]$  maps $\xi_1^r,\xi_2^r :[0,1]^2 \to [0,1]$ that constitute fixed-ends homotopies between $\xi_1$, $\xi_2$  and the map $(s,t)\mapsto t$. That is, we have $\xi_1^1=\xi_1$ and $\xi_2^1=\xi_2$, and $\xi_1^0(s,t)=\xi_2^0(s,t)=t$, as well as $\xi_1^r(s,0)=\xi_2^r(s,0)=0$  and $\xi_1^r(s,1)=\xi_2^r(s,1)=1$, for all $s,t\in [0,1]$. We consider the path $\chi_{\gamma}^{r}$ defined as in \cref{eq:gammaSigma} but using $\xi_1^r$ and $\xi_2^r$ instead of $\xi_1$ and $\xi_2$. We regard $h:[0,1]^2 \to\ LM: (r,t) \mapsto\ \chi_{\gamma}^r(t)$ as a homotopy between the paths $\chi_{\gamma}$ and $\chi_{\gamma}^0$ in $LM$. Calculating the latter paths explicitly, we notice that $\chi_{\gamma}^{0}$ is a path trough flat loops.
The homotopy $h$ fixes the end-loops. We claim that the adjoint map $h^{\vee}: [0,1]^2 \times S^1 \to M$ has rank two; which can be checked explicitly using \cref{eq:gammaSigma}.
Since the connection on $\inf L$ is superficial (see \cite[Lemma 2.2.3]{waldorf10}) and  $\nu$  is parallel, we have 
$pt_{\chi_{\gamma}}(\nu_0)=pt_{\chi_{\gamma}^0}(\nu_0)=\nu_1$. Comparing with \cref{eq:transreg:ptcan}, we have the claim. 
\end{proof}

So far we have provided the data $(\varphi,\xi)$ for a LBG morphism. It remains to show that it respects the fusion representation and the lifted path concatenation, see \cref{re:epsilonunique,rem:alphaunique}.  This is done in the following two lemmas.

\begin{lemma}
The bundle isomorphism $\xi_{ij}$ respects  the fusion representations, i.e. the diagram
\begin{equation}
\label{eq:transreg:fusion:2}
\small
\alxydim{@C=2cm@R=3em}{\inf L'|_{\tau} \otimes \inf R_{ij}'|_{\gamma_2 } \ar[d]_{\varphi \otimes \xi_{ij}} \ar[r]^-{\phi_{ij}'|_{\gamma_1,\gamma_2}} &  \inf R'_{ij}|_{\gamma_1 } \ar[d]^{\xi_{ij}} \\ \inf L|_{\tau} \otimes \inf R_{ij}|_{\gamma_2} \ar[r]_-{\phi_{ij}|_{\gamma_1,\gamma_2}} & \inf R_{ij}|_{\gamma_1}}
\end{equation}
is commutative for all
 $\gamma_1,\gamma_2 \in P_{ij}$ with $\gamma_1(0)=\gamma_2(0)$ and $\gamma_1(1)=\gamma_2(1)$, and $\tau := \gamma_1  \cup \gamma_2 $. 
\end{lemma}

\begin{proof}
We set $x :=\gamma_1(0) $ and $y := \gamma_1(1)$. 
Let $\mathcal{T}$ be a trivialization of $\tau^{*}\mathcal{G}$, and let $p := \varphi(\mathcal{T})\in \inf L|_{\tau}$. We choose a path $\gamma_0$ connecting $x_0$  with $x$ and consider the corresponding sections 
\begin{equation*}
s_a: [0,1] \to P_{x_0}M: t\mapsto \gamma_a^t \star \gamma_0
\end{equation*}
 along $\gamma_a$ into the surjective submersions of $\mathcal{G}$, for $a=1,2$. We have to study  the induced trivializations $\mathcal{T}_a: \gamma_a^{*}\mathcal{G} \to \mathcal{I}_0$: their line bundles are $T_a|_{(t,\beta)} := \inf L_{\beta \cup s_a(t)}$ and their   isomorphisms are $\sigma_a|_{(t,\beta_1,\beta_2)}:= \lambda_{\beta_1,\beta_2,s_a(t)}$, for $t\in [0,1]$ and paths $\beta,\beta_1,\beta_2\in P_{x_0}M$ ending at $\gamma_a^{t}(1)$;  see \cref{rem:sectionsandtrivializations}. Let $\iota_a^{*}\mathcal{T}$ be the  restrictions of $\mathcal{T}$ to trivializations along $\gamma_a$. Then, there exist parallel unit-length sections $\sigma_a$ in $\Des(\mathcal{T}_a,\iota_a^{*}\mathcal{T})$. They induce unit-length sections into $\Des(\mathcal{T}_2|_0,\mathcal{T}_1|_0)$ and  into $\Des(\mathcal{T}_2|_1,\mathcal{T}_1|_1)$. Since $s_1(0)=s_2(0)$, we have $\mathcal{T}_1|_0=\mathcal{T}_2|_0$ and  can thus assume that the induced section into $\Des(\mathcal{T}_2|_0,\mathcal{T}_1|_0)$ is the trivial one. The induced section  into  $\Des(\mathcal{T}_2|_1,\mathcal{T}_1|_1)$ defines a 2-isomorphism $\sigma: \Des(\mathcal{E}_j|_{y},\mathcal{T}_2|_1) \Rightarrow \Des(\mathcal{E}_j|_{y},\mathcal{T}_1|_1)$. From the fact that the sections $\sigma_k$ are parallel together with the definition of $p$, we conclude that $\sigma$ is induced by   
\begin{equation}
\label{eq:transreg:fusion:1}
\alxydim{@C=2cm}{T_1|_{(1,\beta)}= \inf L_{\beta,\gamma_1\pcomp \gamma_0}\ar[r]^-{\lambda(-\otimes p)} &  \inf L_{\beta, \gamma_2\pcomp \gamma_0}  = T_2|_{(1,\beta)}\text{.} }
\end{equation}
We apply this to $\beta=\gamma_2\pcomp\gamma_0$ and obtain the diagram shown in  \cref{diag:1}. Its outer vertical arrows are the maps of \cref{rem:sectionsandtrivializations} that have been used in order to define the isomorphisms of \cref{eq:ident:1} over $\gamma_1$ and $\gamma_2$, respectively.
\begin{figure}[t]
\begin{equation*}
\small
\alxydim{@C=6em@R=3em}{
\Des(\mathcal{E}_j|_{y},\mathcal{T}_2|_1) \ar[dr] \ar[rrr]^{\sigma} \ar[ddd] &&& \Des(\mathcal{E}_j|_{y},\mathcal{T}_1|_1) \ar[dl]\ar[ddd]
\\& \mqqqquad\inf L|_{ (\gamma_2 \star \gamma_0) \cup (\gamma_2 \star \gamma_0)}^{*}\otimes E_j|_{\gamma_2\pcomp \gamma_0} \ar[r]^{\lambda(-\otimes p)^{tr} \otimes \id} \ar[d]_{\tilde\lambda\otimes\id} &\inf L|_{ (\gamma_2 \star \gamma_0) \cup (\gamma_1 \star \gamma_0)}^{*}\otimes E_j|_{\gamma_2\pcomp \gamma_0}\mqqqquad \ar[d]^{\tilde\lambda \otimes \id} &
\\& \mqqqquad\inf L|_{(\gamma_2 \star \gamma_0)\cup (\gamma_2 \star \gamma_0)}\otimes E_j|_{\gamma_2\pcomp \gamma_0}\ar[r]_{\lambda(p \otimes-) \otimes \id} \ar[dl]_{\zeta_j} &\inf L|_{ (\gamma_1 \star \gamma_0) \cup (\gamma_2 \star \gamma_0)}\otimes E_j|_{\gamma_2\pcomp \gamma_0}\mqqqquad  \ar[dr]^{\zeta_j} &
\\
 E_j|_{\gamma_2 \star \gamma_0} \ar[rrr]_{\zeta_j(p \otimes -)} &&& E_j|_{\gamma _1\star \gamma_0} }
\end{equation*}
\caption{}
\label{diag:1}
\end{figure}
The diagram in   \cref{diag:1} is commutative: the subdiagrams on the sides give the construction of  \cref{rem:sectionsandtrivializations}, the subdiagram on the bottom is the relation between $\zeta_j$ and $\lambda$ from the definition of the 1-morphism $\mathcal{E}_j$, the subdiagram in the middle is obviously commutative, and the one on top is \cref{eq:transreg:fusion:1}. We embed the diagram of  \cref{diag:1} twice as subdiagram B into a new diagram shown in \cref{diag:2}.
\begin{figure}[t]
\begin{equation*}
\small
\alxydim{@C=-2em@R=2em}{\inf R'_{ij}|_{\gamma_2}(\mathcal{T}_2) \ar@{=}[drdd] \ar[dddddddd]_{\xi_{ij}|_{\gamma_2}}  &\inf R_{ij}'|_{\gamma_2}(\iota_2^{*}\mathcal{T})\ar@{=}[rr] \ar@{=}[dr] \ar[l]_{r_{\psi_{2}}}&&\inf R_{ij}'|_{\gamma_1}(\iota_1^{*}\mathcal{T}) \ar@{=}[dl]   \ar[r]^{r_{\psi_{1}}}& \inf R_{ij}'|_{\gamma_1}(\mathcal{T}_1) \ar[dddddddd]^{\xi_{ij}|_{\gamma_1}}
 \\ && \mqquad{\mathrm{Hom}(\Des(\mathcal{E}_i|_{x},\mathcal{T}|_0),\Des(\mathcal{E}_j|_{y}),\mathcal{T}|_{\frac{1}{2}})}\mqquad \ar[ddl]|{\tilde\sigma_2}\ar[dr]|{\tilde\sigma_1}&&
 \\ && *+[F-:<10pt>]{A}& \mathrm{Hom}(\Des(\mathcal{E}_i|_{x},\mathcal{T}_1|_0),\Des(\mathcal{E}_j|_{y},\mathcal{T}_1|_1))\ar[dd]^{\labelcref{eq:ident:1}}
 \ar@{=}[uur]
\\ & \mathrm{Hom}(\Des(\mathcal{E}_i|_{x},\mathcal{T}_2|_0),\Des(\mathcal{E}_j|_{y},\mathcal{T}_2|_1)) \mquad\ar[dd]_{\labelcref{eq:ident:1}}
\ar[urr]|{\sigma \circ -} &&& \\ &&*+[F-:<10pt>]{B}& \mathrm{Hom}(E_i|_{\cp_{x} \star \gamma_0} , E_j|_{\gamma _1\star \gamma_0})
\ar[dd]^{\labelcref{eq:ident:2}}\\& \mathrm{Hom}(E_i|_{\cp_{x} \star \gamma_0} , E_j|_{\gamma_2 \star \gamma_0}) \ar[urr]|{\zeta_j(p \otimes -)\circ -} \ar[dd]_{\labelcref{eq:ident:2}} &*+[F-:<10pt>]{C}
\\&&& \mathrm{Hom}(E_i|_{\gamma_0} ,\inf R_{ij}|_{\gamma_1}\otimes _{\inf A_i|_{x}}  E_i|_{\gamma_0}) \ar@{<-}[ddr]_{\labelcref{eq:ident:3}}
\\
& \mathrm{Hom}(E_i|_{\gamma_0} ,\inf R_{ij}|_{\gamma_2}\otimes _{\inf A_i|_{x}}   E_i|_{\gamma_0}) \ar[urr]|{(\phi_{ij}(p\otimes -)\otimes \id)\circ -} \ar@{<-}[dl]^<<<<{\labelcref{eq:ident:3}}
\\ *[r]{\inf R_{ij}|_{\gamma_2}}  \ar[rrrr]_{\phi_{ij}(p\otimes -)} &&&& *[l]{\inf R_{ij}|_{\gamma_1} }}
\end{equation*}
\caption{}
\label{diag:2}
\end{figure}
In that diagram, subdiagram A commutes by construction of $\sigma$, and C commutes by definition of $\epsilon_j$ under regression. All other subdiagrams are obviously commutative. Since the equality $\inf R_{ij}'|_{\gamma_2}(\iota_2^{*}\mathcal{T})=\inf R_{ij}'|_{\gamma_1}(\iota_1^{*}\mathcal{T})$ on top of this diagram is realized by $\phi'_{ij}|_{\gamma_1\cup \gamma_2}([\mathcal{T},1] \otimes -)$, it is now straightforward to conclude the commutativity of \cref{eq:transreg:fusion:2}. \end{proof}

\begin{lemma}
The bundle isomorphism $\xi_{ij}$ respects the lifted path concatenations, i.e., the diagram
\begin{equation*}
\small
\alxydim{@R=3em}{\inf R'_{jk}|_{\gamma_{23}} \otimes \inf R'_{ij}|_{\gamma_{12}} \ar[r]^-{\chi'_{ijk}} \ar[d]_{\xi_{jk} \otimes \xi_{ij}} &  \inf R'_{ik}|_{\gamma_{23} \pcomp \gamma_{12}} \ar[d]^{\xi_{ik}} \\ \inf R_{jk}|_{\gamma_{23}} \otimes \inf R_{ij}|_{\gamma_{12}} \ar[r]_-{\chi_{ijk}} &  \inf R_{ik}|_{\gamma_{23} \pcomp \gamma_{12}}}
\end{equation*} 
is commutative for all composable paths $\gamma_{12}\in P_{ij}$  and $\gamma_{23}\in P_{jk}$. 
\end{lemma}

\begin{proof}
We choose a path $\gamma_0$ connecting $x_0$ with $\gamma_{12}(0)$, induce the  section $s$ of \cref{eq:usualsection} along $\gamma_{23}\pcomp \gamma_{12}$, and obtain a trivialization $\mathcal{T}_{13}$ of $(\gamma_{23} \pcomp \gamma_{12})^{*}\mathcal{G}$. We have its restrictions to trivializations  $\mathcal{T}_{12}$ and $\mathcal{T}_{23}$ over $\gamma_{12}$ and $\gamma_{23}$, respectively; then we can use the definition of $\chi'_{ijk}$ under transgression, see \cref{sec:liftpathcommp}. It is straightforward to see that $\mathcal{T}_{12}$ is induced by the section $s_{12} := s \circ \iota_1$, and $\mathcal{T}_{23}$ is induced by the section $s_{23} := s \circ \iota_2$. 
While $s_{12}$ is again of the form \cref{eq:usualsection}, and thus can be used to define $\xi_{ij}$, the section $s_{23}$ is not of this form. Instead, we have $s_{23}(t) = (\gamma_{23}^{t} \pcomp \gamma_{12})\pcomp\gamma_{12}$.  In order to describe $\xi_{jk}$, we use the trivialization $\mathcal{T}_{23}'$ induced by the section $s'_{23}$ defined by $s'_{23}(t) := \gamma_{23}^{t} \pcomp (\gamma_{12} \pcomp \gamma_0)$. In order to compare $\mathcal{T}_{23}$ and $\mathcal{T}_{23}'$, we consider the section $\tilde s: [0,1] \to \inf L$ along $(s_{23},s_{23}')$, defined by
\begin{equation*}
\tilde s(t) := d_{s_{23}(t) \cup s_{23}(t),s_{23}(t) \cup s_{23}'(t)}(\nu_{s_{23}(t)})\text{.}
\end{equation*} 
Then, by \cref{rem:sectionsandtrivializations}, we obtain a 2-isomorphism $\psi: \mathcal{T}_{23} \Rightarrow \mathcal{T}_{23}'$ and  a commutative diagram
\begin{equation}
\label{eq:diagramtopleft}
\small
\alxydim{@C=1em@R=3em}{\inf R'_{jk}|_{\gamma_{23}}(\mathcal{T}'_{23}) \ar@{=}[r] \ar[d]_-{r_{\psi}} & \mathrm{Hom}(\Des(\mathcal{E}_j|_y,\mathcal{T}_{23}'|_0),\Des(\mathcal{E}_k|_{z},\mathcal{T}_{23}'|_1))  \ar[d]|{\mathrm{Hom}(\psi_0, \psi_1)} \ar[r] & \mathrm{Hom}(E_j|_{\cp_y \pcomp (\gamma_{23} \pcomp \gamma_{12})},E_k|_{\gamma_{23}\pcomp (\gamma_{12} \pcomp \gamma_0)}) \ar[d]|{\mathrm{Hom}(\zeta_j(\tilde s(0) \otimes -),\zeta_k(\tilde s(1) \otimes -))} \\ \inf R'_{jk}|_{\gamma_{23}}(\mathcal{T}_{23}) \ar@{=}[r] & \mathrm{Hom}(\Des(\mathcal{E}_j|_y,\mathcal{T}_{23}|_0),\Des(\mathcal{E}_k|_{z},\mathcal{T}_{23}|_1)) \ar[r] & \mathrm{Hom}(E_j|_{\gamma_{23}\pcomp\gamma_{12}},E_k|_{(\gamma_{23}\pcomp \gamma_{12})  \pcomp \gamma_0})\text{.}}
\end{equation} 
Using a thin homotopy $\Gamma_t$ between $s'_{23}(t) := \gamma_{23}^{t} \pcomp (\gamma_{12} \pcomp \gamma_0)$ and $s_{23}(t)=(\gamma_{23}^{t} \pcomp \gamma_{12}) \pcomp \gamma_0$, the fact that $\zeta_i$ is connection-preserving, and the fact that the canonical elements $\nu_{s_{23}(t)}$ are neutral under under $\zeta_i$, one can show 
that
$\zeta_j(\tilde s(0) \otimes -) = pt_{\Gamma_0}$ and $\zeta_k(\tilde s(1) \otimes -) = pt_{\Gamma_1}$, so that the morphism on the right hand side of the previous diagram is just $\mathrm{Hom}(pt_{\Gamma_0},pt_{\Gamma_1})$.
We now replace $E_i$ by its explicit form as obtained from regression. In the first component,  $pt_{\Gamma_0}$ then becomes
\begin{equation*}
d_{\cp_y \pcomp (\gamma_{12} \pcomp \gamma_0),\gamma_{12} \pcomp \gamma_0} \otimes_{\mathcal{A}_{x_0}} \id_{\inf F_0}: \inf R_{i_0j}|_{\cp_y \pcomp (\gamma_{12} \pcomp \gamma_0)}\otimes_{\mathcal{A}_{x_0}}\inf F_0 \to \inf R_{i_0j}|_{\gamma_{12} \pcomp \gamma_0}\otimes _{\mathcal{A}_{x_0}}\inf F_0\text{,}
\end{equation*}
which can be identified via \cref{eq:lsg:unit} with $\chi_{i_0jj}|_{\gamma_{12} \pcomp\gamma_0,\cp_y}(\epsilon_j(y)\otimes -)^{-1}$. In the second component, $pt_{\Gamma_1}$ becomes
\begin{equation*}
d_{\gamma_{23}\pcomp (\gamma_{12} \pcomp \gamma_0),(\gamma_{23}\pcomp \gamma_{12})  \pcomp \gamma_0} \otimes_{\mathcal{A}_{x_0}} \id_{\inf F_0}: \inf R_{i_0k}|_{\gamma_{23}\pcomp (\gamma_{12} \pcomp \gamma_0)}\otimes _{\mathcal{A}_{x_0}}\inf F_0 \to   \inf R_{i_0k}|_{(\gamma_{23}\pcomp \gamma_{12})  \pcomp \gamma_0}\otimes _{\mathcal{A}_{x_0}}\inf F_0\text{.}
\end{equation*}
In combination with the commutativity of diagram \cref{eq:diagramtopleft}, this shows the commutativity of a diagram, which we find as a subdiagram in the upper left corner of the diagram shown in \cref{diag:3}.
\begin{figure}
\begin{equation*}
\small
\alxydim{@C=-4.9em@R=4em}{*[r]\txt{$\mquad\inf R'_{jk}|_{\gamma_{23}}(\mathcal{T}_{23}') \otimes \inf R'_{ij}|_{\gamma_{12}}(\mathcal{T}_{12})$} \ar[dddd]|{\labelcref{eq:ident:1}} &\hspace{22em} & *[l]\txt{$\inf R'_{jk}|_{\gamma_{23}}(\mathcal{T}_{23}) \otimes \inf R'_{ij}|_{\gamma_{12}}(\mathcal{T}_{12})$}  \ar[ll]_-{r_{\mathcal{T}_{23},\mathcal{T}_{23}' }\otimes \id} \ar[r]^-{\chi_{ijk}'} & \txt{$\inf R'_{ik}|_{\gamma_{23} \pcomp \gamma_{12}}(\mathcal{T})$} \ar@{=}[dl]\ar[dd]^>>>>>>>>>>{\labelcref{eq:ident:1}}
\\ & \mqqqquad\txt{$\mathrm{Hom}(\Des(\mathcal{E}_j|_{y},\mathcal{T}_{23}|_0),\Des(\mathcal{E}_k|_z,\mathcal{T}_{23} |_1)) $\\$ \otimes \mathrm{Hom}(\Des(\mathcal{E}_i|_x,\mathcal{T}_{12}|_0),\Des(\mathcal{E}_j|_y,\mathcal{T}_{12}|_{1}))$} \ar@{=}[ru] \ar[r]_-{\circ } \ar[d]|{\mathrm{Hom}(\varphi_{s_{23}(0)},\varphi_{s_{23}(1)}) \otimes \mathrm{Hom}(\varphi_{s_{12}(0)},\varphi_{s_{12}(1)})}  & \mathrm{Hom}(\Des(\mathcal{E}_i|_{x},\mathcal{T}|_0),\Des(\mathcal{E}_k|_z,\mathcal{T} |_1))\qquad\qquad \ar[dr]_{\mathrm{Hom}(\varphi_{s(0)},\varphi_{s(1)})\quad}  \mqquad  & \\   & \txt{$\mathrm{Hom}(E_j|_{\gamma_{12} \pcomp \gamma_0},E_k|_{(\gamma_{23} \pcomp \gamma_{12}) \pcomp \gamma_0})$\\$ \otimes \mathrm{Hom}(E_i|_{\cp_x \pcomp \gamma_0},E_j|_{\gamma_{12} \pcomp \gamma_0})$}  \ar@{<-}[rd]|{\mathrm{Hom}(\id,\chi_{i_0ik}) \otimes \mathrm{Hom}(\chi_{i_0ii},\id)} \ar[rr]_-{\circ} && *[l]\txt{$\mathrm{Hom}(E_i|_{\cp_x \pcomp \gamma_{0}},E_k|_{(\gamma_{23} \pcomp \gamma_{12})\pcomp \gamma_0})\mquad$} \ar[ddddd]|<<<<<<<<<<<<<<<<<<<<<<<<<<<<<<<{ \labelcref{eq:ident:3}^{-1}\circ \labelcref{eq:ident:2}} \\
 &  & \mqqquad\txt{$\mathrm{Hom}(E_j|_{\gamma_{12} \pcomp \gamma_0},\inf R_{ik}|_{\gamma_{23} \pcomp \gamma_{12}} \otimes_{\inf A_i|_{x}} E_i|_{\gamma_0})$\\$ \otimes \mathrm{Hom}(E_i|_{\gamma_0},E_j|_{\gamma_{12} \pcomp \gamma_0}$)}   & 
\\ *[r]\txt{$\mquad\mathrm{Hom}(E_j|_{\cp_y \pcomp (\gamma_{12} \pcomp \gamma_0)},E_k|_{\gamma_{23} \pcomp (\gamma_{12} \pcomp \gamma_0)})$\\$ \otimes \mathrm{Hom}(E_i|_{\cp_x \pcomp \gamma_0},E_j|_{\gamma_{12} \pcomp \gamma_0})$} \ar[ddd]|{\labelcref{eq:ident:3}^{-1}\circ \labelcref{eq:ident:2}} \ar@{<-}[rdd]|{\quad\mathrm{Hom}(\chi_{i_0jj},\chi_{i_0jk}) \otimes \mathrm{Hom}(\chi_{i_0ii},\chi_{i_0ij})} \ar[ruu]|-{\mathrm{Hom}(\chi_{i_0jj}^{-1},d) \otimes \id} &&& 
\\
&& \mqqquad\txt{$\mathrm{Hom}(E_j|_{\gamma_{12} \pcomp \gamma_0},\inf R_{jk}|_{\gamma_{23}} \otimes_{\inf A_{j}|_y} \inf R_{ij}|_{\gamma_{12} }\otimes_{\inf A_i|_{x}} E_i|_{\gamma_0})$\\$ \otimes \mathrm{Hom}(E_i|_{\gamma_0},\inf R_{ij}|_{\gamma_{12}} \otimes_{\inf A_i|_{x}}E_i|_{\gamma_0})$}  \ar[uu]|{\mathrm{Hom}(\id,\chi_{ijk} \otimes \id) \otimes \mathrm{Hom}(\id,\chi_{i_0ij})}
\\
&\txt{$\mathrm{Hom}(E_j|_{\gamma_{12} \pcomp \gamma_0},\inf R_{jk}|_{\gamma_{23}} \otimes_{\inf A_{j}|_y} E_j|_{\gamma_{12} \pcomp \gamma_0})$\\$ \otimes \mathrm{Hom}(E_i|_{\gamma_0},\inf R_{ij}|_{\gamma_{12}} \otimes_{\inf A_i|_{x}} E_i|_{\gamma_0})$}  \ar@{<-}[dl]_<<<<<<<<<<<<<<<<{\labelcref{eq:ident:3}} \ar@{<-}[ru]|-{\mathrm{Hom}(\id,\id \otimes \chi_{i_0ij}) \otimes \id}&&  && 
\\
 *[r]\txt{\mquad$\inf R_{jk}|_{\gamma_{23}} \otimes \inf R_{ij}|_{\gamma_{12}} $} 
\ar[rrr]_{\chi_{ijk}} &&& \txt{$\inf R_{ik}|_{\gamma_{23} \pcomp \gamma_{12}} $}}
\end{equation*}
\caption{}
\label{diag:3}
\end{figure}
The commutativity of this diagram is what we want to show. The triangular subdiagrams commute by definition of one of their arrows. The four-sided diagram at the top is the definition of $\chi'_{ijk}$  under regression. The four-sided diagram below is just linear algebra, given that $s(0)=s_{12}(0)$, $s_{12}(1)=s_{23}(0)$, and $s_{23}(1)=s(1)$. The Pentagon diagram in the middle has four tensor factors, split as $\mathrm{Hom}(-,-) \otimes \mathrm{Hom}(-,-)$, which commute separately. Indeed, commutativity in the first, third, and fourth factor is obvious, and in the second factor it is precisely the Pentagon diagram of \cref{eq:lsg:pentagon}. Finally, there is a strangely shaped diagram at the lower right corner; this diagram commutes again by pure linear algebra.
\end{proof}

\FloatBarrier

So far we have completed the definition of a LBG morphism  $(\varphi,\xi)$ between the transgression of the regression  of a LBG object and this LBG object. 
It remains to prove that this construction depends naturally on the given LBG object; this is the content of the following lemma.

\begin{lemma}
Let $(\varphi,\xi):(\inf L_1,\lambda_1, \inf R_1,\phi_1,\chi_1,\epsilon_1,\alpha_1) \to (\inf L_2,\lambda_2, \inf R_2,\phi_2,\chi_2,\epsilon_2,\alpha_2)$ be a LBG morphism. We label the transgression of regressions of LBG objects and morphisms with primes, and denote by
\begin{align*}
(\varphi_1,\xi_1):(\inf L_1',\lambda_1', \inf R_1',\phi_1',\chi_1',\epsilon_1',\alpha_1') \to (\inf L_1,\lambda_1, \inf R_1,\phi_1,\chi_1,\epsilon_1,\alpha_1)
\\
(\varphi_2,\xi_2):(\inf L_2',\lambda_2', \inf R_2',\phi_2',\chi_2',\epsilon_2',\alpha_2') \to (\inf L_2,\lambda_2, \inf R_2,\phi_2,\chi_2,\epsilon_2,\alpha_2)
\end{align*}
the LBG morphisms associated to the two LBG objects. Then, we have
\begin{equation*}
(\varphi_2,\xi_2) \circ (\varphi',\xi') =(\varphi,\xi) \circ (\varphi_1,\xi_1)\text{.} \end{equation*}
\end{lemma}

\begin{proof}
Commutativity of the line bundle isomorphisms, $\varphi_2 \circ \varphi'=\varphi \circ \varphi_1$, has been shown in \cite[Section 6.2]{waldorf10}. For the vector bundle isomorphisms, the relevant statement is the commutativity of the  diagram
\begin{equation}
\label{eq:diagramxixiprime}
\small
\alxydim{@R=3em}{\inf R'_{1,ij}|_{\gamma} \ar[r]^{\xi'_{ij}}  \ar[d]_{\xi_{1,ij}} & \inf R'_{2,ij}|_{\gamma} \ar[d]^{\xi_{2,ij}}  \\ \inf R_{1,ij}|_{\gamma} \ar[r]_{\xi_{ij}} & \inf R_{2,ij}|_{\gamma} \text{,}}
\end{equation}
for all paths $\gamma\in P_{ij}$. We pick an arbitrary $\gamma\in P_{ij}$ and set $x:=\gamma(0)$ and $y:=\gamma(1)$. For the construction of $\xi_{1,ij}$ and $\xi_{2,ij}$  we need to choose a path $\gamma_0\in P_{i_0i}$ with $\gamma_0(1)=x$, inducing sections $s_1$ and $s_2$ into the surjective submersions of $\gamma^{*}\mathcal{G}_1$ and $\gamma^{*}\mathcal{G}_2$, respectively, where $\mathcal{G}_1,\mathcal{G}_2$ are the regressed bundle gerbes, see \cref{eq:usualsection}. These sections, in turn, induce trivializations $\mathcal{T}_1$ and $\mathcal{T}_2$  of $\gamma^{*}\mathcal{G}_1$ and $\gamma^{*}\mathcal{G}_2$, respectively.

Let $(\mathcal{A},\psi)$ be the regression of $(\varphi,\xi)$, i.e. $\mathcal{A}:\mathcal{G}_1 \to \mathcal{G}_2$ is a 1-isomorphism induced from  the line bundle morphism $\cup^{*}\varphi$, and $\psi=\{\psi_i\}_{i\in I}$ consists of 2-isomorphisms $\psi_{i}: \mathcal{E}_{1,i} \Rightarrow \mathcal{E}_{2,i} \circ \mathcal{A}$  induced from the vector bundle homomorphism $\tilde\psi_{i} := \xi_{i_0i} \otimes \ev_0^{*}f_0:E_{1,i} \to E_{2,i}$, see \cref{sec:regression}. The transgression of $(\mathcal{A},\varphi)$ is defined by
\begin{equation*}
\xi_{ij}': \inf R_{1,ij}|_{\gamma}(\mathcal{T}_1) \to \inf R_{2,ij}|_{\gamma}(\mathcal{T}_2): \varphi \mapsto \psi_{i,1} \circ \varphi \circ \psi_{i,0}^{-1}\text{,}
\end{equation*}
see \cref{sec:nattrans}.
Evaluating the construction of the isomorphisms $\psi_{i,1}$ and $\psi_{i,0}$ in the present situation, using that the trivializations $\mathcal{T}_1$ and $\mathcal{T}_2$ are induced from sections, we obtain commutative diagrams
\begin{equation*}
\small
\alxydim{@C=1.3cm@R=3em}{\Des(\mathcal{E}_{1,i}|_x,\mathcal{T}_1|_0) \ar[d]_{\varphi_s} \ar[r]^-{\psi_{i,0}}  & \Des(\mathcal{E}_{2,i}|_x ,\mathcal{T}_2|_0)\ar[d]^{\varphi_s}  \\ E_{1,i}|_{\cp_x \pcomp \gamma_0} \ar[r]_{\tilde\psi_i} & E_{2,i}|_{\cp_x \pcomp \gamma_0}}
\quand
\alxydim{@C=1.3cm@R=3em}{\Des(\mathcal{E}_{1,j}|_y,\mathcal{T}_1|_1) \ar[d]_{\varphi_s} \ar[r]^{\psi_{i,1}} & \Des(\mathcal{E}_{2,j}|_y,\mathcal{T}_2|_1 )\ar[d]^{\varphi_s}  \\ E_{1,j}|_{\gamma \pcomp \gamma_0} \ar[r]_{\tilde\psi_j} & E_{2,j}|_{\gamma \pcomp \gamma_0}}
\end{equation*}
where vertical arrows are the bundle morphisms of \cref{rem:sectionsandtrivializations}. This shows the commutativity of the first subdiagram of the diagram shown in \cref{diag:4}.
\begin{figure}[t]
\begin{equation*}
\small
\alxydim{@C=12em@R=3em}{\inf R'_{1,ij}|_{\gamma}(\mathcal{T}_1) \ar[d]_{\labelcref{eq:ident:1}} \ar[r]^{\xi'_{ij}} & \inf R'_{2,ij}|_{\gamma}(\mathcal{T}_2) \ar[d]^{\labelcref{eq:ident:1}}
\\
\mathrm{Hom}(E_{1,i}|_{\cp_x \pcomp \gamma_0},E_{1,j}|_{\gamma \pcomp \gamma_0}) \ar[d]_{\labelcref{eq:ident:2}} \ar[r]^{\mathrm{Hom}(\tilde\psi_i,\tilde\psi_j)} & \mathrm{Hom}(E_{1,i}|_{\cp_x \pcomp \gamma_0},E_{1,j}|_{\gamma \pcomp \gamma_0}) \ar[d]^{\labelcref{eq:ident:2}}
\\
\mathrm{Hom}(E_{1,i}|_{\gamma_0},\inf R_{1,ij}|_{\gamma}\otimes _{\inf A_{1,i}|_{x}}  E_{1,i}|_{\gamma_0}) \ar[d]_{\labelcref{eq:ident:3}} \ar[r]^{\mathrm{Hom}(\xi_{i_0i} \otimes \ev_0^{*}f_0,\xi_{ij} \otimes \xi_{i_0i} \otimes \ev_0^{*}f_0)} & \mathrm{Hom}(E_{2,i}|_{\gamma_0},\inf R_{2,ij}|_{\gamma}\otimes _{\inf A_{2,i}|_{x}}  E_{2,i}|_{\gamma_0})\ar[d]^{\labelcref{eq:ident:3}}
\\
\inf R_{1,ij}|_{\gamma} \ar[r]_{\xi_{ij}} & \inf R_{2,ij}|_{\gamma}\text{.}}
\end{equation*} 
\caption{}
\label{diag:4}
\end{figure}
The diagram in the middle is commutative because $\xi_{ij}$ is compatible with the lifted path concatenation, and the diagram at the bottom is obviously commutative. The outer shape of the diagram of \cref{diag:4} is \cref{eq:diagramxixiprime}. 
\end{proof}

\begin{appendix}

\setsecnumdepth{2}

\section{Vector bundles and algebra bundles over diffeological spaces}

\subsection{Diffeological vector bundles}

\label{sec:diffvectorbundle}

In this section we provide the basics about vector bundles over diffeological spaces. For the definition we follow  \cite[Art. 8.9]{iglesias1}.
Let $X$ be a diffeological space.

\begin{definition}
\label{def:vb}
A \emph{complex vector bundle} of rank $k$ over $X$ consists of the following structure:
\begin{enumerate}[(a)]

\item 
a diffeological space $E$, the total space,

\item
a smooth map $\pi:E \to X$, the projection,

\item
a complex vector space structure on each fiber $E|_x := \pi^{-1}(\{x\})$, for all $x\in X$. 

\end{enumerate}
The following condition has to be satisfied: for each plot $c:U\to X$ and each point $u\in U$ there exists an open neighborhood $u\in W \subset U$ and a diffeomorphism
\begin{equation*}
\phi: W \times \C^{k} \to W \times_X E
\end{equation*}
that covers the identity map on $W$, and restricts to a linear map $\phi|_w: \C^{k} \to E|_{c(w)}$   over each point $w\in W$. 

\end{definition}

A  \emph{morphism} between vector bundles $E$ and $E'$ over $X$ is a smooth map $\varphi:E\to E'$  that commutes with the projections and restricts to a linear map  $\varphi|_x: E|_x \to E|_x'$ over each fiber. 
Hermitian vector bundles and unitary bundle morphisms are defined in the obvious analogous way.

As usual, vector bundles can be associated to principal bundles via representations. For principal bundles over diffeological spaces we use the definition of \cite{waldorf9}.

\begin{lemma}
\label{lem:assbun}
Suppose $G$ is a Lie group and $\rho:G \to \mathrm{GL}(\C^{k})$ is a Lie group homomorphism. Suppose further that $P$ is a principal $G$-bundle over $X$. Let $E := P \times_G \C^{k}$ be equipped with the quotient diffeology, the projection induced from $P$, and the fibrewise vector space structure of $\C^{k}$. Then, $E$ is a vector bundle over $X$. 
\end{lemma}

\begin{proof}
Since the projection $\pi:P\to X$ of a principal $G$-bundle is a subduction, every plot $c:U \to X$ and every point $u\in U$ admit an open neighborhood $u\in W \subset U$ with a lift:  a plot $\tilde c: W \to P$ such that $\pi\circ \tilde c=c$. Then we define a local trivialization of $E$ by
\begin{equation*}
\phi(w,v) := (w,[\tilde c(w),v])\text{.}
\end{equation*} 
This is obviously smooth and fiber-wise linear. An inverse is defined in the following way. Suppose $(w,[p,v])\in W \times_X E$. Since $P$ is a principal $G$-bundle, there exists a unique $g_{p,w}\in G$ such that $p=\tilde c(w)g_{p,w}$. We set $\phi^{-1}(w,[p,v]) := (w,\rho(g_{p,w})(v))$. It is easy to check that this is inverse to $\phi$.
In order to check the smoothness of $\phi^{-1}$, let $d:W' \to W \times_X E$ be a plot, i.e. $d=(d_1,d_2)$ where $d_1: W' \to W$ is a smooth map and $d_2: W' \to E$ is a plot of $E$, such that $\pi\circ d_2=c \circ d_1$. We have to show that $\phi^{-1} \circ d:W' \to\ W \times \C^{k}$ is smooth, which can be done  locally. 
 By definition of the quotient diffeology of $E$, $W'$  can be covered by smaller open sets $W''$ such that $d_2|_{W''} = [p,v]$, where $p: W'' \to P$ is a plot of $P$ and $v: W'' \to \C^{k}$ is smooth. Now we have to check that 
\begin{equation*}
\phi^{-1}\circ d|_{W''}: W'' \to W\times \C^{k}:w \mapsto (d_1(w),\rho(g_{p(w),d_1(w)})(v(w)))
\end{equation*}
is smooth. This follows from the definition of principal $G$-bundles, according to which the map
$\delta: P \times_X P \to G$ that induces $g_{p,v} := \delta(p,v)$ is smooth. 
\end{proof}

\begin{remark}
\begin{enumerate}[(a)]

\item 
If the image of $\rho$ in \cref{lem:assbun} is contained in $U(\C^{k})$, then $E$ is a hermitian vector bundle.

\item
It is easy to check that all familiar operations with vector bundles can be performed: pullback, tensor product, dual bundles, Hom-bundles etc. 

\item
If $X$ is a smooth manifold, considered as a diffeological space, and $E$ is a vector bundle over $X$ in the sense of \cref{def:vb}, then there exists a unique smooth manifold structure on $E$ that induces the given diffeology and gives $E$ the structure of a smooth vector bundle over $X$. Indeed, this smooth manifold structure is defined  via local trivializations of $E$, whose transition functions are smooth. 

\end{enumerate}
\end{remark}

We continue with connections on vector bundles over diffeological spaces, for which no established definition exists. Since tangent vectors in diffeological spaces are notoriously difficult to handle, we define connections  via their parallel transport. 
First  we recall the following prerequisite. By a \emph{path} in $X$ we understand a smooth map $\gamma:[0,1] \to X$ with sitting instants, and we denote by $PX \subset C^{\infty}([0,1],X)$ the space of paths, equipped with its natural diffeology. A smooth map $f$ between smooth manifolds is said to have \emph{rank $k$} if $\mathrm{rk}(\mathrm{d}f_x)\leq k$ for all points $x$ in its domain.
 A smooth map has rank $k$ if and only if the pullback of every $(k+1)$-forms vanishes \cite[Lemma 4.2]{schreiber5}. That condition makes sense for smooth maps between diffeological spaces, and we define the rank of maps between diffeological spaces in this way. 

\begin{definition}
\label{def:thinhomotopic}
Two paths $\gamma_1,\gamma_2\in PX$ in a diffeological space $X$ are called \emph{thin homotopic}, if there exists a path $h\in PPX$ such that
\begin{enumerate}[(a)]

\item 
it is a homotopy:
$h(0)=\gamma_1$ and $h(1)=\gamma_2$

\item
it fixes  end-points: $h(s)(0)=\gamma_1(0)=\gamma_2(0)$ and $h(s)(1)=\gamma_1(1)=\gamma_2(1)$ for all $s\in [0,1]$. 

\item
it is thin: the map $h^{\vee}:[0,1]^2\to X:(s,t) \mapsto h(s)(t)$ has rank one.  

\end{enumerate}
\end{definition}

Now we are in position to define a connection on a vector bundle $E$ over a diffeological space $X$.

\begin{definition}
\label{def:connection}
A \emph{connection} on $E$ is a family of  linear maps $pt_{\gamma}:E|_{\gamma(0)} \to E|_{\gamma(1)}$, for each path $\gamma$ in $X$, such that the following conditions are satisfied:
\begin{enumerate}[(a)]

\item 
\label{def:connection:a}
$pt_{\gamma}$ depends only on the thin homotopy class of the path $\gamma$.

\item
\label{def:connection:b}
$pt_{\gamma_2\pcomp \gamma_1}=pt_{\gamma_2} \circ pt_{\gamma_1}$ for all composable paths in $X$.

\item
\label{def:connection:c}
for each local trivialization $\phi: W \times \C^{k} \to W \times_X E$ of $E$ there exists a 1-form $\omega_{\phi}\in\Omega^1(W,\mathfrak{gl}(\C^{k}))$ such that for every  $\gamma \in PW$ the diagram
\begin{equation*}
\small
\alxydim{@C=2cm@R=3em}{\C^{k} \ar[r]^{\exp(\omega_{\phi})(\gamma)}  \ar[d]_{\phi|_{\gamma(0)}} & \C^{k} \ar[d]^{\phi|_{\gamma(1)}} \\E_{c(\gamma(0))} \ar[r]_{pt_{c\circ\gamma}} & E_{c(\gamma(1))}}
\end{equation*}
is commutative, where $\exp(\omega_{\phi})(\gamma)\in \mathrm{GL}(\C^{k})$ is the path-ordered exponential of $\omega_{\phi}$ along $\gamma$.

\end{enumerate} 
\end{definition}

\begin{remark}
Under the assumption that conditions \cref{def:connection:a*} and \cref{def:connection:b*} hold,  \cref{def:connection:c*} is equivalent to  the following  condition. 
\begin{enumerate}[(c')]

\item 
\label{def:connection:cprime}
for each local trivialization $\phi: W \times \C^{k} \to W \times_X E$ of $E$ the map $p_{\phi}:PW \to \mathrm{GL}(\C^{k})$, where $p_{\phi}(\gamma)\in \mathrm{GL}(\C^{k})$ is the linear isomorphism
\begin{equation*}
\alxydim{}{\C^{k} \ar[r]^-{\phi|_{\gamma(0)}} &  E_{c(\gamma(0))} \ar[r]^{pt_{c\circ\gamma}} & E_{c(\gamma(1))} \ar[r]^-{\phi|_{\gamma(1)}^{-1}} & \C^{k}\text{,}}
\end{equation*}
is smooth. 
\end{enumerate}
The equivalence uses the theory of smooth functors \cite{schreiber3}. Indeed, if $\omega_{\phi}$ exists, then the map $p_{\phi}$ is $\exp(\omega_{\phi}):PW \to \mathrm{GL}(\C^{k})$, and hence smooth. Conversely, if $p_{\phi}$ is smooth, then it follows from \cref{def:connection:a*,def:connection:b*} that it defines a smooth functor $p_{\phi}:\mathcal{P}_1(W) \to B\mathrm{GL}(\C^{k})$, corresponding to a 1-form $\omega_{\phi}$ such that $p_{\phi}=\exp(\omega_{\phi})$. 

\end{remark}

If $E$ is hermitian, then a connection is called \emph{unitary} if $pt_{\gamma}$ is unitary and the 1-forms $\omega_{\phi}$ of all local trivializations $\phi$ take values in $\mathfrak{u}(n)$.  A connection is called \emph{flat} if $pt_{\gamma}$ depends only on the homotopy class of $\gamma$.

Let $E$ and $E'$ be vector bundles over $X$ equipped with connections $pt$ and $pt'$, respectively. 
A bundle morphism $\varphi:E \to E'$ is called \emph{connection-preserving} if it commutes with the parallel transport, i.e. the diagram
\begin{equation*}
\small
\alxydim{@R=3em}{E|_x \ar[r]^{pt_{\gamma}} \ar[d]_{\varphi} & E|_y \ar[d]^{\varphi} \\ E'|_{x} \ar[r]_{pt'_{\gamma}} & E'|_y}
\end{equation*}
is commutative for all paths $\gamma\in PX$, with $x:=\gamma(0)$ and $y := \gamma(1)$. 

We shall verify that our notion of a connection is compatible with the existing notion of a connection on a principal bundle (defined in \cite{waldorf9} as a Lie algebra-valued 1-form) under the associated bundle construction of \cref{lem:assbun}.

\begin{lemma}
\label{lem:assbuncon}
Suppose $P$ is a principal $G$-bundle over $X$ and $\omega\in \Omega^1(P,\mathfrak{g})$ is a connection on $P$. We denote by $\tau^{\omega}_{\gamma}: P_{\gamma(0)} \to P_{\gamma(1)}$ the parallel transport along a path $\gamma \in PX$. Let $\rho: G \to \mathrm{GL}(\C^{k})$ be a Lie group homomorphism. Then, the formula
\begin{equation*}
pt_{\gamma}([p,v])  := [\tau_{\gamma}^{\omega}(p),v]
\end{equation*}
defines a connection on the associated vector bundle $P \times_{G} \C^{k}$. \end{lemma}

\begin{proof}
We use  \cite[Prop. 3.2.10]{waldorf9} for the properties of the parallel transport $\tau^{\omega}$. It is $G$-equivariant by item (b); hence $pt_{\gamma}$ is well-defined. Further, $pt_{\gamma}$ is  linear by construction. It is compatible with path composition due to item (a), and it only depends on the thin homotopy class due to item (b). It remains to check the compatibility with a local trivialization $\phi$, obtained as described in \cref{lem:assbun} as $\phi(w,v) := (w,[\tilde c(w),v])$. We  set $\omega_{\phi} := \rho_{*}(\tilde c^{*}\omega)$. Suppose $\gamma\in PW$. Set $\tilde\gamma := \tilde c\circ \gamma$; this is a lift of $c \circ \gamma \in PX$. By \cite[Def. 3.2.9]{waldorf9} we have $\tau_{c \circ \gamma}^{\omega}(\tilde\gamma(0)) = \tilde\gamma(1)\cdot \exp(\omega)(\tilde\gamma)$. The calculus for path ordered exponentials implies that 
\begin{equation}
\label{eq:calcpoe}
\rho(\exp(\omega)(\tilde\gamma))=\exp(\rho_{*}(\tilde c^{*}\omega))(\gamma)=\exp(\omega_{\phi})(\gamma)\text{.}
\end{equation}  
Now, the commutativity of the diagram in \cref{def:connection} is straightforward to check.
\end{proof}

\begin{remark}
If $X$ is a smooth manifold, then a connection in the sense of \cref{def:connection} furnishes a transport functor \cite{schreiber3}. These are equivalent to ordinary connections on ordinary vector bundles. Thus, our approach to vector bundles and connections over diffeological spaces reduces consistently to the classical theory over smooth manifolds. 
\end{remark}

\begin{remark}
\label{re:curvature}
The treatment of curvature in terms of parallel transport involves transport 2-functors, and drops a bit out of the context of this article, see \cite[Sec. 7.2]{schreiber3} and \cite[Sec. 3.4]{schreiber2}. The content of \cite[Lemma 3.4.3]{schreiber2} is that  the curvature of a connection on a vector bundle $E$ is given locally by the endomorphism valued 2-form 
\begin{equation*}
\Omega_{\phi} := \mathrm{d}\omega_{\phi} + \frac{1}{2}[\omega_{\phi} \wedge \omega_{\phi}] \in \Omega^2(W,\mathfrak{gl}(\C^{k}))\text{,}
\end{equation*}
where $\omega_{\phi}\in\Omega^1(W,\mathfrak{gl}(\C^{k}))$ is the 1-form of a local trivialization. Globally, one can consider its trace, which is a globally-defined 2-form $\mathrm{tr}(\Omega_{\phi})= \mathrm{tr}(\mathrm{d}\omega_{\phi})$.
\end{remark}

\subsection{Superficial connections on path spaces}

\label{sec:superficial}

Let $M$ be a smooth manifold, and let $PM$ denote the diffeological space of paths in $M$ with sitting instants. A path $\Gamma \in PPM$ is called \emph{thin} if $\Gamma^{\vee}:[0,1]^2 \to M$ has rank one, and it is called  \emph{fixed-ends}, if the end-paths $s\mapsto \Gamma(s)(0)$ and $s\mapsto \Gamma(s)(1)$ are constant. A fixed-ends thin path $\Gamma$ makes the paths $\Gamma(0),\Gamma(1) \in PM$ thin homotopic in the sense of \cref{def:thinhomotopic}.

\begin{definition}
\label{def:ranktwohomotopic}
Two paths $\Gamma_1,\Gamma_2\in PPM$ are called \emph{rank-two-homotopic}, if  there exists $h\in PPPM$ satisfying the following conditions:
\begin{enumerate}[(a)]

\item 
\label{def:ranktwohomotopic:a}
It is a homotopy, i.e. $h(0)=\Gamma_1$ and $h(1)=\Gamma_2$\text{.}

\item
\label{def:ranktwohomotopic:b}
It fixes the paths of end-points: for all $r,s\in [0,1]$ we have
\begin{equation*}
h(r)(s)(0)=\Gamma_1(s)(0)=\Gamma_2(s)(0)
\quand
h(r)(s)(1)=\Gamma_1(s)(1)=\Gamma_2(s)(1)\text{.}
\end{equation*}

\item
\label{def:ranktwohomotopic:c}
$h^{\vee}:[0,1]^3 \to M$ has rank two. 

\end{enumerate}
\end{definition}

For such a homotopy we write $h_0 \in PPM$ for the path $h_0(r) := h(r)(0)$ connecting $\Gamma_1(0)$ with $\Gamma_2(0)$, and $h_1\in PPM$ for the path $h_1(r) := h(r)(1)$ connecting $\Gamma_1(1)$ with $\Gamma_2(1)$. Note that $h_0$ and $h_1$ are fixed-ends paths by \cref{def:ranktwohomotopic:b*}. 
\begin{definition}
\label{def:superficial}
Let $E$ be a vector bundle over $PM$. A connection $pt$ on $E$ is called \emph{superficial}, if the following two conditions are satisfied:
\begin{enumerate}[(i)]

\item 
\label{def:superficial:i}
Parallel transport along a fixed-ends thin path is independent of that path. More precisely,  if  $\Gamma_1,\Gamma_2\in PPM$ are fixed-ends thin paths with $\Gamma_1(0)=\Gamma_2(0)$ and $\Gamma_1(1)=\Gamma_2(1)$, then 
$pt_{\Gamma_1}=pt_{\Gamma_2}$.

\item
\label{def:superficial:ii}
If $\Gamma_1,\Gamma_2\in PPM$ are rank-two-homotopic via $h\in PPPM$, then the following diagram is commutative: 
\begin{equation*}
\small
\alxydim{@R=3em}{E|_{\Gamma_1(0)} \ar[r]^{pt_{\Gamma_1}} \ar[d]_{pt_{h_0}} & E|_{\Gamma_1(1)}  \ar[d]^{pt_{h_1}} \\ E|_{\Gamma_2(0)} \ar[r]_{pt_{\Gamma_2}}  & E|_{\Gamma_2(1)} }
\end{equation*}

\end{enumerate}
\end{definition}

Via \cref{def:superficial:i}, a superficial connection on $E$ determines for each pair $(\gamma,\gamma')$ of thin homotopic paths a canonical map $d_{\gamma,\gamma'}:E|_{\gamma} \to E|_{\gamma'}$ by putting $d_{\gamma,\gamma'} := pt_{\Gamma}$ for some fixed-ends thin path $\Gamma$ connecting $\gamma$ with $\gamma'$. We also note that \cref{def:superficial:ii} implies that $pt_{\Gamma_1}=pt_{\Gamma_2}$ if the rank two homotopy $h$ fixes the end-points.

\begin{remark}
Two paths $\Gamma_1,\Gamma_2\in PPM$ can be rank-two-homotopic in the sense of \cref{def:ranktwohomotopic} or thin homotopic as paths in $X=PM$ in the sense of \cref{def:thinhomotopic}. In general, both conditions are different and none implies the other.  
\end{remark}

\subsection{Bundles of algebras and bimodules}

\label{sec:algebrabundles}

By an algebra we will always mean a unital, associative algebra over $\C$, and all algebra homomorphisms and representations will be unital. We first fix some terminology. Let $X$ be a diffeological space.\begin{enumerate}[(a)]

\item
An \emph{algebra structure} on a vector bundle $E$ over $X$ is a bundle morphism $\mu:E \otimes E \to E$ over $X$ such that over each point $x\in X$  the  map $\mu|_x: E|_x \otimes E|_x \to E|_x$ equips  $E|_x$ with the structure of an algebra, and the section $x \mapsto 1_x$ of unit elements is smooth.

\item 
An algebra structure on a vector bundle $E$ is called \emph{local}, if  for each plot $c: U \to X$ and each point $x\in U$ there exist an algebra $A_{c,x}$, an open neighborhood $x\in V \subset U$ and a diffeomorphism 
$\phi: V \times A_{c,x} \to V \times_X E$
that induces the identity on $V$ and its restriction $\phi|_v: A_{c,x} \to E|_{c(v)}$   to the fiber over each $v\in V$ is an algebra isomorphism. 

\item 
A vector bundle $E$ with local algebra structure is called \emph{algebra bundle} or \emph{bundle of algebras}, if the algebras $A_{c,x}$ can be chosen independently of  the plot $c$   and the point $x$. 
\end{enumerate}
Analogous terminology will be used for various types of algebras, for instance, involutive algebras and Frobenius algebras.
We remark that the necessity of carefully distinguishing between these types of  bundles is \emph{not} caused by the fact that we work over diffeological spaces; the same types exist over smooth manifolds and have to be distinguished.

The following results explain in a nice way the role of connections in relation to  algebra bundles.

\begin{lemma}
\label{lem:algloctriv}
Suppose a vector bundle $E$ with algebra structure $\mu$ admits a connection $pt$ for which  $\mu$ is connection-preserving. Then, $\mu$ is local. Moreover, the restriction of  $E$ to each path-connected component of $X$ is an algebra bundle. 
\end{lemma}

\begin{proof}
That $\mu$ is connection-preserving means that the isomorphisms $pt_{\gamma}:E|_x \to E|_y$ are algebra isomorphisms. 
Consider a local trivialization
$\phi: V \times \C^{k} \to V \times_X E$
of the vector bundle $E$, for a plot $c:U \to X$ and $V \subset U$ a contractible open set. For a fixed smooth contraction of $V$ to a point $x_0\in V$, we obtain a smooth map $\gamma: V \to PV$ assigning to a point $x\in V$ a path $\gamma_x$ from $x_0$ to $x$. We induce an algebra structure on $\C^{k}$ such that $\phi|_{x_0}:\C^{k} \to E|_{c(x_0)}$ is an algebra isomorphism, and denote that algebra  by $A_{c,x_0}$. In general, $\phi|_x: A_{c,x_0} \to E|_{c(x)}$ is not an algebra homomorphism for  $x\neq x_0$. However, consider the new trivialization $\phi' :V \times A_{c,x_0} \to V \times_X E$ defined by $\phi'(x,v):=\phi(x,\exp(\omega_{\phi})(\gamma_x)v)$, where $\omega_{\phi}$ is a local connection 1-form for $\phi$. We claim that $\phi'|_x: A_{c,x_0} \to E|_{c(x)}$ is an algebra homomorphism for all $x\in V$. Indeed, we get from \cref{def:connection:c} 
\begin{equation*}
\phi'|_x = \phi|_x \circ \exp(\omega_{\phi})(\gamma_x) =pt_{c \circ \gamma_x} \circ \phi|_{x_0}\text{,}
\end{equation*}  
and this is a composition of algebra homomorphisms. This shows that $\mu$ is local. 

Now, fix an arbitrary plot $c_0:U_0 \to X$, a point $x_0\in U_0$, an algebra $A_{c_0,x_0}$ with a local algebra trivialization $\phi_0$ around $x_0$. For any other plot $c:U \to X$, $x\in U$ and local algebra trivialization $\phi: V \times A_{c,x} \to V \times_X E$ defined in $x\in V \subset U$, choose a path $\gamma\in PX$ connecting $x_0$ with $x$. Parallel transport and the algebra isomorphisms $\phi_0|_{x_0}$ and $\phi|_x$ determine an algebra isomorphism $A_{c_0,x_0} \cong A_{c,x}$. Pre-composing with $\phi$ produces a new local trivialization $\phi'$ defined over $V$ with typical fiber $A_{c_0,x_0}$.
Thus, $E$ is an algebra bundle over the path-connected component of $x_0$.
\end{proof}

If $X$ is not path-connected, then local algebra structures have  non-isomorphic typical fibers over the different connected components, in general. However, since the underlying vector bundle has the same rank everywhere, all these algebras have the same dimension. Since simple algebras of the same dimension are necessarily isomorphic, we obtain the following.

\begin{lemma}
\label{re:loctriv:c}
If $E$ carries a local algebra structure such that for all $x\in X$ the algebra $E|_x$ is simple, then $E$  is an algebra bundle. \qed
\end{lemma}

Next, we introduce terminology for bundles of bimodules. Let $A$ and $B$ be vector bundles over $X$ with  algebra structures.  
\begin{enumerate}[(a)]

\item 
An \emph{$A$-$B$-bimodule structure} on a vector bundle $M$ over $X$ is a pair $(\lambda,\rho)$ of  vector bundle morphisms $\lambda: A \otimes M \to M$ and $\rho: M \otimes B \to M$ such that over each point $x\in X$ the linear maps $\lambda|_x: A|_x \otimes M|_x \to\ M|_x$ and $\rho|_x : M|_x \otimes B|_x \to M|_x$ define commuting  left and right algebra actions. 

\item
An $A$-$B$-bimodule structure is called \emph{local}, if for each plot $c:U \to X$ and each $x\in U$ there exist an open neighborhood $x\in V \subset U$,  algebras $A_0$ and $B_0$, an $A_0$-$B_0$-bimodule $M_0$ 
and diffeomorphisms
\begin{align*}
\phi_A&: V \times A_0 \to V \times_X A
\;\;\text{, }\;\;
\phi_B&: V \times B_0 \to V \times_X B
\;\;\text{, }\;\;
\phi_M&: V \times M_0 \to V \times_X M
\end{align*}
covering the identity on $V$, such that for each $v\in V$ the restrictions $\phi_A|_v: A_0 \to A|_{c(v)}$ and $\phi_B|_v: B_0 \to B|_{c(v)}$  are  algebra isomorphisms, and the restriction $\phi_M|_v: M_0 \to M|_{c(v)}$ is a bimodule intertwiner (along $\phi_A|_v$ and $\phi_B|_v$).

\item
We say that $M$ is a \emph{bundle of $A$-$B$-bimodules}, or \emph{$A$-$B$-bimodule bundle}, if  $A$  and $B$ are  bundles of  algebras,  with typical fibers $A_0$ and $B_0$, respectively, and $M_0$  can be chosen independently of $c$ and $x$ as an $A_0$-$B_0$-bimodule. 
 
\end{enumerate}

 \begin{lemma}
\label{lem:bimodloctriv}
Suppose $A,B$ are vector bundles over $X$ with connections and connection-preserving algebra structures, and suppose $M$ is a vector bundle over $X$ with connection and a connection-preserving $A$-$B$-bimodule structure $(\lambda,\rho)$. Then, $(\lambda,\rho)$ is local. Moreover, the restriction of $M$ to each path-connected component is a bundle of $A$-$B$-bimodules.
\end{lemma}

\begin{proof}
The proof is analogous to \cref{lem:algloctriv} and left out for brevity.
\end{proof}

\begin{lemma}
\label{re:bimod:a}
Suppose $A$ and $B$ are bundles of simple algebras over $X$, and  $M$ is a  vector bundle with a  local  $A$-$B$-bimodule structure $(\lambda,\rho)$. Suppose further that $M$ is  faithfully balanced, i.e., the induced maps $\tilde \lambda: A \to \mathrm{End}_B(M)$ and $\tilde\rho: B \to \mathrm{End}_A(M)^{op}$ are fiber-wise isomorphisms. Then, $M$ is a  bundle of $A$-$B$-bimodules. 
\end{lemma} 

\begin{proof}
Faithfully balanced bimodules establish  Morita equivalences. The claim  then follows  from the statement that Morita equivalent simple algebras are Morita equivalent in a unique way (up to isomorphism). Indeed, every Morita equivalence has to be irreducible (as an $A\otimes B^{op}$-module), because if it was a direct sum of two bimodules, it would not   be invertible. However, if $A$ and $B$ are simple, then $A \otimes B^{op}$ is again simple, so it has a unique irreducible module. 
\end{proof}

Finally, we discuss the composition of  bimodule bundles, i.e. their tensor product over an algebra bundle. This tends to be difficult and is treated in \cite{Kristel2022}. 
Let $A,B,C$ be vector bundles over $X$ with algebra structures, let $M$ be a vector bundle with   $A$-$B$-bimodule structure $(\lambda_M,\rho_M)$, and let $N$ be a vector bundle with   $B$-$C$-bimodule structure $(\lambda_N,\rho_N)$. Over each point $x\in X$ we consider the subspace $K_x \subset M|_x \otimes N|_x$ generated by elements of the form
\begin{equation*}
\rho_M|_x(m \otimes b) \otimes n - m \otimes \lambda_N|_x(b \otimes n)
\end{equation*}
for all $m\in M|_x$, $n\in N|_x$ and $b\in B|_x$. We consider the disjoint union of the quotient spaces $(M|_x \otimes N|_x)/K_x$ for all $x\in X$ and denote it by $M \otimes_B N$. It will be equipped with the obvious left $A$-action $\lambda_M$ and the right $C$-action $\rho_N$,  and be equipped with the unique diffeology making the projection $M \otimes N \to M \otimes_B N$ a subduction. 
In general, $M \otimes_B N$ will not even be a vector bundle. Under assumptions of locality and semisimplicity, however, we have the following result \cite[Thm. 4.2.6 \& Cor. 3.1.12]{Kristel2022}: 

\begin{lemma}
If the algebra structures on $A,B,C$ are  local and semisimple, and the bimodule structures on $M$ and $N$ are local, then $M \otimes_B N$ is a vector bundle, and $(\lambda_M,\rho_N)$ is a local $A$-$C$-bimodule structure.
\end{lemma}

\begin{remark}
\label{re:bimodulebundle:connections}
Consider again  bundles  of algebras $A,B,C$ and bundles $M$  of $A$-$B$-bimodules and $N$ of $B$-$C$-bimodules. We assume that all bundles are equipped with connections, in such a way that the algebra structures and the bimodule structures are connection-preserving. Then, there is a naturally defined   connection on $M \otimes_B N$, for which the $A$-$C$-bimodule structure is connection-preserving. In order to see this, we only have to observe that the sub-vector bundle $U$ is invariant under  parallel transport.  
\end{remark}

\section{Fiber integration for smooth 2-functors}

\label{sec:fibreintegration}

\setsecnumdepth{1}

In this section we provide a result for the theory of smooth functors of \cite{schreiber3,schreiber5}. It is used in \cref{lem:632} as one step to establish our main theorem, but also might be interesting in other contexts. Let $X$ be a smooth manifold. We recall that there is a bijection
\begin{equation}
\label{eq:fibint:1}
\fun^{\infty}(\mathcal{P}_1(X),B\ueins) \cong \Omega^1(X)\text{,}
\end{equation}
where the left hand side consists of smooth functors defined on the smooth path groupoid of $X$  with values in the Lie groupoid $B\ueins$ (it has a single object and $\ueins$ as the manifold of morphisms) \cite[Prop. 4.7]{schreiber3}. Analogously, there is a bijection
\begin{equation}
\label{eq:fibint:2}
\fun^{\infty}(\mathcal{P}_2(X),BB\ueins) \cong \Omega^2(X)\text{,}
\end{equation}
where the left hand side consists of smooth 2-functors defined on the smooth path 2-groupoid of $X$ with values in the Lie 2-groupoid $BB\ueins$ (one object, one 1-morphism, and $\ueins$ as the manifold of 2-morphisms) \cite[Theorem 2.21]{schreiber5}. We define  a ,,fiber integration map``  for smooth 2-functors, i.e. a map
\begin{equation}
\label{eq:fibint}
\int_{[0,1]} : \fun^{\infty}(\mathcal{P}_2([0,1] \times X),BB\ueins) \to \fun^{\infty}(\mathcal{P}_1(X),B\ueins)\text{,}
\end{equation}
and prove the following result.

\begin{proposition}
The fiber integration map  \cref{eq:fibint} corresponds under above bijections to the ordinary fiber integration of differential forms, i.e.
 the diagram
\begin{equation}
\label{eq:fibintdiagram}
\small
\alxydim{@C=4em@R=3em}{\fun^{\infty}(\mathcal{P}_2([0,1] \times X),BB\ueins) \ar[r]^-{\int_{[0,1]}}\ar[d] & \fun^{\infty}(\mathcal{P}_1(X),B\ueins) \ar[d] \\ \Omega^2([0,1] \times X) \ar[r]_{\int_{[0,1]}} & \Omega^1(X)}
\end{equation}
is commutative. 
\end{proposition}

The fiber integration map \cref{eq:fibint} is defined as follows. Suppose $F$ is a smooth 2-functor on $[0,1] \times X$, and $\gamma\in PX$ is a path in $X$. Then,
\begin{equation}
\label{eq:deffibint}
\left ( \int_{[0,1]} F  \right )(\gamma) := F(\Sigma_{\gamma})
\quith
\Sigma_{\gamma} := (\id \times \gamma)_{*}(\Sigma_{1,1})\text{.}
\end{equation}
Here, $\Sigma_{s,t}$ is the \emph{standard bigon} in $[0,1]^2$, and $\id \times \gamma: [0,1]^2 \to [0,1] \times X$ pushes it to a bigon in $[0,1] \times X$. The standard bigon $\Sigma_{s,t}$ is the uniquely defined bigon that fills the rectangle spanned by $(0,0)$ and $(s,t)$, see \cite[Sec. 2.2.1]{schreiber5}. We note that $\Sigma_{\gamma}$ is a bigon between the path $(\gamma_1,\cp_{\gamma(1)}) \pcomp (\cp_0 ,\gamma)$ and the path $(\cp_1,\gamma) \pcomp (\gamma_1,\cp_{\gamma(0)})$, where $\gamma_1$ is the standard path in $[0,1]$.

\begin{lemma}
\Cref{eq:deffibint} defines a smooth functor $\displaystyle \int_{[0,1]} F: \mathcal{P}_1(X) \to B\ueins$.
\end{lemma}

\begin{proof}
We have to check that $\smash{\int_{[0,1]} F}$ respects the composition, that it only depends on the thin homotopy class of $\gamma$, and that it is smooth. The latter follows directly from the smoothness of $F$. If $h:[0,1]^2 \to X$ is a thin homotopy between $\gamma$ and $\gamma'$, then it induces a homotopy
\begin{equation*}
[0,1]^3 \to [0,1] \times X: (r,s,t) \mapsto (\id \times h(r,-))(\Sigma_{1,1}(s,t))
\end{equation*}
between $(\id \times \gamma)_{*}(\Sigma_{1,1})$  and $(\id \times \gamma')_{*}(\Sigma_{1,1})$. This homotopy has rank  two, because $h$ has rank one, and thus the product $\id \times h$ has rank two. By inspection one checks that it restricts to a thin homotopy between the boundary paths.
This shows that $\Sigma_{\gamma}$ and $\Sigma_{\gamma'}$ are rank-two-homotopic, which implies $F(\Sigma_{\gamma})=F(\Sigma_{\gamma'})$. Finally, if  $\gamma$ and $\gamma'$ are composable paths in $X$, then 
\begin{equation*}
\Sigma_{\gamma' \pcomp \gamma} = (\Sigma_{\gamma'}\circ \cp_{\gamma}) \bullet  (\cp_{\gamma'}\circ \Sigma_{\gamma})\text{,}
\end{equation*}
where $\cp_{\gamma}$ is the constant bigon for  $\gamma$, and $\bullet$ and $\circ$ denote the vertical and horizontal composition of 2-morphisms in $\mathcal{P}_2(X)$, respectively.
Since $F$ is a 2-functor and $BB\ueins$-valued, this gives
$F(\Sigma_{\gamma' \pcomp \gamma}) = F(\Sigma_{\gamma'})\cdot F(\Sigma_{\gamma})$.
\end{proof}

Now we prove the commutativity of the diagram \cref{eq:fibintdiagram}.  Consider a point $x\in X$ and a tangent vector $v\in T_xX$ represented by a curve $\gamma:\R \to X$ with $\gamma(0)=x$. Using the description of the bijections \cref{eq:fibint:1,eq:fibint:2}, and starting with a 2-functor $F$, following  diagram \cref{eq:fibintdiagram} clockwise yields
\begin{equation*}
\left . -\frac{\mathrm{d}}{\mathrm{d}t} \right|_0 \left ( \int_{[0,1]} F  \right )(\gamma_{*}(\gamma_t))=\left . \frac{\mathrm{d}}{\mathrm{d}t} \right|_0 F(\Sigma_{\gamma_{*}(\gamma_t)})\text{,}
\end{equation*}
where $\gamma_t$ is the standard path in $\R$ from $0$ to $t$.
Counter-clockwise, we have
\begin{equation*}
-\int_{0}^1 \mathrm{d}r \; \left. \frac{\mathrm{d}^2}{\mathrm{d}s\mathrm{d}t}\right|_{0,0} F(\Gamma_{r}(\Sigma_{s,t}))\text{,}
\end{equation*}
where $\Gamma_r: \R^2 \to [0,1] \times X$ represents the tangent vectors $\partial_r\in T_r[0,1]$ and $v\in T_xX$; for instance, we can put $\Gamma_r(s,t) := (r+s,\gamma(t))$. In order to compare these two expressions, we consider the map 
\begin{equation*}
\tilde \Gamma_r: \R^2 \to \R^2: (s,t) \to (r+s,t)\text{.}
\end{equation*}
Then, we have
\begin{equation*}
(\tilde \Gamma_r(\Sigma_{s,t}) \circ \cp) \bullet (\cp \circ  \Sigma_{r,t}) = \Sigma_{s+r,t}
\end{equation*}
as bigons in $\R^2$. Further, $\Gamma_r = (\id \times \gamma) \circ \tilde\Gamma_r$, and thus
\begin{equation*}
(\Gamma_r(\Sigma_{s,t}) \circ\cp) \bullet (\id \times \gamma)(\cp \circ  \Sigma_{r,t}) = (\id \times \gamma)(\Sigma_{s+r,t})\text{.}
\end{equation*}
We apply $F$ and obtain
\begin{equation*}
F(\Gamma_r(\Sigma_{s,t}))\cdot F((\id \times \gamma)(\Sigma_{r,t})) = F((\id \times \gamma)(\Sigma_{s+r,t}))\text{.}
\end{equation*}
Differentiating, we get
\begin{align*}
\left . \frac{\mathrm{d}}{\mathrm{d}s}\right|_{s=0} F(\Gamma_r(\Sigma_{s,t})) =  \left .\frac{\mathrm{d}}{\mathrm{d}s} \right|_{s=r} F((\id \times \gamma)(\Sigma_{s,t})) \cdot F((\id \times \gamma)(\Sigma_{r,t}))^{-1}\text{.}
\end{align*}
This is the pullback of the Maurer-Cartan form on $\ueins$ along $r \mapsto F((\id \times \gamma)(\Sigma_{r,t}))$; hence we have
\begin{equation*}
\exp\left ( \int_{0}^1 \mathrm{d}r \left . \frac{\mathrm{d}}{\mathrm{d}s}\right|_0 F(\Gamma_r(\Sigma_{s,t}))   \right ) = F((\id \times \gamma)(\Sigma_{1,t}))\text{.} \end{equation*}
Taking the derivative with respect to $t$ and evaluating at $t=0$ gives
\begin{equation*}
\int_{0}^1 \mathrm{d}r \left . \frac{\mathrm{d}^2}{\mathrm{d}s\mathrm{d}t}\right|_{0,0} F(\Gamma_r(\Sigma_{s,t})) = \left . \frac{\mathrm{d}}{\mathrm{d}t} \right|_0 F((\id \times \gamma)(\Sigma_{1,t}))\text{.}
\end{equation*}
It remains to observe that the bigons $\Sigma_{\gamma_{*}(\gamma_t)}=(\id \times \gamma_{*}(\gamma_t))(\Sigma_{1,1})$ and $(\id \times \gamma)(\Sigma_{1,t})$  are thin homotopic, which is straightforward to see. 

\section{Pullback of bundle gerbes along rank-one maps}

\setsecnumdepth{1}

We provide the following,  general result about bundle gerbes and bundle gerbe morphisms. We need it in the proof of \cref{lem:connection}.

\begin{theorem}
\label{th:rankonepullback}
Let $M$ be a smooth manifold,  $X$ be a compact smooth manifold, and  $\phi:X \to M$ be a smooth map of rank at most one. 
\begin{enumerate}[(a)]

\item 
\label{th:rankonepullback:a}
If $\mathcal{G}$ is a bundle gerbe with connection over $M$, then its pullback along $\phi$ admits a parallel  trivialization, i.e. a 1-isomorphism $\mathcal{T}: \phi^{*}\mathcal{G} \to \mathcal{I}_0$ in $\ugrbcon X$. 

\item
\label{th:rankonepullback:b}
If $\mathcal{A}:\mathcal{G} \to \mathcal{G}'$ is a 1-morphism in $\ugrbcon M$, and $\mathcal{T}: \phi^{*}\mathcal{G}\to \mathcal{I}_0$ and $\mathcal{T}': \phi^{*}\mathcal{G}'\to \mathcal{I}_0$ are parallel trivializations, then there exists a hermitian vector bundle $E$ with flat connection, and a 2-isomorphism in $\ugrbcon X$:
\begin{equation*}
\small
\alxydim{@R=3em}{\phi^{*}\mathcal{G} \ar[r]^{\phi^{*}\mathcal{A}} \ar[d]_{\mathcal{T}} & \phi^{*}\mathcal{G}' \ar@{=>}[dl] \ar[d]^{\mathcal{T}'} \\ \mathcal{I}_0 \ar[r]_{E} & \mathcal{I}_0}
\end{equation*}

\end{enumerate}
\end{theorem}

Before proving \cref{th:rankonepullback}, we shall point out the following corollary, which appeared already as \cite[Prop. 3.3.1]{waldorf10}.
It follows from \cref{th:rankonepullback:a}  and the definition of surface holonomy as the exponential of the integral over $X$ of the 2-form $\rho$ of any trivialization $\mathcal{T}:\phi^{*}\mathcal{G} \to \mathcal{I}_{\rho}$. 

\begin{corollary}
If $X$ is an oriented closed surface and $\mathcal{G}$ is a bundle gerbe with connection over $M$, then the surface holonomy of $\mathcal{G}$ around $\phi$ is trivial. 
\qed
\end{corollary}

In the remainder of this section we proof \cref{th:rankonepullback}.
We work on the level of cocycle data, with respect to an open cover $\{U_{\alpha}\}_{\alpha\in A}$ of $M$. A bundle gerbe $\mathcal{G}$ with connection is given by a triple $(B,A,g)$, where $g_{\alpha\beta\gamma}:U_{\alpha}\cap U_{\beta}\cap U_{\gamma} \to \ueins$, $A_{\alpha\beta}\in \Omega^1(U_{\alpha} \cap U_{\beta})$, and $B_{\alpha}\in \Omega^2(U_{\alpha})$ satisfy the cocycle conditions
\begin{align*}
B_{\beta} - B_{\alpha} = \mathrm{d}A_{\alpha\beta}
\quomma
A_{\beta\gamma}-A_{\alpha\gamma}+A_{\alpha\beta} = \mathrm{dlog}(g_{\alpha\beta\gamma})
\quand
g_{\beta\gamma\delta}\,g_{\alpha\beta\delta} &= g_{\alpha\gamma\delta}\,g_{\alpha\beta\gamma}\text{.}
\end{align*}
We can assume that each cocycle is normalized in the sense that $A_{\alpha\alpha}=0$ and $g_{\alpha\beta\gamma}=1$ whenever $|\{\alpha,\beta,\gamma\}|<3$. With respect to cocycles $(B,A,g)$ and $(B',A',g')$ for bundle gerbes $\mathcal{G}$ and $\mathcal{G}'$, respectively, a 1-morphism $\mathcal{A}: \mathcal{G} \to \mathcal{G}'$ is given by a pair $(\Pi,G)$, where $G_{\alpha\beta}:U_\alpha\cap U_{\beta} \to \mathrm{U}(n)$ and $\Pi_{\alpha}\in \Omega^1(U_{\alpha},\mathfrak{u}(n))$ satisfy
\begin{align*}
B_{\alpha}'&=B_{\alpha}+ \textstyle\frac{1}{n}\mathrm{tr}(\mathrm{d}\Pi_{\alpha})
\\
\Pi_{\beta}+A_{\alpha\beta} &= A_{\alpha\beta}'+\mathrm{Ad}^{-1}_{G_{\alpha\beta}}(\Pi_{\alpha})+\mathrm{dlog}(G_{\alpha\beta})
\\
G_{\alpha\gamma} \cdot g_{\alpha\beta\gamma} &= g_{\alpha\beta\gamma}'\cdot G_{\alpha\beta}\cdot G_{\beta\gamma}\text{;}
\end{align*}
see, e.g. \cite{gawedzki4}. Here, $n$ is the rank of the vector bundle of the 1-morphism, and we regard $\ueins \subset \mathrm{U}(n)$ as the central subgroup, and similarly we view $\R \subset \mathfrak{u}(n)$ as diagonal matrices.

We construct a refinement of the open cover $\{U_{\alpha}\}_{\alpha\in A}$ with properties adapted to the map $\phi:X \to M$. 
Let $K:=\phi(X)\subset M$ be equipped with the subspace topology, so that $\{U_{\alpha} \cap K\}_{\alpha\in A}$ is an open cover of $K$.  By \cite[Theorem 2]{sard1} and  \cite[Proposition 1.3]{Church1963} the covering dimension of $K$ is  $\dim(K)\leq 2$. Thus, there exists a refinement $\{V_{\alpha}\}_{\alpha\in B}$ consisting of open sets $V_{\alpha}\subset K$ and a refinement map $r:B \to A$ with $V_{\alpha} \subset U_{r(\alpha)} \cap K$, such that all non-trivial 3-fold intersections are empty. Since $V_{\alpha}$ is open in $K$, there exists an open set $\smash{\tilde V_{\alpha}\subset M}$ such that $\smash{V_{\alpha}=\tilde V_{\alpha}\cap K}$.  Now we collect all open sets $\tilde V_{\alpha}$ and all sets $U_{\alpha} \setminus K$, which are open since $K\subset M$ is closed, as it is a compact subset of a Hausdorff space. This results in a new open cover  $\{W_{\alpha}\}_{\alpha\in C}$ of $M$ that is a refinement of the original one, and no point $x\in K$ is contained in an intersection $W_{\alpha} \cap W_{\beta}\cap W_{\gamma}$ with $\alpha,\beta,\gamma$ distinct.

Since we have a refinement, we can assume that our cocycles $(B,A,g)$, $(B',A',g')$ and $(\Pi,G)$ are defined with respect to $\{W_{\alpha}\}_{\alpha\in C}$. Now we prove \cref{th:rankonepullback:a*}. Let $\{\psi_{\alpha}\}_{\alpha\in C}$ be a smooth partition of unity subordinated to the open cover $\{W_{\alpha}\}_{\alpha\in C}$. 
We define
\begin{equation*}
\rho_{\alpha} :=  \sum_{\beta\in C} \psi_{\beta} A_{\alpha\beta}  \in \Omega^1(W_{\alpha})
\end{equation*}
and check that
\begin{equation*}
\rho_{\alpha}-\rho_{\beta} = A_{\alpha\beta} + \sum_{\gamma\in C} \psi_{\gamma}  \mathrm{dlog}(g_{\alpha\beta\gamma})\text{.}
\end{equation*}
We change our cocycle $(B,A,g)$ by the 1-forms $\rho_{\alpha}$, and obtain an equivalent cocycle $(\tilde B,\tilde A,g)$ with $\tilde B_{\alpha} = B_{\alpha} + \mathrm{d}\rho_{\alpha}$ and  $\tilde A_{\alpha\beta} = A_{\alpha\beta} - \rho_{\alpha} + \rho_{\beta}$. Then we perform the pullback along $\phi$. The fact that $\phi$ is of rank one implies $\phi^{*}\tilde B_{\alpha}=0$. Since there are no non-trivial 3-fold intersections, we have $\phi^{*}g_{\alpha\beta\gamma}=1$, and the above calculation implies $\phi^{*}\tilde A_{\alpha\beta}=0$.  Thus, the pullback results in $(0,0,1)$, which is a cocycle for $\mathcal{I}_0$. Translating between cocycle data and geometric objects, this implies the existence of the parallel trivialization $\mathcal{T}$.

It remains to prove \cref{th:rankonepullback:b*}. Under the change of local data from $(B,A,g)$ to $(\tilde B,\tilde A,g)$, and similarly for $(B',A',g')$, we obtain new local data $(\tilde\Phi,\tilde G)$ for the 1-morphism. Pulling back along $\phi$, we obtain local data $(\phi^{*}\tilde\Pi,\phi^{*}\tilde G)$ for a 1-endomorphism of $(0,0,1)$. Since  $\phi$ is of rank one, we have $\mathrm{d}(\phi^{*}\tilde\Pi)=\phi^{*}(\mathrm{d}\tilde\Pi)=0$, and the absence of non-trivial 3-fold intersections implies the usual cocycle condition for $\phi^{*}\tilde G$. Thus, $(\phi^{*}\tilde\Pi,\phi^{*}\tilde G)$ is local data for a vector bundle $E$ with flat connection over $X$.

\end{appendix}

\tocsection{Table of Notation}

\label{sec:notation}

\begin{minipage}[t]{0.55\textwidth}
\raggedright
\begin{enumerate}[align=left,labelsep=0em,leftmargin=4em,labelwidth=4em]

\item[LBG] loop space brane geometry, \cref{sec:targetspacegeometry}

\item[TBG] target space brane geometry, \cref{sec:loopspacegeometry}

\item[$PX$]  the space of smooth paths in $X$ with sitting instants

\item[$\gamma_1\pcomp\gamma_2$] denotes the concatenation of paths

\item[$\overline{\gamma}$] denotes the reversed path

\item[$\cp_x$] denotes the constant path at a point $x$

\item[$\gamma_1\cup\gamma_2$]   the loop $\overline{\gamma_2}\pcomp \gamma_1$, when $\gamma_1$ and $\gamma_2$ have a common initial point and a common end point

\item[$d_{\gamma_1,\gamma_2}$] the parallel transport of a superficial connection along a (arbitrary) thin path connecting $\gamma_1$ with $\gamma_2$ 

\item[$\Des(\mathcal{E},\mathcal{F})$] a vector bundle obtained from two twisted vector bundles, see \cref{eq:des}

\end{enumerate}
\end{minipage}
\hfill
\begin{minipage}[t]{0.4\textwidth}
\raggedright
\begin{enumerate}[align=left,labelsep=0em,leftmargin=2em,labelwidth=2em]

\item[$\overline{V}$] the complex conjugate vector space

\item[$V^{*}$] the dual vector space, $V^{*}:=\mathrm{Hom}(V,\C)$

\item[$\varphi^{*}$] the adjoint of a linear map between complex inner product spaces

\item[$S^1$] the circle, $S^1=\R/\Z$.

\item[$\iota_1$] the map $[0,1] \to S^1:t \mapsto \frac{1}{2}t$ 

\item[$\iota_2$] the map $[0,1] \to S^1:t \mapsto 1-\frac{1}{2}t$ 

\item[$i_x$]
for  $x\in X$, is the map
$[0,1] \to X \times [0,1] :t\mapsto (x,t)$

\item[$j_t$]
for $t\in [0,1]$ and some space $X$, is 

$X \to X \times [0,1] :x\mapsto (x,t)$

\item[$f^{\vee}$] for $f:X \to C^{\infty}(Y,Z)$, denotes the adjoint map $X \times Y \to Z:(x,y)\mapsto f(x)(y)$.

\end{enumerate}
\end{minipage}

\def\bibdir{F:/uni/bibliothek/tex/}

\bibliographystyle{kobib}
\bibliography{kobib}

\end{document}

 the inclusion of \quot{flat} loops $PM \to